\documentclass[letterpaper,11pt]{article}
\usepackage{amsmath}
\usepackage{amssymb}
\usepackage{amsthm}
\usepackage{amsfonts}
\usepackage{setspace}
\usepackage{bbm}
\usepackage{color}
\usepackage{enumerate}
\usepackage{mathrsfs}
\usepackage{titling}
\usepackage{graphicx}
\usepackage{subcaption}
\usepackage{hyperref}
\usepackage{enumitem}
\usepackage[usenames,dvipsnames,svgnames,table]{xcolor}
\usepackage{multirow}
\usepackage[margin=0.90in]{geometry}
\usepackage{enumitem}
\setlist{nosep}

\newtheorem{theorem}{Theorem}[section]
\newtheorem{lemma}{Lemma}[section]

\newtheorem{assumption}{Assumption}

\newtheorem{corollary}{Corollary}[section]

\newtheorem{remark}{Remark}[section]
\newtheoremstyle{theoremd}% name of the style to be used
  {\topsep}% measure of space to leave above the theorem. E.g.: 3pt
  {\topsep}% measure of space to leave below the theorem. E.g.: 3pt
  {\itshape}% name of font to use in the body of the theorem
  {0pt}% measure of space to indent
  {\bfseries}% name of head font
  {.\!\!$^{\,\boldsymbol\prime}$}% punctuation between head and body
  { }% space after theorem head; " " = normal interword space
  {\thmname{#1}\thmnumber{ #2}\thmnote{ (#3)}}

\theoremstyle{theoremd}
\newtheorem{aprime}{Assumption}

\newcommand{\mb}[1]{\mathbb{#1}}
\newcommand{\wh}[1]{\widehat{#1}}
\newcommand{\wt}[1]{\widetilde{#1}}
\newcommand{\mf}[1]{\mathsf{#1}}
\newcommand{\mr}[1]{\mathrm{#1}}

\newcommand{\mc}[1]{\mathcal{#1}}

\newcommand{\ind}{1\!\mathrm{l}}
\newcommand{\ol}[1]{\overline{#1}}
\newcommand{\ul}[1]{\underline{#1}}

\hypersetup{
     colorlinks   = true,
     citecolor    = Blue,
     urlcolor     = Black,
     linkcolor    = Blue
}

\makeatletter
\renewcommand\paragraph{\@startsection{paragraph}{4}{\z@}%
                                    {0pt \@plus1ex \@minus.2ex}%
                                    {-1em}%
                                    {\normalfont\normalsize\bfseries}}
\makeatother

\usepackage{bibunits}
\usepackage[longnamesfirst]{natbib}

\begin{document}

\defaultbibliography{supgen-2}
\defaultbibliographystyle{chicago}
\begin{bibunit}

\author{%
Xiaohong Chen\thanks{%
Cowles Foundation for Research in Economics, Yale
University, Box 208281, New Haven, CT 06520, USA. E-mail address: \texttt{xiaohong.chen@yale.edu}}
\and
Timothy M. Christensen\thanks{%
Department of Economics, New York University, 19 W. 4th Street, 6th floor, New York, NY 10012, USA. E-mail address: \texttt{timothy.christensen@nyu.edu}}
}

\title{%
Optimal Sup-norm Rates and Uniform Inference on Nonlinear Functionals of Nonparametric IV Regression\thanks{%
This paper is a revised version of the preprint arXiv:1508:03365v1 \citep{ChenChristensen-adaptive-npiv}, which was in turn a major extension of Sections 2 and 3 of the preprint arXiv:1311.0412 \citep{ChenChristensen-npiv}. We are grateful to Y. Sun for careful proof-reading and useful comments and M. Parey for sharing the gasoline demand data set. We thank L.P. Hansen, R. Matzkin, W. Newey, J. Powell, A. Tsybakov and participants of SETA2013, AMES2013, SETA2014, the 2014 International Symposium in honor of Jerry Hausman, the 2014 Cowles Summer Conference,
the 2014 SJTU-SMU Econometrics Conference, the 2014 Cemmap Celebration Conference,
the 2015 NSF Conference - Statistics for Complex Systems, the 2015 International Workshop for Enno Mammen's 60th birthday, the 2015 World Congress of ES meetings, and seminars at various universities for comments. Support from the Cowles Foundation is
gratefully acknowledged.
}
}

\date{First version: August 2013. Revised: January 2017. }

\maketitle

\begin{abstract}
\singlespacing
\noindent
This paper makes several important contributions to the literature about nonparametric instrumental variables (NPIV) estimation and inference on a structural function $h_0$ and its functionals. First, we derive sup-norm convergence rates for computationally simple sieve NPIV (series 2SLS) estimators of $h_0$ and its derivatives. Second, we derive a lower bound that describes the best possible (minimax) sup-norm rates of estimating $h_0$ and its derivatives, and show that the sieve NPIV estimator can attain the minimax rates when $h_0$ is approximated via a spline or wavelet sieve. Our optimal sup-norm rates surprisingly coincide with the optimal root-mean-squared rates for severely ill-posed problems, and are only a logarithmic factor slower than the optimal root-mean-squared rates for mildly ill-posed problems. Third, we use our sup-norm rates to establish the uniform Gaussian process strong approximations and the score bootstrap uniform confidence bands (UCBs) for collections of nonlinear functionals of $h_0$ under primitive conditions, allowing for mildly and severely ill-posed problems.
Fourth, as applications, we obtain the first asymptotic pointwise and uniform inference results for plug-in sieve t-statistics of exact consumer surplus (CS) and deadweight loss (DL) welfare functionals under low-level conditions when demand is estimated via sieve NPIV. Empiricists could read our real data application of UCBs for exact CS and DL functionals of gasoline demand that reveals interesting patterns and is applicable to other markets.

\medskip


\medskip


\end{abstract}

\pagenumbering{arabic}

\newpage
\section{Introduction}

Well founded empirical evaluation of economic policy is often based upon inference on nonlinear welfare functionals of nonparametric or semiparametric structural models.  This paper makes several important contributions to estimation and inference on a flexible (i.e. nonparametric) structural function $h_0$ and nonlinear functionals of $h_0$ within the framework of a nonparametric instrumental variables (NPIV) model:
\begin{equation} \label{npreg}
 Y_i = h_0(X_i) + u_i~\quad\quad E[u_i| W_i]=0~
\end{equation}
where $h_0$ is an unknown function, $X_i$ is a vector of continuous \emph{endogenous} regressors, $W_i$ is a vector of (conditional) instrumental variables, and the conditional distribution of $X_i$ given $W_i$ is unspecified.

Given a random sample $\{(Y_i, X_i, W_i)\}_{i=1}^n$ (of size $n$) from the NPIV model (\ref{npreg}), our first two main theoretical results address how well one may estimate $h_0$ and its derivatives simultaneously in sup-norm loss, i.e. we bound
\[
 \sup_x \Big|\wh h(x) - h_0(x)\Big| \quad \quad \mbox{and} \quad \quad \sup_x \Big|\partial^k \wh h(x) - \partial^k h_0(x)\Big|
\]
for estimators $\wh h$ of $h_0$, where $\partial^k h(x)$ denotes $k$-th partial derivatives of $h$ with respect to components of $x$.
We first provide upper bounds on sup-norm convergence rates for the computationally simple sieve NPIV (i.e., series two stage least squares (2SLS)) estimators \citep{NeweyPowell,AiChen2003,BCK}. We then derive a lower bound that describes the best possible (i.e., minimax) sup-norm convergence rates among all estimators for $h_0$ and its derivatives, and show that the sieve NPIV estimator can attain the minimax lower bound when spline or wavelet base is used to approximate $h_0$.\footnote{The optimal sup-norm rates for estimating $h_0$ were in the first version \citep{ChenChristensen-npiv}. The optimal sup-norm rates for estimating derivatives of $h_0$ were in the second version \citep{ChenChristensen-adaptive-npiv}.}
Next, we apply our sup-norm rate results to establish the uniform Gaussian process strong approximation and the validity of score bootstrap uniform confidence bands (UCBs) for collections of possibly nonlinear functionals of $h_0$ under primitive conditions.\footnote{The uniform strong approximation and the score bootstrap UCBs results were in the second version \citep{ChenChristensen-adaptive-npiv}; see Theorem B.1 and its proof in that version.}
This includes valid score bootstrap UCBs for $h_0$ and its derivatives as special cases. Finally, as important applications, we establish first pointwise and uniform inference results for two leading nonlinear welfare functionals of a nonparametric demand function $h_0$ estimated via sieve NPIV, namely the exact consumer surplus (CS) and deadweight loss (DL) arising from price changes at different income levels when prices (and possibly income) are endogenous.\footnote{The pointwise inference results on exact CS and DL were in the second version \citep{ChenChristensen-adaptive-npiv}.} We present two real data applications to illustrate the easy implementation and usefulness of the score bootstrap UCBs based on sieve NPIV estimators. The first is to nonparametric exact CS and DL functionals of gasoline demand and the second is to nonparametric Engel curves and their derivatives. The UCBs reveal new interesting and sensible patterns in both data applications. We note that the score bootstrap UCBs for exact CS and DL nonlinear functionals are new to the literature even when the prices might be exogenous. Empiricists could jump to Section \ref{s-example} to read the sieve score bootstrap UCBs procedure and these real data applications without the need to read the rest of more theoretical sections.

Regardless of whether the regressor $X_i$ is endogenous or not, sup-norm convergence rates provide sharper measures of how well $h_0$ and its derivatives can be estimated nonparametrically than the usual $L^2$-norm (i.e., root-mean-squared) rates. This is also why, in the existing literature on nonparametric models without endogeneity, consistent specification tests in sup-norm (i.e., Kolmogorov-Smirnov type statistics) are widely used. Further, sup-norm rates are particularly useful for controlling nonlinearity bias when conducting inference on \emph{highly nonlinear} (i.e., beyond quadratic) functionals of $h_0$. In addition to being useful in constructing pointwise and uniform confidence bands for nonlinear functionals of $h_0$ via plug-in estimators, the sup-norm rates for estimating $h_0$ are also useful in semiparametric two-step procedures when $h_0$ enters the second-stage moment conditions (equalities or inequalities) nonlinearly.

Despite the usefulness of sup-norm convergence rates in nonparametric estimation and inference, as yet
there are no published results on optimal sup-norm convergence rates for estimating $h_0$ or its derivatives in the NPIV model (\ref{npreg}).
This is because, unlike nonparametric least squares (LS) regression (i.e. estimation of $h_0 (x) = E[Y_i |X_i =x]$ when $X_i$ is exogenous), estimation of $h_0$ in the NPIV model (\ref{npreg}) is a difficult ill-posed inverse problem with an unknown operator \citep{NeweyPowell,CFR2007}. Intuitively, $h_0$ in model (\ref{npreg}) is identified by the integral equation
\[
E[Y_i | W_i =w] =Th_0 (w) :=\int h_0 (x) f_{X|W}(x|w)\, \mr dx
\]
where $T$ must be inverted to obtain $h_0$. Since integration smoothes out features of $h_0$, a small error in estimating $E[Y_i | W_i =w]$ using the data $\{(Y_i, X_i, W_i)\}_{i=1}^n$ may lead to a large error in estimating $h_0$. In addition, the conditional density $f_{X|W}$ and hence the operator $T$ is generally unknown, so $T$ must be also estimated from the data. Due to the difficult ill-posed inverse nature, even the $L^2$-norm convergence rates for estimating $h_0$ in model (\ref{npreg}) have not been established until recently.\footnote{See e.g., \cite{HallHorowitz,BCK,ChenReiss,DFFR,Horowitz2011,ChenPouzo2012,GagliardiniScaillet,FlorensSimoni,Kato2013} and references therein.} In particular, \cite{HallHorowitz} derived minimax $L^2$-norm convergence rates for mildly ill-posed NPIV models and showed that their estimators can attain the optimal $L^2$-norm rates for $h_0$. \cite{ChenReiss} derived minimax $L^2$-norm convergence rates for mildly and severely ill-posed NPIV models and showed that sieve NPIV estimators can attain the optimal rates.\footnote{Appendix \ref{s-l2} extends the results in \cite{ChenReiss} to $L^2$-norm optimality for estimating  derivatives of $h_0$.} Moreover, it is generally much harder to obtain optimal nonparametric convergence rates in sup-norm than in $L^2$-norm.\footnote{Even for the simple nonparametric LS regression of $h_0$ (without endogeneity), the optimal sup-norm rates for series LS estimators of $h_0$ were not obtained till recently in \cite{CattaneoFarrell} for locally partitioning series LS, \cite{BCCK2014} for spline LS and \cite{ChenChristensen-reg} for wavelet LS.}

In this paper, we derive the best possible (i.e., minimax) sup-norm convergence rates of any estimator of $h_0$ and its derivatives in mildly and severely ill-posed NPIV models. Surprisingly, the optimal sup-norm convergence rates for estimating $h_0$ and its derivatives coincide with the optimal $L^2$-norm rates for severely ill-posed problems and are only a power of $\log n$ slower than optimal $L^2$-norm rates for mildly ill-posed problems. We also obtain sup-norm convergence rates for sieve NPIV estimators of $h_0$ and its derivatives. We show that a sieve NPIV estimator using a spline or wavelet basis to approximate $h_0$ can attain the minimax sup-norm rates for estimating both $h_0$ and its derivatives. When specializing to series LS regression (without endogeneity), our results automatically imply that spline and wavelet series LS estimators will also achieve the optimal sup-norm rates of \cite{Stone1982} for estimating the derivatives of a nonparametric LS regression function, which strengthen the recent sup-norm optimality results in \cite{BCCK2014} and \cite{ChenChristensen-reg} for estimating regression function $h_0$ itself. We focus on the sieve NPIV estimator because it has been used in empirical work, can be implemented as easily as 2SLS, and can reduce to simple series LS when the regressor $X_i$ is exogenous. Moreover, both $h_0$ and its derivatives may be simultaneously estimated at their respectively optimal convergence rates via a sieve NPIV estimator when the same sieve dimension is used to approximate $h_0$.
This is a desirable property to practitioners. In addition, the sieve NPIV estimator for $h_0$ in model (\ref{npreg}) and our proof of its sup-norm rates could be easily extended to estimating unknown
functions in other semiparametric models with nonparametric
endogeneity, such as a system of shape-invariant Engel curve IV regression model \citep{BCK}.

We provide two important applications of our results on sup-norm convergence rates in details; both are about inferences on nonlinear functionals of $h_0$ based on plug-in sieve NPIV estimators; see Section \ref{s-end} for discussions of additional applications.
Inference on highly nonlinear (i.e., beyond quadratic) functionals of $h_0$ in a NPIV model is very difficult because of the combined effects of nonlinearity bias and the slow convergence rates (in sup-norm and $L^2$-norm) of any estimators of $h_0$. Indeed, our minimax rate results show that \emph{any} estimator of $h_0$ in an ill-posed NPIV model must necessarily converge slower than their nonparametric LS counterpart. For example, the optimal sup- and $L^2$-norm rates for estimating $h_0$ in a severely ill-posed NPIV model is $(\log n )^{-\gamma}$ for some $\gamma>0$. It is well-known that a plug-in series LS estimate of a weighted quadratic functional could be root-$n$ consistent. But, a plug-in sieve NPIV estimate of a weighted quadratic functional of $h_0$ in a severely ill-posed NPIV model fails to be root-$n$ consistent \citep{ChenPouzo2014}. In fact, we establish the minimax convergence rate of any estimators of a simple weighted quadratic functional of $h_0$ in a severely ill-posed NPIV model is as slow as $(\log n )^{-a}$ for some $a>0$ (see Appendix \ref{s-quad}).

In the first application, we extend the seminal work of \cite{HausmanNewey1995} about pointwise inference on \emph{exact} CS and DL functionals of nonparametric demand without endogeneity to allow for prices, and possibly incomes, to be endogenous. According to \cite{Hausman1981} and \cite{HausmanNewey1995,HausmanNewey2016,HausmanNewey2017},  exact CS and DL functionals are the most widely used welfare and economic efficiency measures. Exact CS is a leading example of a complicated nonlinear functional of $h_0$, which is defined as the solution to a differential equation involving a demand function \citep{Hausman1981}. \cite{HausmanNewey1995} were the first to establish the pointwise asymptotic normality of plug-in kernel estimators of exact CS and DL functionals of a nonparametric demand without endogeneity.
\cite{Vanhems2010} was the first to estimate exact CS via the plug-in \cite{HallHorowitz} kernel NPIV estimator of $h_{0}$ when price is endogenous, and derived its convergence rate in $L^2$-norm for the mildly ill-posed case, but did not establish any inference results (such as the pointwise asymptotic normality).  Our paper is the first to provide low-level sufficient conditions to establish inference results for
plug-in (spline and wavelet) sieve NPIV estimators of exact CS and DL functionals, allowing for both mildly and severely ill-posed NPIV models.
Precisely, we use our sup-norm convergence rates for sieve NPIV estimators of $h_0$ and its derivatives to locally linearize plug-in estimators of exact CS and DL, which then leads to asymptotic normality of sieve $t$-statistics for exact CS and DL under primitive sufficient conditions. We also establish the asymptotic normality of plug-in sieve NPIV $t$-statistic for an \emph{approximate} CS functional, extending \cite{Newey1997}'s result from nonparametric exogenous demand to endogenous demand. Recently, \cite{ChenPouzo2014} presented a set of high-level conditions for the pointwise asymptotic normality of sieve $t$-statistics of possibly nonlinear functionals of $h_0$ in a general class of nonparametric conditional moment restriction models (including the NPIV model as a special case). They verified their high-level conditions for pointwise asymptotic normality of sieve $t$-statistics for linear and quadratic functionals. But, without sup-norm convergence rate result, \cite{ChenPouzo2014} were unable to provide low-level sufficient conditions for pointwise asymptotic normality of plug-in sieve NPIV estimators for complicated nonlinear (beyond quadratic) functionals such as the exact CS functional. This was actually the original motivation for us to derive sup-norm convergence rates for sieve NPIV estimators of $h_0$ and its derivatives.

In the second important application of our sup-norm rate results, we establish the uniform Gaussian process strong approximation and the validity of score bootstrap uniform confidence bands (UCBs) for collections of possibly nonlinear functionals of $h_0$, under primitive sufficient conditions that allow for mildly and severely ill-posed NPIV models. The low-level sufficient conditions for Gaussian process strong approximation and UCBs are applied to complicated nonlinear functionals such as collections of exact CS and DL functionals of nonparametric demand with endogenous price (and possibly income). When specializing to collections of linear functionals of the NPIV function $h_0$, our Gaussian process strong approximation and sieve score bootstrap UCBs for $h_0$ and its derivatives are valid under mild sufficient conditions. In particular, for a NPIV model with a scalar endogenous regressor, our sufficient conditions are comparable to those in \cite{HorowitzLee2012} for their notion of UCBs with a growing number of grid points by interpolation for $h_0$ estimated via the modified orthogonal series NPIV estimator of \cite{Horowitz2011}.
When specializing to a nonparametric LS regression (with exogenous $X_i$), our results on the Gaussian strong approximation and score bootstrap UCBs for collections of nonlinear functionals of $h_0$, such as exact CS and DL functionals, are still new to the literature and complement the important results in \cite{CLR} for $h_0$ and \cite{BCCK2014} for linear functionals of $h_0$ estimated via series LS.

Our sieve score bootstrap UCBs procedure is extremely easy to implement since it computes the sieve NPIV estimator only once using the data, and then perturbs the sieve score statistics by random weights that are mean zero and independent of the data. So it should be very useful to empirical researchers who conduct nonparametric estimation and inference on structural functions with endogeneity in diverse subfields of applied economics, such as consumer theory, IO, labor economics, public finance, health economics, development and trade, to name only a few. Two real data illustrations are presented in Section \ref{s-example}. In the first, we construct UCBs for exact CS and DL welfare functionals for a range of gasoline taxes at different income levels. For this illustration, we use the same data set as in \cite{BlundellHorowitzParey2012,BlundellHorowitzParey2013} and estimate household gasoline demand via spline sieve NPIV (other data sets and other goods could be used). Despite the slow convergence rates of NPIV estimators, the UCBs for exact CS are particularly informative.
In the second empirical illustration, we use the same data set as in \cite{BCK} to estimate Engel curves for households with kids via a spline sieve NPIV and construct UCBs for Engel curves and their derivatives for various categories of household expenditure.

The rest of the paper is organized as follows. Section \ref{s-example} presents the sieve NPIV estimator, the score bootstrap UCBs procedure and two real-data applications. This section aims at empirical researchers. Section \ref{npiv sec} establishes the minimax optimal sup-norm rates for estimating a NPIV function $h_0$ and its derivatives. Section \ref{s-ucb} presents low-level sufficient conditions for the uniform Gaussian process strong approximation and sieve score bootstrap UCBs for collections of general nonlinear functionals of a NPIV function. Section \ref{s-inference} deals with pointwise and uniform inferences on exact CS and DL, and approximate CS functionals in nonparametric demand estimation with endogeneity. Section \ref{s-end} concludes with discussions of additional applications of the sup-norm rates of sieve NPIV estimators. Appendix \ref{s-sup-lemmas} contains additional results on sup-norm convergence rates. Appendix \ref{s-l2} presents optimal $L^2$-norm rates for estimating derivatives of a NPIV function under extremely weak conditions. Appendix \ref{s-quad} establishes the minimax lower bounds for estimating quadratic functionals of a NPIV function. The main online supplementary appendix contains pointwise normality of sieve $t$ statistics for nonlinear functionals of NPIV under lower-level sufficient conditions than those in \cite{ChenPouzo2014} (Appendix \ref{s-pw});
background material on B-spline and wavelet sieves (Appendix \ref{ax-basis}); and useful lemmas on random matrices (Appendix \ref{ax-supp}). The secondary online appendix contains additional lemmas and all of the proofs (Appendix \ref{ax-proofs}).

\section{Estimator and motivating applications to UCBs}\label{s-example}

This section describes the sieve NPIV estimator and a score bootstrap UCBs procedure for collections of functionals of the NPIV function. It mentions intuitively why sup-norm convergence rates of a sieve NPIV estimator are needed to formally justify the validity of the computationally simple score bootstrap UCBs procedure. It then present two real data applications of uniform inferences on functionals of a NPIV function: UCBs for exact CS and DL functionals of nonparametric demand with endogenous price, and UCBs for nonparametric Engel curves and their derivatives when the total expenditure is endogenous. This section is presented to practitioners.

\textbf{Sieve NPIV estimators}. Let $\{(Y_i,X_i,W_i)\}_{i=1}^n$ denote a random sample from the NPIV model (\ref{npreg}). The sieve NPIV estimator $\wh h$ of $h_0$ is simply the 2SLS estimator applied to some basis functions of $X_i$ (the endogenous regressors) and $W_i$ (the conditioning variables), namely
\begin{equation} \label{def-series-2SLS}
 \wh h(x) = \psi^J(x)'\wh c~~\mbox{ with }~~ \wh c = [\Psi'B(B'B)^-B'\Psi]^- \Psi'B (B'B)^- B'{Y}
\end{equation}
where ${Y} = (Y_1,\ldots,Y_n)'$,
\begin{align}
 \psi^J(x) & = (\psi_{J1}(x),\ldots,\psi_{JJ}(x))' & & \Psi = (\psi^J(X_1),\ldots,\psi^J(X_n))' \label{def-psi} \\
 b^K(w) & = (b_{K1}(w),\ldots,b_{KK}(w))'& & B  = (b^K(W_1),\ldots,b^K(W_n))' \label{def-b}
\end{align}
and $\{\psi_{J1},\ldots,\psi_{JJ}\}$ and $\{b_{K1},\ldots,b_{KK}\}$ are collections of basis functions of dimension $J$ and $K$ for approximating $h_0$ and the instrument space, respectively \citep{BCK,ChenPouzo2012,Newey2013}). The \emph{regularization} parameter $J$ is the dimension of the sieve for approximating $h_0$. The \emph{smoothing} parameter $K$ is the dimension of the instrument sieve. From the analogy with 2SLS, it is clear that we need $K \geq J$. \cite{BCK,ChenReiss,ChenPouzo2012} have previously shown that $\lim_{J} (K/J) = c \in [1,\infty)$ can lead to the optimal $L^2$-norm convergence rate for sieve NPIV estimator. Thus we assume that $K$ grows to infinity at the same rate as that of $J$, say $J \leq K \leq c J$ for some finite $c>1$ for simplicity.\footnote{Monte Carlo evidences in \citep{BCK,ChenPouzo2014} and others suggest that sieve NPIV estimators often perform better with $K > J$ than with $K = J$, and that the regularization parameter $J$ is important for finite sample performance while the parameter $K$ is not as important as long as it is larger than $J$. See our second version \citep{ChenChristensen-adaptive-npiv} for data-driven choice of $J$.} When $K=J$ and $b^K =\psi^J$ being an orthogonal series basis, the sieve NPIV estimator becomes \cite{Horowitz2011}'s modified orthogonal series NPIV estimator.
Note that the sieve NPIV estimator (\ref{def-series-2SLS}) reduces to a series LS estimator $\wh h(x) = \psi^J(x)' [\Psi'\Psi]^- \Psi'{Y}$ when $X_i =W_i$ is exogenous, $J=K$ and $\psi^J(x)=b^K(w)$ \citep{Newey1997,Huang1998}.

\subsection{Uniform confidence bands for nonlinear functionals}

One important motivating application is to \emph{uniform inference} on a collection of nonlinear functionals $\{f_t(h_0) : t \in \mc T\}$ where $\mc T$ is an index set (e.g. an interval). Uniform inference may be performed via uniform confidence bands (UCBs) that contain the function $t \mapsto f_t(h_0)$ with prescribed coverage probability. UCBs for $h_0$ (or its derivatives) are obtained as a special case with $\mc T = \mc X$ (support of $X_i$) and $f_t(h_0) = h_0(t)$ (or $f_t(h_0) = \partial^k h_0(t)$ for $k$-th derivative). We present applications below to uniform inference on exact CS and DL functionals over a range of price changes as well as UCBs for Engel curves and their derivatives.

A $100(1-\alpha)\%$ bootstrap-based UCB for $\{f_t(h_0) : t \in \mc T\}$ is constructed as
\begin{equation} \label{e-ucb}
 t \mapsto \bigg[f_t(\widehat h) - z_{1-\alpha}^* \frac{\wh\sigma (f_t)}{\sqrt n}  \; , \; f_t(\widehat h) + z_{1-\alpha}^* \frac{\wh\sigma (f_t)}{\sqrt n}  \bigg] \,.
\end{equation}
In this display $f_t(\wh h)$ is the plug-in sieve NPIV estimator of $f_t(h_0)$, $\wh\sigma^2 (f_t)$ is a \emph{sieve variance estimator} for $f_t(\wh h)$, and $z_{1-\alpha}^*$ is a bootstrap-based critical value to be defined below.

To compute the sieve variance estimator for $f_t(\widehat h)$ with $\wh h(x) = \psi^J(x)'\wh c$ given in (\ref{def-series-2SLS}), one would first compute the 2SLS covariance matrix estimator (but applied to basis functions) for $\wh c$:
\begin{equation} \label{cov-2sls}
 \wh \mho = [\wh S' \wh G_b^{-1} \wh S]^{-1} \wh S' \wh G_b^{-1} \wh \Omega \wh G_b^{-1} \wh S[\wh S' \wh G_b^{-1} \wh S]^{-1}
\end{equation}
where $\wh S = B'\Psi/n$, $\wh G_b = B'B/n$, $\wh \Omega = n^{-1} \sum_{i=1}^n \wh u_i^2 b^K(W_i) b^K(W_i)'$ and $\wh u_i = Y_i - \wh h(X_i)$. One then compute a ``delta-method'' correction term, a $J \times 1$ vector $Df_t(\wh h)[\psi^J]:=\big(Df_t(\wh h)[\psi_{J1}],\ldots,Df_t(\wh h)[\psi_{JJ}]\big)'$, by calculating $Df_t(\wh h)[v] = \lim_{\delta \rightarrow 0^{+}} [\delta^{-1} f_t(\wh h +\delta v)]$, which is the (functional directional) derivative of $f_t$ at $\wh h$ in direction $v$, for $v=\psi_{J1},\ldots,\psi_{JJ}$. The sieve variance estimator for $f_t(\widehat h)$ is then
\begin{equation} \label{e-sieve-var}
 \wh\sigma^2 (f_t) =\big( Df_t(\wh h)[\psi^J]\big)' \,\wh \mho\, \big(Df_t(\wh h)[\psi^J]\big) \,.
\end{equation}

We use the following \emph{sieve score bootstrap} procedure to calculate the critical value $z_{1-\alpha}^*$. Let $\varpi_{1},\ldots,\varpi_{n}$ be IID random variables independent of the data with mean zero, unit variance and finite 3rd moment, e.g. $N(0,1)$.\footnote{Other examples of distributions with these properties include the re-centered exponential (i.e. $\varpi_i = \mr{Exp}(1)-1$), Rademacher (i.e. $\pm 1$ each with probability $\frac{1}{2}$), or the two-point distribution of \cite{Mammen1993} (i.e. $(1-\sqrt 5)/2$ with probability $(\sqrt 5 + 1)/(2 \sqrt 5)$ and $(\sqrt 5+1)/\sqrt 2$ with remaining probability.} We define the bootstrap sieve $t$-statistic process $\{\mb Z_n^*(t) : t \in \mathcal T\}$ as
\begin{equation} \label{bscore}
 \mb Z_n^*(t) := \frac{(D f_t(\wh h)[\psi^J])' [\wh S' \wh G_b^{-1} \wh S]^{-1} \wh S' \wh G_b^{-1}}{\wh\sigma (f_t)} \left(\frac{1}{\sqrt n} \sum_{i=1}^n  b^K(W_i) \wh u_i \varpi_i \right) \quad \mbox{for each $t \in \mathcal T$} \,.
\end{equation}
To compute $z_{1-\alpha}^*$, one would calculate $\sup_{t \in \mc T} |\mb Z_n^*(t)|$ for a large number of independent draws of $\varpi_{1},\ldots,\varpi_{n}$. The critical value $z_{1-\alpha}^*$ is the $(1-\alpha)$ quantile of $\sup_{t \in \mc T} |\mb Z_n^*(t)|$ over the draws. Note that this sieve score bootstrap procedure is different from the usual nonparametric bootstrap (based on resampling the data, then recomputing the estimator): here we only compute the estimator once, and then perturb the sieve $t$-statistic process by the innovations $\varpi_1,\ldots,\varpi_n$.

An intuitive description of why sup-norm rates are very useful to justify this procedure is as follows. Under regularity conditions, the sieve $t$-statistic for an individual functional $f_t(h_0)$ admits an expansion
\begin{equation} \label{e-sieve-tstat}
 \frac{\sqrt n (f_t(\wh h) - f_t(h_0))}{\wh\sigma (f_t)} \; = \; \wh{ \mb Z}_n (t) \; + \; \mbox{nonlinear remainder term}
\end{equation}
(see equation \ref{bahadur-Z}) for the definition of $\wh{ \mb Z}_n (t)$). The term $\wh{ \mb Z}_n (t)$ is a CLT term, i.e. $\wh{ \mb Z}_n (t) \to_d N(0,1)$ for each fixed $t \in \mc T$. Therefore, the sieve $t$-statistic for $f_t(h_0)$ also converges to a $N(0,1)$ random variable provided that the ``nonlinear remainder term'' is asymptotically negligible (i.e. $o_p(1)$) (see Assumption 3.5 in \cite{ChenPouzo2014}). Our sup-norm rates are very useful for providing weak regularity conditions under which the remainder is $o_p(1)$ for fixed $t$.\footnote{\cite{ChenPouzo2014} verified their high-level Assumption 3.5 for a plug-in sieve estimator of a weighted quadratic functional example. Without sup-norm convergence rates, it is difficult to verify their Assumption 3.5 for nonlinear functionals (such as the exact CS) that are more complicated than quadratic functionals.} This justifies constructing confidence intervals for \emph{individual} functionals $f_t(h_0)$ for any fixed $t \in \mc T$ by inverting the sieve $t$-statistic (on the left-hand side of display (\ref{e-sieve-tstat})) and using $N(0,1)$ critical values. However, for \emph{uniform} inference the usual $N(0,1)$ critical values are no longer appropriate as we need to consider the sampling error in estimating the whole process $t \mapsto f_t(h_0)$. For this purpose, display (\ref{e-sieve-tstat}) is strengthened to be valid uniformly in $t \in \mc T$ (see Lemma \ref{lem-bahadur}). Under some regularity conditions, $\sup_{t\in \mc T} |\wh{\mb Z}_n(t)|$ converges in distribution to the supremum of a (non-pivotal) Gaussian process. As its critical values are generally not available, we use the sieve score bootstrap procedure to estimate its critical values.

Section \ref{s-ucb} formally justifies the use of this procedure for constructing UCBs for $\{f_t(h_0) : t \in \mc T\}$. The sup-norm rates are useful for controlling the nonlinear remainder terms for UCBs for collections of nonlinear functionals. Theorem \ref{t-dist-u} appears to be the first to establish the consistency of sieve score bootstrap UCBs for general nonlinear functionals of NPIV under low-level conditions, allowing for mildly and severely ill-posed problems. It includes as special cases the score bootstrap UCBs for nonlinear functionals of $h_0$ under exogeneity when $h_0$ is estimated via series LS, and the score bootstrap UCBs for the NPIV function $h_0$ and its derivatives.\footnote{One also needs to use sup-norm convergence rates of $\wh h$ to $h_0$ to build a valid UCB for $\{h_0 (t): t \in \mc X \}$.}
Theorem \ref{t-dist-u} is applied in Section \ref{s-inference} to formally justify the validity of score bootstrap UCBs for exact CS and DL functionals over a range of price changes when demand is estimated nonparametrically via sieve NPIV.

\subsection{Empirical application 1: UCBs for nonparametric exact CS and DL functionals}\label{s-ucb-dl}

Here we apply our methodology to study the effect of gasoline price changes on household welfare. We extend the important work by \cite{HausmanNewey1995} on pointwise confidence bands for exact CS and DL of demand without endogeneity to UCBs for exact CS and DL of demand with endogeneity.

Let demand of consumer $i$ be
\[
  \mf Q_i = h_0(\mf P_i,\mf Y_i) + \mf u_i
\]
where $\mf Q_i$ is quantity, $\mf P_i$ is price, which may be endogenous, $\mf Y_i$ is income of consumer $i$, and $\mf u_i$ is an error term.\footnote{Endogeneity may also be an issue in the estimation of static models of labor supply, in which $\mf Q_i$ represents hours worked, $\mf P_i$ is the wage, and $\mf Y_i$ is other income. In this setting it is reasonable to allow for endogeneity of both $\mf P_i$ and $\mf Y_i$ (see \cite{BlundellDuncanMeghir}, \cite{BlundellMaCurdyMeghir}, and references therein). } \cite{Hausman1981} shows that the exact CS from a price change from $\mf p^0$ to $\mf p^1$ at income level $\mf y$, denoted $\mf S_{\mf y}(\mf p^0)$, solves
\begin{equation}\label{e-cs}
 \begin{array}{rcl}
 \displaystyle \frac{\partial \mf S_{\mf y}(\mf p(u))}{\partial u} & = & -h_0\big(\mf p(u),\mf y - \mf S_{\mf y}(\mf p(u))\big)\displaystyle \frac{\mathrm d \mf p(u)}{\mathrm d u} \\[8pt]
 \mf S_{\mf y}(\mf p(1)) & = & 0 \end{array}
\end{equation}
where $\mf p :  [0,1] \to \mb R$ is a twice continuously differentiable path with $\mf p(0) = \mf p^0$ and $\mf p(1) = \mf p^1$. The corresponding DL functional $\mf D_{\mf y}(\mf p^0)$ is
\begin{equation} \label{e-dwl}
 \mf D_{\mf y}(\mf p^0) = \mf S_{\mf y}(\mf p^0) - (\mf p^1-\mf p^0)h_0(\mf p^1,\mf y)\,.
\end{equation}
As is evident from (\ref{e-cs}) and (\ref{e-dwl}), exact CS and DL are (typically nonlinear) functionals of  $h_0$. An exception is when demand is independent of income, in which case exact CS and DL are linear functionals of $h_0$. Let $t = (\mf p^0,\mf p^1,\mf y)$ index the initial price, final price, and income level and let $\mc T \subseteq [\ul {\mf p}^0 , \ol {\mf p}^0] \times [\ul {\mf p}^1 , \ol {\mf p}^1] \times [\ul{\mf y}, \ol{\mf y}]$ denote a range of price changes and/or incomes over which inference is to be performed. To denote dependence on $h_0$, we use the notation
\begin{eqnarray}
 f_{CS,t}(h) & = & \mbox{solution to (\ref{e-cs}) with $h$ in place of $h_0$} \label{e-fcs} \\
 f_{DL,t}(h) & = & f_{CS,t}(h) - (\mf p^1 - \mf p^0)h(\mf p^1,\mf y) \label{e-fdwl}
\end{eqnarray}
so $\mf S_{\mf y}(\mf p^0) = f_{CS,t}(h_0)$ and $\mf D_{\mf y}(\mf p^0) = f_{DL,t}(h_0)$.

We estimate exact CS and DL using the plug-in estimators $f_{CS,t}(\wh h)$ and $f_{DL,t}(\wh h)$. The sieve variance estimators $\wh \sigma^2 (f_{CS,t})$ and $\wh \sigma^2 (f_{DL,t})$ are as described in (\ref{e-sieve-var}) with the delta-method correction terms
\begin{eqnarray}
 Df_{CS,t}(\wh h)[\psi^J] & = & \int_{0}^{1}  \psi^J( \mf p(u),\mf y-\wh{\mf S}_{\mf y}(\mf p(u)))  e^{-\int_{0}^u \partial_2 \wh h( \mf p(v),\mf y-\wh{\mf S}_{\mf y}(\mf p(v))) \mf p'(v)\,\mathrm dv} \mf p'(u) \,\mathrm du  \label{e-Df_CS} \\
 Df_{DL,t}(\wh h)[\psi^J] & = & Df_{CS,t}(\wh h)[\psi^J] - (\mf p^1 - \mf p^0) \psi^J(\mf p^1,\mf y) \label{e-Df_DL}
\end{eqnarray}
where $\mf p'(u) = \frac{\mathrm d \mf p(u)}{\mathrm d u}$, $\partial_2 h$ denotes the partial derivative of $h$ with respect to its second argument and $\wh{\mf S}_{\mf y}(\mf p(u))$ denotes the solution to (\ref{e-cs}) with $\wh h$ in place of $h_0$.

We use the 2001 National Household Travel Survey gasoline demand data from \cite{BlundellHorowitzParey2012,BlundellHorowitzParey2013}.\footnote{We are grateful to Matthias Parey for sharing the dataset with us. We refer the reader to section 3 of \cite{BlundellHorowitzParey2012} for a detailed description of the data.} The main variables are annual household gasoline consumption (in gallons), average price (in dollars per gallon) in the county in which the household is located, household income, and distance from the Gulf coast to the capital of the state in which the household is located. Due to censoring, we consider the subset of households with incomes less than \$100,000 per year. To keep households somewhat homogeneous, we select household with incomes above \$25,000 per year (the 8th percentile), with at most 6 inhabitants, and 1 or 2 drivers. The resulting sample has size $n = 2753$.\footnote{We also exclude one household that reports 14,635 gallons; the next largest is 8089 gallons. Similar results are obtained using the full set of $n = 4811$ observations.} Table \ref{tab-gas} presents summary statistics.

\begin{table}[t]
 \centering {
\begin{tabular}{c|ccc}
\hline \hline
 & Quantity (gal) & Price (\$/gal) & Income (\$) \\ \hline
mean & 1455 & 1.33 & 58307 \\
25th \% & 871 & 1.28 & 42500 \\
median & 1269 & 1.32 & 57500 \\
75th \% & 1813 & 1.40 & 72500 \\
std dev & 894 & 0.07 & 19584 \\    \hline \hline
\end{tabular}
\parbox{5.0in}{\caption{ \label{tab-gas} \small Summary statistics for gasoline demand data. } }}
\end{table}

We estimate the household gasoline demand function in levels via sieve NPIV using distance as instrument for price. To implement the estimator, we form $\Psi_J$ by taking a tensor product of quartic B-spline bases of dimension 5 for both price and income (so $J = 25$) and $B_K$ by taking a tensor product of quartic B-spline bases of dimension 8 for distance and 5 for income (so $K = 40$) with interior knots spaced evenly at quantiles.

We consider exact CS and DL resulting from price increases from $\mf p^0 \in [\$1.20,\$1.40]$ to $\mf p^1 = \$1.40$ at income levels of $\mf y = \$42,500$ (low) and $\mf y = \$72,500$ (high). We estimate exact CS at each initial price level by solving the ODE (\ref{e-cs}) by backward differences. We construct UCBs for exact CS as described above by setting $\mc T = [\$1.20,\$1.40] \times \{\$1.40\} \times \{\$42,500\}$ for the low-income group and $\mc T = [\$1.20,\$1.40] \times \{\$1.40\} \times \{\$72,500\}$ for the high-income group, $f_t(h) = f_{CS,t}(h)$  from display (\ref{e-fcs}), and  $D f_{t}(\wh h)[\psi^J] = D f_{CS,t}(\wh h)[\psi^J]$ from display (\ref{e-Df_CS}). The ODE (\ref{e-cs}) is solved numerically by backward differences and the integrals in (\ref{e-Df_CS}) are computed numerically. UCBs for DL are formed similarly, $f_t(h) = f_{DL,t}(h)$  from display (\ref{e-fdwl}), and  $D f_{t}(\wh h)[\psi^J] = D f_{DL,t}(\wh h)[\psi^J]$ from display (\ref{e-Df_DL}). We draw the bootstrap innovations $\varpi_i$ from Mammen's two-point distribution with $1000$ bootstrap replications.

\begin{figure}[p]
  \centering
  \includegraphics[trim = 20mm 70mm 20mm 75mm, clip, width=0.7\textwidth]{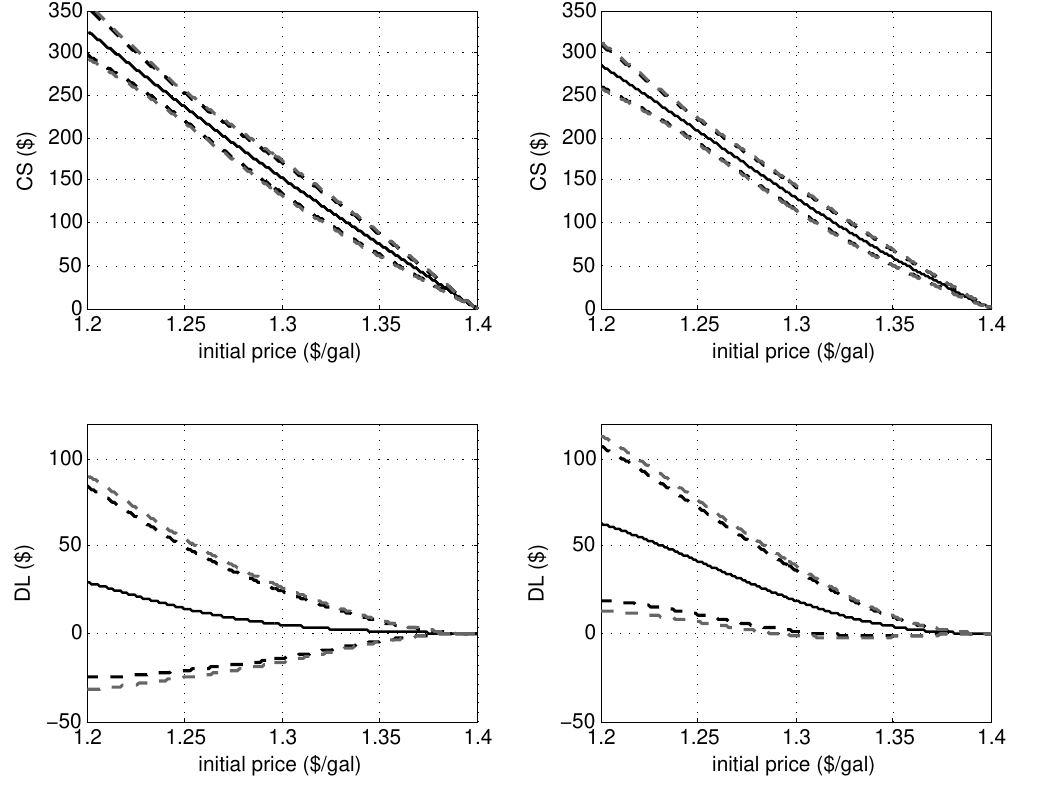}
  \parbox{5.0in}{\caption{ \label{f-cs-dwl} \small Estimated CS and DL from a price increase to \$1.40/gal (solid black line) and their  bootstrap UCBs (dashed black lines are 90\%, dashed grey lines are 95\%) when demand is estimated via sieve NPIV. Left panels are for household income of \$72,500; right panels are for household income of \$42,500.} }
\end{figure}

\begin{figure}[p]
  \centering
  \includegraphics[trim = 20mm 70mm 20mm 75mm, clip, width=0.7\textwidth]{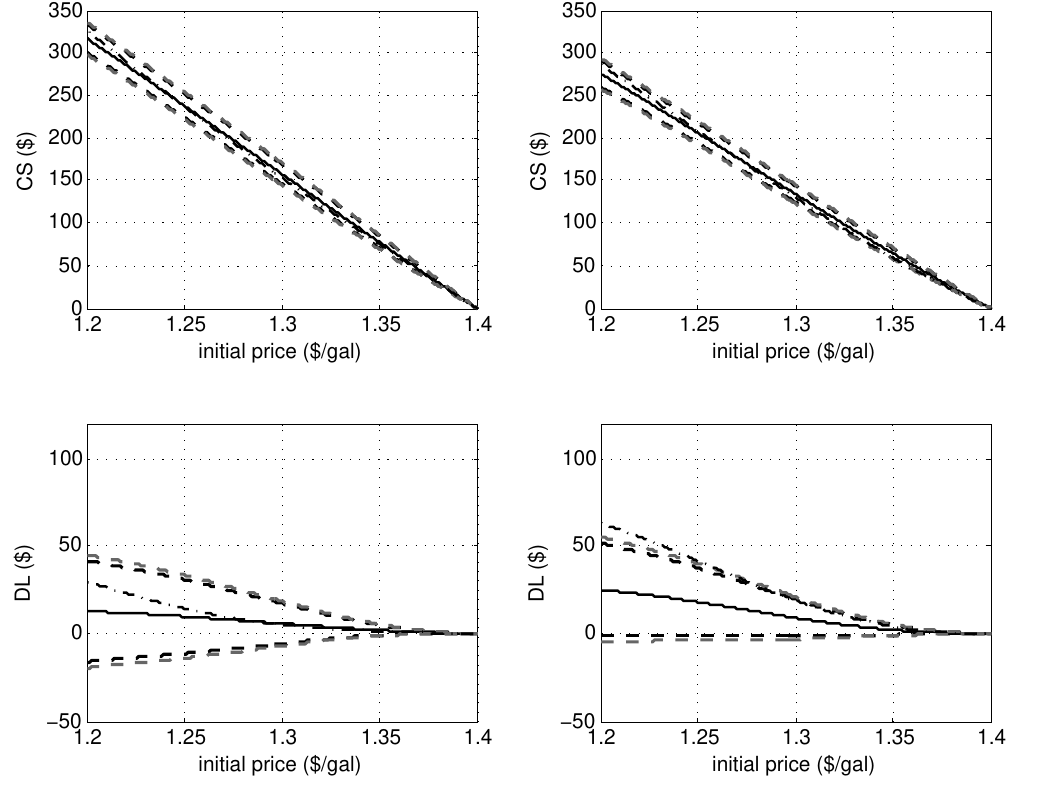}
  \parbox{5.0in}{\caption{ \label{f-cs-dwl-ls} \small Estimated CS and DL from a price increase to \$1.40/gal (solid black lines) and their bootstrap UCBs  (dashed black lines are 90\%, dashed grey lines are 95\%) when demand is estimated via series LS. CS and DL when demand is estimated via NPIV are also shown (black dash-dot lines). Left panels are for household income of \$72,500; right panels are for household income of \$42,500.} }
\end{figure}

The exact CS and DL estimates are presented in Figure \ref{f-cs-dwl} together with their UCBs. It is clear that exact CS is much more precisely estimated than DL. This is to be expected, since exact CS is computed by essentially integrating over one argument of the estimated demand function and is therefore smoother than the DL functional, which depends on $h_0$ estimated at the point $(\mf p^1,\mf y)$. In fact, even though the sieve NPIV $\wh h$ itself converges slowly, the UCBs for exact CS are still quite informative. At their widest point (with initial price \$1.20), the 95\% UCBs for exact CS for low-income households are $[\$259,\$314]$. In terms of comparison across high- and low-income households, the exact CS estimates are higher for the high-income households whereas DL estimates are higher for the low-income households.

Figure \ref{f-cs-dwl-ls} displays estimates obtained when we treat price as exogenous and estimate demand ($h_0$) by series LS regression. This is a special case of the preceding analysis with $X_i = W_i = (\mf P_i,\mf Y_i)'$, $K = J$ and $\psi^J = b^K$. These estimates display several notable features. First, the exact CS estimates are very similar whether demand is estimated via series LS or via sieve NPIV. Second, the UCBs for exact CS estimates are of a similar width to those obtained when demand was estimated via sieve NPIV, even though NPIV is an ill-posed inverse problem whereas nonparametric LS regression is not. Third, the UCBs for DL are noticeably narrower when demand is estimated via series LS than when demand is estimated via sieve NPIV. Fourth, the DL estimates for LS and sieve NPIV are similar for high income households but quite different for low income households. This is consistent with \cite{BlundellHorowitzParey2013}, who find some evidence of endogeneity in gasoline prices for low income groups.

\subsection{Empirical application 2: UCBs for Engel curves and their derivatives}

Engel curves describe the household budget share for expenditure categories as a function of total household expenditure. Following \cite{BCK}, we use sieve NPIV to estimate Engel curves, taking log total household income as an instrument for log total household expenditure. We use data from the 1995 British Family Expenditure Survey, focusing on the subset of married or cohabitating couples with one or two children, with the head of household aged between 20 and 55 and in work. This leaves a sample of size $n = 1027$.  We consider six categories of nondurables and services expenditure: food in, food out, alcohol, fuel, travel, and leisure.

We construct UCBs for Engel curves as described above by setting $\mc T = [4.75,6.25]$ (approximately the 5th to 95th percentile of log expenditure), $f_t(h) = h(t)$, and $D f_t(\wh h)[\psi^J] = \psi^J(t)$. We also construct UCBs for derivatives of the Engel curves by setting $\mc T = [4.75,6.25]$, $f_t(h)$ to be the derivative of $h$ evaluated at $t$, and $D f_t(\wh h)[\psi^J]$ to be the vector formed by taking derivatives of $\psi_{J1},\ldots,\psi_{JJ}$ evaluated at $t$. For both constructions, we use a quartic B-spline basis of dimension $J = 5$ for $\Psi_J$ and a quartic B-spline basis of dimension $K = 9$ for $B_K$, with interior knots evenly spaced at quantiles (an important feature of sieve estimators is that the same sieve dimension can be used for optimal estimation of the function and its derivatives; this is not the case for kernel-based estimators). We draw the bootstrap innovations $\varpi_i$ from Mammen's two-point distribution with $1000$ bootstrap replications.

\begin{figure}[p]
  \centering
  \includegraphics[trim = 20mm 70mm 20mm 75mm, clip, width=0.7\textwidth]{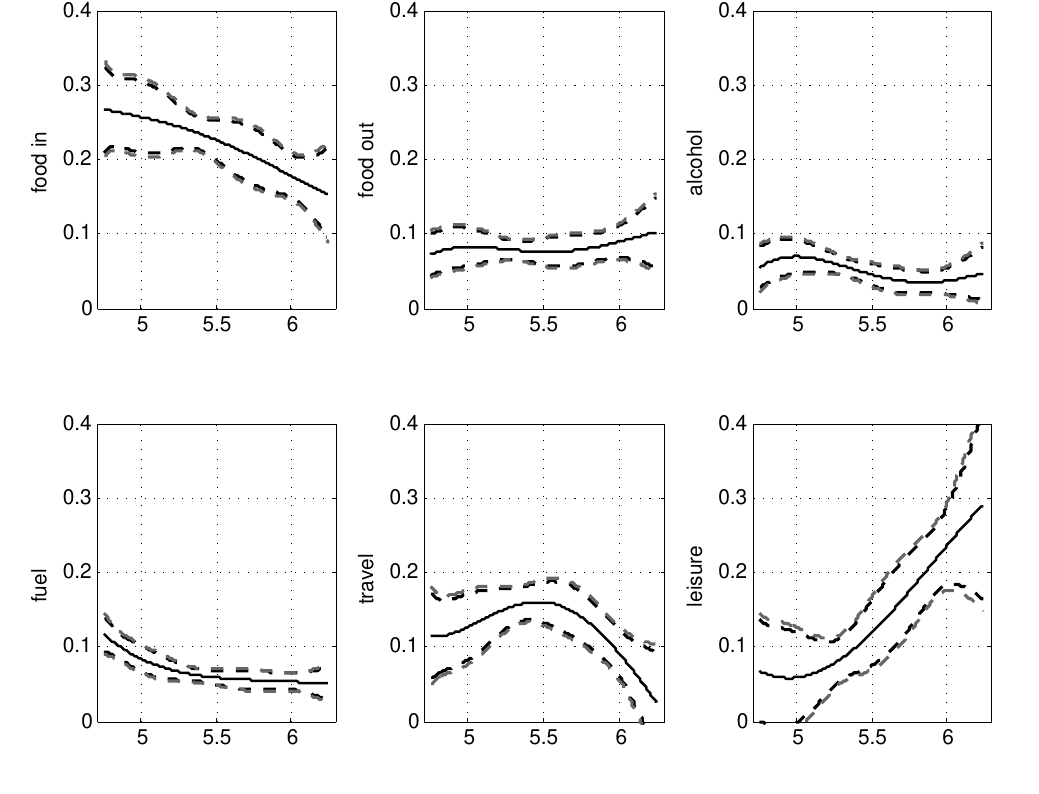}
  \parbox{5.0in}{\caption{ \label{f-engel} \small Estimated Engel curves (black line) with bootstrap uniform confidence bands (dashed black lines are 90\%, dashed grey lines are 95\%). The $x$-axis is log total household expenditure, the $y$-axis is household budget share.} }
\end{figure}

\begin{figure}[p]
  \centering
  \includegraphics[trim = 20mm 70mm 20mm 75mm, clip, width=0.7\textwidth]{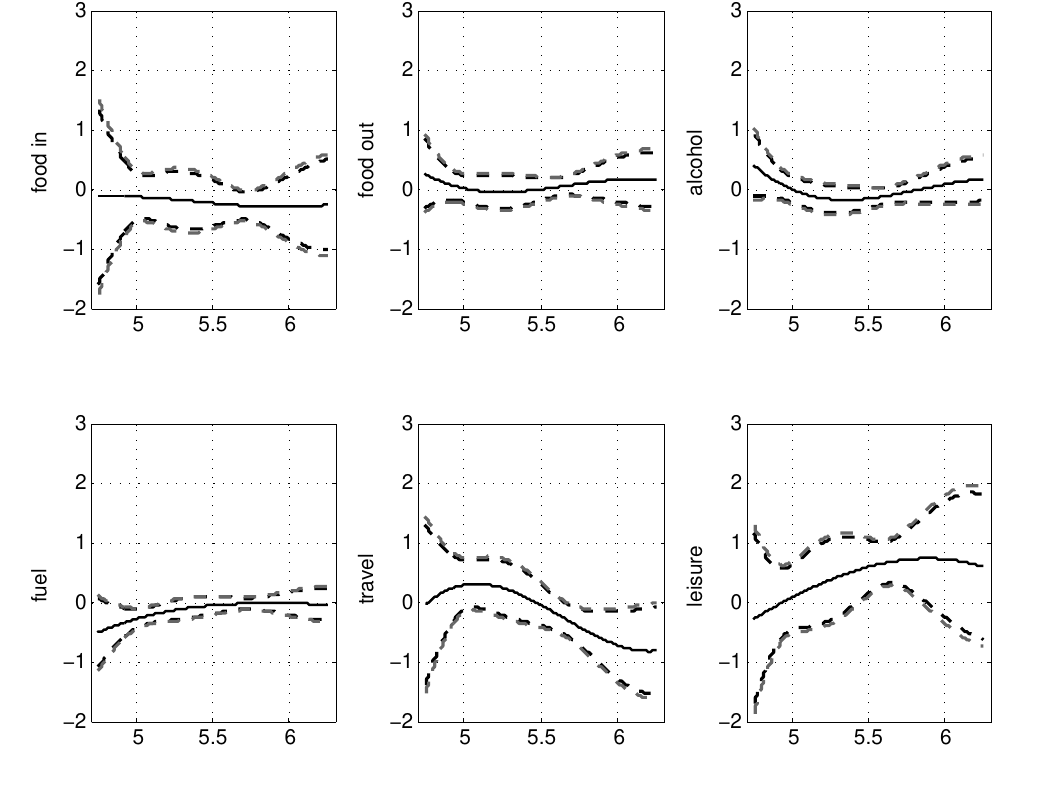}
  \parbox{5.0in}{\caption{ \label{f-engel-deriv} \small Estimated Engel curve derivatives (black line) with bootstrap uniform confidence bands (dashed black lines are 90\%, dashed grey lines are 95\%).} }
\end{figure}

The Engel curves presented in Figure \ref{f-engel} and their derivatives presented in Figure \ref{f-engel-deriv} exhibit several interesting features. The curves for food-in and fuel (necessary goods) are both downward sloping, with the curve for fuel exhibiting a pronounced downward slope at lower income levels. The derivative of the curve for fuel is negative, though the UCBs are positive at the extremities. In contrast, the curve for leisure expenditure (luxury good) is strongly upwards sloping and its derivative is positive except at low income levels. Remaining curves for food-out, alcohol and travel appear to be non-monotonic.

\section{Optimal sup-norm convergence rates}\label{npiv sec}

This section presents several results on sup-norm convergence rates. Subsection \ref{s-upper} presents upper bounds on sup-norm convergence rates of NPIV estimators of $h_0$ and its derivatives. Subsection \ref{s-lower} presents (minimax) lower bounds. Subsection \ref{s-exog} considers NPIV models with endogenous and exogenous regressors that are useful in empirical studies.

\textbf{Notation:}
We work on a probability space $(\Omega,\mathcal F,\mb P)$.  $\mathcal A^c$ denotes the complement of an event $\mathcal A \in \mathcal F$. We abbreviate ``with probability approaching one'' to ``wpa1'', and say that a sequence of events $\{\mathcal A_n \} \subset \mathcal F$ holds wpa1 if $\mb P(\mathcal A_n^c) = o(1)$. For a random variable $X$ we define the space $L^q(X)$ as the equivalence class of all measurable functions of $X$ with finite $q$th moment if $1 \leq q < \infty$; when $q = \infty$ we denote $L^\infty(X)$ as the set of all bounded measurable functions $g : \mathcal X \to \mb R$ endowed with the sup norm $\|g\|_\infty = \sup_x |g(x)|$. Let $\langle\cdot,\cdot\rangle_X$ denote the inner product on $L^2(X)$. For matrix and vector norms, $\|\cdot\|_{\ell^q}$ denotes the vector $\ell^q$ norm when applied to vectors and the operator norm induced by the vector $\ell^q$ norm when applied to matrices. If $a$ and $b$ are scalars we let $a \vee b := \max\{a,b\}$ and $a \wedge b := \min\{a,b\}$. Minimum and maximum eigenvalues are denoted by $\lambda_{\min}$ and $\lambda_{\max}$. If $\{a_n\}$ and $\{b_n\}$ are sequences of positive numbers, we say that $a_n \lesssim b_n$ if $\limsup_{n \to \infty} a_n/b_n < \infty$ and we say that $a_n \asymp b_n$ if $a_n \lesssim b_n$ and $b_n \lesssim a_n$.

\textbf{Sieve measure of ill-posedness}. For a NPIV model (\ref{npreg}), an important quantity is the \emph{measure of ill-posedness} which, roughly speaking, measures how much the conditional expectation $h \mapsto E[h(X_i)|W_i = w]$ smoothes out $h$. Let $T : L^2(X) \to L^2(W)$ denote the conditional expectation operator given by
\[
 T h(w) = E[h(X_i)|W_i = w]\,.
\]
Let $\Psi_J = clsp\{\psi_{J1},\ldots,\psi_{JJ}\} \subset L^2(X)$ and $B_K = clsp \{b_{K1},\ldots,b_{KK}\}\subset L^2(W)$ denote the sieve spaces for the endogenous variables and instrumental variables, respectively. Let $\Psi_{J,1} = \{h \in \Psi_J : \|h\|_{L^2(X)} = 1\}$. The \emph{sieve $L^2$ measure of ill-posedness} is
\begin{equation*}
 \tau_J  =  \sup_{h \in \Psi_J : h \neq 0} \frac{\|h\|_{L^2(X)}}{\|Th\|_{L^2(W)}} =  \frac{1}{ \inf_{h \in \Psi_{J,1}} \|Th\|_{L^2(W)}}\,.
\end{equation*}
Following \cite{BCK}, we call a NPIV model (\ref{npreg}) with $X_i$ being a $d$-dimensional random vector: \\
(i) \emph{mildly ill-posed} if $\tau_J = O(J^{\varsigma/d})$ for some $\varsigma >
0 $; and \\
(ii) \emph{severely ill-posed} if $\tau_J = O(\exp(\frac{1}{2}
J^{\varsigma /d}))$ for some $\varsigma > 0$.

See our second version \citep{ChenChristensen-adaptive-npiv} for simple consistent estimation of the sieve measure of ill-posedness $\tau_J$.

\subsection{Sup-norm convergence rates}\label{s-upper}

We first introduce some basic conditions on the basic NPIV model (\ref{npreg}) and the sieve spaces.

\begin{assumption} \label{a-data}
(i) $X_i$ has compact rectangular support $\mathcal X \subset \mb R^d$ with nonempty interior and the density of $X_i$ is uniformly bounded away from $0$ and $\infty$ on $\mathcal X$; (ii) $W_i$ has compact rectangular support $\mathcal W \subset \mb R^{d_w}$ and the density of $W_i$ is uniformly bounded away from $0$ and $\infty$ on $\mathcal W$; (iii) $T : L^2(X) \to L^2(W)$ is injective; and (iv) $h_0 \in \mathcal H \subset L^\infty(X)$, and ${\cup_J \Psi_J}$ is dense in $(\mathcal H,\|\cdot\|_{L^\infty (X)})$.
\end{assumption}

\begin{assumption} \label{a-residuals}
(i) $\sup_{w\in \mathcal W} E[u_i^2|W_i = w] \leq \overline \sigma^2 < \infty$; and (ii) $E[|u_i|^{2+\delta}] < \infty$ for some $\delta > 0$.
\end{assumption}

The following assumptions concern the basis functions. Define
\[
 \begin{array}{rcccl}
 G_\psi & = & G_{\psi,J} & = & E[\psi^J(X_i)\psi^J(X_i)']  =  E[ \Psi'\Psi/n]\\
 G_b & = & G_{b,K} & = & E[b^K(W_i)b^K(W_i)']  =  E[ B'B/n]\\
 S & = & S_{KJ} & = & E[b^K(W_i)\psi^J(X_i)']  =  E[ B'\Psi/n]\,.
 \end{array}
\]
We assume throughout that the basis functions are not linearly dependent, i.e. $S$ has full column rank $J$ and  $G_{\psi,J}$ and $G_{b,K}$ are positive definite for each $J$ and $K$, i.e. $e_J = \lambda_{\min}(G_{\psi,J}) >0$ and $e_{b,K} = \lambda_{\min}(G_{b,K}) >0$, although $e_J$ and $e_{b,K}$ could go to zero as $K\geq J$ goes to infinity. Let
\begin{align*}
 \zeta_{\psi} & = \zeta_{\psi,J}  =  \sup_{x} \|G_\psi^{-1/2} \psi^J(x)\|_{\ell^2} & & \zeta_{b}  = \zeta_{b,K}  =  \sup_{w} \|G_b^{-1/2} b^K(w)\|_{\ell^2}\\
  \xi_{\psi} & = \xi_{\psi,J} = \sup_{x} \|\psi^J (x)\|_{\ell^1}
\end{align*}
for each $J$ and $K$ and define $\zeta = \zeta_J = \zeta_{b,K} \vee \zeta_{\psi,J}$. Note that $\zeta_{\psi,J}$ has some useful properties: $\|h\|_\infty \leq \zeta_{\psi,J} \|h\|_{L^2(X)}$ for all $h \in \Psi_J$, and $\sqrt J =(E[\|G_\psi^{-1/2} \psi^J(X)\|_{\ell^2}^2])^{1/2} \leq \zeta_{\psi,J} \leq \xi_{\psi,J} / \sqrt{e_J}$; clearly $\zeta_{b,K}$ has similar properties.

We say that the sieve basis for $\Psi_J$ is H\"older continuous if there exist finite constants $\omega \geq 0, \omega' >0$ such that $\|G_{\psi,J}^{-1/2}\{\psi^J(x) - \psi^J(x')\}\|_{\ell^2} \lesssim J^\omega \|x - x'\|_{\ell^2}^{\omega'}$ for all $x,x' \in \mathcal X$.

\begin{assumption} \label{a-sieve}
(i) the basis spanning $\Psi_J$ is H\"older continuous; (ii) $\tau_J \zeta^2 /\sqrt n = O(1)$; and (iii) $\zeta^{(2+\delta)/\delta} \sqrt{(\log n)/n} = o(1)$.
\end{assumption}

Let  $\Pi_J: L^2(X) \to \Psi_J$ denote the $L^2(X)$ orthogonal (i.e. least squares) projection onto $\Psi_J$, namely $\Pi_J h_0 = \mathrm{arg}\min_{h \in \Psi_J} \|h_0 - h\|_{L^2(X)}$ and let $\Pi_K :L^2(W) \to B_K$ denote the $L^2(W)$ orthogonal (i.e. least-squares) projection onto $B_K$.
Let $Q_J h_0 = \mathrm{arg}\min_{h \in \Psi_J} \|\Pi_K T(h_0 - h)\|_{L^2(W)}$ denote the sieve 2SLS projection of $h_0$ onto $\Psi_J$. We may write $Q_J h_0 = \psi^J(\cdot)'c_{0,J}$ where
\[
c_{0,J} = [S' G_b^{-1} S]^{-1} S' G_b^{-1} E[b^K(W_i)h_0 (X_i)]\,.
\]

\begin{assumption}\label{a-approx}
(i) $\sup_{h \in \Psi_{J,1}}\|(\Pi_K T - T)h\|_{L^2(W)} = o(\tau_J^{-1})$; (ii) $\tau_J \times \|T(h_0 - \Pi_J h_0)\|_{L^2(W)} \leq \mathrm{const} \times \|h_0 - \Pi_J h_0\|_{L^2(X)}$; and (iii) $\|Q_J (h_0 - \Pi_J h_0) \|_\infty \leq O(1) \times \|h_0 - \Pi_J h_0\|_{\infty}$.
\end{assumption}

\textbf{Discussion of Assumptions}.
Assumption \ref{a-data} is standard. Assumption \ref{a-data}(iii) is stronger than needed for convergence rates in sup-norm only. We impose it as a common sufficient condition for convergence rates in both sup-norm and $L^2$-norm (Appendix \ref{s-l2}). For sup-norm convergence rate only, Assumption \ref{a-data}(iii) could be replaced by the following weaker identification condition:\\
 \textbf{Assumption \ref{a-data}} \emph{(iii-sup) $h_0 \in \mathcal H \subset L^\infty(X)$, and $T[h-h_0]=0 \in L^2(W)$ for any $h\in \mathcal H$ implies that $\|h-h_0\|_\infty =0$.}
\\
This in turn is implied by the injectivity of $T:L^\infty(X)\to L^2(W)$ (or the bounded completeness), which is weaker than the injectivity of $T : L^2(X) \to L^2(W)$ (i.e., the $L^2$-completeness). Bounded completeness or $L^2$-completeness condition is often assumed in models with endogeneity (e.g. \cite{NeweyPowell,CFR2007,BCK,Andrews2011,CCLN}) and is generically satisfied according to \cite{Andrews2011}. The parameter space $\mathcal H$ for $h_0$ is typically taken to be a H\"older or Sobolev class of smooth functions. Assumption \ref{a-data}(i) could be relaxed to unbounded support, and the proofs need to be modified slightly using wavelet basis and weighted compact embedding results, see, e.g., \cite{BCK,ChenPouzo2012,Triebel2006} and references therein. To present the sup-norm rate results in a clean way we stick to the simplest Assumption \ref{a-data}. Assumption \ref{a-residuals} is also imposed for sup-norm convergence rates for series LS regression under exogeneity (e.g., \cite{ChenChristensen-reg}). Assumption \ref{a-sieve}(i) is satisfied by many commonly used sieve bases, such as splines, wavelets, and cosine bases. Assumption \ref{a-sieve}(ii)(iii) restrict the rate at which $J$ can grow with $n$. Upper bounds for $\zeta_{\psi,J}$ and $\zeta_{b,K}$ are known for commonly used bases. For instance, under Assumption \ref{a-data}(i)(ii), $\zeta_{b,K} = O(\sqrt K)$ and $\zeta_{\psi,J} = O(\sqrt J)$ for (tensor-product) polynomial spline, wavelet and cosine bases, and $\zeta_{b,K} = O(K)$ and $\zeta_{\psi,J} = O(J)$ for (tensor-product) orthogonal polynomial bases; see, e.g., \cite{Newey1997}, \cite{Huang1998} and main online Appendix \ref{ax-basis}. Assumption \ref{a-approx}(i) is a mild condition on the approximation properties of the basis used for the instrument space and is similar to the first part of Assumption 5(iv) of \cite{Horowitz2014}. In fact, $\|(\Pi_K T - T)h\|_{L^2(W)} = 0$ for all $h \in \Psi_J$ when the basis functions for $B_K$ and $\Psi_J$ form either a Riesz basis or eigenfunction basis for the conditional expectation operator. Assumption \ref{a-approx}(ii) is the usual $L^2$ ``stability condition'' imposed in the NPIV literature (cf. Assumption 6 in \cite{BCK} and Assumption 5.2(ii) in \cite{ChenPouzo2012}). Assumption \ref{a-approx}(iii) is a new $L^{\infty}$ ``stability condition'' to control the sup-norm bias. It turns out that Assumption \ref{a-approx}(ii) and \ref{a-approx}(iii) are also automatically satisfied by Riesz bases; see Appendix \ref{s-sup-lemmas} for further discussions and sufficient conditions.

To derive the sup-norm (uniform) convergence rate we split $\|\wh h - h_0\|_\infty$ into so-called ``bias'' and ``standard deviation'' terms and derive sup-norm convergence rates for the two terms. Specifically, let
\begin{equation*}
 \widetilde h(x)  = \psi^J(x)'\widetilde c ~~\mbox{ with }~~ \widetilde c = [\Psi'B(B'B)^-B'\Psi]^- \Psi'B (B'B)^- B'H_0
\end{equation*}
where $H_0 = (h_0(X_1),\ldots,h_0(X_n))'$. We refer loosely to $\|\widetilde h - h_0\|_\infty$ as the  ``bias'' term and $\|\wh h - \widetilde h\|_\infty$ as the ``standard deviation'' (or sometimes ``variance'') term. Both are random quantities. We first bound the sup-norm ``standard deviation'' term in the following lemma.

\begin{lemma}\label{lem-chat}
Let Assumptions \ref{a-data}(i)(iii), \ref{a-residuals}(i)(ii), \ref{a-sieve}(ii)(iii), and \ref{a-approx}(i) hold. Then:\\
(1) $\|\wh h - \widetilde h\|_\infty = O_p \big( \tau_J \xi_{\psi,J} \sqrt{ (\log J)/ (n e_J )} \big)$.\\
(2) If Assumption \ref{a-sieve}(i) also holds, then: $\|\wh h - \widetilde h\|_\infty = O_p \big(  \tau_J \zeta_{\psi,J} \sqrt{(\log n)/n} \big)\,.$
\end{lemma}

Recall that $\sqrt{J} \leq \zeta_{\psi,J} \leq \xi_{\psi,J} / \sqrt{e_J}$. Result (2) of Lemma \ref{lem-chat} provides a slightly tighter upper bound on the variance term than Result (1) does, while Result (1) allows for slightly more general basis to approximate $h_0$. For splines and wavelets, we show in Appendix \ref{ax-basis} that $\xi_{\psi,J} / \sqrt{e_J} \lesssim \sqrt{J}$, so Results (1) and (2) produce the same tight upper bound $\|\wh h - \widetilde h\|_\infty = O_p(\tau_J \sqrt{(J\log n )/n})$ when $J \asymp n^r$ for some constant $r > 0$.

Before we present an upper bound on the ``bias'' term in Theorem \ref{t-rate} part (1) below, we mention one more property of the sieve space $\Psi_J$ that is crucial for sharp bounds on the sup-norm bias term. Let $h_{0,J} \in \Psi_J$ denote the best approximation to $h_0$ in sup-norm, i.e. $h_{0,J}$ solves $\inf_{h \in \Psi_J}\|h_0 - h\|_\infty$. Then by Lebesgue's Lemma \cite[p. 30]{DeVoreLorentz}:
\[
 \|h_0 - \Pi_J h_0 \|_\infty   \leq  (1 + \|\Pi_J \|_{\infty}) \times \|h_0 - h_{0,J}\|_\infty
\]
where $\|\Pi_J \|_{\infty}$ is the Lebesgue constant for the sieve $\Psi_J$. Recently it has been established that $\|\Pi_J\|_\infty \lesssim 1$ when $\Psi_J$ is spanned by a tensor product B-spline basis (\cite{Huang2003}) or a tensor product Cohen-Daubechies-Vial (CDV) wavelet basis (\cite{ChenChristensen-reg}).\footnote{See \cite{DeVoreLorentz} and \cite{BCCK2014} for examples of other bases with bounded Lebesgue constant or with Lebesgue constant diverging slowly with the sieve dimension.} Boundedness of the Lebesgue constant is crucial for attaining optimal sup-norm rates.

\begin{theorem}\label{t-rate}
(1) Let Assumptions \ref{a-data}(iii), \ref{a-sieve}(ii) and \ref{a-approx} hold.  Then:
\[
 \|\widetilde h -  h_0\|_\infty = O_p \left(  \|h_0 - \Pi_J h_0\|_\infty \right) \,.
\]
(2) Let Assumptions \ref{a-data}(i)(iii)(iv), \ref{a-residuals}(i)(ii), \ref{a-sieve}(ii)(iii), and \ref{a-approx} hold. Then:
\[
 \|\wh h - h_0\|_\infty = O_p \left( \|h_0 - \Pi_J h_0\|_\infty +   \tau_J \xi_{\psi,J} \sqrt{ (\log J)/ (n e_J )} \right)\,.
\]
(3) Further, if the linear sieve $\Psi_J$ satisfies $\|\Pi_J\|_\infty \lesssim 1$ and $\xi_{\psi,J} / \sqrt{e_J} \lesssim \sqrt{J}$, then
\[
 \|\wh h - h_0\|_\infty   =   O_p \left( \|h_0 - h_{0,J}\|_\infty +  \tau_J  \sqrt{(J \log J)/n} \right)\,.
\]
\end{theorem}

Theorem \ref{t-rate}(2)(3) follows directly from part (1) (for bias) and Lemma \ref{lem-chat}(1) (for standard deviation). See Appendix \ref{s-sup-lemmas} for additional details about bound on sup-norm bias.

The following corollary provides concrete sup-norm convergence rates of $\wh h$ and its derivatives. To introduce the result, let $B^p_{\infty,\infty}$ denote the H\"older space of smoothness $p > 0$ and $\|\cdot\|_{B^p_{\infty,\infty}}$ denote its norm (see Section 1.11.10 of \cite{Triebel2006}). Let $B_\infty(p,L) = \{h \in B^p_{\infty,\infty} : \|h\|_{B^p_{\infty,\infty}} \leq L\}$ denote a H\"older ball of smoothness $p > 0$ and radius $ L \in (0, \infty )$. Let $\alpha_1,\ldots,\alpha_d$ be non-negative integers, let $|\alpha| = \alpha_1 + \ldots + \alpha_d $, and define
\[
 \partial^\alpha h(x) := \frac{\partial ^{|\alpha|} h}{\partial^{\alpha_1} x_1 \cdots \partial ^{\alpha_d} x_d} h(x) \,.
\]
Of course, if $|\alpha| = 0$ then $\partial^\alpha h = h$.\footnote{If $|\alpha| > 0$ then we assume $h$ and its derivatives can be continuously extended to an open set containing $\mathcal X$.}

\begin{corollary}\label{c-deriv}
Let Assumptions \ref{a-data}(i)(ii)(iii) and \ref{a-approx} hold. Let $h_0 \in B_\infty(p,L)$, $\Psi_J$ be spanned by a B-spline basis of order $\gamma > p$ or a CDV wavelet basis of regularity $\gamma > p$, $B_K$ be spanned by a cosine, spline or wavelet basis.\\
(1) If Assumption \ref{a-sieve}(ii) holds, then
\[
 \|\partial^\alpha \widetilde h - \partial^\alpha h_0 \|_\infty = O_p \Big( J^{-(p-|\alpha|)/d} \Big)~~\mbox{ for all }~~0 \leq |\alpha| < p\,.
\]
(2) If Assumptions \ref{a-residuals}(i)(ii) and \ref{a-sieve}(ii)(iii) hold, then
\[
 \|\partial^\alpha \wh h - \partial^\alpha h_0 \|_\infty = O_p \Big( J^{-(p-|\alpha|)/d} +  \tau_J J^{|\alpha|/d} \sqrt{(J\log J)/n} \Big)~~\mbox{ for all }~~0 \leq |\alpha| < p\,.
\]
(2.a) Mildly ill-posed case: with $p \geq d/2$ and $\delta \geq d/(p + \varsigma)$, choosing $ J \asymp (n/\log n)^{d/(2(p+\varsigma)+d)}$ implies that Assumption \ref{a-sieve}(ii)(iii) holds and
\begin{equation*}
 \|\partial^\alpha \wh h - \partial^\alpha h_0 \|_\infty = O_p ( (n/\log n)^{-(p-|\alpha|)/(2(p+\varsigma)+d)})\,.
\end{equation*}
(2.b) Severely ill-posed case: choosing $J = (c_0\log n)^{d/\varsigma}$ with $c_0 \in (0,1)$ implies that Assumption \ref{a-sieve}(ii)(iii) holds and
\begin{equation*}
 \|\partial^\alpha  \wh h - \partial^\alpha  h_0 \|_{\infty} = O_p ( (\log n)^{-(p-|\alpha|)/\varsigma})\,.
\end{equation*}
\end{corollary}

Corollary \ref{c-deriv} shows that, for sieve NPIV estimators, taking derivatives has the same impact on the bias and standard deviation terms in terms of the order of convergence, and that the same choice of sieve dimension $J$ can lead to optimal sup-norm convergence rates for estimating $h_0$ and its derivatives simultaneously (since they match the lower bounds in Theorem \ref{npiv lower bound} below). When specializing to series LS regression (without endogeneity, i.e., $\tau_J =1$), Corollary \ref{c-deriv}(2.a) with $\varsigma=0$ automatically implies that spline and wavelet series LS estimators will also achieve the optimal sup-norm rates of \cite{Stone1982} for estimating the derivatives of a nonparametric LS regression function. This strengthens the recent results in \cite{BCCK2014} and \cite{ChenChristensen-reg} for sup-norm rate optimality of spline and wavelet LS estimators of the regression function $h_0$ itself.
This is in contrast to kernel based LS regression estimators where different choices of bandwidth are needed for the optimal rates of estimating $h_0$ and its derivatives.

Corollary \ref{c-deriv} is useful for estimating functions with certain shape properties. For instance, if $h_0 : [a,b] \to \mb R$ is strictly monotone and/or strictly concave/convex, then knowing that $\partial {\wh h}(x)$ and/or $\partial^2 {\wh h}(x)$ converge uniformly to $\partial h_0 (x)$ and/or $\partial^2 h_0(x)$ implies that $\wh h$ will also be strictly monotone and/or strictly concave/convex wpa1. In this paper, we shall illustrate the usefulness of Corollary \ref{c-deriv} in controlling the nonlinear remainder terms for pointwise and uniform inferences on highly nonlinear (i.e., beyond quadratic) functionals of $h_0$; see Sections \ref{s-ucb} and \ref{s-inference} for details.

\subsection{Lower bounds}\label{s-lower}

We now establish that the sup-norm rates obtained in Corollary \ref{c-deriv} are the best possible (i.e. minimax) sup-norm convergence rates for estimating $h_0$ and its derivatives.

To establish a lower bound, we require a \emph{link condition} that relates smoothness of $T$ to the parameter space for $h_0$. Let $\widetilde \psi_{j,k,G}$ denote a tensor-product CDV wavelet basis for $[0,1]^d$ of regularity $\gamma > p$. Appendix \ref{ax-basis} provides details on the construction and properties of this basis.

\paragraph{Condition LB} \emph{(i) Assumption \ref{a-data}(i)--(iii) holds; (ii)  $E[u_i^2 |W_i = w] \geq \underline \sigma^2 > 0$ uniformly for $w \in \mathcal W$; and (iii) there is a positive decreasing function $\nu$ s.t. $\|T h\|_{L^2(W)}^2 \lesssim \sum_{j,G,k} [\nu(2^j)]^2 \langle h, \widetilde \psi_{j,k,G} \rangle_X^2$ holds for all $h \in B_\infty(p,L)$.}

Condition LB is standard in the optimal rate literature (see \cite{HallHorowitz} and \cite{ChenReiss}).
The mildly ill-posed case corresponds to choosing $\nu(t) = t^{-\varsigma}$, and says roughly that the conditional expectation operator $T$ makes $p$-smooth functions of $X$ into $(\varsigma+p)$-smooth functions of $W$. The severely ill-posed case, which corresponds to choosing $\nu(t) = \exp(-\frac{1}{2}t^{\varsigma})$ and says roughly that $T$ maps smooth functions of $X$ into ``supersmooth'' functions of $W$.

\begin{theorem} \label{npiv lower bound}
Let Condition LB hold for the NPIV model with a random sample $\{(X_{i},Y_{i},W_i)\}_{i=1}^n$. Then for any $0 \leq |\alpha| < p$:
\begin{equation*}
 \liminf_{n \to \infty} \inf_{\wh g_n} \sup_{h \in B_\infty(p,L)} \mb P_h \left( \|\wh g_n - \partial^\alpha h\|_\infty \geq c r_n \right) \geq c'>0
\end{equation*}
where
\[
 r_n = \left[ \begin{array}{ll}
 (n/\log n)^{-(p-|\alpha|)/(2(p+\varsigma)+d)} & \mbox{in the mildly ill-posed case } \\
 (\log n)^{-(p-|\alpha|)/\varsigma} & \mbox{in the severely ill-posed case,}
 \end{array} \right.
\]
$\inf_{\wh g_n}$ denotes the infimum over all estimators of $\partial^\alpha h$ based on the sample of size $n$, $\sup_{h \in B_\infty(p,L)} \mb P_h$ denotes the sup over $h \in B_\infty(p,L)$  and distributions of $(X_i,W_i,u_i)$ that satisfy Condition LB with fixed $\nu$, and the finite positive constants $c, c'$ do not depend on $n$.
\end{theorem}

According to Theorem \ref{npiv lower bound} and Theorem \ref{t-lb-l2} (in Appendix \ref{s-l2}), the minimax lower bounds in sup-norm  for estimating $h_0$ and its derivatives coincide with those in $L^2$  for severely ill-posed NPIV problems, and are only a factor of $[\log( n)]^{\epsilon}$
(with $\epsilon =\frac{p-|\alpha|}{2(p+\varsigma)+d}<\frac{p}{2p+d}<\frac{1}{2}$) worse than those in $L^2$ for mildly ill-posed problems.
Our proof of sup-norm lower bound for NPIV models is similar to that of \cite{ChenReiss} for $L^2$-norm lower bound. Similar sup-norm lower bounds for density deconvolution were recently obtained by \cite{LouniciNickl}.

\subsection{Models with endogenous and exogenous regressors}\label{s-exog}

In many empirical studies, some regressors might be endogenous while others are exogenous. Consider the model
\begin{equation}\label{e-p-exog}
 Y_i = h_0(X_{1i},Z_i) + u_i
\end{equation}
where $X_{1i}$ is a vector of endogenous regressors and $Z_i$ is a vector of exogenous regressors. Let $X_i = (X_{1i}',Z_i')'$. Here the vector of instrumental variables $W_i$ is of the form $W_i = (W_{1i}',Z_i')'$ where $W_{1i}$ are instruments for $X_{1i}$. We refer to this as the ``partially endogenous case''.
The sieve NPIV estimator is implemented in exactly the same way as the ``fully endogenous'' setting in which $X_i$ consists only of endogenous variables, just like 2SLS with endogeneous and exogenous regressors.\footnote{ All that changes here is that $J$ may grow more quickly as the degree of ill-posedness will be smaller. In contrast, other NPIV estimators based on estimating the conditional densities of the regressors and instrumental variables must be implemented separately for each value of $z$  \citep{HallHorowitz,Horowitz2011,GagliardiniScaillet}.}
Our convergence rates presented in Section \ref{s-upper} and Appendix \ref{s-l2} apply equally to the partially endogenous model (\ref{e-p-exog}) under the stated regularity conditions: all that differs between the two cases is the interpretation of the sieve measure of ill-posedness.

Consider first the fully endogenous case where $T : L^2(X) \to L^2(W)$ is compact under mild conditions on the conditional density of $X$ given $W$ (see, e.g., \cite{NeweyPowell,BCK,DFFR,Andrews2011}). Then $T$ admits a singular value decomposition (SVD) $\{\phi_{0j},\phi_{1j},\mu_j\}_{j=1}^\infty$ where $(T^*T)^{1/2} \phi_{0j} = \mu_j \phi_{0j}$, $\mu_j \geq \mu_{j+1}$ for each $j$ and $\{\phi_{0j}\}_{j=1}^\infty$ and $\{\phi_{1j}\}_{j=1}^\infty$ are orthonormal bases for $L^2(X)$ and $L^2(W)$, respectively. Suppose that $\Psi_J$ spans $\phi_{0j},\ldots,\phi_{0J}$. Then the sieve measure of ill-posedness is $\tau_J = \mu_J^{-1}$.

Now consider the partially endogenous case. Similar to \cite{Horowitz2011}, we suppose that for each value of $z$ the conditional expectation operator $T_z : L^2(X_1|Z=z) \to L^2(W_1|Z=z)$ given by $(T_z h)(w_1) = E[h(X_1)|W_{1i} = w_1,Z_i = z]$ is compact. Then each $T_z$ admits a SVD $\{\phi_{0j,z},\phi_{1j,z},\mu_{j,z}\}_{j=1}^\infty$ where $T_z \phi_{0j,z} = \mu_{j,z} \phi_{1j,z}$, $(T^*_{z}T^{\phantom *}_{z})^{1/2} \phi_{0j,z} = \mu_{j,z} \phi_{0j,z}$, $(T^{\phantom *}_{z}T^*_{z})^{1/2} \phi_{1j,z} = \mu_{j,z} \phi_{1j,z}$,
$\mu_{j,z} \geq \mu_{j+1,z}$ for each $j$ and $z$, and $\{\phi_{0j,z}\}_{j=1}^\infty$ and $\{\phi_{1j,z}\}_{j=1}^\infty$ are orthonormal bases for $L^2(X_1|Z=z)$ and $L^2(W_1|Z=z)$, respectively, for each $z$.
The following result adapts Lemma 1 of \cite{BCK} to the partially endogenous setting.

\begin{lemma}\label{lem-bck}
Let $T_z$ be compact with SVD $\{\phi_{0j,z},\phi_{1j,z},\mu_{j,z}\}_{j=1}^\infty$ for each $z$. Let $\mu_j^{2}=E[\mu_{j,Z_i}^2]$ and $\phi_{0j}(\cdot,z)=\phi_{0j,z}(\cdot)$ for each $z$ and $j$. Then:
(1) $\tau_J \geq \mu_J^{-1}$.\\
(2) If, in addition, $\phi_{01},\ldots,\phi_{0J} \in \Psi_J$, then: $\tau_J \leq \mu_J^{-1}$.
\end{lemma}

Consider the following partially-endogenous stylized example from \cite{HoderleinHolzmann}. Let $X_{1i}$, $W_{1i}$ and $Z_i$ be scalar random variables with
\[
 \left( \begin{array}{c} X_{1i} \\ W_{1i} \\ Z_i \end{array} \right)
 \sim N \left( \left( \begin{array}{c} 0 \\ 0 \\ 0 \end{array} \right) , \left( \begin{array}{ccc}
 1 & \rho_{xw} & \rho_{xz} \\
 \rho_{xw} & 1 & \rho_{wz} \\
 \rho_{xz} & \rho_{wz} & 1 \end{array} \right) \right) \,.
\]
Then
\begin{equation} \label{e-mvar-partial}
 \left( \left. \begin{array}{c} \frac{X_{1i} - \rho_{xz} z}{\sqrt{1-\rho^2_{xz}}} \\ \frac{W_{1i} - \rho_{wz} z}{\sqrt{1-\rho^2_{wz}}}  \end{array} \right| Z_i = z \right)
 \sim N \left( \left( \begin{array}{c} 0 \\ 0 \end{array} \right) , \left( \begin{array}{ccc}
 1 & \rho_{xw| z} \\
 \rho_{xw| z}  & 1 \end{array} \right) \right)
\end{equation}
where
\[
 \rho_{xw| z} = \frac{\rho_{xw} - \rho_{xz} \rho_{wz}}{\sqrt{(1-\rho^2_{xz})(1-\rho^2_{wz})}}
\]
is the partial correlation between $X_{1i}$ and $W_{1i}$ given $Z_i$.
For each $j \geq 1$ let $H_j$ denote the $j$th Hermite polynomial (the Hermite polynomials form an orthonormal basis with respect to Gaussian density). Since $T_z : L^2(X_1|Z=z) \to L^2(W_1|Z=z)$ is compact for each $z$, it follows from Mehler's formula that $T_z$ has a SVD $\{\phi_{0j,z},\phi_{1j,z},\mu_{j,z}\}_{j=1}^\infty$ with
\[
 \phi_{0j,z}(x_1) = H_{j-1} \bigg( \frac{x_1 - \rho_{xz}z}{\sqrt{1-\rho_{xz}^2}} \bigg), \quad
 \phi_{1j,z}(w_1) = H_{j-1} \bigg( \frac{w_1 - \rho_{wz}z}{\sqrt{1-\rho_{wz}^2}} \bigg), \quad
 \mu_{j,z} = |\rho_{xw| Z}|^{j-1}
\]
for each $z$.
Since  $\mu_{J,z} = |\rho_{xw | z}|^{J-1}$ for each $z$, we have $\mu_J = |\rho_{xw | z}|^{J-1}\asymp |\rho_{xw | z}|^{J}$. If $X_{1i}$ and $W_{1i}$ are uncorrelated with $Z_i$, then $\mu_J = |\rho|^{J-1}$ where $\rho = \rho_{xw}$.

In contrast, consider the following fully-endogenous model in which $X_i$ and $W_i$ are bivariate with
\[
 \left( \begin{array}{c} X_{1i} \\ X_{2i} \\ W_{1i} \\ W_{2i} \end{array} \right)
 \sim N \left( \left( \begin{array}{c} 0 \\ 0 \\ 0 \\ 0 \end{array} \right) , \left( \begin{array}{cccc}
 1 & 0 & \rho_1 & 0 \\
 0 & 1 & 0 & \rho_2 \\
 \rho_1 & 0 & 1 & 0 \\
 0 & \rho_2 & 0 & 1 \end{array} \right) \right)
\]
where $\rho_1$ and $\rho_2$ are such that the covariance matrix is invertible. It is straightforward to verify that $T$ has singular value decomposition with
\[
 \phi_{0j}(x) = H_{j-1}(x_1) H_{j-1}(x_2) \, \quad \phi_{1j}(w) = H_{j-1}(w_1) H_{j-2}(w_2), \quad \mu_j = |\rho_1 \rho_2|^{j-1}\,,
\]
and $\mu_J = \rho^{2(J-1)}$ if $\rho_1 = \rho_2 = \rho$. Thus, the measure of ill-posedness diverges faster in the fully-endogenous case ($\mu_J = \rho^{2(J-1)}$) than that in the partially endogenous case ($\mu_J =|\rho|^{J-1}$).

\section{Uniform inference on collections of nonlinear functionals}\label{s-ucb}

In this section we apply our sup-norm rate results and tight bounds on random matrices (in main online Appendix \ref{ax-supp}) to establish uniform Gaussian process strong approximation and the consistency of the score bootstrap UCBs defined in (\ref{e-ucb}) for collections of (possibly) nonlinear functionals $\{f_t(\cdot) : t \in \mathcal T\}$ of a NPIV function $h_0$. See Section \ref{s-end} for discussions of other applications.

We consider functionals $f_t : \mc H \subset L^\infty(X) \to \mb R$ for each $t \in \mathcal T$ for which $Df_t(h)[v] = \lim_{\delta \rightarrow 0^{+}} [\delta^{-1} f_t(h +\delta v)]$ exists for all $v\in \mc H - \{h_0\}$ for all $h$ in a small neighborhood of $h_0$ (where the neighborhood is independent of $t$). This is trivially true for, say, $f_t(h) = h(t)$ with $\mc T \subseteq \mc X$ for UCBs for $h_0$. Let $\Omega = E[u_i^2 b^K(W_i) b^K(W_i)']$. Then the ``2SLS covariance matrix'' for $\wh c$ (given in (\ref{def-series-2SLS})) is
$$\mho = [S' G_b^{-1}  S]^{-1}  S'  G_b^{-1}  \Omega  G_b^{-1}  S[ S'  G_b^{-1}  S]^{-1}~,$$ and the sieve variance for $f_t (\wh h)$ is
\begin{equation*}
[\sigma_n(f_t)]^2   = \big(Df_t(h_{0})[\psi^J]\big)' \mho  \big(Df_t(h_{0})[\psi^J]\big)~.
\end{equation*}

\setcounter{assumption}{1}

\begin{assumption}[continued]%\label{a-residuals}
 (iii) $E[u_i^2 |W_i = w] \geq \underline \sigma^2 > 0$ uniformly for all $w \in \mathcal W$; and (iv) $\sup_w E[|u_i|^3|W_i = w] < \infty$.
\end{assumption}
Assumptions \ref{a-residuals}(iii)(iv) are reasonably mild conditions used to derive the uniform limit theory.
Define
\begin{align*}
 v_{n}(f_t)(x) & = \psi^J(x)' [S' G_b^{-1} S]^{-1} Df_t(h_{0})[\psi^J]~, &
\wh v_{n}(f_t)(x) & = \psi^J(x)' [S' G_b^{-1} S]^{-1} Df_t(\wh h) [\psi^J]~,
\end{align*}
where, for each fixed $t$, $v_{n}(f_t)$ could be viewed as a ``sieve 2SLS Riesz representer''. Note that $v_{n}(f_t)=\wh v_{n}(f_t)$ whenever $f_t$ is linear. Under Assumption \ref{a-residuals}(i)(iii) we have that
\begin{equation*}
[\sigma_n(f_t)]^2   \asymp Df_t(h_{0})[\psi^J]' [S' G_b^{-1} S]^{-1} Df_t(h_{0})[\psi^J]=\|\Pi_K T  v_{n}(f_t))\|_{L^2(W)}^2~~~\text{uniformly in}~t~.
\end{equation*}
Following \cite{ChenPouzo2014}, we call $f_t (\cdot)$ an irregular functional of $h_0$ (i.e., slower than $\sqrt n$-estimable) if $\sigma_n(f_t)\nearrow +\infty$ as $n \to \infty$. This includes the evaluation functionals $h_0(t)$ and $\partial^{\alpha}h_0(t)$ as well as $f_{CS,t}(h_0)$ and $f_{DL,t}(h_0)$. In this paper we shall focus on applications of sup-norm rate results to inference on irregular functionals.

\setcounter{assumption}{4}

\begin{assumption}
\label{a-functional}
Let $\eta_n$ and $\eta_n'$ be sequences of nonnegative numbers such that $\eta_n = o(1)$ and $\eta_n' = o(1)$. Let $\sigma_n(f_t) \nearrow +\infty$ as $n \to \infty$ for each $t \in \mathcal T$. Either (a) or (b) of the following holds:
\begin{itemize}
\item[\emph{(a)}] \emph{$f_t $ is a linear functional for each $t \in \mathcal T$ and $\sup_{t \in \mathcal T}\sqrt{n}%
(\sigma_n(f_t))^{-1} |f_t(\widetilde{h})-f_t(h_{0})| = O_{p}(\eta_n)$; or}
\item[\emph{(b)}] \emph{(i) $v\mapsto Df_t(h_{0})[v]$
is a linear functional for each $t \in \mathcal T$; (ii) }
\[
 \sup_{t \in \mathcal T} \left| \sqrt n  \frac{f_t(\wh h) - f(h_0) }{\sigma_n(f_t)} - \sqrt n  \frac{Df_t(h_0)[\wh h - \widetilde h]}{\sigma_n(f_t)} \right| = O_p(\eta_n)\,;
\]
\emph{and (iii) $\sup_{t \in \mathcal T} \frac{\|\Pi_K T (\wh v_{n}(f_t) - v_{n}(f_t))\|_{L^2(W)}}{\sigma_n(f_t)} = O_p(\eta_n')$.}
\end{itemize}
\end{assumption}
Assumption \ref{a-functional}(a)(b)(i)(ii) are similar to uniform-in-$t$ versions of Assumption 3.5 of \cite{ChenPouzo2014}. Assumption \ref{a-functional}(b)(iii) controls any additional error arising in the estimation of $\sigma_n(f_t)$ by $\wh\sigma(f_t)$ (given in equation \eqref{e-sieve-var}) due to nonlinearity of $f_t(\cdot)$, and is automatically satisfied with $\eta_n'=0$ when $f_t(\cdot)$ is a linear functional.

The next remark presents a set of sufficient conditions for Assumption \ref{a-functional} when $\{f_t : t \in \mc T\}$ are irregular functionals of $h_0$. Since the functionals are irregular, the quantity $\ul \sigma_n := \inf_{t \in \mathcal T} \sigma_n(f_t)$ will typically satisfy $\ul \sigma_n \nearrow +\infty$ as $n \to \infty$. Our sup-norm rates for $\wh h$ and $\wt h$, together with divergence of $\ul \sigma_n$, helps to control the nonlinearity bias terms.

\begin{remark} \label{rmk-a5suff}
Let $\mathcal H_n \subseteq \mathcal H$ be a sequence of neighborhoods of $h_0$ with $\wh h,\widetilde h \in \mathcal H_n$ wpa1 and assume $\ul \sigma_n :=\inf_{t \in \mathcal T} \sigma_n(f_t) >0$ for each $n$. Then: Assumption \ref{a-functional}(a) is implied by (a'), and Assumption \ref{a-functional}(b) is implied by (b'), where
\begin{itemize}
\item[(a')] (i) $f_t$ is a linear functional for each $t \in \mathcal T$ and there exists $\alpha$ with $|\alpha| \geq 0$ s.t. $\sup_t |f_t(h-h_0)| \lesssim \|\partial^\alpha h- \partial^\alpha h_0 \|_\infty$ for all $h \in \mathcal H_n$; and (ii) $n^{1/2}\ul \sigma_n^{-1}\|\partial^\alpha \widetilde h - \partial^\alpha h_0\|_\infty = O_{p}(\eta_n)$; or
\item[(b')] (i) $v\mapsto Df_t(h_{0})[v]$
is a linear functional for each $t\in \mathcal T$ and there exists $\alpha$ with $|\alpha| \geq 0$ s.t. $\sup_t|D f_t(h_{0})[h-h_0]| \lesssim \|\partial^\alpha h - \partial^\alpha h_0\|_\infty$ for all $h \in \mathcal H_n$;\\
(ii) there are $\alpha_1$, $\alpha_2$ with $|\alpha_1|,|\alpha_2| \geq 0$ s.t.
\begin{align*}
(ii.1)& \; \sup_t \left| f_t(\wh h) - f_t(h_0) -Df_t(h_{0})[\wh h - h_0] \right| \lesssim \|\partial^{\alpha_1} \wh h - \partial^{\alpha_1} h_0\|_\infty \|\partial ^{\alpha_2} \wh h - \partial^{\alpha_2} h_0\|_\infty\,~~and \\
(ii.2)& \; n^{1/2} \ul \sigma_n^{-1} \big( \|\partial^{\alpha_1}\wh h - \partial^{\alpha_1}h_0\|_\infty \|\partial ^{\alpha_2} \wh h - \partial^{\alpha_2} h_0\|_\infty + \|\partial^\alpha \widetilde h - \partial^\alpha h_0\|_\infty \big) = O_{p}(\eta_n)~;
\end{align*}
and (iii) $\sup_{t\in \mathcal T} \frac{(\tau_J)\sqrt{\sum_{j=1}^J \left( Df_{t}(\wh h)[(G_\psi^{-1/2} \psi^J)_j] - Df_{t}(h_0)[(G_\psi^{-1/2} \psi^J)_j] \right)^2}}{\sigma_n(f_{t})}=O_p(\eta_n')$.
\end{itemize}
\end{remark}

Condition (a')(i) is automatically satisfied by functionals of the form $f_t(h) = \partial^\alpha h(t)$ with $\mathcal T \subseteq \mathcal X$ and $\mathcal H_n = \mathcal H$. Conditions (a')(i) and (b')(i)(ii) are sufficient conditions that are formulated to take advantage of the sup-norm rate results in Section \ref{npiv sec}. For example, condition (b')(i)(ii.1) is easily satisfied by exact CS and DL functionals (lemma A.1 of \cite{HausmanNewey1995}). Condition (b')(ii.2) is simply satisfied by applying our sup-norm rate results.
Condition (b')(iii) is a sufficient condition for Assumption \ref{a-functional}(b)(iii), and is needed for uniform-in-$t$ consistent estimation of $\sigma_n(f_t)$ by $\wh\sigma (f_t)$ only, and is automatically satisfied with $\eta_n'=0$ when $f_t(\cdot)$ is a linear functional.

The next assumption concerns the set of normalized sieve 2SLS Riesz representers, given by
\[
 u_n(f_t)(x) = v_n(f_t)(x)/\sigma_n(f_t) \,.
\]
Let $d_n$ denote the semi-metric on $\mathcal T$ given by $d_n(t_1,t_2)^2 = E[(u_n(f_{t_1})(X_i) - u_n(f_{t_2})(X_i))^2]$ and $N(\mathcal T,d_n,\epsilon)$ be the $\epsilon$-covering number of $\mathcal T$ with respect to $d_n$. Let $\eta_n$ and $\eta_n'$ be from Assumption \ref{a-functional}, and $\delta_{h,n}$ be a sequence of positive constants such that $\|\wh h - h_0\|_\infty = O_p(\delta_{h,n})=o_p (1)$. Denote $\delta_{V,n} \equiv \big[ \zeta_{b,K}^{(2+\delta)/\delta} \sqrt{(\log K)/n} \big]^{\delta/(1+\delta)}  + \tau_J \zeta \sqrt{(\log J)/n} + \delta_{h,n}$.

\setcounter{assumption}{5}

\begin{assumption} \label{a-Lipschitz}
(i)  there is a sequence of finite constants $c_n \gtrsim 1$ that could grow to infinity such that
\[
 1 + \int_0^\infty \sqrt{\log N(\mathcal T,d_n,\epsilon)}\,\mathrm d \epsilon = O( c_n)\,;
\]
and (ii) there is a sequence of constants $r_n>0$ decreasing to zero slowly such that\\
(ii.1) $r_n c_n \lesssim 1$ and $\frac{\zeta_{b,K} J^2}{r_n^3 \sqrt n} = o(1)$; and\\
(ii.2) $\tau_J \zeta \sqrt{(J \log J)/n} + \eta_n + ( \delta_{V,n} + \eta_n' )\times c_n  = o(r_n)$,
with $\eta_n' \equiv 0$ when $f_t(\cdot)$ is linear.
\end{assumption}

Assumption \ref{a-Lipschitz}(i) is a mild regularity condition requiring that the class $\{u_n(f_t) : t \in \mathcal T\}$ not be too complex; see Remark \ref{rmk-a6suff} below for sufficient conditions to bound $c_n$. Assumption \ref{a-Lipschitz}(ii) strengthens conditions on the growth rate of $J$. Condition $\frac{\zeta_{b,K} J^2}{r_n^3 \sqrt n} = o(1)$ of Assumption \ref{a-Lipschitz}(ii.1) is used to apply Yurinskii's coupling \cite[Theorem 10, p. 244]{CLR,PollardUGMTP} to derive uniform Gaussian process strong approximation to the linearized sieve process $\{\wh {\mb Z}_n(t):t\in \mathcal T \}$ (defined in equation (\ref{bahadur-Z})). This condition could be improved if other types of strong approximation probability tools are used.
Assumption \ref{a-Lipschitz}(ii.2) ensures that both the \emph{nonlinear} remainder terms and the error in estimating $\sigma_n(f_t)$ by $\wh \sigma(f_t)$ vanish sufficiently fast. While the consistency of $\wh \sigma(f)$ is enough for the pointwise asymptotic normality of the plug-in sieve $t$-statistic for $f(h_0)$ (see Theorem \ref{t-dist} in the main online Appendix \ref{s-pw}), we need the following rate of convergence for uniform inference
\[
 \sup_{t \in \mathcal T} \left| \frac{\sigma_n(f_t)}{\wh\sigma (f_t)} -1 \right| = O_p( \delta_{V,n} + \eta_n')~,
\]
which is established using our results on sup-norm convergence rates of sieve NPIV; see Lemma \ref{lem-varest-u} in the secondary online Appendix \ref{ax-proofs}.

\begin{remark} \label{rmk-a6suff}
Let Assumptions \ref{a-data}(iii) and \ref{a-approx}(i) hold. Let $\mathcal T$ be a compact subset in $\mb R^{d_T}$, and there exist positive sequences $\Gamma_n$ and $\gamma_n$ such that for any $t_1, t_2 \in \mathcal T$,
\[
 \sup_{h \in \Psi_J : \|h\|_{L^2(X)} = 1} \left| \left( Df_{t_1}(h_0)[h] -  Df_{t_2}(h_0)[h] \right) \right| \leq \Gamma_n \|t_1 - t_2 \|_{\ell^2}^{\gamma_n}~.
\]
Then: Assumption \ref{a-Lipschitz}(i) holds with $c_n = 1 + \int_0^\infty \sqrt{ \{ (d_T /\gamma_n )\log ( \Gamma_n \tau_J/(\epsilon \ul \sigma_n ))\} \vee 0 }\,\mathrm d \epsilon$.
\end{remark}

The next lemma is about uniform Bahadur representation and uniform Gaussian process strong approximation for the sieve $t$-statistic process for (possibly) nonlinear functionals of NPIV. Define
\begin{align}
 \wh {\mb Z}_n(t) & =  \frac{(D f_t(h_0)[\psi^J])' [ S'  G_b^{-1}  S]^{-1}  S'  G_b^{-1/2}}{\sigma_n(f_t)} \left(\frac{1}{\sqrt n} \sum_{i=1}^n G_b^{-1/2} b^K(W_i) u_i\right)~,\label{bahadur-Z} \\
 \mb Z_n(t) & =  \frac{(D f_t(h_0)[\psi^J])' [ S'  G_b^{-1}  S]^{-1}  S'  G_b^{-1/2}}{\sigma_n(f_t)}  \mathcal Z_n~\nonumber
\end{align}
with $\mathcal Z_n \sim N(0, G_b^{-1/2}\Omega G_b^{-1/2} )$. Note that $\mb Z_n(t)$ is a Gaussian process indexed by $t \in \mc T$.

\begin{lemma}\label{lem-bahadur}
Let Assumptions \ref{a-data}(iii), \ref{a-residuals}, \ref{a-sieve}(ii)(iii), \ref{a-approx}(i), \ref{a-functional} and \ref{a-Lipschitz} hold. Then:
\begin{equation} \label{e:sa}
 \sup_{t \in \mc T} \left| \frac{\sqrt n (f_t(\wh h) - f_t(h_0))}{\wh\sigma (f_t)} - {\mb Z}_n(t) \right| = \sup_{t \in \mc T} \left| \frac{\sqrt n (f_t(\wh h) - f_t(h_0))}{\wh\sigma (f_t)} - \wh{\mb Z}_n(t) \right| + o_p(r_n ) = o_p(r_n)\,.
\end{equation}
\end{lemma}

Lemma \ref{lem-bahadur} is used in this paper to establish the consistency of the sieve score bootstrap for estimating the critical values of the uniform sieve $t$-statistic process, $\sup_{t \in \mathcal T} \left|\frac{\sqrt n (f_t(\wh h) - f_t(h_0))}{\wh\sigma (f_t)} \right|$, for a NPIV model. The strong approximation result, however, is also useful for various applications to testing equality and/or inequality (such as shape) constraints on $f_t(h_0)$, and is therefore of independent interest.

In what follows, $\mb P^*(\cdot)$ denotes a probability measure conditional on the data $Z^n:= \{(X_i,Y_i,W_i)\}_{i=1}^n$. Recall that $\mb Z_n^*(t)$ is defined in equation (\ref{bscore}).

\begin{theorem} \label{t-dist-u}
Let conditions of Lemma \ref{lem-bahadur} hold. Let $\eta_n' \sqrt J = o(r_n)$ for nonlinear $f_t()$. Let the bootstrap weights $\{\varpi_{i}\}_{i=1}^n$ be IID with zero mean, unit variance and finite 3rd moment, and independent of the data. Then:
\begin{equation}\label{boots-ucb}
 \sup_{s \in \mb R} \left| \mb P \left( \sup_{t \in \mathcal T} \left|\frac{\sqrt n (f_t(\wh h) - f_t(h_0))}{\wh\sigma (f_t)} \right| \leq s \right) - \mb P^*\left( \sup_{t \in \mathcal T} |\mb Z_n^*(t)| \leq s\right)  \right| = o_p(1)\,.
\end{equation}
\end{theorem}

Theorem \ref{t-dist-u} appears to be the first to establish consistency of a sieve score bootstrap for uniform inference on general nonlinear functionals of NPIV under low-level conditions. When specializing to collections of linear functionals,
Lemma \ref{lem-bahadur}, Theorem \ref{t-dist-u} and Corollary \ref{c-deriv} immediately imply the following result.

\begin{corollary} \label{cor-ucb}
Consider a collection of linear functionals $\{f_t(h_0) =\partial^\alpha h_0 (t): t\in \mathcal T \}$ of the NPIV function $h_0$, with $\mathcal T$ a compact convex subset of $\mathcal X$. Let Assumptions \ref{a-data}(i)(ii)(iii) and \ref{a-residuals} (with $\delta \geq 1$) hold, $h_0 \in B_\infty(p,L)$, $\Psi_J$ be formed from a B-spline basis of regularity $\gamma > (p \vee 2 + |\alpha|)$, $B_K$ be a B-spline, wavelet or cosine basis, and $\sigma_n(f_t) \asymp \tau_J J^{a}$ uniformly in $t$ with $a =\frac{1}{2} + \frac{|\alpha|}{d}$. For $\kappa \in [1/2,1]$ we set $J^5(\log n)^{6\kappa} /n = o(1)$, $\tau_J J (\log J)^{\kappa+0.5} / \sqrt{n}= o(1)$ and $J^{-p/d}=o([\log J]^{-\kappa} \tau_J \sqrt{J/n})$. Then: Results (\ref{e:sa}) (with $r_n=(\log J)^{-\kappa}$) and (\ref{boots-ucb}) hold for $f_t(h_0) =\partial^\alpha h_0 (t)$.
\end{corollary}

Recently \cite{HorowitzLee2012} developed a notion of UCBs for a NPIV function $h_0$ of a scalar endogenous regressor $X_i \in [0,1]$ based on interpolation over a growing number of uniformly generated random grid points on $[0,1]$, with $h_0$ estimated via the modified orthogonal series NPIV estimator of \cite{Horowitz2011}.\footnote{Remark 4 in \cite{HorowitzLee2012} mentioned that their notion of UCB is different from the standard UCBs. They also proved the consistency of their bootstrap confidence bands over fixed finite number of grid points.}
When specializing Corollary \ref{cor-ucb} to a NPIV function of a scalar regressor (i.e., $d=1$ and $|\alpha|=0$), our sufficient conditions are comparable to theirs (see their theorem 4.1). Our score bootstrap UCBs would be computationally much simpler for a NPIV function of a multivariate endogenous regressor $X_i$, however.

When $X_i$ is exogenous, the sieve NPIV estimator $\wh h$ reduces to the series LS estimator of a nonparametric regression $h_0(x) = E[Y_i | W_i = x]$ with $X_i = W_i$, $K =J$ and $b^K = \psi^J$ with $\tau_J =1$. Lemma \ref{lem-bahadur} and Theorem \ref{t-dist-u} immediately imply the validity of Gaussian strong approximation and sieve score bootstrap UCBs for collections of general nonlinear functionals of a nonparametric LS regression. We note that the regularity conditions in Lemma \ref{lem-bahadur} and Theorem \ref{t-dist-u} are much weaker for models with exogenous regressors. For instance, when specializing Corollary \ref{cor-ucb} to a nonparametric LS regression with exogenous regressor $X_i$, the conditions on $J$ simplify to $J^5(\log n)^{6\kappa} /n = o(1)$ and $J^{-p/d}=o([\log J]^{-\kappa}\sqrt{J/n} )$ for $\kappa \in [1/2,1]$, and Results (\ref{e:sa}) (with $r_n = [\log J]^{-\kappa}$) and (\ref{boots-ucb}) both hold for linear functionals $\{f_t(h_0) = \partial^\alpha h(t_0 ): t\in \mathcal T \}$ of $h_0(\cdot) =E[Y_i | X_i = \cdot]$. These conditions on $J$ are the same as those in \cite{CLR} for $h_0$ (see their theorem 7) and \cite{BCCK2014} for linear functionals of $h_0$ (see their theorem 5.5 with $r_n = [\log J]^{-1/2}$) estimated via series LS.

To the best of our knowledge, there is no published work on uniform Gaussian process strong approximation and sieve score bootstrap for general \emph{nonlinear functionals} of sieve NPIV or series LS regression. The results in this section are thus presented as non-trivial applications of our sup-norm rate results for sieve NPIV, and are not aimed at weakest sufficient conditions.

\subsection{Monte Carlo}

We now evaluate the finite sample performance of our sieve score bootstrap UCBs for $h_0$ in NPIV model (\ref{npreg}). We use the experimental design of \cite{NeweyPowell}, in which IID draws are generated from
\[
 \left( \begin{array}{c} u_i \\ V_i^* \\ W_i^* \end{array} \right) \sim N \left( \left( \begin{array}{c} 0 \\ 0 \\ 0 \end{array} \right) , \left( \begin{array}{ccc} 1 & 0.5 & 0 \\ 0.5 & 1 & 0 \\ 0 & 0 & 1 \end{array} \right) \right)
\]
from which we then set $X_i^* = W_i^* + V_i^*$. To ensure compact support of the regressor and instrument, we rescale $X_i^*$ and $W_i^*$ by defining $X_i = \Phi(X_i^*/\sqrt 2)$ and $W_i = \Phi(W_i^*)$ where $\Phi$ is the Gaussian cdf. We use $h_0(x) = 4x-2$ for our \emph{linear} design and $h_0(x) = \log(|16x-8|+1)\mathrm{sgn}(x-\frac{1}{2})$ for our \emph{nonlinear} design (our nonlinear $h_0$ is a re-scaled version of the $h_0$ used in \cite{NeweyPowell}). Note that $p$ for the nonlinear $h_0$ is between $1$ and $2$, so $h_0$ is not particularly smooth ($h_0'(x)$ has a kink at $x = \frac{1}{2}$).

We generate 1000 samples of length 1000 and implement our procedure using a B-spline basis for $B_K$ and $\Psi_J$. For each simulation, we calculate the 90\%, 95\%, and 99\% uniform confidence bands for $h_0$ over the support $[0.05,0.95]$ with 1000 bootstrap replications for each simulation. We draw the bootstrap innovations $\varpi_i$ from the two-point distribution of \cite{Mammen1993}. We then calculate the MC coverage probabilities of our uniform confidence bands.

\begin{table}[p]
 \centering {
\begin{tabular}{cccc|cccccc}
\hline \hline
 & & & &  \multicolumn{3}{c}{Design 1: Linear $h_0$} & \multicolumn{3}{c}{Design 2: nonlinear $h_0$} \\
  $\Psi_J$ & $B_K$ & $J$ & $K$ & 90\% CI & 95\% CI & 99\% CI & 90\% CI & 95\% CI & 99\% CI \\ \hline
  C  &  C  & 5 & 5 & 0.962 & 0.983 & 0.996 & 0.896 & 0.942 & 0.987 \\
  C  &  C  & 5 & 6 & 0.957 & 0.983 & 0.996 & 0.845 & 0.924 & 0.981 \\
  C  &  Q  & 5 & 5 & 0.961 & 0.982 & 0.996 & 0.884 & 0.939 & 0.985 \\
  C  &  Q  & 5 & 6 & 0.958 & 0.983 & 0.997 & 0.846 & 0.921 & 0.981 \\
  Q  &  Q  & 5 & 5 & 0.964 & 0.984 & 0.997 & 0.913 & 0.948 & 0.989 \\
  Q  &  Q  & 5 & 6 & 0.961 & 0.985 & 0.996 & 0.886 & 0.937 & 0.983 \\
   \hline \hline
\end{tabular} \parbox{5.0in}{\caption{ \label{tab-ucb-1} \small MC coverage probabilities of uniform confidence bands for $h_0$. Results are presented for cubic (C) and quartic (Q) B-spline bases for $\Psi_J$ and $B_K$.   } }}
\end{table}

\begin{figure}[p]
  \centering {
  \includegraphics[trim = 20mm 70mm 20mm 75mm, clip, width=0.7\textwidth]{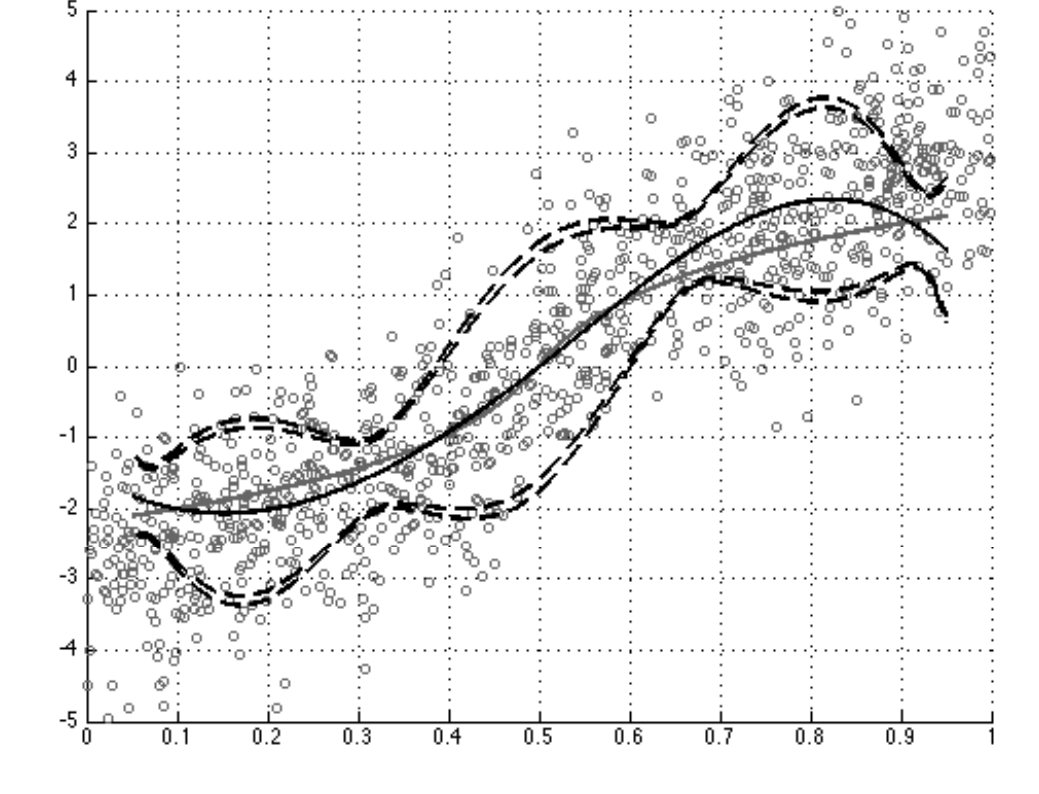}
  \parbox{5.0in}{\caption{ \label{f-ucb} \small 90\% and 95\% uniform confidence bands for $h_0$ (dashed lines; innermost are 90\%), NPIV estimate $\wh h$ (solid black line), true structural function $h_0$ (solid grey line) for the nonlinear design.} }}
\end{figure}

Figure \ref{f-ucb} displays the estimated structural function $\wh h$ and confidence bands together with a scatterplot of the sample $(X_i,Y_i)$ data for the nonlinear design. The true function $h_0$ is seen to lie inside the UCBs.
The results of this MC experiment are presented in Table \ref{tab-ucb-1}. Comparing the MC coverage probabilities with their nominal values, it is clear that the uniform confidence bands for the linear design are slightly too conservative. However, the uniform confidence bands for the nonlinear design using cubic B-splines to approximate $h_0$ have MC converge much closer to the nominal coverage probabilities.

\section{Pointwise and uniform inference on nonparametric welfare functionals}\label{s-inference}

We now apply our sup-norm rate results to study pointwise and uniform inference on nonlinear welfare functionals in nonparametric demand estimation with endogeneity. First, we provide mild sufficient conditions under which plug-in sieve $t$-statistics for exact CS and DL and approximate CS functionals are asymptotically $N(0,1)$, allowing for mildly and severely ill-posed NPIV models (subsections \ref{ss-cs-endog} and \ref{ss-apcs}). Second, under stronger sufficient conditions but still allowing for severely ill-posed NPIV models, the validity of uniform Gaussian process strong approximations and sieve score bootstrap UCBs for exact CS and DL over a range of taxes and/or incomes (subsection \ref{s-ucb-cs}) are presented. When specialized to inference on exact CS and DL and approximate CS functionals of nonparametric demand estimation without endogeneity, our pointwise asymptotic normality results are valid under sufficient conditions weaker than those in the existing literature, while our uniform inference results appear to be new (subsection \ref{ss-cs-exog}).

Previously, \cite{HausmanNewey1995} and \cite{Newey1997} provided sufficient conditions for pointwise asymptotic normality for plug-in nonparametric LS estimators of exact CS and DL functionals and of approximate CS functionals respectively, when prices and incomes are exogenous. \cite{Vanhems2010} studied consistency and convergence rates of kernel-based plug-in estimators of CS functional allowing for mildly ill-posed NPIV models. \cite{BlundellHorowitzParey2012} and \cite{HausmanNewey2016} estimated CS and DL of nonparametric gasoline demand allowing for prices to be endogenous, but did not provide theoretical justification for their inference approach under endogeneity.
Therefore, although presented as applications of our sup-norm rate results, our inference results contribute nicely to the literature on nonparametric welfare analysis.

\subsection{Pointwise inference on exact CS and DL with endogeneity} \label{ss-cs-endog}

Here we present primitive regularity conditions for pointwise asymptotic normality of the sieve $t$-statistics for exact CS and DL. We suppress dependence of the functionals on $t = (\mf p^0, \mf p^1,\mf y)$.

Let $\mf X_i = (\mf P_i,\mf Y_i)$. We assume in what follows that the support of both $\mf P_i$ and $\mf Y_i$ is bounded away from zero. If both $\mf P_i$ and $\mf Y_i$ are endogenous, let $\mf W_i$ be a $2\times 1$ vector of instruments. Let $T : L^2(\mf X) \to L^2(\mf W)$ be compact and injective with singular value decomposition (SVD) $\{\phi_{0j},\phi_{1j},\mu_j\}_{j=1}^\infty$ where
\[
 T \phi_{0j} = \mu_j \phi_{1j}, \quad (T^*T)^{1/2} \phi_{0j} = \mu_j \phi_{0j}, \quad (TT^*)^{1/2} \phi_{1j} = \mu_j \phi_{1j}
\]
and $\{\phi_{0j}\}_{j=1}^\infty$ and $\{\phi_{0j}\}_{j=1}^\infty$ are orthonormal bases for $L^2(\mf X)$ and $L^2(\mf W)$, respectively. If $\mf P_i$ is endogenous but $\mf Y_i$ is exogenous, we take $\mf W_i = (\mf W_{1i},\mf Y_i)'$ with $\mf W_{1i}$ an instrument for $\mf P_i$. Let $T_{\mf y} : L^2(\mf P|\mf Y = \mf y) \to L^2(\mf W_1|\mf Y=\mf y)$ be compact and injective with SVD $\{\phi_{0j,\mf y},\phi_{1j,\mf y},\mu_{j,\mf y}\}_{j=1}^\infty$ for each $\mf y$ where
\[
 T_{\mf y} \phi_{0j,\mf y} = \mu_{j,\mf y} \phi_{1j,\mf y}, \quad (T^*_{\mf y}T^{\phantom *}_{\mf y})^{1/2} \phi_{0j,\mf y} = \mu_{j,\mf y} \phi_{0j,\mf y}, \quad (T^{\phantom *}_{\mf y}T^*_{\mf y})^{1/2} \phi_{1j,\mf y} = \mu_{j,\mf y} \phi_{1j,\mf y}
\]
and $\{\phi_{0j,\mf y}\}_{j=1}^\infty$ and $\{\phi_{0j,\mf y}\}_{j=1}^\infty$ are orthonormal bases for $L^2(\mf P|\mf Y = \mf y)$ and $L^2(\mf W_1|\mf Y = \mf y)$, respectively. In this case, we define $\phi_{0j}(\mf p,\mf y) = \phi_{0j,\mf y}(\mf p)$, $\phi_{1j}(\mf w_1,\mf y) = \phi_{1j,\mf y}(\mf w_1)$, and $\mu_j^2 = E[\mu_{j,\mf Y_i}^2]$ (see Section \ref{s-exog} for further details).

In both cases, we follow \cite{ChenPouzo2014} and assume that $\Psi_J$ and $B_K$ are Riesz bases in that they span $\phi_{01},\ldots,\phi_{0J}$ and $\phi_{11},\ldots,\phi_{1K}$, respectively. This implies that $\tau_J \asymp \mu_J^{-1}$.
For fixed $\mf p^0$, $\mf p^1$, and $\mf y$ we define
\begin{equation*}
 a_j = a_j(\mf p^0,\mf p^1,\mf y) = \int_{0}^{1} \left( \phi_{0j}(\mf p(u),\mf y-\mf S_{\mf y}(\mf p(u))) e^{-\int_{0}^u \partial_2 h_0(\mf p(v),\mf y-\mf S_{\mf y}(\mf p(v)))\mf p'(v)\,\mathrm dv} \mf p'(u) \right)\mathrm du \,
\end{equation*}
for the exact CS functional.

\paragraph{Assumption CS} \emph{(i) $\mf X_i$ and $\mf W_i$ both have compact rectangular support and densities bounded away from $0$ and $\infty$; (ii) $h_0 \in B_{\infty}(p,L)$ with $p > 2$ and $0 < L < \infty$; (iii) $E[\mf u_i^2|\mf W_i = w]$ is uniformly bounded away from $0$ and $\infty$, $E[|\mf u_i|^{2+\delta}]$ is finite for some $\delta>0$, and $
\sup_w E[u_i^2 \{ |u_i | > \ell(n)\}|\mf W_i = w] = o(1)$ for any positive sequence with $\ell(n) \nearrow \infty$; (iv) $\Psi_J$ is spanned by a (tensor-product) B-spline basis of order $\gamma > p$ or continuously differentiable wavelet basis of regularity $\gamma > p$ and $B_K$ is spanned by a (tensor-product) B-spline, wavelet or cosine basis; (v) $J^{(2+\delta)/(2\delta)}\sqrt{(\log n) / n} = o(1)$ and
\[
 \frac{\sqrt n}{ \big( \sum_{j=1}^J (a_j/\mu_j)^2 \big)^{1/2} } \times \bigg( J^{-p/2} +  \mu_J^{-2} \frac{J^2 \sqrt{\log J}}{n}  \bigg) = o(1)\,.
\]
}

Assumption CS(i)--(iv) is standard even for series LS regression without endoegenity. Let $[\sigma_n(f_{CS})]^2   = \big(Df_{CS} (h_{0})[\psi^J]\big)' \mho  \big(Df_{CS} (h_{0})[\psi^J]\big)$ be the sieve variance of the plug-in sieve NPIV estimator $f_{CS} (\wh h_0 )$. Then these assumptions imply that $[\sigma_n(f_{CS})]^2 \asymp \sum_{j=1}^J (a_j/\mu_j)^2 \lesssim J \mu_J^{-2}$. Assumption CS(v) is sufficient for Remark \ref{rmk-a5suff}(b') for a fixed $t$.

Our first result is for exact CS functionals, established by applying Theorem \ref{t-dist} in the main online Appendix \ref{s-pw}. Let
\[
 \wh \sigma^2 (f_{CS})  = Df_{CS}(\wh h)[\psi^J]'\, \wh \mho\, Df_{CS}(\wh h)[\psi^J]
\]
with
\[
 Df_{CS}(\wh h)[\psi^J]  = \int_{0}^{1}  \psi^J( \mf p(u),\mf y-\wh{\mf S}_{\mf y}(\mf p(u)))  e^{-\int_{0}^u \partial_2 \wh h( \mf p(v),\mf y-\wh{\mf S}_{\mf y}(\mf p(v))) \mf p'(v)\,\mathrm dv} \mf p'(u) \,\mathrm du \,.
\]

\begin{theorem}\label{t-cs}
Let Assumption CS hold. Then: the sieve $t$-statistic for $f_{CS}(h_0)$ is asymptotically $N(0,1)$, i.e.,
\[
 \sqrt n\frac{ f_{CS}(\wh h) - f_{CS}(h_0)}{\wh \sigma (f_{CS})} \to_d N(0,1)\,.
\]
\end{theorem}

Since $\mu_j >0$ decreases as $j$ increases, we could use the following relation
\begin{equation}\label{sigma-lb}
\mu_J^{-2} J  \gtrsim \mu_J^{-2}\sum_{j=1}^J a_j^2 \geq \sum_{j=1}^J (a_j/\mu_j)^2 \geq \max \left( (\min_{1\leq j\leq J} a_j^2) \sum_{j=1}^J \mu_j^{-2}, \max_{1\leq j\leq J} ( a_j^2 \mu_j^{-2} ), \mu_1^{-2}\sum_{j=1}^J a_j^2 \right)
\end{equation}
to provide simpler sufficient conditions for Assumption CS(v) that could be satisfied by both mildly and severely ill-posed NPIV models. Corollary \ref{c-cs} provides one set of concrete sufficient conditions for Assumption CS(v).

\begin{corollary}\label{c-cs}
Let Assumption CS(i)--(iv) hold and $a_j^2 \asymp j^{a}$ for $a\leq 0$. Then: $[\sigma_n(f_{CS})]^2 \asymp \sum_{j=1}^J (j^{a} \mu_j^{-2})$.

(1) Mildly ill-posed case: let $\mu_j \asymp j^{-\varsigma/2}$ for $\varsigma \geq 0, a + \varsigma > -1$. Then:
\[
[\sigma_n(f_{CS})]^2 \asymp J^{(a+\varsigma)+1}~;
\]
further, if $\delta \geq 2/(2+\varsigma- a)$, $n J^{-(p+a+\varsigma+1)} = o(1)$ and $J^{3+\varsigma - a} (\log n)/n = o(1)$, then: Assumption CS(v) is satisfied, and the sieve $t$-statistic for $f_{CS}(h_0)$ is asymptotically $N(0,1)$.

(2) Severely ill-posed case: let $\mu_j \asymp \exp(-\frac{1}{2} j^{\varsigma/2})$, $\varsigma > 0$ and $J =( \log (n/(\log n)^{\varrho}))^{2/\varsigma}$ for $\varrho>0$. Then:
\[
 [\sigma_n(f_{CS})]^2 \gtrsim \frac{n}{(\log n)^{\varrho}} \times (\log (n/(\log n)^\varrho))^{2a/\varsigma}\,;
\]
further, if $\varrho>0$ is chosen such that $2 p > \varrho \varsigma - 2a$ and $\varrho \varsigma > 8-2a $, then: Assumption CS(v) is satisfied, and the sieve $t$-statistic for $f_{CS}(h_0)$ is asymptotically $N(0,1)$.
\end{corollary}

Note that in Corollary \ref{c-cs}, $J$ may be chosen to satisfy the stated conditions in the mildly ill-posed case whenever $p > 2-2a$, and in the severely ill-posed case whenever $p > 4 - 2a$.

Our next result is for DL functionals. Note that DL is the sum of CS and a tax receipts functional, namely $(\mf p^1 - \mf p^0)h_0(\mf p^1,\mf y)$. Note that the tax receipts functional is typically less smooth and hence converges slower than that of CS functional.
Therefore, $[\sigma_n(f_{DL})]^2   = \big(Df_{DL}(h_{0})[\psi^J]\big)' \mho  \big(Df_{DL}(h_{0})[\psi^J]\big)$ will typically grow at the order of $(\tau_J \sqrt J)^2$, which is the growth order of the sieve variance term for estimating the unknown NPIV function $h_0$ at a fixed point. For this reason we do not derive the joint asymptotic distribution of $f_{CS}(\wh h)$ and $f_{DL}(\wh h)$. The next result adapts Theorem \ref{t-cs} to derive asymptotic normality of plug-in sieve $t$-statistics for DL functionals. Let
\[
 \wh \sigma^2 (f_{DL})  = Df_{DL}(\wh h)[\psi^J]'\, \wh \mho\, Df_{DL}(\wh h)[\psi^J]
\]
with
\[
 Df_{DL}(\wh h)[\psi^J]  = Df_{CS}(\wh h)[\psi^J] - (\mf p^1 - \mf p^0) \psi^J(\mf p^1,\mf y)\,.
\]

\begin{theorem} \label{t-dl}
Let Assumption CS(i)--(iv) hold. Let $\sigma_n (f_{DL}) \asymp \mu_J^{-1} \sqrt J$, $\sqrt n \mu_J J^{-(p+1)/2} = o(1)$ and $(J^{(2+\delta)/(2\delta)}\sqrt{\log n} \vee \mu_J^{-1}J^{3/2}\sqrt{\log J}) /\sqrt n= o(1)$. Then:
\[
 \sqrt n\frac{ f_{DL}(\wh h) - f_{DL}(h_0)}{\wh \sigma (f_{DL})} \to_d N(0,1)\,.
\]
\end{theorem}

\subsection{Pointwise inference on approximate CS with endogeneity} \label{ss-apcs}

Suppose instead that demand of consumer $i$ for some good is estimated in logs, i.e.
\begin{equation} \label{e-logQ}
\log \mf Q_i = h_0(\log \mf P_i ,\log \mf Y_i) + u_i \,.
\end{equation}
As $h_0$ is the log-demand function, any linear functional of demand is a nonlinear functional of $h_0$. One such example is the weighted average demand functional of the form
\begin{equation*}
 f_A(h) = \int w(\mf p) e^{h(\log \mf p,\log \mf y)}\, \mathrm d \mf p
\end{equation*}
where  $w(\mf p)$ is a non-negative weighting function and $\mf y$ is fixed. With $w(\mf p) = \ind \{ \ul {\mf p} \leq \mf p \leq \ol{\mf p}\}$, the functional $f(h)$ may be interpreted as the approximate CS. The functional is  defined for fixed $\mf y$, so it will typically be an irregular functional of $h_0$.

The setup is similar to the previous subsection. Let $\mf X_i = (\log \mf P_i,\log \mf Y_i)$. If both $\mf P_i$ and $\mf Y_i$ are endogenous, we let $\mf W_i$ be a $2\times 1$ vector of instruments and $T : L^2(\mf X) \to L^2(\mf W)$ be compact with SVD $\{\phi_{0j},\phi_{1j},\mu_j\}_{j=1}^\infty$. If $\mf P_i$ is endogenous but $\mf Y_i$ is exogenous, we let $\mf W_i = (\mf W_{1i},\log \mf Y_i)'$ with $\mf W_{1i}$ an instrument for $\mf P_i$, and let $T_{\mf y} : L^2(\log \mf P|\log \mf Y = \log \mf y) \to L^2(\mf W_1|\log \mf Y=\log \mf y)$ be compact with SVD $\{\phi_{0j,\mf y},\phi_{1j,\mf y},\mu_{j,\mf y}\}_{j=1}^\infty$ for each $\mf y$. In this case, we define $\phi_{0j}(\log \mf p,\log \mf y) = \phi_{0j,\mf y}(\log \mf p)$, $\phi_{1j}(\mf w_1,\log \mf y) = \phi_{1j,\mf y}(\mf w_1)$, and $\mu_j^2 = E[\mu_{j,\mf Y_i}^2]$.
We again assume that $\Psi_J$ and $B_K$ are Riesz bases. For each $j \geq 1$, define
\begin{align*}
 a_j & = a_j(\mf y) = \int w(\mf p) e^{h_0(\log \mf p,\log \mf y)}\phi_{0j}(\log \mf p,\log \mf y) \, \mathrm d \mf p \,.
\end{align*}

The next result follows from Theorem \ref{t-dist} (in the main online Appendix \ref{s-pw}). Let
\[
 \wh \sigma^2 (f_A)  = Df_{A}(\wh h)[\psi^J]'\, \wh \mho\, Df_{A}(\wh h)[\psi^J]
\]
with
\[
 Df_A(\wh h)[\psi^J]  = \int w(\mf p) e^{\wh h(\log \mf p,\log \mf y)}\psi^J(\log \mf p,\log \mf y) \, \mathrm d \mf p \,.
\]

\begin{theorem} \label{t-apcs} Let Assumption CS(i)--(iv) hold for the log-demand model (\ref{e-logQ}) with $p > 0$, and let $ J^{(2+\delta)/(2\delta)}\sqrt{(\log n)/n} = o(1)$ and
\[
 \frac{\sqrt n}{ \big( \sum_{j=1}^J (a_j/\mu_j)^2 \big)^{1/2} } \times \bigg( J^{-p/2} + \mu_J^{-2} \frac{J^{3/2} \sqrt{\log J}}{n}   \bigg) = o(1)~.
\]
Then:
\[
 \frac{\sqrt n ( f_A(\widehat h) - f_A(h_0))}{ \wh \sigma (f_A)} \to_d N(0,1)\,.
\]
\end{theorem}

\subsection{Uniform inference on collections of exact CS and DL functionals with endogeneity}\label{s-ucb-cs}

Here we apply Lemma \ref{lem-bahadur} and Theorem \ref{t-dist-u} to present sufficient conditions for uniform Gaussian process strong approximations and bootstrap UCBs for exact CS and DL under endogeneity. We maintain the setup described at the beginning of Subsection \ref{ss-cs-endog}. We take $t = (\mf p^0,\mf p^1,\mf y) \in \mc T = [\ul {\mf p}^0 , \ol {\mf p}^0] \times [\ul {\mf p}^1 , \ol {\mf p}^1] \times [\ul{\mf y}, \ol{\mf y}]$, where the intervals $[\ul {\mf p}^0 , \ol {\mf p}^0]$ and $[\ul {\mf p}^1 , \ol {\mf p}^1]$ are in the interior of the support of $\mf P_i$ and $[\ul{\mf y}, \ol{\mf y}]$ is in the interior of the support of $\mf Y_i$. For each $t \in \mc T$ we let
\begin{equation}
 a_{j,t} = a_{j,t}(\mf p^0,\mf p^1,\mf y) = \int_{0}^{1} \left( \phi_{0j}(\mf p(u),\mf y-\mf S_{\mf y}(\mf p(u))) e^{-\int_{0}^u \partial_2 h_0(\mf p(v),\mf y-\mf S_{\mf y}(\mf p(v)))\mf p'(v)\,\mathrm dv} \mf p'(u) \right)\mathrm du
\end{equation}
for each $j \geq 1$ (where $\mf p(u)$ is a smooth price path from $\mf p^0 = \mf p(0)$ to $\mf p^1 = \mf p(1)$). Also define $\ul \sigma_n = \inf_{t \in \mc T}(( \sum_{j=1}^J (a_{j,t}/\mu_j)^2 )^{1/2} $.

\paragraph{Assumption U-CS} \emph{(i) $E[\mf u_i^2|\mf W_i = w]$ is uniformly bounded away from $0$, $E[|\mf u_i|^{2+\delta}]$ is finite with $\delta \geq 1$, and $\sup_w E[|\mf u_i|^3 |W_i = w]$ is finite; (ii) the H\"older condition in Remark \ref{rmk-a6suff} holds with $\gamma_n = \gamma$ and $\Gamma_n \lesssim J^c$ for some finite positive constants $\gamma$ and $c$; (iii) $J^5(\log n)^{3} /n = o(1)$, $\frac{\sqrt {n(\log J)}}{ \ul \sigma_n } J^{-p/2} =o(1)$; (iv) let $\eta_n'  =  \frac{ J^{3/2} \mu_J^{-1}}{ \ul \sigma_n } \left( J^{-p/2} + \mu_J^{-1} \sqrt{J(\log J)/n} \right)$, either (iv.1) $\eta_n' (\log J)=o(1)$, or (iv.2) $\eta_n' \sqrt{J(\log J)} =o(1)$.
}

Assumption U-CS (i) is slightly stronger than Assumption CS(iii) (since $\delta = 1$ in Assumption U-CS(i) is enough). Assumption U-CS(ii) is made for simplicity to verify Assumption \ref{a-Lipschitz}(i); other sufficient conditions could also be used. Assumption U-CS(iii)(iv.1) strengthens Assumption CS(v) to ensure uniform Gaussian process strong approximation with an error rate of $r_n =(\log J)^{-1/2}$. Again, one could use bounds on $\ul \sigma_n$ that is analogous to Relation (\ref{sigma-lb}) to provide sufficient conditions for Assumption U-CS(iii)(iv) that could be satisfied by mildly and severely ill-posed NPIV models. See Remark \ref{rm-ucs} below for one concrete set of such sufficient conditions.

\begin{remark}\label{rm-ucs}
Let $\ul \sigma_n^2 \gtrsim \sum_{j=1}^J (j^{a}\mu_j^{-2})$ for $a \leq 0$.\\
(1) Mildly ill-posed case: let $\mu_j \asymp j^{-\varsigma/2}$ for $\varsigma \geq 0$ and $ a+\varsigma  > -1$. Let $J^{5 \vee (4 + \varsigma-a)} (\log n)^{3} /n = o(1)$ and $nJ^{-(p+a+\varsigma+1)}{(\log J)} = o(1)$. Then Assumption U-CS(iii)(iv) holds.\\
(2) Severely ill-posed case: let $\mu_J \asymp \exp(-\frac{1}{2}j^{\varsigma/2})$, $\varsigma > 0$. Let
%$p > 4-a$ and
$J = (\log(n/(\log n)^\varrho))^{2/\varsigma}$ with $\varrho > 0$ chosen such that $2p > \varrho \varsigma - 2a$ and $\varrho \varsigma > 10-2a$. Then Assumption U-CS(iii)(iv) holds.
\end{remark}

The next results are about the uniform Gaussian process strong approximation and validity of score bootstrap UCBs for exact CS and DL functionals.

\begin{theorem}\label{t-ucb-cs} Let Assumptions CS(i)(ii)(iv) and U-CS(i)(ii)(iii) hold. Then:\\
(1) If Assumption U-CS(iv.1) holds, then Result (\ref{e:sa}) (with $r_n=(\log J)^{-1/2}$) holds for $f_t = f_{CS,t}$;\\
(2) If Assumption U-CS(iv.2) holds, then Result (\ref{boots-ucb}) also holds for $f_t = f_{CS,t}$.
\end{theorem}

In the next theorem the condition $\ul \sigma_n \asymp \mu_J^{-1} \sqrt J$ is implied by the assumption that $\sigma_n (f_{DL,t}) \asymp \mu_J^{-1} \sqrt J$ uniformly for $t \in \mc T$, which is reasonable for the DL functional.

\begin{theorem}\label{t-ucb-dl}
Let Assumptions CS(i)(ii)(iv) and U-CS(i)(ii)(iii) hold with $\ul \sigma_n \asymp \mu_J^{-1} \sqrt J$. Then:\\
(1) If Assumption U-CS(iv.1) holds, then Result (\ref{e:sa}) (with $r_n=(\log J)^{-1/2}$) holds for $f_t = f_{DL,t}$;\\
(2) If Assumption U-CS(iv.2) holds, then Result (\ref{boots-ucb}) also holds for $f_t = f_{DL,t}$.
\end{theorem}

\subsection{Inference on welfare functionals without endogeneity}\label{ss-cs-exog}

This subsection specializes the pointwise and uniform inference results for welfare functionals from the preceding subsections to nonparametric demand estimation with exogenous price and income. Precisely, we let $\mf X_i = \mf W_i$, $J = K$, $b^K = \psi^J$, $\mu_J \asymp 1$, $\tau_J \asymp 1$ and so the sieve NPIV estimator reduces to the usual series LS estimator of $h_0 (x)=E[Y_i |W_i=x]$.

The next two corollaries are direct consequences of our Theorems \ref{t-cs}, \ref{t-dl} and \ref{t-apcs} for pointwise asymptotic normality of sieve $t$ statistics for exact CS and DL and approximate CS functionals under exogeneity, and hence the proofs are omitted.

\begin{corollary}\label{c-cs-exog}
Let Assumption CS(i)--(iv) hold with $\mf X_i = \mf W_i$, $J = K$, $b^K = \psi^J$ and $\mu_J \asymp 1$ and let $\sum_{j=1}^J a_j^2 \gtrsim J^{a+1}$ with $ 0\geq a \geq -1$.

(1) Let $n J^{-(p+a+1)} = o(1)$, $J^{3-a} (\log J)/n = o(1)$, and $\delta \geq 2/(2-a)$. Then: the sieve $t$-statistic for $f_{CS}(h_0)$ is asymptotically $N(0,1)$.

(2) Let $n J^{-(p+1)} = o(1)$, $J^3 (\log J)/n = o(1)$, and $a=0$, $\delta \geq 1$. Then: the sieve $t$-statistic for $f_{DL}(h_0)$ is asymptotically $N(0,1)$.
\end{corollary}

Previously \cite{HausmanNewey1995} established the pointwise asymptotic normality of $t$-statistics for exact CS and DL based on plug-in kernel LS estimators of demand without endogeneity. They also established root-$n$ asymptotic normality of t-statistics for \emph{averaged} exact CS and DL (i.e. CS/DL averaged over a range of incomes) based on plug-in power series LS estimator of demand without endogeneity, under some regularity conditions including that $\sup_{x} E[|u_i|^4|\mf X_i = x]<\infty$ (which, in our notation, implies $\delta = 2$), $p =\infty$ (i.e., $h_0$ is infinitely times differentiable) and $J^{22}/n = o(1)$. Corollary \ref{c-cs-exog} complements their work by providing conditions for the pointwise asymptotic normality of exact CS and DL functionals based on spline and wavelet LS estimators of demand.

\begin{corollary} \label{c-apcs-exog}
Let Assumption CS(i)--(iv) hold for the log-demand model (\ref{e-logQ}) with $\mf X_i = \mf W_i$, $J = K$, $b^K = \psi^J$,  $\mu_J \asymp 1$ and $p > 0$, let $\sum_{j=1}^J a_j^2 \gtrsim J^{c+1}$ with $0\geq c \geq -1$. Let $n J^{-(p+c+1)} = o(1)$, $ J^{2-c} (\log J)/n = o(1)$ and $\delta \geq 2/(1-c)$. Then: the sieve $t$ statistic for $f_{A}(h_0)$ is asymptotically $N(0,1)$.
\end{corollary}

Previously \cite{Newey1997} established the pointwise asymptotic normality of $t$-statistics for \emph{approximate} CS functionals based on plug-in series LS estimators of exogenous demand under some regularity conditions including that $\sup_{x} E[|u_i|^4|\mf X_i = x]<\infty$ (which implies $\delta = 2$), $n J^{-p}=o(1)$ and either $J^6/n = o(1)$ for power series or $J^4/n=o(1)$ for splines.

The final corollary is a direct consequence of our Theorems \ref{t-ucb-cs} and \ref{t-ucb-dl} and Remark \ref{rm-ucs} for uniform inferences based on sieve $t$ processes for exact CS and DL nonlinear functionals under exogeneity, and hence its proof is omitted.

\begin{corollary}\label{c-ucb-exog}
Let Assumptions CS(i)(ii)(iv) and U-CS(i)(ii) hold with $\mf X_i = \mf W_i$, $J = K$, $b^K = \psi^J$ and $\mu_J \asymp 1$. Let $\ul \sigma_n^2 \gtrsim J^{a+1}$ with $ 0\geq a \geq -1$. Let $J^5(\log n)^{3} /n = o(1)$ and $nJ^{-(p+a+1)}{(\log J)} = o(1)$. Then: Results (\ref{e:sa}) (with $r_n=(\log J)^{-1/2}$) and (\ref{boots-ucb}) hold for $f_t =f_{CS,t}, f_{DL,t}$.
\end{corollary}

We note that $\ul \sigma_n^2 \asymp J$ (or $a=0$) for $f_t =f_{DL,t}$. Corollary \ref{c-ucb-exog} appears to be a new addition to the existing literature. The sufficient conditions for uniform inference for collections of nonlinear exact CS and DL functionals of nonparametric demand estimation under exogeneity are mild and simple.

\section{Conclusion}\label{s-end}

This paper makes several important contributions to inference on nonparametric models with endogeneity. We derive the minimax sup-norm convergence rates for estimating the structural NPIV function $h_0$ and its derivatives. We also provide upper bounds for sup-norm convergence rates of computationally simple sieve NPIV (series 2SLS) estimators using any sieve basis to approximate unknown $h_0$, and show that the sieve NPIV estimator using spline or wavelet basis can attain the minimax sup-norm rates. These rate results are particularly useful for establishing validity of pointwise and uniform inference procedures for nonlinear functionals of $h_0$. In particular, we use our sup-norm rates to establish the uniform Gaussian process strong approximation and the validity of score bootstrap-based UCBs for collections of nonlinear functionals of $h_0$ under primitive conditions, allowing for mildly and severely ill-posed problems. We illustrate the usefulness of our UCBs procedure with two real data applications to nonparametric demand analysis with endogeneity. We establish the pointwise and uniform limit theories for sieve $t$-statistics for exact (and approximate) CS and DL nonlinear functionals under low-level conditions when the demand function is estimated via sieve NPIV. Our theoretical and empirical results for CS and DL are new additions to the literature on nonparametric welfare analysis.

We conclude the paper by mentioning some further extensions and applications of sup-norm convergence rates of sieve NPIV estimators.

\textbf{Extensions to semiparametric IV models}. Although our rate results are presented for purely nonparametric IV models, the results may be adapted easily to some semiparametric models with nonparametric endogeneity, such as partially linear IV regression \citep{AiChen2003,Florens-linear}, shape-invariant Engel curve IV regression \citep{BCK}, and single index IV regression \citep{CCLN}, to list a few.
For example, consider the partially linear NPIV model
\[
 Y_i = X_{1i}'\beta_0 + h_0(X_{2i}) + u_i~\quad\quad E[u_i|W_{1i},W_{2i}]=0~
\]
where $X_{1i}$ and $X_{2i}$ are of dimensions $d_1$ and $d_2$ and do not contain elements in common, and $W_i = (W_{1i},W_{2i})$ is the (conditional) IV. See \cite{Florens-linear,CCLN} for identification of $(\beta_0, h_0)$ in this model. We can still estimate $(\beta_0, h_0)$ via sieve NPIV or series 2SLS as before, replacing $\Psi$ and $B$ in equations (\ref{def-psi})--(\ref{def-b}) by:
\begin{align*}
 \psi^J(x) & = (x_1',\psi^J_2(x_2)')'  & \psi^J_2(x) & = (\psi_{J1}(x_2),\ldots,\psi_{JJ}(x_2))'  \\
  b^K(w) & = (w_1',b^K_2(w_2)')' & b^K_2(w) & = (b_{K1}(w_2),\ldots,b_{KK}(w_2))'
\end{align*}
where $x = (x_1',x_2')'$, $w = (w_1',w_2')'$, $\psi_{J1},\ldots,\psi_{JJ}$ denotes a sieve of dimension $J$ for approximating $h_0(x_2)$ and $b_{K1},\ldots,b_{KK}$ denotes a sieve of dimension $K$ for the instrument space for $W_2$.
We then partition $\wh c$ in (\ref{def-series-2SLS}) into $\wh c = (\wh \beta'_{\phantom{2}},\wh c_2')'$
and set $\widehat h(x) = \psi^J_2(x)'\widehat c_2$. Note that $\wh \beta$ is root-$n$ consistent and asymptotically normal for $\beta_0$ under mild conditions (see \cite{AiChen2003,ChenPouzo2009}), and hence would not affect the optimal convergence rate of $\wh h$ to $h_0$. Our rate results may be slightly altered to derive sup-norm convergence rates for $\widehat h$ and its derivatives.

\textbf{Nonparametric specification testing in NPIV models}. Structural models may specify a parametric form $m_{\theta_0}(x)$ where $\theta_0 \in \Theta \subseteq \mb R^{d_\theta}$ for the unknown structural function $h_0(x )$ in NPIV model (\ref{npreg}). We may be interested in testing the parametric model $\{m_\theta : \theta \in \Theta\}$ against a nonparametric alternative that only assumes some smoothness on $h_0$.
Specification tests for nonparametric regression without endogeneity have typically been performed via either a quadratic-form-based statistic or a Kolmogorov-Smirnov (KS) type sup statistic.\footnote{See, e.g., \cite{Bierens82}, \cite{HardleMammen}, \cite{HongWhite}, \cite{FanLi}, \cite{LavergneVuong}, \cite{StinchcombeWhite1998} and \cite{HorowitzSpokoiny} to list a few.} However, specification tests for NPIV models have so far only been performed via quadratic-form-based statistics; see, e.g.,  \cite{Horowitz2006,Horowitz2011,Horowitz2012,BlundellHorowitz2007,Breunig2015}. Equipped with our sup-norm rate and UCBs results for NPIV function and its derivatives, one could also perform specification tests in NPIV models using KS type statistics of the form
\[
 T_n = \sup_{x} \frac{\Big| \wh h (x) - \wh m(x,\wh \theta) \Big|}{s_n(x)}
\]
where $\wh \theta$ is a first-stage estimator of $\theta_0$, and $\wh m (x,\wh \theta)$ is obtained from series 2SLS regression of $m(X_1,\wh \theta),\ldots,m(X_n,\wh \theta)$ on the same basis functions as in $\wh h$, and $s_n(x)$ is a normalization factor. Alternatively, one could consider a KS statistic formed in terms of the projection of $[\wh h (x) - \wh m(x,\wh \theta)]$ onto the instrument space. Sup-norm convergence rates and uniform limit theory derived in this paper would be useful in deriving the large-sample distribution of these KS type statistics. Further, based on our rate results (in sup- and $L^2$-norm) for estimating derivatives of $h_0$ in a NPIV model,
one could also perform nonparametric tests of significance by testing whether partial derivatives of NPIV function $h_0$ are identically zero, via KS or quadratic-form-based test statistics.

If one is interested in specifications or inferences on functionals directly, then one might consider KS type sup statistics for (possibly nonlinear) functionals directly.
For example, if one is interested in exact CS functional of a demand and concerns about the potential endogeneity of price. Then one could estimate exact CS functional using a series LS estimated demand (under exogeneity) and series 2SLS estimated demand (under endogeneity), and then compare the two estimated exact CS functionals via KS type or quadratic-form-based test. In fact, the score bootstrap-based UCBs reported in Figure \ref{f-cs-dwl-ls} indicates that such a test based on exact CS functional directly could be quite informative.

\textbf{Semiparametric 2-step procedures with NPIV first stage}. Many semiparametric two-step or multi-step estimation and inference procedures involve a nonparametric first stage. There are many theoretical results when the first stage is a purely nonparametric LS regression (without endogeneity) and its sup-norm convergence rate is used to assistant subsequent analysis. For structural estimation and inference, it is natural to allow for the presence of nonparametric endogeneity in the first stage as well. For instance, if there is endogeneity present in the conditional moment inequality application of the famous intersection bound paper of \cite{CLR}, one could simply use our sup-norm rate and UCBs results for sieve NPIV instead of their series LS regression in the first stage. As another example, consider semiparametric two-step GMM models
$E[g(Z_i,\theta_0,h_0(X_i))] =0$, where $h_0$ is the NPIV function in model (\ref{npreg}), $g$ is a $\mb R^{d_g}$-valued vector of moment functions with $d_g \geq d_\theta$, and $\theta_0 \in \mb R^{d_\theta}$ is a finite-dimensional parameter of interest, such as the average exact CS parameter of a nonparametric demand function with endogeneity. A popular estimator $\wh \theta$ of $\theta_0$ is a solution to the semiparametric two-step GMM with a weighting matrix $\wh W$:
\[
 \min_{\theta } \left( \frac{1}{n} \sum_{i=1}^n g(Z_i,\theta,\wh h(X_i)) \right)' \wh W \left( \frac{1}{n} \sum_{i=1}^n g(Z_i,\theta,\wh h(X_i)) \right)
\]
where $\wh h$ is a sieve NPIV estimator of $h_0$. When $h_0$ enters the moment function $g(\cdot)$ nonlinearly, sup-norm convergence rates of $\wh h$ to $h_0$ are useful in deriving the asymptotic properties of $\wh \theta$.

\bigskip

\appendix

{\small \singlespacing
\putbib
}

\section{Additional lemmas for sup-norm rates} \label{s-sup-lemmas}

Let $s_{\min}(A)$ denote the minimum singular value of a rectangular matrix $A$. For a positive-definite symmetric matrix $A$ we let $A^{1/2}$ be its positive definite square root. We define $s_{JK} = s_{\min} (G_b^{-1/2} S G_\psi^{-1/2})$, which satisfies
\begin{equation*}
 s_{JK}^{-1}  =  \sup_{h \in \Psi_J : h \neq 0} \frac{\|h\|_{L^2(X)}}{\|\Pi_K Th\|_{L^2(W)}} \geq \tau_J
\end{equation*}
for all $K \geq J >0$. The following lemma is used throughout the paper.

\begin{lemma}\label{lem-ip}
Let Assumptions \ref{a-data}(iii) and \ref{a-approx}(i) hold. Then: $ (1-o(1)) s_{JK}^{-1} \leq \tau_J \leq s_{JK}^{-1}$ as $J \to \infty$.
\end{lemma}

Before we provide a bound on the sup-norm ``bias'' term, we present some sufficient conditions for Assumption \ref{a-approx}(iii). This involves three projections of $h_0$ onto the sieve approximating space $\Psi_J$. These projections imply different, but closely related, approximation biases for $h_0$. Recall that $\Pi_J: L^2(X) \to \Psi_J$ is the $L^2(X)$ orthogonal (i.e. least squares) projection onto $\Psi_J$, namely $\Pi_J h_0 = \mathrm{arg}\min_{h \in \Psi_J} \|h_0 - h\|_{L^2(X)}$, and $Q_J h_0 = \mathrm{arg}\min_{h \in \Psi_J} \|\Pi_K T(h_0 - h)\|_{L^2(W)}$ is the sieve 2SLS projection of $h_0$ onto $\Psi_J$. Let $\pi_J h_0 = \mathrm{arg}\min_{h \in \Psi_J} \|T(h_0 - h)\|_{L^2(W)}$ denote the IV projection of $h_0$ onto $\Psi_J$. Note that each of these projections are non-random.

Instead of Assumption \ref{a-approx}(iii), we could impose:

 \textbf{Assumption \ref{a-approx}} \emph{(iii') $ (\zeta_{\psi,J} \tau_J) \times  \|(\Pi_K T - T) (Q_J h_0 - \pi_Jh_0 )\|_{L^2(W)} \leq \mathrm{const} \times   \|Q_J h_0 - \pi_J h_0 \|_{L^2(X)}$.}

Assumption \ref{a-approx}(iii') seems mild and is automatically satisfied by Riesz basis. This is because $\|(\Pi_K T - T)h\|_{L^2(W)} = 0$ for all $h \in \Psi_J$ when the basis functions for $B_K$ and $\Psi_J$ form either a Riesz basis or eigenfunction basis for the conditional expectation operator.
The following lemma collects some useful facts about the approximation properties of $\pi_J h_0$.

\begin{lemma}\label{lem-approx-new}
 Let Assumptions \ref{a-data}(iii) and \ref{a-approx}(ii) hold. Then: \\
 (1) $\|h_0 - \pi_J h_0\|_{L^2(X)}\asymp \|h_0 - \Pi_J h_0\|_{L^2(X)} $; \\
 (2) If Assumption \ref{a-approx}(i) also holds, then: $\|Q_J h_0 - \pi_J h_0 \|_{L^2(X)} \leq o(1) \times \|h_0 - \pi_J h_0\|_{L^2(X)}$. \\
(3) Further, if Assumption \ref{a-approx}(iii') and
\begin{equation} \label{e-piapprox}
 \|\Pi_J h_0 - \pi_J h_0\|_\infty  \leq \mathrm{const} \times \|h_0 - \Pi_J h_0\|_\infty
\end{equation}
hold then Assumption \ref{a-approx}(iii) is satisfied.
\end{lemma}

In light of Lemma \ref{lem-approx-new} parts (1) and (2), Condition (\ref{e-piapprox}) seems mild. In fact, Condition (\ref{e-piapprox}) is trivially satisfied when the basis for $\Psi_J$ is a Riesz basis because then $\pi_J h_0 = \Pi_J h_0$ (see section 6 in \cite{ChenPouzo2014}). See Lemma \ref{lem-approx} in the online Appendix \ref{ax-proofs} for more detailed relations among $\Pi_J h_0$, $\pi_J h_0$ and $Q_J h_0$.

The next lemma provides a bound on the sup-norm ``bias'' term.
\begin{lemma}\label{lem-bias}
Let Assumptions \ref{a-data}(iii), \ref{a-sieve}(ii) and \ref{a-approx} hold. Then:\\
(1) $\|\widetilde h - \Pi_J h_0\|_\infty \leq O_p(1) \times \|h_0 - \Pi_J h_0\|_{\infty}$.\\
(2) $\|\widetilde h - h_0\|_\infty \leq O_p\left( 1 + \|\Pi_J \|_{\infty} \right) \times \|h_0 - h_{0,J}\|_\infty$.
\end{lemma}

\section{Optimal $L^2$-norm rates for derivatives}\label{s-l2}

Here we show that the sieve NPIV estimator can attain the optimal $L^2$-norm convergence rates for estimating $h_0$ and its derivatives under much weaker conditions. The optimal $L^2$-norm rates for sieve NPIV derivative estimation presented in this section are new, and should be very useful for inference on some nonlinear functionals involving derivatives such as $f(h) = \|\partial^\alpha h\|_{L^2(X)}^2$.

Instead of Assumption \ref{a-data}(iii), we impose the following condition for identification in $(\mathcal H, \|\cdot \|_{L^2(X)})$:

 \textbf{Assumption \ref{a-data}} \emph{(iii')
$h_0 \in \mathcal H \subset L^2 (X)$, and $T[h-h_0]=0 \in L^2(W)$ for any $h\in \mathcal H$ implies that $\|h-h_0\|_{L^2(X)} =0$.}

\begin{theorem}\label{t-l2}
Let Assumptions \ref{a-data}(iii') and \ref{a-approx}(i)(ii) hold and let $\tau_J\zeta \sqrt{(\log J)/n} = o(1)$. Then:\\
(1) $\|\widetilde h - h_0 \|_{L^2(X)} \leq O_p (1) \times \|h_0 - \Pi_J h_0\|_{L^2(X)}$.\\
(2) Further, if Assumption \ref{a-residuals}(i) holds then
\[
 \|\wh h - h_0\|_{L^2(X)} = O_p\left( \|h_0 - \Pi_J h_0\|_{L^2(X)} + \tau_J\sqrt{J/n}\right)\,.
\]
\end{theorem}

The following corollary provides concrete $L^2$ norm convergence rates of $\wh h$ and its derivatives.
Let $B^p_{2,2}$ denote the Sobolev space of smoothness $p > 0$, $\|\cdot\|_{B^p_{2,2}}$ denote a Sobolev norm of smoothness $p$, and $B_2(p,L)= \{h \in B^p_{2,2} : \|h\|_{B^p_{2,2}} \leq L\}$ where radius $0 < L < \infty$ \cite[Section 1.11]{Triebel2006}.

\begin{corollary}\label{c-rate2}
Let Assumptions \ref{a-data}(i)(ii)(iii') and \ref{a-approx}(i)(ii) hold. Let $h_0 \in B_2(p,L)$, $\Psi_J$ be spanned by a cosine basis, B-spline basis of order $\gamma > p$, or CDV wavelet basis of regularity $\gamma > p$, $B_K$ be spanned by a cosine, spline, or wavelet basis. Let $\tau_J \sqrt{(J \log J)/n} = o(1)$ hold. Then:\\
(1) $\|\partial^\alpha \widetilde h - \partial^\alpha h_0 \|_{L^2(X)} = O_p \left( J^{-(p-|\alpha|)/d} \right)$ for all $0 \leq |\alpha| < p$.\\
(2) Further if Assumption \ref{a-residuals}(i) holds, then
\[
 \|\partial^\alpha \wh h - \partial^\alpha h_0 \|_{L^2(X)} = O_p \left( J^{-(p-|\alpha|)/d} +  \tau_J J^{|\alpha|/d} \sqrt{J/n} \right)~~\mbox{ for all }~~0 \leq |\alpha| < p\,.
\]
(2.a) Mildly ill-posed case: choosing $ J \asymp n^{d/(2(p+\varsigma)+d)}$ yields $\tau_J \sqrt{(J \log J)/n} = o(1)$ and
\begin{equation*}
 \|\partial^\alpha \wh h - \partial^\alpha h_0 \|_{L^2(X)} = O_p ( n^{-(p-|\alpha|)/(2(p+\varsigma)+d)}).
\end{equation*}
(2.b) Severely ill-posed case: choosing $J = (c_0\log n)^{d/\varsigma}$ for any $c_0 \in (0,1)$ yields $\tau_J \sqrt{(J \log J)/n} = o(1)$ and
\begin{equation*}
 \|\partial^\alpha \wh h - \partial^\alpha h_0 \|_{L^2(X)} = O_p ( (\log n)^{-(p-|\alpha|)/\varsigma})\,.
\end{equation*}
\end{corollary}

The conclusions of Corollary \ref{c-rate2} remain true for any basis $B_K$ under the condition $\tau_J \zeta_b \sqrt{(\log J)/n} = o(1)$. Previously, assuming some rates on estimating the unknown operator $T$, \cite{Johannesetal2011} obtained similar $L^2$-norm rates for derivatives of iteratively Tikhonov-regularized estimators in a NPIV model with scalar regressor $X_i$ and scalar instrument $W_i$.

Our next theorem shows that the rates obtained in Corollary \ref{c-rate2} are optimal. It extends the earlier work by \cite{ChenReiss} on $L^2$-norm lower-bounds for $h_0$ to lower bounds for derivative estimation.

\begin{theorem} \label{t-lb-l2}
Let Condition LB hold with $B_2(p,L)$ in place of $B_\infty(p,L)$ for the NPIV model with a random sample $\{(X_{i},Y_{i},W_i)\}_{i=1}^n$. Then for any $0 \leq |\alpha| < p$:
\begin{equation*}
 \liminf_{n \to \infty} \inf_{\wh g_n} \sup_{h \in B_2(p,L)} \mb P_h \left( \|\wh g_n - \partial^\alpha h\|_{L^2(X)} \geq c n^{-(p-|\alpha|)/(2(p+\varsigma)+d)} \right) \geq c'>0
\end{equation*}
in the mildly ill-posed case, and
\begin{equation*}
 \liminf_{n \to \infty} \inf_{\wh g_n} \sup_{h \in B_2(p,L)} \mb P_h \left( \|\wh g_n - \partial^\alpha h\|_{L^2(X)} \geq c(\log n)^{-(p-|\alpha|)/\varsigma} \right) \geq c'>0
\end{equation*}
in the severely ill-posed case,
where $\inf_{\wh g_n}$ denotes the infimum over all estimators of $\partial^\alpha h$ based on the sample of size $n$, $\sup_{h \in B_2(p,L)} \mb P_h$ denotes the sup over $h \in B_2(p,L)$  and distributions of $(X_i,W_i,u_i)$ that satisfy Condition LB with $\nu$ fixed, and the finite positive constants $c, c'$ do not depend on $n$.
\end{theorem}

\section{Lower bounds for quadratic functionals} \label{s-quad}

In this section we study quadratic functionals of the form
\[
 f(h) = \int (\partial^{\alpha} h(x))^2 \mu(x) \mr d x
\]
where $\mu(x) \geq \ul \mu > 0$ is a positive weighting function. These functionals are very important for nonparametric specification and goodness-of-fit testing, as outlined in the conclusion section. We derive lower bounds on convergence rates of estimators of the functional $f(h_0)$.

\begin{theorem}\label{t-lbquad}
Let Condition LB hold with $B_2(p,L)$ in place of $B_\infty(p,L)$ for the NPIV model with a random sample $\{(X_{i},Y_{i},W_i)\}_{i=1}^n$. Then for any $0 \leq |\alpha| < p$:
\[
 \liminf_{n \to \infty} \inf_{\wh g_n} \sup_{h \in B_2(p,L)} \mb P_h \left( |\wh g_n - f(h) | > c r_n \right) \geq c' > 0
\]
where
\[
 r_n = \left[ \begin{array}{ll}
 n^{-1/2} & \mbox{in the mildly ill-posed case when } p \geq \varsigma + 2|\alpha| + d/4 \\
 n^{-4(p-|\alpha|)/(4(p+\varsigma)+d)} & \mbox{in the mildly ill-posed case when } \varsigma < p < \varsigma + 2|\alpha| + d/4 \\
 (\log n)^{-2(p-|\alpha|)/\varsigma} & \mbox{in the severely ill-posed case,}
 \end{array} \right.
\]
$\inf_{\wh g_n}$ denotes the infimum over all estimators of $f(h)$ based on the sample of size $n$, $\sup_{h \in B_2(p,L)} \mb P_h$ denotes the sup over $h \in B_2(p,L)$  and distributions  $(X_i,W_i,u_i)$ which satisfy Condition LB with $\nu$ fixed, and the finite positive constants $c, c'$ do not depend on $n$.
\end{theorem}

In the mildly ill-posed case, Theorem \ref{t-lbquad} shows that the rate exhibits a so-called elbow phenomenon, in which $f(h_0)$ is $\sqrt n$-estimable when $p \geq \varsigma + 2|\alpha| + d/4$ and irregular otherwise. Moreover, $f(h_0)$ is always irregular in the severely ill-posed case.

Consider estimation using the plug-in estimator $f(\widehat h)$. Expanding the quadratic, we see that
\[
  f(\wh h) - f(h_0) = \int \partial^\alpha h_0(x) (\partial^\alpha \wh h(x) - \partial^\alpha h_0(x))\mu(x)\,\mr d x +  \|  \partial^\alpha \wh h - \partial^\alpha h_0\|^2_{L^2(\mu)}\,.
\]
Under appropriate normalization, the first term on the right-hand side will be the ``CLT term''. Consider the quadratic remainder term. Since $\mu$ is bounded away from zero and the density of $X_i$ is bounded away from zero and infinity, the quadratic remainder term behaves like $\|  \partial^\alpha \wh h - \partial^\alpha h_0\|^2_{L^2(X)}$. In the mildly ill-posed case, the optimal convergence rate of this term has been shown to be $O_p ( n^{-2(p-|\alpha|)/(2(p+\varsigma)+d)})$ (see Appendix \ref{s-l2}). This term vanishes faster than $n^{-1/2}$ provided that $p > \varsigma + 2|\alpha| + d/2$, which is a stronger condition than that is required for $f(h_0)$ to be $\sqrt n$-estimable. Therefore, when $\varsigma + 2|\alpha| + d/4 < p < \varsigma + 2|\alpha| + d/2$, the weighted quadratic functional $f(h_0)$ is $\sqrt n$-estimable but its simple plug-in estimator $f(\wh h)$ fails to attain the optimal rate.

\end{bibunit}

%%%%%%%%% MAIN ONLINE APPENDIX %%%%%%%%%%
\newpage
\begin{bibunit}

\pagenumbering{arabic}\renewcommand{\thepage}{\arabic{page}}

\begin{center} \doublespacing
{\Large Main Online Appendix to}
\vskip 12pt
{\LARGE \thetitle}
\vskip 24pt
{\large Xiaohong Chen \quad Timothy M. Christensen
\vskip 12 pt
\thedate}

\end{center}

{

This main online supplementary appendix contains material to support our paper ``\thetitle''. Appendix \ref{s-pw} presents pointwise normality of sieve $t$ statistics for nonlinear functionals of NPIV under low-level sufficient conditions. Appendix \ref{ax-basis} contains background material on B-spline and wavelet bases and the equivalence between Besov and wavelet sequence norms. Appendix \ref{ax-supp} contains material on useful matrix inequalities and convergence results for random matrices. The secondary online supplementary appendix contains additional technical lemmas and all of the proofs (Appendix \ref{ax-proofs}).

\section{Pointwise asymptotic normality of sieve $t$-statistics}\label{s-pw}

In this section we derive the pointwise asymptotic normality of sieve $t$-statistics for nonlinear functionals of a NPIV function under low-level sufficient conditions. Previously under some high-level conditions, \cite{ChenPouzo2014} established the pointwise asymptotic normality of sieve $t$ statistics for (possibly) nonlinear functionals of $h_0$ satisfying general semi/nonparametric conditional moment restrictions including NPIV and nonparametric quantile IV models as special cases.
As the sieve NPIV estimator $\wh h$ has a closed-form expression and for the sake of easy reference, we derive the limit theory directly rather than appealing to the general theory in \cite{ChenPouzo2014}. Our low-level sufficient conditions are tailored to the case in which the functional $f(\cdot )$ is \emph{irregular} in $h_0$ (i.e. slower than root-$n$ estimable), so that they are directly comparable to the sufficient conditions for the uniform inference theory in Section \ref{s-ucb}.

We consider a functional $f : \mc H \subset L^\infty(X) \to \mb R$ for which $Df(h)[v] = \lim_{\delta \rightarrow 0^{+}} [\delta^{-1} f(h +\delta v)]$ exists for all $v\in \mc H - \{h_0\}$ for all $h$ in a small neighborhood of $h_0$. Recall that the sieve 2SLS Riesz representer of $Df(h_0)$ is
\[
 v_n(f)(x) = \psi^J(x)' [S' G_b^{-1} S]^{-1} Df(h_{0})[\psi^J]~,
\]
and let
\[
 [s_n(f)]^2 = \| \Pi_K T v_n(f)\|_{L^2(W)}^2 =  (Df(h_{0})[\psi^J])'[S' G_b^{-1} S]^{-1} Df(h_{0})[\psi^J]
\]
denote its weak norm. \cite{ChenPouzo2014} called that the functional $f(\cdot)$ is an irregular (i.e. slower than $\sqrt n$-estimable) functional of $h_0$ if $s_n(f) \nearrow \infty$ and a regular (i.e. $\sqrt n$-estimable) functional of $h_0$ if $\lim_{n} s_n(f) < \infty$. Denote
\[
\wh v_n(f)(x) = \psi^J(x)' [S' G_b^{-1} S]^{-1} Df(\wh h)[\psi^J]~.
\]
It is clear that $v_n(f) = \wh v_n(f)$ whenever $f(\cdot)$ is linear.

Recall that $\Omega = E[u_i^2 b^K(W_i) b^K(W_i)']$, and the ``2SLS covariance matrix'' for $\wh c$ (given in equation (\ref{def-series-2SLS})) is
$$\mho = [S' G_b^{-1}  S]^{-1}  S'  G_b^{-1}  \Omega  G_b^{-1}  S[ S'  G_b^{-1}  S]^{-1}~,$$ and the sieve variance for $f (\wh h)$ is
\begin{equation*}
[\sigma_n(f)]^2   = \big(Df(h_{0})[\psi^J]\big)' \mho  \big(Df(h_{0})[\psi^J]\big)~.
\end{equation*}
Under Assumption \ref{a-residuals}(i)(iii) we have that $[\sigma_n(f)]^2 \asymp [s_n(f)]^2$.
Therefore $f()$ is an irregular functional of $h_0$ iff $\sigma_n(f)\nearrow +\infty$ as $n \to \infty$.
Recall the sieve variance estimator is
\begin{equation*}
[\wh\sigma (f)]^2   = \big(Df(\wh h)[\psi^J]\big)' \wh \mho  \big(Df(\wh h)[\psi^J]\big)~.
\end{equation*}
where $\wh \mho$ is defined in equation (\ref{cov-2sls}).

\setcounter{assumption}{1}

\begin{assumption}[continued]\label{a-residuals-app}
(iv') $\sup_w E[u_i^2 \{ |u_i | > \ell(n)\}|W_i = w] = o(1)$ for any positive sequence with $\ell(n) \nearrow \infty$.
\end{assumption}
Assumption \ref{a-residuals-app}(iv') is a mild condition which is trivially satisfied if $E[|u_i|^{2+\epsilon}|W_i = w]$ is uniformly bounded for some $\epsilon > 0$.

\setcounter{aprime}{4}

\begin{aprime}\label{a-functional-prime}
Assumption \ref{a-functional} holds with $f_t = f$ and $\mc T$ a singleton.
\end{aprime}
Assumption \ref{a-functional-prime}'(a) and \ref{a-functional-prime}'(b)(i)(ii) is similar to Assumption 3.5 of \cite{ChenPouzo2014}. Assumption \ref{a-functional-prime}'(b)(iii) controls any additional error arising in the estimation of $\sigma_n(f)$ due to nonlinearity of $f(\cdot)$ and is automatically satisfied when $f(\cdot)$ is a linear functional.

\begin{remark}
Remark \ref{rmk-a5suff} presents sufficient conditions for Assumption \ref{a-functional-prime}' as a special case, with $f_t = f$, $\ul \sigma_n = \sigma_n(f)$, and $\mc T$ a singleton.
\end{remark}

Again these sufficient conditions are formulated to take advantage of the sup-norm rate results in Section \ref{npiv sec}. Denote
\begin{equation*}
 \wh{\mb Z}_n \equiv \frac{(D f(h_0)[\psi^J])' [S' G_b^{-1} S]^{-1} S' G_b^{-1} }{\sigma_n(f)} \frac{1}{\sqrt n } \sum_{i=1}^n b^K( W_i) u_i \,,
\end{equation*}
and $\delta_{V,n} \equiv \big[ \zeta_{b,K}^{(2+\delta)/\delta} \sqrt{(\log K)/n} \big]^{\delta/(1+\delta)}  + \tau_J \zeta \sqrt{(\log J)/n} + \delta_{h,n}$, where $\delta_{h,n}=o_p (1)$ is a positive finite sequence such that $\|\wh h - h_0\|_\infty = O_p(\delta_{h,n})$.

\begin{theorem} \label{t-dist}
(1) Let Assumptions \ref{a-data}(iii), \ref{a-residuals}(i)(iii)(iv'), \ref{a-approx}(i), and either \ref{a-functional-prime}'(a) or \ref{a-functional-prime}'(b)(i)(ii) hold, and let $\tau_J \zeta \sqrt{(J \log J)/n} = o(1)$. Then:
\begin{equation*}
 \sqrt{n} \frac{(f(\wh{h})-f(h_{0}))}{\sigma_n(f)} = \wh {\mb Z}_n + o_p(1) \to_d N(0,1)\,.
\end{equation*}
(2) If $\|\wh h - h_0\|_\infty = o_p(1)$ and Assumptions \ref{a-residuals}(ii) and \ref{a-sieve}(iii) hold (and \ref{a-functional-prime}'(b)(iii) also holds if $f$ is nonlinear), then:
\begin{eqnarray*}
 \left| \frac{\wh\sigma (f)}{\sigma_n(f) } - 1 \right| & = & = O_p( \delta_{V,n} + \eta_n')=o_p(1)~,
\end{eqnarray*}
and
\begin{eqnarray*}
 \sqrt{n} \frac{(f(\wh{h})-f(h_{0}))}{\wh\sigma (f)} = \wh{\mb Z}_n + o_p(1) \to_d N(0,1) \,.
\end{eqnarray*}
\end{theorem}

By exploiting the closed form expression of the sieve NPIV estimator and by applying exponential inequalities for random matrices, Theorem \ref{t-dist} derives the pointwise limit theory under lower-level sufficient conditions than those in \cite{ChenPouzo2014} for irregular nonlinear functionals. In particular, when specialized to the exogenous case of $X_i =W_i$, $h_0(x) = E[Y_i | W_i = x]$, $K =J$ and $b^K = \psi^J$ with $\tau_J =1$, the regularity conditions for Theorem \ref{t-dist} become about the same mild conditions for Theorem 3.2 in \cite{ChenChristensen-reg} on asymptotic normality of sieve $t$ statistics for nonlinear functionals of series LS estimators. It is now obvious that one could also derive the asymptotic normality of sieve $t$-statistics for regular (i.e., root-$n$ estimable) nonlinear functionals of a NPIV function under lower-level sufficient conditions by using our sup-norm rates results to verify Assumption 3.5(ii) and Remark 3.1 in \cite{ChenPouzo2014}.

\section{Spline and wavelet bases}\label{ax-basis}

In this section we bound the terms $\xi_{\psi,J}$, $e_J = \lambda_{\min}(G_{\psi,J})$ and $\kappa_\psi(J)$ for B-spline and CDV wavelet bases. Although we state the results for the space $\Psi_J$, they may equally be applied to $B_K$ when $B_K$ is constructed using B-spline or CDV wavelet bases.

\subsection{Spline bases}

We construct a univariate B-spline basis of order $r \geq 1$ (or degree $r-1 \geq 0$) with $m \geq 0$ interior knots and support $[0,1]$ in the following way. Let $0=t_{-(r-1)}=\ldots =t_{0}\leq t_{1}\leq \ldots \leq t_{m}\leq
t_{m+1}=\ldots =t_{m+r}=1$
denote the extended knot sequence and let $I_{1}=[t_{0},t_{1}),\ldots ,I_{m}=[t_{m},t_{m+1}]$. A
basis of order $1$ is constructed by setting
\begin{equation*}
N_{j,1}(x)= \left\{ \begin{array}{rl} 1 & \mbox{if } x \in  I_j \\ 0 & \mbox{otherwise} \end{array} \right.
\end{equation*}%
for $j=0,\ldots m$. Bases of order $r>1$ are generated recursively
according to
\begin{equation*}
N_{j,r}(x)=\frac{x-t_{j}}{t_{j+r-1}-t_{j}}N_{j,r-1}(x)+%
\frac{t_{j+r}-x}{t_{j+r}-t_{j+1}}N_{j+1,r-1}(x)
\end{equation*}%
for $j=-(r-1),\ldots ,m$ where we adopt the convention $\frac{1}{0}:=0$ (see Section 5 of \cite{DeVoreLorentz}).
This results in a total of $m+r$ splines of order $r$, namely $N_{-(r-1),r},\ldots,N_{m,r}$. Each spline is a
polynomial of degree $r-1$ on each interior interval $I_{1},\ldots ,I_{m}$
and is $(r-2)$-times continuously differentiable on $[0,1]$ whenever $r\geq
2 $. The mesh ratio is defined as
\begin{equation*}
\mbox {mesh}(m)=\frac{\max_{0\leq j\leq m}(t_{j+1}-t_{j})}{\min_{0\leq j\leq
m}(t_{j+1}-t_{j})}\,.
\end{equation*}%
Clearly $\mbox{mesh}(m) = 1$ whenever the knots are placed evenly (i.e. $t_i = \frac{i}{m+1}$ for $i = 1,\ldots,m$ and $m \geq 1$) and we say that the mesh ratio is \emph{uniformly
bounded} if $\mbox {mesh}(m)\lesssim 1$ as $m \to \infty$. Each of  has continuous derivatives of orders $\leq r-2$ on $(0,1)$. We let the space $\mbox {BSpl}(r,m,[0,1])$ be the closed linear span of the $%
m+r$ splines $N_{-(r-1),r},\ldots,N_{m,r}$.

We construct B-spline bases for $[0,1]^d$ by taking tensor products of univariate bases. First generate $d$ univariate bases $N_{-(r-1),r,i},\ldots,N_{m,r,i}$ for each of the $d$ components $x_i$ of $x$ as described above. Then form the vector of basis functions $\psi^J$ by taking the tensor product of the vectors of univariate basis functions, namely:
\[
 \psi^J(x_1,\ldots,x_d) = \bigotimes_{i=1}^d \left( \begin{array}{c} N_{-(r-1),r,i}(x_i) \\ \vdots \\ N_{m,r,i}(x_i) \end{array} \right)\,.
\]
The resulting vector $\psi^J$ has dimension $J = (r+m)^d$. Let $\psi_{J1},\ldots,\psi_{JJ}$ denote its $J$ elements.

\paragraph{Stability properties:}

The following two Lemmas bound $\xi_{\psi,J}$, and the  minimum eigenvalue and condition number of $G_\psi = G_{\psi,J} = E[\psi^J(X_i)\psi^J(X_i)']$ when $\psi_{J1},\ldots,\psi_{JJ}$ is constructed using univariate and tensor-products of B-spline bases with uniformly bounded mesh ratio.

\begin{lemma}\label{lem-spline1}
Let $X$ have support $[0,1]$ and let $\psi_{J1} = N_{-(r-1),r},\ldots,\psi_{JJ}=N_{m,r}$ be a univariate B-spline basis of order $r \geq 1$ with $m = J-r \geq 0$ interior knots and uniformly bounded mesh ratio.
Then: (a) $\xi_{\psi,J} = 1$ for all $J \geq r$; (b) If the density of $X$ is uniformly bounded away from $0$ and $\infty$ on $[0,1]$, then there exists finite positive constants $c_\psi$ and $C_\psi$ such that $c_\psi J \leq \lambda_{\max}(G_\psi)^{-1} \leq \lambda_{\min}(G_\psi)^{-1} \leq C_\psi J$ for all $J \geq r$; (c) $\lambda_{\max}(G_\psi)/\lambda_{\min}(G_\psi) \leq C_\psi/c_\psi$ for all $J \geq r$.
\end{lemma}

\begin{lemma}\label{lem-spline2}
Let $X$ have support $[0,1]^d$ and let $\psi_{J1},\ldots,\psi_{JJ}$ be a B-spline basis formed as the tensor product of $d$ univariate bases of order $r \geq 1$ with $m = J^{1/d}-r \geq 0$ interior knots and uniformly bounded mesh ratio. Then: (a) $\xi_{\psi,J} = 1$ for all $J \geq r^d$; (b) If the density of $X$ is uniformly bounded away from $0$ and $\infty$ on $[0,1]^d$, then there exists finite positive constants $c_\psi$ and $C_\psi$ such that $c_\psi J \leq \lambda_{\max}(G_\psi)^{-1} \leq \lambda_{\min}(G_\psi)^{-1} \leq C_\psi J$ for all $J \geq r^d$; (c) $\lambda_{\max}(G_\psi)/\lambda_{\min}(G_\psi) \leq C_\psi/c_\psi$ for all $J \geq r^d$.
\end{lemma}

\subsection{Wavelet bases}

We construct a univariate wavelet basis with support $[0,1]$ following \cite{CDV1993} (CDV hereafter).
Let $(\varphi,\psi)$ be a Daubechies pair such that $\varphi$ has support $%
[-N+1,N]$. Given $j$ such that $2^j -2N > 0$, the orthonormal (with respect
to the $L^2([0,1])$ inner product) basis for the space $V_j$ includes $%
2^j-2N$ interior scaling functions of the form $\varphi_{j,k}(x) = 2^{j/2}
\varphi(2^j x - k)$, each of which has support $[2^{-j}(-N+1+k),2^{-j}(N+k)]$
for $k = N,\ldots,2^j-N-1$. These are augmented with $N$ left scaling
functions of the form $\varphi^0_{j,k}(x) = 2^{j/2}\varphi_k^l(2^j x)$ for $%
k = 0,\ldots,N-1$ (where $\varphi^l_0,\ldots,\varphi^l_{N-1}$ are fixed
independent of $j$), each of which has support $[0,2^{-j}(N+k)]$, and $N$
right scaling functions of the form $\varphi_{j,2^j-k}(x) = 2^{j/2}
\varphi^r_{-k}(2^j(x-1))$ for $k = 1,\ldots,N$ (where $\varphi^r_{-1},%
\ldots,\varphi^r_{-N}$ are fixed independent of $j$), each of which has
support $[1-2^{-j}(1-N-k),1]$. The resulting $2^j$ functions $%
\varphi^0_{j,0},\ldots,\varphi^0_{j,N-1},\varphi_{j,N},\ldots,%
\varphi_{j,2^j-N-1},\varphi^1_{j,2^j-N},\ldots,\varphi^1_{j,2^j-1}$ form an
orthonormal basis (with respect to the $L^2([0,1])$ inner product) for their closed linear span $V_j$.

An orthonormal wavelet basis for the space $W_j$, defined as the orthogonal
complement of $V_j$ in $V_{j+1}$, is similarly constructed form the mother
wavelet. This results in an orthonormal basis of $2^j$ functions, denoted $%
\psi^0_{j,0},\ldots,\psi^0_{j,N-1},\psi_{j,N},\ldots,\psi_{j,2^j-N-1},%
\psi^1_{j,2^j-N},\ldots,\psi^1_{j,2^j-1}$ (we use this conventional notation without confusion with the $\psi_{Jj}$ basis functions spanning $\Psi_J$) where the ``interior'' wavelets $\psi_{j,N},\ldots,\psi_{j,2^j-N-1}$ are of the form $\psi_{j,k}(x) = 2^{j/2} \psi(2^j x-k)$. To simplify notation we ignore
the $0$ and $1$ superscripts on the left and right wavelets and scaling
functions henceforth. Let $L_0$ and $L$ be integers such that $2N< 2^{L_0} \leq 2^L$. A wavelet
space at resolution level $L$ is the $2^{L+1}$-dimensional set of functions given by
\[
\mathrm{Wav}(L,[0,1]) = \left\{ \sum_{k=0}^{2^{L_0}-1} a_{L_0,k} \varphi_{L_0,k} +
\sum_{j=L_0}^L \sum_{k=0}^{2^j-1} b_{j,k} \psi_{j,k} : a_{L_0,k}, b_{j,k}
\in \mb{R }\right\} \,.
\]
We say that $\mbox {Wav}(L,[0,1])$ has \emph{regularity} $\gamma $ if $\psi \in C^\gamma$ (which can be achieved by choosing $N$ sufficiently large) and write $\mbox {Wav}(L,[0,1],\gamma )$ for a wavelet space
of regularity $\gamma$ with continuously differentiable basis functions.

We construct wavelet bases for $[0,1]^d$ by taking tensor products of univariate bases. We again take $L_0$ and $L$ to be integers such that $2N< 2^{L_0} \leq 2^L$. Let $\widetilde \psi_{j,k,G}(x)$ denote an orthonormal tensor-product wavelet for $L^2([0,1]^d)$ at resolution level $j$ where $k = (k_1,\ldots,k_d) \in \{0,\ldots,2^{j}-1\}^d$ and where $G \in G_{j,L} \subseteq \{w_\varphi,w_\psi\}^d$ denotes which elements of the tensor product are $\psi_{j,k_i}$ (indices corresponding to $w_\psi$) and which are $\varphi_{j,k_i}$ (indices corresponding to $w_\varphi$). For example, $\widetilde \psi_{j,k,w_\psi^d} = \prod_{i=1}^d \psi_{j,k_i}(x_i)$. Note that each $G \in G_{j,L}$ with $j > L$ has an element that is $w_\psi$ (see \cite{Triebel2006} for details). We have $\#( G_{L_0,L_0}) = 2^d$, $\# (G_{j,L_0}) = 2^d-1$ for $j > L_0$. Let $\mbox{Wav}(L,[0,1]^d,\gamma)$ denote the space
\begin{equation} \label{e-wav-multi}
 \mbox{Wav}(L,[0,1]^d,\gamma) = \left\{ \sum_{j = L_0}^L \sum_{G \in G_{j,L_0}} \sum_{k \in \{0,\ldots,2^{j}-1\}^d} a_{j,k,G} \widetilde \psi_{j,k,G} : a_{j,k,G} \in \mb R \right\}
\end{equation}
where each univariate basis has regularity $\gamma$. This definition clearly reduces to the above definition for $\mbox{Wav}(L,[0,1],\gamma)$ in the univariate case.

\paragraph{Stability properties:}

The following two Lemmas bound $\xi_{\psi,J}$, as well as the minimum eigenvalue and condition number of $G_\psi = G_{\psi,J} = E[\psi^J(X_i)\psi^J(X_i)']$ when $\psi_{J1},\ldots,\psi_{JJ}$ is constructed using univariate and tensor-products of CDV wavelet bases.

\begin{lemma} \label{lem-wavelet1}
Let $X$ have support $[0,1]$ and let be a univariate CDV wavelet basis of resolution level $L = \log_2(J)-1$.
Then: (a) $\xi_{\psi,J} = O(\sqrt J)$ for each sieve dimension $J = 2^{L+1}$; (b) If the density of $X$ is uniformly bounded away from $0$ and $\infty$ on $[0,1]$, then there exists finite positive constants $c_\psi$ and $C_\psi$ such that $c_\psi \leq \lambda_{\max}(G_\psi)^{-1} \leq \lambda_{\min}(G_\psi)^{-1} \leq C_\psi$ for each $J$; (c)  $\lambda_{\max}(G_\psi)/\lambda_{\min}(G_\psi) \leq C_\psi/c_\psi$ for each $J$.
\end{lemma}

\begin{lemma} \label{lem-wavelet2}
Let $X$ have support $[0,1]^d$ and let $\psi_{J1} , \ldots ,\psi_{JJ}$ be a wavelet basis formed as the tensor product of $d$ univariate bases of resolution level $L$.
Then: (a) $\xi_{\psi,J} = O(\sqrt J)$ each $J$; (b) If the density of $X$ is uniformly bounded away from $0$ and $\infty$ on $[0,1]^d$, then there exists finite positive constants $c_\psi$ and $C_\psi$ such that $c_\psi \leq \lambda_{\max}(G_\psi)^{-1} \leq \lambda_{\min}(G_\psi)^{-1} \leq C_\psi$ for each $J$; (c) $\lambda_{\max}(G_\psi)/\lambda_{\min}(G_\psi) \leq C_\psi/c_\psi$ for each $J$.
\end{lemma}

\paragraph{Wavelet characterization of Besov norms:}

When the wavelet basis just described is of regularity $\gamma > 0$, the norms $\|\cdot\|_{B^p_{\infty,\infty}}$ for $p < \gamma$ can be restated in terms of the wavelet coefficients. We briefly explain the multivariate case as it nests the univariate case. Any $f \in L^2([0,1]^d)$ may be represented as
\[
 f = \sum_{j,G,k} a_{j,k,G}(f) \widetilde \psi_{j,k,G}
\]
with the sum is understood to be taken over the same indices as in display (\ref{e-wav-multi}). If $f \in B^p_{\infty,\infty}([0,1]^d)$ then
\[
 \|f\|_{B^p_{\infty,\infty}} \asymp \|f\|_{b^p_{\infty,\infty}} := \sup_{j,k,G} 2^{j(p+d/2)}|a_{j,k,G}(f)|
\,.
\]
and if $f \in B^p_{2,2}([0,1])$ then
\[
 \|f\|_{B^p_{2,2}}^2 \asymp \|f\|_{b^p_{2,2}}^2 := \sum_{j,k,G} 2^{j p} a_{j,k,G}(f)^2
\]
See \cite{Johnstone2013} and \cite{Triebel2006} for more thorough discussions.

\section{Useful results on random matrices}\label{ax-supp}

\textbf{Notation:} For a $r \times c$ matrix $A$ with $r \leq c$ and full row rank $r$ we let $A^-_l$ denote its left pseudoinverse, namely $(A'A)^{-}A'$ where $'$ denotes transpose and $^{-}$ denotes generalized inverse. We let $s_{\min}(A)$ denote the minimum singular value of a rectangular matrix $A$. For a positive-definite symmetric matrix $A$ we let $\lambda_{\min}(A)$ and $\lambda_{\max}(A)$ denote its minimum and maximum eigenvalue, respectively.

\subsection{Some matrix inequalities}

The following Lemmas are used throughout the proofs in this paper and are stated here for convenience.

\begin{lemma}[Weyl's inequality]\label{lem-Weyl}
Let $A,B \in \mb R^{r \times c}$ and let $s_i(A)$, $s_i(B)$ denote the $i$th (ordered) singular value of $A$ and $B$ respectively, for $1 \leq i \leq (r \wedge c)$. Then: $|s_i(A) - s_i(B)| \leq \|A - B\|_{\ell^2}$
for all $1 \leq i \leq (r \wedge c)$. In particular, $|s_{\min}(A) - s_{\min}(B)| \leq \|A - B\|_{\ell^2}$.
\end{lemma}

\begin{lemma}\label{lem-inv}
Let $A \in \mb R^{r \times r}$ be nonsingular. Then:
 $\|A^{-1} - I_r\|_{\ell^2} \leq \|A^{-1}\|_{\ell^2} \|A - I_r\|_{\ell^2}\,.$
\end{lemma}

\begin{lemma}[\cite{Schmitt1992}]\label{lem-sqrtm}
Let $A,B \in \mb R^{r \times r}$ be positive definite. Then:
\[
 \|A^{1/2} - B^{1/2}\|_{\ell^2} \leq \frac{1}{\sqrt{\lambda_{\min}(B)} + \sqrt{\lambda_{\min}(A)}} \|A - B\|_{\ell^2}\,.
\]
\end{lemma}

\begin{lemma}\label{lem-linv}
Let $A,B \in \mb R^{r \times c}$ with $r \leq c$ and let $A$ and $B$ have full row rank $r$. Then:
\[
 \| B^-_l - A^-_l \|_{\ell^2} \leq \frac{1 + \sqrt 5}{2} (s_{\min}(A)^{-2} \vee s_{\min}(B)^{-2}) \| A - B\|_{\ell^2} \,.
\]
If, in addition, $\|A - B\|_{\ell^2} \leq \frac{1}{2} s_{\min}(A)$ then
\[
 \| B^-_l - A^-_l \|_{\ell^2} \leq 2(1 + \sqrt 5) s_{\min}(A)^{-2} \|A - B\|_{\ell^2} \,.
\]
\end{lemma}

\begin{lemma}\label{lem-mpnorm}
Let $A \in \mb R^{r \times c}$ with $r \leq c$ have full row rank $r$. Then: $\|A^-_l\|_{\ell^2} \leq s_{\min}(A)^{-1}$.
\end{lemma}

\begin{lemma}\label{lem-projm}
Let $A,B \in \mb R^{r \times c}$ with $r \leq c$ and let $A$ and $B$ have full row rank $r$. Then:
\[
 \| A'(AA')^{-1}A - B'(BB')^{-1}B \|_{\ell^2} \leq ( s_{\min}(A)^{-1} \vee s_{\min}(B)^{-1} ) \|A - B\|_{\ell^2}\,.
\]
\end{lemma}

\subsection{Convergence of the matrix estimators}

Before presenting the following lemmas, we define the \emph{orthonormalized} matrix estimators
\begin{eqnarray*}
 \wh G_b^o & = & G_b^{-1/2} \wh G_b G_b^{-1/2} \\
 \wh G_\psi^o & = & G_\psi^{-1/2} \wh G_\psi G_\psi^{-1/2} \\
 \wh S^o & = & G_b^{-1/2} \wh S G_\psi^{-1/2}
\end{eqnarray*}
and let $G_b^o = I_K$, $G_\psi^o = I_J$ and $S^o$ denote their respective expected values.

\begin{lemma}\label{lem-matl2}
The orthonormalized matrix estimators satisfy the exponential inequalities:
\begin{eqnarray*}
 \mb P \left( \|\wh G_\psi^o - G_\psi^o \|_{\ell^2} > t \right) & \leq & 2\exp\left\{\log J - \frac{t^2/2}{\zeta_{\psi,J}^2(1+ 2 t/3)/n} \right\} \\
 \mb P \left( \|\wh G_b^o - G_b^o \|_{\ell^2} > t \right) & \leq & 2\exp\left\{\log K - \frac{t^2/2}{\zeta_{b,K}^2(1+2t/3)/n} \right\} \\
 \mb P \left( \|\wh S^o - S^o\|_{\ell^2} > t \right) & \leq & 2\exp\left\{\log K - \frac{t^2/2}{(\zeta_{b,K}^2 \vee \zeta_{\psi,J}^2)/n+ 2 \zeta_{b,K} \zeta_{\psi,J} t/(3n)} \right\}
\end{eqnarray*}
and therefore
\begin{eqnarray*}
 \|\wh G_\psi^o - G_\psi^o \|_{\ell^2} & = & O_p (\zeta_{\psi,J} \sqrt{( \log J)/n}) \\
 \|\wh G_b^o - G_b^o \|_{\ell^2} & = & O_p (\zeta_{b,K} \sqrt{(\log K)/n}) \\
 \|\wh S^o - S^o\|_{\ell^2} & = & O_p ((\zeta_{b,K} \vee \zeta_{\psi,J})\sqrt{(\log K)/n})\,.
\end{eqnarray*}
as $n,J,K \to \infty$ provided $(\zeta_{b,K} \vee \zeta_{\psi,J})\sqrt{(\log K)/n} = o(1)$.
\end{lemma}

\begin{lemma}[\cite{Newey1997}, p. 162]\label{lem-BuL2}
Let Assumption \ref{a-residuals}(i) hold. Then: $\|G_b^{-1/2}B'u/n\|_{\ell^2} = O_p(\sqrt{K/n})$.
\end{lemma}

\begin{lemma}\label{lem-BHJ}
Let $h_J(x) = \psi^J(x)'c_J$ for any deterministic $c_J \in \mb R^J$ and $H_J = (h_J(X_1),\ldots,h_J(X_n))'=\Psi c_J$. Then:
\begin{eqnarray*}
 & & \|G_b^{-1/2} (B'(H_0 - \Psi c_J)/n - E[ b^K(W_i)(h_0(X_i) - h_J(X_i))])\|_{\ell^2} \\
 & & \quad = \quad O_p \left( \Big( \sqrt{K/n} \times \|h_0 - h_J\|_{\infty} \Big) \wedge \Big( \zeta_{b,K}/\sqrt{n} \times \|h_0 - h_J\|_{L^2(X)}  \Big) \right) \,.
\end{eqnarray*}
\end{lemma}

\begin{lemma}\label{lem-SGl2}
Let $s_{JK}^{-1} \zeta \sqrt{(\log J)/n} = o(1)$ and let $J \leq K = O(J)$. Then:
\begin{eqnarray*}
 (a) & & \|(\wh G_b^{-1/2} \wh S)^-_l\wh G_b^{-1/2}G_b^{1/2} - (G_b^{-1/2} S)^-_l \|_{\ell^2} = O_p \Big(s_{JK}^{-2} \zeta \sqrt{(\log J)/(n e_J)}\Big) \\
 (b) & & \|G_\psi^{1/2}\{(\wh G_b^{-1/2} \wh S)^-_l\wh G_b^{-1/2}G_b^{1/2} - (G_b^{-1/2} S)^-_l\} \|_{\ell^2} = O_p \Big(s_{JK}^{-2} \zeta \sqrt{(\log J)/n)}\Big) \\
 (c) & & \|G_b^{-1/2} S\{(\wh G_b^{-1/2} \wh S)^-_l\wh G_b^{-1/2}G_b^{1/2} - (G_b^{-1/2} S)^-_l\} \|_{\ell^2} = O_p \Big(s_{JK}^{-1} \zeta \sqrt{(\log J)/n}\Big)\,.
\end{eqnarray*}
\end{lemma}

\singlespacing
\small
\putbib
}
\end{bibunit}

%%%%%%%%% SECONDARY ONLINE APPENDIX %%%%%%%%%%
\newpage
\begin{bibunit}

\pagenumbering{arabic}\renewcommand{\thepage}{\arabic{page}}

\begin{center} \doublespacing
{\Large Secondary Online Appendix to}
\vskip 12pt
{\LARGE \thetitle}
\vskip 24pt
{\large Xiaohong Chen \quad Timothy M. Christensen
\vskip 12 pt
\thedate}

\end{center}

{\small

\section{Supplementary Lemmas and Proofs} \label{ax-proofs}

All the notation follow from the main text and the main online appendix. For a $r \times c$ matrix $A$ with $r \leq c$ and full row rank $r$ we let $A^-_l$ denote its left pseudoinverse, namely $(A'A)^{-}A'$ where $'$ denotes transpose and $^{-}$ denotes generalized inverse. We let $s_{\min}(A)$ denote the minimum singular value of a rectangular matrix $A$.

Let $s_{JK} = s_{\min} (G_b^{-1/2} S G_\psi^{-1/2})$. Throughout the proofs in the appendix we use the identity
\begin{eqnarray*}
 \psi^J(x)'(G_b^{-1/2}S )_l^- & = & \psi^J(x)'(  S' G_b^{-1} S)^{-1}  S' G_b^{-1/2} \\
 & = & \psi^J(x)'G_\psi^{-1/2}( G_\psi^{-1/2} S' G_b^{-1} S G_\psi^{-1/2})^{-1} G_\psi^{-1/2} S' G_b^{-1/2} \\
 & = & \psi^J(x)'G_\psi^{-1/2}(G_b^{-1/2} S G_\psi^{-1/2})_l^-
\end{eqnarray*}
which implies that
\begin{eqnarray}
 \| \psi^J(x)'(G_b^{-1/2}S)_l^- \|_{\ell^2} & \leq & \| \psi^J(x)'G_\psi^{-1/2}\|_{\ell^2} \|(G_b^{-1/2} S G_\psi^{-1/2})_l^-\|_{\ell^2} \notag \\
 & \leq & \zeta_{\psi,J} \|(G_b^{-1/2} S G_\psi^{-1/2})_l^-\|_{\ell^2} \notag \\
 & \leq & \zeta_{\psi,J} s_{JK}^{-1} \label{e-xident}
\end{eqnarray}
by definition of $\zeta_{\psi,J}$ and the fact that $\|A^-_l\|_{\ell^2} \leq s_{\min}(A)^{-1}$ (see Lemma \ref{lem-mpnorm}).

\subsection{Proofs for Appendix \ref{s-sup-lemmas} and Section \ref{s-upper}}

Since the proofs of results in Section \ref{s-upper} built upon those for results in Appendix \ref{s-sup-lemmas}, we shall present the proofs for Appendix \ref{s-sup-lemmas} first.

\subsubsection{Proofs for Appendix \ref{s-sup-lemmas}}

\begin{proof}[\textbf{Proof of Lemma \ref{lem-ip}}]
First note that $\tau_J > 0$ for all $J$ by compactness and injectivity of $T$. Then:
\[
 s_{JK} = \inf_{h \in \Psi_J : \|h\|_{L^2(X)} = 1} \|\Pi_K T h\|_{L^2(W)} \leq \inf_{h \in \Psi_J : \|h\|_{L^2(X)} = 1} \| T h\|_{L^2(W)} = \tau_J^{-1}
\]
holds uniformly in $J$ because $\Pi_K$ is a contraction, whence $\tau_J \leq s_{JK}^{-1}$. To derive a lower bound on $\tau_J$, the triangle inequality and Assumption \ref{a-approx}(i) yield:
\begin{eqnarray*}
 s_{JK} & = & \inf_{h \in \Psi_{J,1}} \|\Pi_K T h\|_{L^2(W)} \\
 & \geq & \inf_{h \in \Psi_{J,1}} \| T h\|_{L^2(W)} - \sup_{h \in \Psi_{J,1}} \|(\Pi_K T - T) h\|_{L^2(W)}  \\
 & = & (1-o(1)) \tau_J^{-1}\,.
\end{eqnarray*}
Therefore, $s_{JK}^{-1} \leq (1-o(1))^{-1} \tau_J$.
\end{proof}

It is clear that Lemma \ref{lem-approx-new} is implied by the following lemma.

\begin{lemma}\label{lem-approx}
 Let Assumptions \ref{a-data}(iii) and \ref{a-approx}(ii) hold. Then: \\
 (1) (a) $\|h_0 - \pi_J h_0\|_{L^2(X)}\asymp \|h_0 - \Pi_J h_0\|_{L^2(X)} $; and \\
 \phantom{(1)} (b) $\tau_J \times \|T(h_0 - \pi_J h_0)\|_{L^2(W)} \leq \mathrm{const} \times \|h_0 - \pi_J h_0\|_{L^2(X)}$. \\
(2) If Assumption \ref{a-approx}(i) also holds, then: (a) $\|Q_J h_0 - \pi_J h_0 \|_{L^2(X)} \leq o(1) \times \|h_0 - \pi_J h_0\|_{L^2(X)}$; and \\ \phantom{(2)} (b) $\|h_0 - \Pi_J h_0\|_{L^2(X)} \asymp \|h_0 - Q_J h_0\|_{L^2(X)}$. \\
(3) If Assumption \ref{a-approx}(iii') also holds, then: $\|Q_J h_0 - \pi_J h_0\|_\infty \leq O(1) \times \|h_0 - \pi_J h_0\|_{L^2(X)}$. \\
(4) Further, if Condition (\ref{e-piapprox}) also holds, then Assumption \ref{a-approx}(iii) is satisfied.
\end{lemma}

 \begin{proof}[\textbf{Proof of Lemma \ref{lem-approx}}]
 In what follows, ``$ \mathrm{const} $'' denotes a generic positive constant that may be different from line to line.
Assumption \ref{a-data}(iii) guarantees $\tau_J$ and $\pi_J h_0$ are well defined. For part (1.a), we have:
\begin{eqnarray*}
 \|h_0 - \Pi_J h_0\|_{L^2(X)} & \leq & \|h_0 - \pi_J h_0\|_{L^2(X)} \\
 & \leq & \|h_0 - \Pi_J h_0\|_{L^2(X)} + \|\Pi_J h_0 - \pi_J h_0\|_{L^2(X)} \\
 & \leq & \|h_0 - \Pi_J h_0\|_{L^2(X)} + \tau_J  \|T (\pi_J h_0 - \Pi_J h_0)\|_{L^2(W)} \\
 & = & \|h_0 - \Pi_J h_0\|_{L^2(X)} + \tau_J  \|T \pi_J (h_0 - \Pi_J h_0 )\|_{L^2(W)} \\
 & \leq & \|h_0 - \Pi_J h_0\|_{L^2(X)} + \tau_J  \|T (h_0 - \Pi_J h_0 )\|_{L^2(W)} \\
 & = & (1+ \mathrm{const}) \times \|h_0 - \Pi_J h_0\|_{L^2(X)}
\end{eqnarray*}
where the third line is by definition of $\tau_J$, the fourth is because $\pi_J h =h$ for all $h \in \Psi_J$, the final line is by Assumption \ref{a-approx}(ii), and the fifth is because $\pi_J$ is a weak contraction under the norm $h \mapsto \|T h\|_{L^2(W)}$. More precisely by the definition of $\pi_J h_0$ we have:
\begin{equation}\label{FOC-pi}
 \langle Th , T(h_0 - \pi_J h_0) \rangle_W  =  0
\end{equation}
for all $h \in \Psi_J$, where $\langle \cdot, \cdot \rangle_W$ denotes the $L^2(W)$ inner product. With $h=\pi_J h_0 - \Pi_J h_0 \in \Psi_J$ this implies
\begin{equation*}
 \langle T(\pi_J h_0 - \Pi_J h_0) , T(h_0 - \pi_J h_0) \rangle_W  =  0.
\end{equation*}
\begin{equation*}
 \langle T(\pi_J h_0 - \Pi_J h_0) , T(h_0 - \Pi_J h_0) \rangle_W  = \langle T(\pi_J h_0 - \Pi_J h_0) , T(\pi_J h_0 - \Pi_J h_0) \rangle_W .
\end{equation*}
Thus $ \|T (\pi_J h_0 - \Pi_J h_0)\|_{L^2(W)} \leq \|T (h_0 - \Pi_J h_0 )\|_{L^2(W)}$.\\
For part (1.b):
\begin{eqnarray*}
 \tau_J \|T(h_0 - \pi_J h_0)\|_{L^2(W)} & \leq & \tau_J \|T(h_0 - \Pi_J h_0)\|_{L^2(W)} \\
 & \leq & \mathrm{const} \times \|h_0 - \Pi_J h_0\|_{L^2(X)} \\
 & \leq & \mathrm{const} \times \|h_0 - \pi_J h_0\|_{L^2(X)}
\end{eqnarray*}
where the first and final inequalities are by definition of $\pi_J h_0$ and $\Pi_J h_0$ and the second inequality is by Assumption \ref{a-approx}(ii).

For part (2.a), Lemma \ref{lem-ip} guarantees that $Q_J h_0$ is well defined  and that $s_{JK}^{-1} \leq 2 \tau_J$ for all $J$ sufficiently large. By definition of $Q_J h_0$ we have:
\begin{equation}\label{FOC-Q}
 \langle \Pi_K Th , T(h_0 - Q_J h_0) \rangle_W  =  0
\end{equation}
for all $h \in \Psi_J$,
where we use the fact that $\langle \Pi_K f, g \rangle_W = \langle \Pi_K f,\Pi_K g \rangle_W$ holds for any $f,g \in L^2(W)$ since $\Pi_K$ is a projection).
Substituting $h = Q_Jh_0 -\pi_Jh_0 \in \Psi_J$ into the two equations (\ref{FOC-pi}) and (\ref{FOC-Q}) yields:
\begin{eqnarray}
 \langle (T-\Pi_K T)(Q_Jh_0 -\pi_Jh_0) , T(h_0 - \pi_J h_0) \rangle_W + \langle \Pi_K T (Q_Jh_0 -\pi_Jh_0) , T(h_0 - \pi_J h_0) \rangle_W & = & 0  \label{e-ipt1n} \\
 \langle \Pi_K T (Q_Jh_0 -\pi_J h_0) , T(h_0 - Q_J h_0) \rangle_W &= & 0 \,. \label{e-ipt2n}
\end{eqnarray}
By subtracting (\ref{e-ipt2n}) from (\ref{e-ipt1n}) we obtain
\begin{eqnarray*}
 \langle (T - \Pi_K T)(Q_J h_0 - \pi_J h_0),T(h_0 - \pi_J h_0) \rangle_W + \| \Pi_K T(Q_J h_0 - \pi_J h_0)\|_{L^2(W)}^2 &= &0
\end{eqnarray*}
We have therefore proved
\begin{eqnarray}
 \| \Pi_K T(Q_J h_0 - \pi_J h_0)\|_{L^2(W)}^2 & = & |\langle (T - \Pi_K T)(Q_J h_0 - \pi_J h_0),T(h_0 - \pi_J h_0) \rangle_W | \,. \label{e-ipt3n}
\end{eqnarray}
It follows from (\ref{e-ipt3n}), the Cauchy-Schwarz inequality, and Assumption \ref{a-approx}(i)  that:
\begin{eqnarray}
 s_{JK}^2 \|Q_J h_0 - \pi_J h_0\|^2_{L^2(X)} & \leq & \| \Pi_K T (Q_J h_0 - \pi_J h_0) \|^2_{L^2(W)} \notag \\
 & \leq & \|(T - \Pi_K T)(Q_J h_0 - \pi_J h_0)\|_{L^2(W)} \|T (h_0 - \pi_J h_0)\|_{L^2(W)} \label{e-ipt6n} \\
 & \leq & o(\tau_J^{-1}) \|Q_J h_0 - \pi_J h_0 \|_{L^2(X)}  \|T(h_0 - \pi_J h_0)\|_{L^2(W)}  \label{e-ipt4n} \,.
\end{eqnarray}
It follows by (\ref{e-ipt4n}) and the relation $s_{JK}^{-1} \leq 2 \tau_J$  for all $J$ large that:
\begin{eqnarray*}
 \|Q_J h_0 - \pi_J h_0\|_{L^2(X)} & \leq & o(1) \times \tau_J \|T (h_0 - \pi_J h_0)\|_{L^2(W)} \\
 & \leq & o(1) \times \mathrm{const} \times \|h_0 - \pi_J h_0\|_{L^2(X)}
\end{eqnarray*}
where the final line is by part (1.b). For part (2.b), by definition of $Q_J$, $\Pi_J$ and results in part (1.a) and part (2.a), we have:
\begin{eqnarray*}
 \|h_0 - \Pi_J h_0\|_{L^2(X)} & \leq & \|h_0 - Q_J h_0\|_{L^2(X)} \\
 & \leq & \|h_0 - \pi_J h_0\|_{L^2(X)} + \|\pi_J h_0 - Q_J h_0\|_{L^2(X)} \\
 & \leq & \|h_0 - \pi_J h_0\|_{L^2(X)} + o(1) \times \|h_0 - \pi_J h_0\|_{L^2(X)} \\
 & = & (1+ \mathrm{const}) \times \|h_0 - \Pi_J h_0\|_{L^2(X)}.
\end{eqnarray*}
This proves part (2.b).

For part (3), it follows from (\ref{e-ipt6n}) and Assumption \ref{a-approx}(iii') that
\[
  s_{JK}^2 \|Q_J h_0 - \pi_J h_0\|_{L^2(X)} \leq  \mathrm{const} \times  (\zeta_{\psi,J} \tau_J)^{-1}   \|T(h_0 - \pi_J h_0)\|_{L^2(W)}  \,.
\]
and hence
\begin{eqnarray}
 \|Q_J h_0 - \pi_J h_0\|_{L^2(X)} & \leq &  \mathrm{const} \times    \zeta_{\psi,J}^{-1}  \times \tau_J  \|T(h_0 - \pi_J h_0)\|_{L^2(W)} \notag  \\
  & \leq & \zeta_{\psi,J}^{-1} \times \mathrm{const} \times  \|h_0 - \pi_J h_0 \|_{L^2(X)} \label{e-ipt8n}
\end{eqnarray}
by Part (1.b) and the fact that $s_{JK}^{-1} \leq 2 \tau_J$ for all $J$ large.
Therefore,
\begin{eqnarray*}
 \|Q_J h_0 - \pi_J h_0\|_\infty
 & \leq & \zeta_{\psi} \|Q_J h_0 - \pi_J h_0\|_{L^2(X)} \\
 & \leq & \mathrm{const} \times   \|h_0 - \pi_J h_0\|_{L^2(X)}
\end{eqnarray*}
where the last inequality is due to (\ref{e-ipt8n}).

For part (4), by the triangle inequality, the results in part (1.a) and (3) and Condition (\ref{e-piapprox}) we have:
\begin{eqnarray*}
 \|Q_J ( h_0 - \Pi_J h_0) \|_\infty
 & \leq &  \|Q_J h_0 - \pi_J h_0\|_\infty + \|\pi_J h_0 - \Pi_J h_0\|_\infty \\
 & \leq & \mathrm{const} \times   \|h_0 - \Pi_J h_0\|_{L^2(X)} + \| h_0 - \Pi_J h_0\|_\infty \\
 & \leq & O(1) \times   \|h_0 - \Pi_J h_0\|_\infty
\end{eqnarray*}
which completes the proof.
\end{proof}

Note that we may write $\Pi_J h_0(x) = \psi^J(x)'c_J$ for some $c_J$ in $\mb R^J$. We use this notation hereafter.

\begin{proof}[\textbf{Proof of Lemma \ref{lem-bias}}]
We first prove \textbf{Result (1)}. We begin by writing
\begin{eqnarray*}
 \widetilde h(x) - \Pi_J h_0(x)
 & = & Q_J (h_0 - \Pi_J h_0)(x) \\
 & & + \psi^J(x)'(G_b^{-1/2}S)^-_l \{ G_b^{-1/2} (B'(H_0 - \Psi c_J)/n - E[b^K(W_i)(h_0(X_i) - \Pi_J h_0(X_i))] )\}\\
 & & + \psi^J(x)'\{(\wh G_b^{-1/2} \wh S)^-_l\wh G_b^{-1/2}G_b^{1/2} - (G_b^{-1/2} S)^-_l  \}G_b^{-1/2} B'(H_0 - \Psi c_J)/n \\
 & =: & T_1 + T_2 + T_3
\end{eqnarray*}
where $Q_J : L^2(X) \to \Psi_J$ is the sieve 2SLS projection operator given by
\[
 Q_J h(x) = \psi^J(x)' [S' G_b^{-1} S]^{-1} S' G_b^{-1} E[b^K(W_i)h(X_i)]\,.
\]
Note that $Q_J h = h$ for all $h \in \Psi_J$.

Control of $\|T_1\|_\infty$: $\|T_1\|_\infty = O(1) \times \|h_0 - \Pi_J h_0\|_\infty$ by Assumption \ref{a-approx}(iii).

Control of $\|T_2\|_\infty$: Using equation (\ref{e-xident}), the Cauchy-Schwarz inequality, and Lemma \ref{lem-BHJ}, we obtain:
\begin{eqnarray*}
 \|T_2\|_{\infty} & \leq & \sup_x \|\psi^J (x)( G_b^{-1/2} S)^-_l \|_{\ell^2} \| G_b^{-1/2}(B'(H_0 - \Psi c_J)/n - E[b^K(W_i)(h_0(X_i) - \Pi_J h_0(X_i))] ) \|_{\ell^2}  \\
 & \leq & \zeta_{\psi,J} s_{JK}^{-1}\| G_b^{-1/2} (B'(H_0 - \Psi c_J)/n - E[b^K(W_i)(h_0(X_i) - \Pi_J h_0(X_i))] ) \|_{\ell^2}  \\
 & = & \zeta_{\psi,J} s_{JK}^{-1} \times O_p( \sqrt{K/n}) \times \|h_0 - \Pi_J h_0 \|_{\infty}  \,.
\end{eqnarray*}
It then follows by the relations $s_{JK}^{-1} \asymp \tau_J$ (Lemma \ref{lem-ip}) and $\zeta_{\psi,J} \geq \sqrt J \asymp \sqrt K$ and Assumption \ref{a-sieve}(ii) that:
\begin{equation*}
 \|T_2\|_{\infty}  =  O_p (\tau_J \zeta_{\psi,J} \sqrt{J/n}) \times \|h_0 - \Pi_J h_0\|_\infty  = O_p(1) \times \|h_0 - \Pi_J h_0\|_\infty \,.
\end{equation*}

Control of $\|T_3\|_{\infty}$: Similar to $T_2$ in the proof of Lemma \ref{lem-chat}, we may use Lemmas \ref{lem-SGl2}(b) and \ref{lem-ip} to obtain:
\begin{eqnarray}
  \|T_3\|_{\infty} & \leq & \zeta_{\psi,J} \|G_\psi^{1/2}\{(\wh G_b^{-1/2} \wh S)^-_l\wh G_b^{-1/2}G_b^{1/2} - (G_b^{-1/2} S)^-_l\} \|_{\ell^2} \| G_b^{-1/2} B'(H_0 - \Psi c_J)/n\|_{\ell^2} \notag \\
  & = & \zeta_{\psi,J} \times O_p( \tau_J^2  \zeta \sqrt{(\log J)/n)} ) \times \| G_b^{-1/2} B'(H_0 - \Psi c_J)/n\|_{\ell^2} \label{e-bt31} \,.
\end{eqnarray}
Then by Lemma \ref{lem-BHJ} and the triangle inequality, we have:
\begin{eqnarray}
 \| G_b^{-1/2} B'(H_0 - \Psi c_J)/n\|_{\ell^2} & \leq &  O_p ( \sqrt{K/n} )\times \|h_0 - \Pi_J h_0\|_{\infty}  + \|\Pi_K T (h_0 - \Pi_J h_0) \|_{L^2(W)} \notag \\
  & \leq &  O_p ( \sqrt{K/n} )\times \|h_0 - \Pi_J h_0\|_{\infty}  + \| T (h_0 - \Pi_J h_0) \|_{L^2(W)}\,. \label{e-bt32}
\end{eqnarray}
Substituting (\ref{e-bt32}) into (\ref{e-bt31}) and using Assumptions \ref{a-sieve}(ii) and \ref{a-approx}(ii):
\begin{eqnarray*}
 \|T_3\|_{\infty} & \leq &  O_p( \tau_J \zeta^2 /\sqrt n) \times  \left( O_p(\tau_J \sqrt{K(\log J)/n} ) \times \|h_0 - \Pi_J h_0\|_\infty + \tau_J  \|T(h_0 - \Pi_J h_0)\|_{L^2(W)}  \right)\\
 & = & O_p(1) \times \left( o_p(1) \times  \|h_0 - \Pi_J h_0\|_{\infty} + O_p(1) \times \|h_0 - \Pi_J h_0\|_{L^2(X)} \right) \\
 & \leq & O_p(1) \times \|h_0 - \Pi_J h_0\|_\infty
\end{eqnarray*}
where the final line is by the relation between the $L^2(X)$ and sup norms.

\textbf{Result (2)} then follows because
\begin{eqnarray*}
 \|\widetilde h - h_0\|_\infty & \leq & \|\widetilde h - \Pi_J h_0\|_\infty + \|\Pi_J h_0 - h_0\|_\infty \\
  & \leq & (1+O_p(1)) \|\Pi_J h_0 - h_0\|_\infty \\
  & \leq & (1+O_p(1)) (1+\|\Pi_J\|_\infty) \|h_0 - h_{0,J}\|_\infty\,.
\end{eqnarray*}
where the second inequality is by Result (1) and the final line is by Lebesgue's lemma.
\end{proof}

\subsubsection{Proofs for Section \ref{s-upper}}

\begin{proof}[\textbf{Proof of Lemma \ref{lem-chat}}]
Let $u = (u_1,\ldots,u_n)'$. Let $M_n$ be a sequence of positive constants diverging to $+\infty$, and decompose $u_i = u_{1,i} + u_{2,i}$ where
\begin{eqnarray*}
 u_{1,i} & = & u_i\{|u_i| \leq M_n \} - E[u_i\{|u_i| \leq M_n \}|W_i] \\
 u_{2,i} & = & u_i\{|u_i| > M_n \} - E[u_i \{|u_i| > M_n\}|W_i] \\
 u_1 & = & (u_{1,1},\ldots,u_{1,n})' \\
 u_2 & = & (u_{2,1},\ldots,u_{2,n})' \,.
\end{eqnarray*}

For \textbf{Result (1)}, recall that $\xi_{\psi,J} = \sup_{x} \|\psi^J(x)\|_{\ell^1}$. By H\"older's inequality we have
\[
 \|\wh h - \widetilde h\|_\infty = \sup_x |\psi^J(x)'(\wh c - \widetilde c)| \leq \xi_{\psi,J} \|(\wh c - \widetilde c)\|_{\ell^\infty} \,.
\]
To derive the sup-norm convergence rate of the standard deviation term $\wh h - \widetilde h$, it suffices to bound the $\ell^\infty$ norm of the $J \times 1$ random vector $(\wh c - \widetilde c)$. Although this appears like a crude bound, $\xi_{\psi,J}$ grows slowly in $J$ for certain sieves whose basis functions have local support. For such bases the above bound, in conjunction with the following result
\begin{equation} \label{result1}
 \|\wh c - \widetilde c\|_{\ell^\infty} = O_p\left (s_{JK}^{-1} \sqrt{(\log J)/(n e_J)}\right)
\end{equation}
leads to a tight bound on the convergence rate of $\|\wh h - \widetilde h\|_\infty$.

To prove (\ref{result1}), we begin by writing
\begin{eqnarray*}
 \wh c - \widetilde c & = & (\wh G_b^{-1/2} \wh S)^-_l\wh G_b^{-1/2} B'u/n \\
 & = & (G_b^{-1/2}S)^-_l G_b^{-1/2} B'u/n + \{(\wh G_b^{-1/2} \wh S)^-_l\wh G_b^{-1/2} - (G_b^{-1/2} S)^-_l G_b^{-1/2} \} B'u /n \\
 & =: & T_1 + T_2 \,.
\end{eqnarray*}
We will show that $\|T_1\|_{\ell^\infty} = O_p \left(s_{JK}^{-1} \sqrt{(\log J)/(n e_J)}\right)$ and $\|T_2\|_{\ell^\infty} = O_p \left(s_{JK}^{-1} \sqrt{(\log J)/(n e_J)}\right)$.

Control of $\|T_1\|_{\ell^\infty}$. Note that $T_1 = (G_b^{-1/2}S)^-_l G_b^{-1/2} B'u_1/n + (G_b^{-1/2}S)^-_l G_b^{-1/2} B'u_2/n$.

Let $(a)_j$ denote the $j$th element of a vector $a$. By the definition of $\|\cdot\|_{\ell^\infty}$ and the union bound,
\begin{eqnarray}
 \mb P \left( \|(G_b^{-1/2}S)^-_l G^{-1/2}_b B'u_1/n\|_{\ell^\infty} > t \right) & \leq & \mb P \left( \bigcup_{j =1}^J |((G_b^{-1/2}S)^-_l G^{-1/2}_b B'u_1/n)_j | > t \right) \notag \\
 & \leq & \sum_{j =1}^J \mb P \left(  |((G_b^{-1/2}S)^-_l G^{-1/2}_b B'u_1/n)_j | > t \right) \notag \\
 & = & \sum_{j =1}^J \mb P \left( \left| \sum_{i=1}^n q_{j,JK}(W_i) u_{1,i}/n \right| > t \right) \label{e-prob}
\end{eqnarray}
where $q_{j,JK}(W_i) = ((G_b^{-1/2}S)^-_l G_b^{-1/2} b^K(W_i))_j$. The summands may be bounded by noting that
\begin{eqnarray}
 |q_{j,JK}(W_i)| & \leq & \|(G_b^{-1/2}S)^-_l\|_{\ell^2} \| G_b^{-1/2} b^K(W_i)\|_{\ell^2} \notag \\
 & = & \|G_\psi^{-1/2} [G_\psi^{-1/2} S'G_b^{-1} S G_\psi^{-1/2}]^{-1} G_\psi^{-1/2} S' G_b^{-1/2}\|_{\ell^2} \| G_b^{-1/2} b^K(W_i)\|_{\ell^2} \notag \\
 & \leq & \|G_\psi^{-1/2}\|_{\ell^2} \| [G_\psi^{-1/2} S'G_b^{-1} S G_\psi^{-1/2}]^{-1} G_\psi^{-1/2} S'G_b^{-1/2}\|_{\ell^2} \| G_b^{-1/2} b^K(W_i)\|_{\ell^2} \notag \\
 & \leq & \frac{\zeta_{b,K}}{s_{JK} \sqrt{e_J}} \label{e-max1}
\end{eqnarray}
uniformly in $i$ and $j$. Therefore,
\begin{equation} \label{e-max}
 |q_{j,JK}(W_i) u_{1,i}/n| \leq \frac{2 M_n \zeta_{b,K}}{ns_{JK} \sqrt{e_J}}
\end{equation}
uniformly in $i$ and $j$.

Let $(A)_{j|}$ denote the $j$th row of the matrix $A$ and let $(A)_{jj}$ denote its $j$th diagonal element. The second moments of the summands may be bounded by observing that
\begin{eqnarray}
 E[q_{j,JK}(W_i)^2] & = & E[ ((G_b^{-1/2}S)^-_l)_{j|} G_b^{-1/2} b^K(W_i))^2 ] \notag \\
 & = & E [((G_b^{-1/2}S)^-_l)_{j|} G_b^{-1/2} b^K(W_i)b^K(W_i)' G_b^{-1/2} ((G_b^{-1/2}S)^-_l)_{j|}'] \notag \\
 & = & ((G_b^{-1/2}S)^-_l ((G_b^{-1/2}S)^-_l)')_{jj} \notag \\
 & = & ((S'G_b^{-1}S)^{-1})_{jj} \notag \\
 & \leq & \| (S' G_b^{-1} S)^{-1} \|_{\ell^2} \notag \\
 & = & \| G_\psi^{-1/2} [G_\psi^{-1/2} S'G_b^{-1} S G_\psi^{-1/2}]^{-1} G_\psi^{-1/2} \|_{\ell^2} \notag \\
 & \leq & \frac{1}{ s_{JK}^2 e_J} \label{e-var0}
\end{eqnarray}
and so
\begin{equation} \label{e-var}
 E[( q_{j,JK}(W_i) u_{1,i}/n )^2] \leq  \frac{\overline \sigma^2}{n^2 s_{JK}^2 e_J}
\end{equation}
by Assumption \ref{a-residuals}(i) and the law of iterated expectations. Bernstein's inequality and expressions (\ref{e-prob}), (\ref{e-max}) and (\ref{e-var}) yield
\begin{eqnarray}
 & & \mb P \left( \|(G_b^{-1/2}S)^-_l G^{-1/2} B'u_1/n\|_{\ell^\infty} > C s_{JK}^{-1} \sqrt{(\log J)/(n e_J)} \right) \notag \\
 & & \quad \leq \quad  2\exp \left\{ \log J - \frac{  C^2 (\log J)/(ns_{JK}^2 e_J)  }{c_1 /(n s_{JK}^2 e_J) + c_2 C M_n \zeta_{b,K} \sqrt{\log J}/(n^{3/2} s_{JK}^2 e_J)} \right\} \notag \\
 & & \quad = \quad  2\exp \left\{ \log J - \frac{  C^2 (\log J)/(ns_{JK}^2 e_J)  }{1 /(n s_{JK}^2 e_J) [c_1 + c_2 C M_n \zeta_{b,K} \sqrt{(\log J)/n}] } \right\} \label{e-prob-lem1}
\end{eqnarray}
for finite positive constants $c_1$ and $c_2$. Then (\ref{e-prob-lem1}) is $o(1)$ for all large $C$ provided $M_n \zeta_{b,K} \sqrt{(\log J)/n} = o(1)$.

By the triangle and Markov inequalities and (\ref{e-max1}), we have
\begin{eqnarray*}
 \mb P \left( \|(G_b^{-1/2}S)^-_l G^{-1/2} B'u_2/n\|_{\ell^\infty} > t \right)
 & = & \mb P \left( \max_{1 \leq j \leq J} \left| \sum_{i=1}^n q_{j,JK}(W_i) u_{2,i}/n \right| > t \right) \\
 & \leq & \mb P \left( \frac{\zeta_{b,K}}{s_{JK} \sqrt{e_J}}  \sum_{i=1}^n |u_{2,i}/n| > t \right) \\
 & \leq &  \frac{2\zeta_{b,K}}{ts_{JK} \sqrt{e_J}}  E[|u_i|\{|u_i| > M_n\}] \\
 & \leq &  \frac{2\zeta_{b,K}}{ts_{JK} \sqrt{e_J}M_n^{1+\delta}}   E[|u_i|^{2+\delta}\{|u_i| > M_n\}]
\end{eqnarray*}
which, by Assumption \ref{a-residuals}(ii), is $o(1)$ when $t = C s_{JK}^{-1} \sqrt{ ( \log J)/(ne_J)}$ provided $\zeta_{b,K}\sqrt{n/(\log J)} = O(M_n^{1+\delta})$.

Choosing $M_n^{1+\delta} \asymp \zeta_{b,K} \sqrt{n/\log J}$ satisfies the condition $\zeta_{b,K} \sqrt{n/(\log J)} = O(M_n^{1+\delta})$ trivially, and satisfies the condition $M_n \zeta_{b,K} \sqrt{(\log J)/n} = o(1)$ provided $\zeta_{b,K}^{(2+\delta)/\delta} \sqrt{(\log J)/n} = o(1)$, which holds by Assumption \ref{a-sieve}(iii).

Control of $\|T_2\|_{\ell^\infty}$: Using the fact that $\|\cdot\|_{\ell^\infty} \leq \|\cdot\|_{\ell^2}$ on $\mb R^J$ and Lemmas \ref{lem-SGl2}(a) and \ref{lem-BuL2}, we have:
\begin{eqnarray*}
 \|T_2\|_{\ell^\infty} & = & \|\{(\wh G_b^{-1/2} \wh S)^-_l\wh G_b^{-1/2}G_b^{1/2} - (G_b^{-1/2} S)^-_l  \} G_b^{-1/2} B'u /n \|_{\ell^\infty} \\
 & \leq & \| (\wh G_b^{-1/2} \wh S)^-_l\wh G_b^{-1/2} G_b^{1/2} - (G_b^{-1/2} S)^-_l  \|_{\ell^2} \| G_b^{-1/2}  B'u/n\|_{\ell^2} \\
 & = & O_p\left ( s_{JK}^{-2} \zeta \sqrt{(\log K)/(n e_J)} \right) \times O_p ( \sqrt{K/n} ) \\
 & = &  O_p \left(s_{JK}^{-1} \sqrt{(\log J)/(n e_J)}\right) \times O_p (s_{JK}^{-1} \zeta \sqrt{K/n} ) \\
 & = &  O_p \left(s_{JK}^{-1} \sqrt{(\log J)/(n e_J)}\right)
\end{eqnarray*}
where the last equality follows from Assumption \ref{a-sieve}(ii) and the facts that $\zeta \geq \sqrt K$ and $ J \asymp K$.

For \textbf{Result (2)}, we begin by writing
\begin{eqnarray*}
 \wh h(x) - \widetilde h(x) & = & \psi^J(x)' (\wh G_b^{-1/2} \wh S)^-_l\wh G_b^{-1/2} B'u/n \\
 & = & \psi^J(x)'(G_b^{-1/2}S)^-_l G_b^{-1/2} B'u/n + \psi^J(x)'\{(\wh G_b^{-1/2} \wh S)^-_l\wh G_b^{-1/2} - (G_b^{-1/2} S)^-_l G_b^{-1/2} \} B'u /n \\
 & =: & T_1 + T_2 \,.
\end{eqnarray*}
We will show that $\|T_1\|_{\infty} = O_p \left(\tau_J \zeta_{\psi,J} \sqrt{(\log n)/n}\right)$ and $\|T_2\|_{\infty} = O_p \left(\tau_J \zeta_{\psi,J} \sqrt{(\log n)/n}\right)$.

Control of $\|T_1\|_{\infty}$. Note that $T_1 = \psi^J(x)'(G_b^{-1/2}S)^-_l G_b^{-1/2} B'u_1/n + \psi^J(x)'(G_b^{-1/2}S)^-_l G_b^{-1/2} B'u_2/n$.

Let $\mathcal X_n \subset \mathcal X$ be a grid of finitely many points such that for each $x \in \mathcal X$ there exits a $\bar x_n(x) \in \mathcal X_n$ such that $\|x - \bar x_n(x)\| \lesssim (\zeta_{\psi,J} J^{-(\omega + \frac{1}{2})})^{1/\omega'}$, where $\omega,\omega'$ are as in Assumption \ref{a-sieve}(i). By compactness and convexity of the support $\mathcal X$ of $X_i$, we may choose $\mathcal X_n$ to have cardinality $\#(\mathcal X_n) \lesssim n^\beta$ for some $0 < \beta < \infty$. Therefore,
\begin{align*}
 & \sup_x \|\psi^J(x)' ( G_b^{-1/2}  S)^-_l G_b^{-1/2} B'u_1/n\|_\infty  \\
 & \leq \max_{ x_n \in \mathcal X_n} |\psi^J( x_n)' ( G_b^{-1/2}  S)^-_l G_b^{-1/2} B'u_1/n| + \sup_x |\{\psi^J(x) - \psi^J(\bar x_n(x))\}' ( G_b^{-1/2}  S)^-_l G_b^{-1/2} B'u_1/n| \\
 & \leq \max_{ x_n \in \mathcal X_n} |\psi^J( x_n)' ( G_b^{-1/2}  S)^-_l G_b^{-1/2} B'u_1/n| + C_\omega J^\omega \zeta_{\psi,J}  J^{-(\omega + \frac{1}{2} )} s_{JK}^{-1} \| G_b^{-1/2} B'u_1/n\|_{\ell^2} \\
 & = \max_{ x_n \in \mathcal X_n} |\psi^J( x_n)' ( G_b^{-1/2}  S)^-_l G_b^{-1/2} B'u_1/n|  + C_\omega J^\omega \zeta_{\psi,J}  J^{-(\omega + \frac{1}{2} )} s_{JK}^{-1} \times O_p(\sqrt{J/n}) \\
 & = \max_{ x_n \in \mathcal X_n} |\psi^J( x_n)' ( G_b^{-1/2}  S)^-_l G_b^{-1/2} B'u_1/n| + o_p( s_{JK}^{-1} \zeta_{\psi,J} \sqrt{(\log J)/n} )
\end{align*}
for some finite positive constant $C_\omega$, where the first inequality is by the triangle inequality, the second is by H\"older continuity of the basis for $\Psi_J$ and similar reasoning to that used in equation (\ref{e-xident}), the first equality is by Lemma \ref{lem-BuL2} and the fact that $J \asymp K$, and the final equality is because $(\log J)^{-1/2} = o(1)$. For each $x_n \in \mathcal X_n$ we may write
\begin{eqnarray*}
 \psi^J( x_n)' ( G_b^{-1/2}  S)^-_l G_b^{-1/2} B'u_1/n & = & \frac{1}{n} \sum_{i=1}^n g_{n,i}(x_n) u_{1,i}, \mbox{ where} \\
 g_{n,i}(x_n) & = & \psi^J( x_n)' ( G_b^{-1/2}  S)^-_l G_b^{-1/2} b^K(W_i)\,.
\end{eqnarray*}
It follows from equation (\ref{e-xident}) and the Cauchy-Schwarz inequality that the bounds
\begin{eqnarray*}
 |g_{n,i}(x_n) | & \leq & s_{JK}^{-1} \zeta_{\psi,J} \zeta_{b,K} \\
 E[g_{n,i}(x_n)^2] & = & \psi^J(x_n)' ( G_b^{-1/2}  S)^-_l (( G_b^{-1/2}  S)^-_l )' \psi^J(x_n) \quad \leq \quad s_{JK}^{-2} \zeta_{\psi,J}^2
\end{eqnarray*}
hold uniformly for $x_n \in \mathcal X_n$. Therefore, by Assumption \ref{a-residuals}(i) and iterated expectations, the bounds
\begin{eqnarray*}
 |g_{n,i}(x_n) u_{1,i}| & \leq &2 s_{JK}^{-1} \zeta_{\psi,J} \zeta_{b,K} M_n\\
 E[g_{n,i}(x_n)^2 u_{1,i}^2] & \leq &\overline \sigma^2 \psi^J(x_n)' ( G_b^{-1/2}  S)^-_l (( G_b^{-1/2}  S)^-_l )' \psi^J(x_n)  \quad \leq \quad \overline \sigma^2 s_{JK}^{-2} \zeta_{\psi,J}^2
\end{eqnarray*}
hold uniformly for $x_n \in \mathcal X_n$. It follows by the union bound and Bernstein's inequality that
\begin{eqnarray}
 & & \mb P \left( \max_{ x_n \in \mathcal X_n} |\psi^J( x_n)' ( G_b^{-1/2}  S)^-_l G_b^{-1/2} B'u_1/n| > C s_{JK}^{-1}\zeta_{\psi,J} \sqrt{(\log n)/n} \right) \notag \\
 & & \quad \leq \quad \#(\mathcal X_n) \max_{x_n \in \mathcal X_n} \mb P \left(  \left|\frac{1}{n} \sum_{i=1}^n g_{n,i}(x_n) u_{1,i}\right| > C s_{JK}^{-1} \zeta_{\psi,J} \sqrt{(\log n)/n} \right) \notag \\
 & & \quad \lesssim \quad  \exp \left\{ \beta \log n - \frac{ C^2 \zeta_{\psi,J}^2 (\log n)/(ns_{JK}^2 )  }{c_1 \zeta_{\psi,J}^2 /(n s_{JK}^2)[1 + (c_2/c_1) (C M_n \zeta_{b,K} \sqrt{(\log n)/n} ]} \right\}  \label{e-prob-lem2n}
\end{eqnarray}
for finite positive constants $c_1$ and $c_2$. Then (\ref{e-prob-lem2n}) is $o(1)$ for all large $C$ provided $M_n \zeta_{b,K} \sqrt{(\log n)/n} = o(1)$.

By the triangle and Markov inequalities and equation (\ref{e-xident}), we have
\begin{eqnarray*}
 \mb P \left( \|\psi^J(x)(G_b^{-1/2}S)^-_l G_b^{-1/2} B'u_2/n\|_{\infty} > t \right)
 & \leq & \mb P \left(  s_{JK}^{-1} \zeta_{\psi,J} \| G_b^{-1/2} B'u_2/n \|_{\ell^2} > t \right) \\
 & \leq & \mb P \left( s_{JK}^{-1} \zeta_{\psi,J} \zeta_{b,K} \sum_{i=1}^n |u_{2,i}/n| > t \right) \\
 & \leq &  \frac{2\zeta_{\psi,J} \zeta_{b,K}}{ts_{JK}}  E[|u_i|\{|u_i| > M_n\}] \\
 & \leq &  \frac{2\zeta_{\psi,J} \zeta_{b,K}}{ts_{JK} M_n^{1+\delta}}   E[|u_i|^{2+\delta}\{|u_i| > M_n\}]
\end{eqnarray*}
which, by Assumption \ref{a-residuals}(ii), is $o(1)$ when $t = C s_{JK}^{-1}\zeta_{\psi,J} \sqrt{ ( \log n)/n}$ provided $\zeta_{b,K}\sqrt{n/(\log n)} = O(M_n^{1+\delta})$.

Choosing $M_n^{1+\delta} \asymp \zeta_{b,K} \sqrt{n/\log n}$ satisfies the condition $\zeta_{b,K} \sqrt{n/(\log n)} = O(M_n^{1+\delta})$ trivially, and satisfies the condition $M_n \zeta_{b,K} \sqrt{(\log n)/n} = o(1)$ provided $\zeta_{b,K}^{(2+\delta)/\delta} \sqrt{(\log n)/n} = o(1)$, which holds by Assumption \ref{a-sieve}(iii). We have therefore proved that $\|T_1\|_\infty = O_p( s_{JK}^{-1} \zeta_{\psi,J} \sqrt{(\log n)/n})$. It follows by the relation $\tau_J \asymp s_{JK}^{-1}$ (Lemma \ref{lem-ip}) that $\|T_1\|_\infty = O_p( \tau_J \zeta_{\psi,J} \sqrt{(\log n)/n})$.

Control of $\|T_2\|_{\infty}$: Using the fact that $\|h\|_{\infty} \leq \zeta_{\psi,J} \|h\|_{L^2(X)}$ on $\Psi_J$ and Lemmas \ref{lem-SGl2}(b) and \ref{lem-BuL2} and the relation $\tau_J \asymp s_{JK}^{-1}$, we have:
\begin{eqnarray*}
 \|T_2\|_{\infty} & \leq &\zeta_{\psi,J} \| G_\psi^{1/2}\{(\wh G_b^{-1/2} \wh S)^-_l\wh G_b^{-1/2}G_b^{1/2} - (G_b^{-1/2} S)^-_l  \} G_b^{-1/2} B'u /n \|_{\ell^2} \\
 & \leq & \zeta_{\psi,J} \| G_\psi^{1/2} \{ (\wh G_b^{-1/2} \wh S)^-_l\wh G_b^{-1/2} G_b^{1/2} - (G_b^{-1/2} S)^-_l\}  \|_{\ell^2} \| G_b^{-1/2}  B'u/n\|_{\ell^2} \\
 & = & \zeta_{\psi,J} O_p\left ( \tau_J^2 \zeta \sqrt{(\log J)/n} \right) \times O_p ( \sqrt{K/n} ) \\
 & = &  O_p \left(\tau_J \zeta_{\psi,J} \sqrt{(\log J)/n}\right) \times O_p (\tau_J \zeta \sqrt{K/n} ) \\
 & = &  O_p \left(\tau_J \zeta_{\psi,J} \sqrt{(\log J)/n}\right) \times O_p(1)
\end{eqnarray*}
where the last equality follows from Assumption \ref{a-sieve}(ii) and the fact that $\zeta \geq \sqrt J \asymp \sqrt K$.
\end{proof}

\begin{proof}[\textbf{Proof of Theorem \ref{t-rate}}]
We decompose $\|\wh h - h_0\|_\infty$ into three parts:
\[
 \|\wh h - h_0\|_\infty \leq \|\wh h - \widetilde h\|_\infty + \|\widetilde h - \Pi_J h_0\|_\infty + \|\Pi_J  h_0 - h_0\|_\infty\,.
\]
where $\|\wh h - \widetilde h\|_\infty = O_p (\tau_J \xi_{\psi,J} \sqrt{ (\log J)/ (n e_J )} )$ by Lemma \ref{lem-chat}(1) and $\|\widetilde h - \Pi_J h_0\|_\infty = O_p (1) \times \|\Pi_J  h_0 - h_0\|_\infty$ by Lemma \ref{lem-bias}.
\end{proof}

\begin{proof}[\textbf{Proof of Corollary \ref{c-deriv}}]
For \textbf{Result (1)}, note that Assumption \ref{a-sieve}(ii) is satisfied with $\zeta = O(J^{1/2} )$ for $\Psi_J$ and $B_K$ being spline, or wavelet or cosine sieves. Next, by the lemmas in Appendix \ref{ax-basis}, $\xi_{\psi,J} /\sqrt{e_J} =O(J^{1/2} )$ for $\Psi_J$ being spline or wavelet sieves. Also, $\|\Pi_J\|_\infty \lesssim 1$ for $\Psi_J$ being a spline sieve (\cite{Huang2003}) or a tensor product CDV wavelet sieve (\cite{ChenChristensen-reg}). For $h_0 \in B_\infty(p,L)$ and $\Psi_J$ being spline or wavelet sieves, Lemma \ref{lem-bias} implies that
\[
\|\widetilde h - h_0\|_\infty = O_p (J^{-p/d}).
\]
Note that Bernstein inequalities (or inverse estimates) from approximation theory imply that
\[
 \| \partial^\alpha h\|_\infty = O(J^{|\alpha|/d}) \|h\|_\infty
\]
for all $h \in \Psi_J$ (see \cite{Schumaker2007} for splines and \cite{Cohen2003} for wavelets on domains). Therefore,
\begin{align*}
 \|\partial^\alpha \widetilde h - \partial^\alpha h_0 \|_\infty & \leq \|\partial^\alpha \widetilde h - \partial^\alpha (\Pi_J h_0) \|_\infty + \|\partial^\alpha (\Pi_J h_0) - \partial^\alpha h_0 \|_\infty \\
 & \leq O(J^{|\alpha|/d}) \|\widetilde h - \Pi_J h_0 \|_\infty + \|\partial^\alpha (\Pi_J h_0) - \partial^\alpha h_0 \|_\infty \\
 & \leq O_p( J^{-(p-|\alpha|)/d} ) + \|\partial^\alpha (\Pi_J h_0) - \partial^\alpha h_0 \|_\infty
\end{align*}
Let $h_J$ be any element of $\Psi_J$. Since $\Pi_J h_J = h_J$, we have:
\begin{align*}
 \|\partial^\alpha (\Pi_J h_0) - \partial^\alpha h_0 \|_\infty & = \|\partial^\alpha (\Pi_J (h_0 -  h_J)) + \partial^\alpha h_J - \partial^\alpha h_0 \|_\infty \notag \\
 & \leq O(J^{|\alpha|/d}) \|\Pi_J (h_0 -  h_J))\|_\infty + \|\partial^\alpha h_J - \partial^\alpha h_0 \|_\infty \notag \\
 & \leq O(J^{|\alpha|/d}) \times \mathrm{const} \times \|h_0 -  h_J\|_\infty + \|\partial^\alpha h_J - \partial^\alpha h_0 \|_\infty \label{e-bias-hj-leb} \,.
\end{align*}
The above inequality holds uniformly in $h_J \in \Psi_J$. Choosing $h_J$ such that $\|h_0 -  h_J\|_\infty = O(J^{-p/d})$ and $\|\partial^\alpha h_J - \partial^\alpha h_0 \|_\infty = O(J^{-(p-|\alpha|)/d})$ yields the desired result.

For \textbf{Result (2)}, Theorem \ref{t-rate} implies that
\[
\|\wh h - h_0\|_\infty = O_p (J^{-p/d} + \tau_J \sqrt{ (J\log J)/n} ).
\]
By similar arguments to the above, we have:
\begin{align*}
 \|\partial^\alpha \wh h - \partial^\alpha h_0 \|_\infty & \leq \|\partial^\alpha \wh h - \partial^\alpha \widetilde h \|_\infty + \|\partial^\alpha \widetilde h - \partial^\alpha h_0 \|_\infty \\
 & \leq O(J^{|\alpha|/d}) \|\wh h - \widetilde h \|_\infty + \|\partial^\alpha \widetilde h - \partial^\alpha h_0 \|_\infty \\
 & \leq O_p\left( J^{|\alpha|/d}  \left( \tau_J \sqrt{(J \log J)/n} \right)  \right) + \|\partial^\alpha \widetilde h - \partial^\alpha h_0 \|_\infty
\end{align*}
and the result follows by Result (1).

For \textbf{Results (2.a) and (2.b)}, Assumption \ref{a-sieve}(ii)(iii) is satisfied if $\tau_J \times J /\sqrt n = O(1)$ and $J^{(2+\delta)/\delta} (\log n)/n = o(1)$. This is satisfied given the stated conditions with the optimal choice of $J$ for mildly ill-posed case and severely ill-posed case respectively.
\end{proof}

\subsection{Proofs for Section \ref{s-lower}}

\begin{proof}[\textbf{Proof of Theorem \ref{npiv lower bound}}]
Consider the Gaussian reduced-form NPIR model with known operator $T$:
\begin{equation}  \label{e-npir}
\begin{array}{rcl}
Y_{i} & = & Th_0(W_{i}) + u_i\\
u_i|W_i & \sim & N(0,\sigma^2 (W_i) ) %\\
\end{array}%
\end{equation}
for $1 \leq i \leq n$, where $W_i$ is continuously distributed over $\mathcal W$ with density uniformly bounded away from $0$ and $\infty$. As in \cite{ChenReiss}, Theorem \ref{npiv lower bound} is proved by (i) noting that the risk (in sup-norm loss) for the NPIV model is at least as large as the risk (in sup-norm loss) for the NPIR model, and (ii) calculating a lower bound (in sup-norm loss) for the NPIR model.  Theorem \ref{npiv lower bound} therefore follows from a sup-norm analogue of Lemma 1 of \cite{ChenReiss} and Theorem \ref{npir lower bound}, which establishes a lower bound on minimax risk over H\"older classes under sup-norm loss for the NPIR model.
\end{proof}

\begin{theorem}\label{npir lower bound}
Let Condition LB hold for the NPIR model (\ref{e-npir}) with a random sample $\{(W_i,Y_i)\}_{i=1}^n$.  Then for any $0 \leq |\alpha| < p$:
\begin{equation*}
 \liminf_{n \to \infty} \inf_{\wh g_n} \sup_{h \in B_\infty(p,L)} \mb P_h \left( \|\wh g_n - \partial^\alpha h\|_\infty \geq c(n/\log n)^{-(p-|\alpha|)/(2(p+\varsigma)+d)} \right) \geq c'>0
\end{equation*}
in the mildly ill-posed case, and
\begin{equation*}
 \liminf_{n \to \infty} \inf_{\wh g_n} \sup_{h \in B_\infty(p,L)} \mb P_h \left( \|\wh g_n - \partial^\alpha h\|_\infty \geq c(\log n)^{-(p-|\alpha|)/\varsigma} \right) \geq c'>0
\end{equation*}
in the severely ill-posed case, where $\inf_{\wh g_n}$ denotes the infimum over all estimators of $\partial^\alpha h$ based on the sample of size $n$, $\sup_{h \in B_\infty(p,L)} \mb P_h$ denotes the sup over $h \in B_\infty(p,L)$  and distributions  $(W_i,u_i)$ which satisfy Condition LB with $\nu$ fixed, and the finite positive constants $c, c'$ depend only on $p,L,d,\varsigma$ and $\sigma_0$.
\end{theorem}

\begin{proof}[\textbf{Proof of Theorem \ref{npir lower bound}}]
We establish the lower bound by applying Theorem 2.5 of \cite{Tsybakov2009} (see Theorem \ref{t-tsybakov} below). We first explain the scalar $(d = 1)$ case in detail. Let $\{\phi_{j,k},\psi_{j,k}\}_{j,k}$ be a wavelet basis of regularity $\gamma > p$ for $L^2([0,1])$ as described in Appendix \ref{ax-basis}. Recall that this basis is generated by a Daubechies pair $(\varphi,\psi)$ where $\varphi$ has support $[-N+1,N]$. We will define a family of submodels in which we perturb $h_0$ by elements of the wavelet space $W_j$, where we choose $j$ deterministically with $n$. For given $j$, recall that the wavelet space $W_j$ consists of $2^j$ functions $\{\psi_{j,k}\}_{0 \leq k \leq 2^j-1}$, such that $\{\psi_{j,k}\}_{r \leq k \leq 2^j-N-1}$ are interior wavelets for which $\psi_{j,k}(\cdot) = 2^{j/2}\psi(2^j(\cdot)-k)$.

By construction, the support of each interior wavelet is an interval of length $2^{-j}(2r-1)$. Thus for all $j$ sufficiently large (hence the $\liminf$ in our statement of the Lemma) we may choose a set $M \subset \{r,\ldots,2^j-N-1\}$ of interior wavelets with $\#(M) \gtrsim 2^j$ such that $\mbox{support}(\psi_{j,m}) \cap \mbox{support}(\psi_{j,m'}) = \emptyset$ for all $m,m' \in M$ with $m \neq m'$. Note also that by construction we have $\#(M) \leq 2^j$ (since there are $2^j-2N$ interior wavelets).

Recall the norms $\|\cdot\|_{b^p_{\infty,\infty}}$ defined in Appendix \ref{ax-basis}. Let $h_0 \in B_\infty(p,L)$ be such that $\|h_0\|_{B^p_{\infty,\infty}} \leq L/2$, and for each $m \in M$ let
\begin{equation*}
 h_m = h_0 + c_0 2^{-j(p+1/2)} \psi_{j,m}
\end{equation*}
where $c_0$ is a positive constant to be defined subsequently. Noting that
\begin{eqnarray*}
  c_0 2^{-j(p+1/2)} \|\psi_{j,m}\|_{B^p_{\infty,\infty}} & \lesssim &  c_0 2^{-j(p+1/2)} \|\psi_{j,m}\|_{b^p_{\infty,\infty}} \\
  & \leq & c_0
\end{eqnarray*}
it follows by the triangle inequality that $\|h_m\|_{B^p_{\infty,\infty}} \leq L$ uniformly in $m$ for all sufficiently small $c_0$. By Condition LB, let $W_i$ be distributed such that $X_i$ has uniform marginal distribution on $[0,1]$. For $m \in \{0\} \cup M$ let $P_m$ be the joint distribution of $\{(W_i,Y_i)\}_{i=1}^n$ with $Y_i = T h_m(W_i) + u_i$ for the Gaussian NPIR model (\ref{e-npir}).

For any $m \in M$
\begin{eqnarray*}
 \| \partial^\alpha h_0 - \partial^\alpha h_m\|_\infty & = & c_0 2^{-j(p+1/2)}\|\partial^\alpha \psi_{j,m} \|_\infty \\
 & = & c_0 2^{-j(p-|\alpha|)}\|\psi^{(|\alpha|)}\|_\infty
\end{eqnarray*}
where $\psi^{(|\alpha|)}$ denotes the $|\alpha|$th derivative of $\psi$. Moreover, for any $m,m' \in M$ with $m \neq m'$
\begin{eqnarray*}
 \|\partial^\alpha h_m - \partial^\alpha h_{m'}\|_\infty & = & c_0 2^{-j(p+1/2)}\| \partial^{\alpha} \psi_{j,m} - \partial^{\alpha}\psi_{j,m'} \|_\infty \\
 & = & 2 c_0 2^{-j(p-|\alpha|)}\|\psi^{(|\alpha|)}\|_\infty
\end{eqnarray*}
by virtue of the disjoint support of $\{\psi_{j,m}\}_{m \in M}$.

By Condition LB(iii),
\[
 \|T \psi_{j,m}(W_i)\|_{L^2(W)} \lesssim \nu(2^j)^2 \langle \psi_{j,m} ,\psi_{j,m} \rangle_X^2 = \nu(2^j)^2
\]
(because $c_0 2^{-j(p+1/2)}\psi_{j,m} \in \mathcal H_2(p,L)$ for sufficiently small $c_0$) where $\nu(2^j) = 2^{-j \varsigma}$ in the mildly ill-posed case and $\nu(2^j) = \exp(-  2^{j \varsigma})$ in the severely ill-posed case.
The KL distance $K(P_m,P_0)$ is
\begin{eqnarray*}
 K(P_m,P_0) & \leq & \frac{1}{2}\sum_{i=1}^n (c_0 2^{-j(p+1/2)})^2 E\left[\frac{(T \psi_{j,m}(W_i))^2 }{\sigma^2(W_i)}\right] \\
 & \leq & \frac{1}{2}\sum_{i=1}^n (c_0 2^{-j(p+1/2)})^2 \frac{E\left[(T \psi_{j,m}(W_i))^2\right] }{\underline \sigma^2}  \\
 & \lesssim &  n(c_0 2^{-j(p+1/2)})^2 \nu(2^j)^2 \,.
\end{eqnarray*}
In the mildly ill-posed case ($\nu(2^j) =  2^{-j \varsigma}$) we choose $2^{j} \asymp (n/(\log n))^{1/(2(p+\varsigma)+1)}$. This yields:
\begin{align*}
 K(P_m,P_0) &\lesssim c_0^2 \log n \quad \mbox{ uniformly in $m$}\\
 \log(\#(M)) &\gtrsim \log n + \log \log n \,.
\end{align*}
since $\#(M) \asymp 2^j$.

In the severely ill-posed case ($\nu(2^j)= \exp(- \frac{1}{2} 2^{j \varsigma})$) we choose $2^j = (c_1 \log n)^{1/\varsigma}$ with $c_1 > 1$. This yields:
\begin{align*}
 K(P_m,P_0) &\lesssim n^{-(c_1-1)} \quad \mbox{ uniformly in $m$}\\
 \log(\#(M)) &\gtrsim  \log \log n \,.
\end{align*}
In both the mildly and severely ill-posed cases, we may choose $c_0$ sufficiently small that both $\|h_m\|_{B^p_{\infty,\infty}} \leq L$ and $K(P_m,P_0) \leq \frac{1}{8} \log (\#(M))$ hold uniformly in $m$ for all $n$ sufficiently large. All conditions of Theorem 2.5 of \cite{Tsybakov2009} are satisfied and hence we obtain the lower bound result.

In the multivariate case ($d > 1$) we let $\widetilde \psi_{j,k,G}(x)$ denote an orthonormal tensor-product wavelet for $L^2([0,1]^d)$ at resolution level $j$ (see Appendix \ref{ax-basis}). We construct a family of submodels analogously to the univariate case, setting $h_m = h_0 + c_0 2^{-j(p+d/2)} \widetilde \psi_{j,m,G}$ where $\widetilde \psi_{j,m,G}$ is now the tensor product of $d$ interior univariate wavelets at resolution level $j$ with $G = (w_\psi)^d$ and where $\#(M) \asymp 2^{jd}$. By condition LB we obtain
\begin{equation*}
 \|\partial^\alpha h_m - \partial^\alpha h_{m'}\|_\infty \gtrsim c_0 2^{-j(p-|\alpha|)}
\end{equation*}
for each $m,m' \in \{0\} \cup M$ with $m \neq m'$, and
\begin{eqnarray*}
 K(P_m,P_0) & \lesssim &  n (c_0 2^{-j(p+d/2)})^2 \nu(2^j)^2
\end{eqnarray*}
for each $m \in M$, where $\nu(2^j) = 2^{-j \varsigma}$ in the mildly ill-posed case and $\nu(2^j) \asymp \exp(-  2^{j \varsigma})$ in the severely ill-posed case. We choose $2^j \asymp (n/\log n)^{1/(2(p+\varsigma)+d)}$ in the mildly ill-posed case and $2^j = (c_1 \log n)^{1/\varsigma}$ in the severely ill-posed case.
The result follows as in the univariate case.
\end{proof}

The following theorem is a special case of Theorem 2.5 on p. 99 of \cite{Tsybakov2009} which we use to prove the minimax lower bounds in sup- and $L^2$-norm loss for $h_0$ and its derivatives. We state the result here for convenience.

\begin{theorem}[\cite{Tsybakov2009}] \label{t-tsybakov}
Assume that $\#(M) \geq 2$ and suppose that $(\mathcal H,\|\cdot\|_{\mathcal H})$ contains elements $\{h_m : m \in \{0\} \cup M\}$ such that:\\
(i) $\|\partial^\alpha h_m - \partial^\alpha h_{m'}\|_{\mathcal H} \geq 2s > 0$ for each $m,m' \in M \cup \{0\}$ with $m \neq m'$;\\
(ii) $P_m \ll P_0$ for each $m \in M$ and
\[
 \frac{1}{\#(M)} \sum_{m \in M} K(P_m,P_0) \leq \mf a \log (\#(M))
\]
with $0 < \mf a < \frac{1}{8}$ and where $P_m$ denotes the distribution of the data when $h = h_m$ for each $m \in \{0\} \cup M$. Then:
\[
 \inf_{ \wh g} \sup_{h \in \mathcal H} \mb P_h(\|\wh g - \partial^\alpha h\|_{\mathcal H} \geq s) \geq \frac{\sqrt{\#(M)}}{1 + \sqrt {\#(M)}} \left( 1- 2 \mf a - \sqrt{\frac{2\mf a}{\log (\#(M))}} \right) > 0\,.
\]
\end{theorem}

\subsection{Proofs for Section \ref{s-exog}}

\begin{proof}[\textbf{Proof of Lemma \ref{lem-bck}}]
We first prove \textbf{Result (1)}. Let $P_{J-1,z} = clsp \{\phi_{01,z},\ldots,\phi_{0J-1,z}\}$ and let $P_{J-1,z}^\perp$ denote its orthogonal complement in $L^2(X_1|Z=z)$. Observe that by definition of the singular values, for each $z$ we have:
\begin{align} \label{e-singular}
 \sup_{h_z \in P_{J-1,z}^\perp: \|h_z\|_{L^2(X_1|Z=z)} = 1} \| T_z h_z \|_{L^2(W_1|Z=z)}^2 & = \sup_{h_z \in P_{J-1,z}^\perp: \|h_z\|_{L^2(X_1|Z=z)} = 1}  \langle (T_z^* T_z^{\phantom *}) h_z,h_z \rangle_{X_1|Z=z} \notag \\
 & = \mu_{J,z}^2\,.
\end{align}
Then let $P_{J-1}^\perp = \{ h(x_1,z) \in L^2(X) : h(\cdot,z) \in P_{J-1,z}^\perp$ for each $z\}$. Note that $\phi_{0j} \in \{h \in P_{J-1}^\perp : \|h(\cdot,z)\|_{L^2(X_1|Z=z)}=1\,\forall z\}$ for each $j \geq J$. Then:
\begin{align}
 \tau_J^{-2} & = \inf_{h \in \Psi_J : \|h\|_{L^2(X)=1}} \|Th\|^2_{L^2(W)} \notag \\
 & \leq \inf_{h \in \Psi_J \cap P_{J-1}^\perp : \|h\|_{L^2(X_1|Z=z)=1 }\,\forall z} \|Th\|^2_{L^2(W)} \notag \\
 & \leq \sup_{h \in \Psi_J \cap P_{J-1}^\perp : \|h\|_{L^2(X_1|Z=z)=1 }\,\forall z} \|Th\|^2_{L^2(W)} \notag \\
 & \leq \sup_{h \in P_{J-1}^\perp : \|h\|_{L^2(X_1|Z=z)=1 }\;\forall z} \|Th\|^2_{L^2(W)} \label{e-bck1}\,.
\end{align}
Let $F_Z$ denote the distribution of $Z$. For any $h \in P_{J-1}^\perp$ let $h_z(x_1) = h(x_1,z)$ and observe that $h_z \in P_{J-1,z}^\perp$. By iterated expectations and (\ref{e-singular}), for any $h \in P^\perp_{J-1}$ with $\|h_z\|_{L^2(X_1|Z=z)} = 1$ for each $z$, we have:
\begin{align}
 \|Th\|^2_{L^2(W)} & = \int \| E[h(X_{1i},z)|W_{1i},Z_i =z] \|_{L^2(W_1|Z=z)}^2 \,\mathrm dF_Z(z) \notag \\
 & = \int \| T_z h_z \|_{L^2(W_1|Z=z)}^2 \,\mathrm dF_Z(z) \notag \\
 & \leq \int \mu_{J,z}^2 \|h_z\|_{L^2(X_1|Z=z)}^2 \,\mathrm dF_Z(z) \notag \\
 & = \int \mu_{J,z}^2 \,\mathrm dF_Z(z) = E[\mu_{J,Z_i}^2] \,. \label{e-bck2}
\end{align}
It follows by substituting (\ref{e-bck2}) into (\ref{e-bck1}) that $\tau_J \geq E[\mu_{J,Z_i}^2]^{-1/2}$.

To prove \textbf{Result (2)}, note that any $h \in \Psi_J$ with $h \neq 0$ can be written as $\sum_{j=1}^J a_j \phi_{0j}$ for constants $a_j = a_j(h)$ where
\[
 \|h\|^2_{L^2(X)}  = E \left[ E \left[ \left. \left( \sum_{j=1}^J a_j \phi_{0j}(X_{1i},Z_i) \right)^2 \right| Z_i \right] \right] = \sum_{j=1}^J a_j^2
\]
since $E[\phi_{0j,z}(X_i) \phi_{0k,z}(X_i)|Z_i = z] = \delta_{jk}$ where $\delta_{jk}$ denotes the Kronecker delta. Moreover:
\begin{align*}
 \|T h\|^2_{L^2(W)} & = E \left[ \left(  E \left[ \left. \sum_{j=1}^J a_j \phi_{0j}(X_{1i},Z_i)  \right| W_{1i},Z_i \right] \right)^2 \right]\\
 & = E \left[ \left( E \left[ \left. \sum_{j=1}^J a_j \phi_{0j,Z_i}(X_{1i}) \right| W_{1i},Z_i \right] \right)^2 \right] \\
 & = E \left[ \left( \sum_{j=1}^J a_j \mu_{j,Z_i} \phi_{1j,Z_i}(W_{1i})\right)^2 \right] \\
 & = E \left[ E \left[ \left. \left( \sum_{j=1}^J a_j \mu_{j,Z_i} \phi_{1j,Z_i}(W_{1i})\right)^2 \right| Z_i \right] \right]  = \sum_{j=1}^J a_j^2  E \left[ \mu_{j,Z_i}^2 \right] \geq \|h\|^2_{L^2(X)} E[\mu_{J,Z_i}^2]
\end{align*}
since $E[\phi_{1j,z}(W_{1i}) \phi_{1k,z}(W_{1i})|Z_i = z] = \delta_{jk}$. Therefore,
\[
 \tau_J =  \sup_{h \in \Psi_J} \frac{\|h\|_{L^2(X)}}{\|Th\|_{L^2(W)}} \leq \frac{1}{E[\mu_{J,Z_i}^2]^{1/2}}
\]
as required.
\end{proof}

\subsection{Proofs for Appendix \ref{s-pw} and Section \ref{s-ucb}}

Since the proofs for uniform inference theories (in Section \ref{s-ucb}) built upon that for the pointwise normality Theorem \ref{t-dist} (in Appendix \ref{s-pw}), we shall present the proof of Theorem \ref{t-dist} first.

\subsubsection{Proofs for Appendix \ref{s-pw}}

\begin{proof}[\textbf{Proof of Theorem \ref{t-dist}}]
We first prove \textbf{Result (1)}. By Assumption \ref{a-functional-prime}'(a) or \ref{a-functional-prime}'(b)(i)(ii) we have:
\[
 \sqrt n   \frac{(f(\wh h) - f(h_0))}{\sigma_n(f)}  = \sqrt n \frac{D f(h_0)[\wh h - \widetilde h]}{\sigma_n(f)} + o_p(1)\,.
\]
Define
\[
 Z_n (W_i)  = \frac{(D f(h_0)[\psi^J])' [S' G_b^{-1} S]^{-1} S' G_b^{-1} b^K(W_i)}{\sigma_n(f)}  = \Pi_K T u_n(f)(W_i)
\]
where $ u_n(f) = v_n(f)/\sigma_n(f)$ is the scaled sieve 2SLS Riesz representer. Note that $E[(Z_n(W_i)u_i)^2] = 1$. Then
\begin{eqnarray*}
 \sqrt n \frac{D f(h_0)[\wh h - \widetilde h]}{\sigma_n(f)}
 & = &  \frac{1}{\sqrt n} \sum_{i=1}^n Z_n(W_i)u_i \label{e-dist} \\
 & & + \frac{(D f(h_0)[\psi^J])'((\wh S'\wh G_b^{-}\wh S)^- \wh S' \wh G_b^{-} - (S'G_b^{-1} S)^{-1} S' G_b^{-1} )(B'u/\sqrt n)}{\sigma_n(f)} \\
 & =: & T_1 + T_2 \,.\notag
\end{eqnarray*}
We first show $T_1 \to_d N(0,1)$ by the Lindeberg-Feller theorem. To verify the Lindeberg condition, note that
\begin{eqnarray*}
 |Z_n(W_i)| & \leq & \frac{\left\| (D f(h_0)[\psi^J])'(S'G_b^{-1}S)^{-1}S' G_b^{-1/2} \right\| \left\|  G_b^{-1/2} b^K(W_i)\right\|}{ (\inf_w E[u_i^2|W_i = w])^{1/2}\left\| (D f(h_0)[\psi^J])'(S'G_b^{-1}S)^{-1}S' G_b^{-1/2} \right\| } \quad \leq \quad  \underline \sigma^{-1} \zeta_b(K)
\end{eqnarray*}
by the Cauchy-Schwarz inequality and Assumption \ref{a-residuals}(iii). Therefore,
\begin{eqnarray*}
E[u_i^2 Z_n(W_i)^2 \{ |Z_n(W_i)u_i | > \eta \sqrt n\}]
& \leq & \sup_w E[u_i^2 \{ |u_i | \gtrsim \eta (\sqrt n/\zeta_b(K))\}|W_i = w] \quad = \quad o(1)
\end{eqnarray*}
by Assumption \ref{a-residuals}(iv') and the condition on $J$. Therefore, $T_1 \to_d N(0,1)$.

For $T_2$, observe that
\begin{eqnarray*}
 |T_2| & = & \left| \frac{(D f(h_0)[\psi^J])'((\wh G_b^{-1/2}\wh S)^-_l \wh G_b^{-1/2} G_b^{1/2} - (G_b^{-1/2}S)^-_l )(G_b^{-1/2} B'u/\sqrt n)}{\sigma_n(f)} \right| \\
 & = & \left| \frac{[(D f(h_0)[\psi^J])' (G_b^{-1/2}S)^-_l ] G_b^{-1/2} S\{(\wh G_b^{-1/2}\wh S)^-_l \wh G_b^{-1/2} G_b^{1/2} - (G_b^{-1/2}S)^-_l \}(G_b^{-1/2} B'u/\sqrt n)}{\sigma_n(f)} \right| \\
 & \leq & \frac{\left\| (D f(h_0)[\psi^J])'(G_b^{-1/2}S)^{-}_l \right\| \left\| G_b^{-1/2} S\{(\wh G_b^{-1/2}\wh S)^-_l \wh G_b^{-1/2} G_b^{1/2} - (G_b^{-1/2}S)^-_l \} \right\|  \left\|  G_b^{-1/2} B'u/\sqrt n \right\|}{ (\inf_w E[u_i^2|W_i = w])^{1/2}\left\| (D f(h_0)[\psi^J])'(G_b^{-1/2}S)^{-}_l  \right\| } \\
 & \leq & \underline \sigma^{-1} \left\| G_b^{-1/2} S\{(\wh G_b^{-1/2}\wh S)^-_l \wh G_b^{-1/2} G_b^{1/2} - (G_b^{-1/2}S)^-_l \} \right\| \left\|  G_b^{-1/2} B'u/\sqrt n \right\| \\
 & = & O_p( s_{JK}^{-1} \zeta \sqrt{(J\log J)/n} )
\end{eqnarray*}
where the first inequality is by the Cauchy-Schwarz inequality, the second is by Assumption \ref{a-residuals}(iii), and the final line is by Lemmas \ref{lem-SGl2}(c) and \ref{lem-BuL2}. The result follows by the equivalence $\tau_J \asymp s_{JK}^{-1}$ (see Lemma \ref{lem-ip}) and the condition $\tau_J \zeta \sqrt{(J \log n)/n} = o(1)$.

\textbf{Result (2)} follows directly from Result (1) and Lemma \ref{lem-varest}.
\end{proof}

\begin{lemma}\label{lem-varest}
Let Assumptions \ref{a-data}(iii), \ref{a-residuals}(i)--(iii), \ref{a-sieve}(iii) and \ref{a-approx}(i) hold, $\tau_J \zeta \sqrt{(\log n)/n} = o(1)$, and Assumption \ref{a-functional-prime}'(b)(iii) hold (with $\eta_n'=0$ if $f(\cdot )$ is linear). Let $\|\wh h - h_0\|_\infty = O_p(\delta_{h,n})= o_p(1)$, and $\delta_{V,n}\equiv\big[ \zeta_{b,K}^{(2+\delta)/\delta} \sqrt{(\log K)/n} \big]^{\delta/(1+\delta)}  + \tau_J \zeta \sqrt{(\log J)/n} +\delta_{h,n} $. Then:
\[
 \left| \frac{\wh\sigma_n(f)}{\sigma_n(f)}-1\right| = O_p( \delta_{V,n} + \eta_n')=o_p (1)\,.
\]
\end{lemma}

\begin{proof}[\textbf{Proof of Lemma \ref{lem-varest}}]
First write
\begin{eqnarray*}
\frac{\wh\sigma_n(f)^2}{\sigma_n(f)^2} -1
 &=& \left( \frac{\wh \gamma_n' \Omega^o \wh \gamma_n}{\sigma_n(f)^2} -1 \right) + \frac{\wh \gamma_n' (\wh \Omega^o - \Omega^o) \wh \gamma_n}{\sigma_n(f)^2} \\
 &=& \left( \frac{(\wh \gamma_n - \gamma_n)' \Omega^o (\wh \gamma_n + \gamma_n)}{\sigma_n(f)^2} \right) + \frac{\wh \gamma_n' (\wh \Omega^o - \Omega^o) \wh \gamma_n}{\sigma_n(f)^2} \quad =: \quad T_1 + T_2
\end{eqnarray*}
where
\begin{align*}
 \wh \Omega^o & = G_b^{-1/2}\wh \Omega G_b^{-1/2} & \wh \gamma_n & = \textstyle G_b^{1/2} \wh G_b^{-1}  \wh S [ \wh S'\wh   G_b^{-1}  \wh S]^{-1} Df( \wh h)[\psi^J] \\
 \Omega^o & = G_b^{-1/2} \Omega G_b^{-1/2} & \gamma_n & = \textstyle G_b^{-1/2}  S [ S'  G_b^{-1}  S]^{-1} Df( h_0)[\psi^J]
\end{align*}
and observe that $\gamma_n'\Omega^o \gamma_n = \sigma_n(f)^2$ and $\wh \gamma_n'\wh \Omega^o \wh \gamma_n = \wh\sigma_n(f)^2$.

Control of $T_1$: We first show that
\begin{equation} \label{e-gambd1}
 \frac{\|\wh \gamma_n - \gamma_n\|_{\ell^2}}{\sigma_n(f)} = O_p( \tau_J \zeta \sqrt{(\log J)/n} + \eta_n') = o_p(1)\,.
\end{equation}
To simplify notation, let
\[
\partial = \frac{Df(h_{0})[\psi^J]}{s_n(f)}  \quad \mbox{and} \quad
\wh{\partial } = \frac{Df(\wh h)[\psi^J]}{s_n(f)}
\]
and note that $\|\partial' (G_b^{-1/2} S)^-_l\|_{\ell^2} = s_n(f)/\sigma_n(f) \asymp 1$ under Assumptions \ref{a-residuals}(i)(iii) and that $\wh \partial  = \partial$ if $f(\cdot)$ is linear. Then we have:
\begin{align*}
 \frac{\|\wh \gamma_n - \gamma_n\|_{\ell^2}}{\sigma_n(f)}
 & = \| \wh \partial ' (\wh G_b^{-1/2} \wh S)^-_l \wh G_b^{-1/2}G_b^{1/2} - \partial ' (G_b^{-1/2} S)^-_l \|_{\ell^2} \notag \\
 & \leq \| \wh \partial ' (G_b^{-1/2} S)^-_l\|_{\ell^2} \| G_b^{-1/2} S\{(\wh G_b^{-1/2} \wh S)^-_l\wh G_b^{-1/2}G_b^{1/2} - (G_b^{-1/2} S)^-_l\}\|_{\ell^2} + \underline \sigma^{-1} \| (\wh \partial' - \partial')(G_b^{-1/2} S)^-_l) \|_{\ell^2} \notag \\
 & = O_p(1) \times O_p (s_{JK}^{-1} \zeta \sqrt{(\log J)/n}) + \underline \sigma^{-1}\frac{\|\Pi_K T(\wh v_n(f)-v_n(f))\|_{L^2(W)}}{s_n(f)} \notag \\
 & = O_p(1) \times O_p (s_{JK}^{-1} \zeta \sqrt{(\log J)/n}) + O_p(\eta_n') \notag
\end{align*}
where the third line is Lemma \ref{lem-SGl2}(c) and the final line is by Assumption \ref{a-functional-prime}'(b)(iii). Therefore, (\ref{e-gambd1}) holds by the equivalence $s_{JK}^{-1} \asymp \tau_J$ (Lemma \ref{lem-ip}) and the condition $\tau_J \zeta \sqrt{(\log n)/n}$.

Finally, since all eigenvalues of $\Omega^o$ are bounded between $\underline \sigma^2$ and $\overline \sigma^2$ under Assumption \ref{a-residuals}(i)(iii), it follows from (\ref{e-gambd1}) and Cauchy-Schwarz that $|T_1| =  o_p(1)$.

Control of $T_2$: Equation (\ref{e-gambd1}) implies that $\|\wh \gamma_n\|/\sigma_n(f) = O_p(1)$. Therefore, $|T_2| \leq O_p(1) \times \|\wh \Omega^o - \Omega^o\|_{\ell^2}=o_p(1)$ by  Lemma \ref{lem-omega}.
\end{proof}

\begin{lemma}\label{lem-omega}
Let Assumptions \ref{a-residuals}(i)(ii) hold, let $\zeta_{b,K} \sqrt{(\log K)/n} = o(1)$, and let $\|\wh h - h_0\|_\infty = O_p(\delta_{h,n})$ with $\delta_{h,n} = o(1)$. Then:
\[
 \|\wh \Omega^o - \Omega^o \|_{\ell^2} = O_p\Big( \big( \zeta_{b,K}^{(2+\delta)/\delta} \sqrt{(\log K)/n} \big)^{\delta/(1+\delta)} + \delta_{h,n}\Big)
\]
\end{lemma}

\begin{proof}[\textbf{Proof of Lemma \ref{lem-omega}}]
By the triangle inequality:
\begin{eqnarray*}
 \|\wh \Omega^o - \Omega^o \|_{\ell^2} & \leq &  \left\| G_b^{-1/2}\left( \frac{1}{n} \sum_{i=1}^n u_i^2 b^K(W_i) b^K(W_i)' \right) G_b^{-1/2} \right\|_{\ell^2} \\
 & & +  \left\| G_b^{-1/2}\left( \frac{1}{n} \sum_{i=1}^n 2 u_i(\wh u_i - u_i) b^K(W_i) b^K(W_i)' \right) G_b^{-1/2} \right\|_{\ell^2} \\
 & & +  \left\| G_b^{-1/2}\left( \frac{1}{n} \sum_{i=1}^n (\wh u_i - u_i)^2 b^K(W_i) b^K(W_i)' \right) G_b^{-1/2} \right\|_{\ell^2} \\
 & \leq & O_p ( ( \zeta_{b,K}^{(2+\delta)/\delta} \sqrt{(\log K)/n} )^{\delta/(1+\delta)}  )  + \|\wh h - h_0\|_{\infty} \times O_p(1) + \|\wh h - h_0\|_{\infty}^2 \times O_p(1)
\end{eqnarray*}
where the first term may easily be deduced from the proof of Lemma 3.1 of \cite{ChenChristensen-reg}, the second then follows because $2 u_i(\wh u_i - u_i) \leq 2(1+u_i^2) \|\wh h - h_0\|_\infty$, and the third follows similarly because $\|\wh G_b^o\|_{\ell^2} = O_p(1)$ by Lemma \ref{lem-matl2}.
\end{proof}

\subsubsection{Proofs for Section \ref{s-ucb}}

\begin{proof}[\textbf{Proof of Lemma \ref{lem-bahadur}}]
Recall that
\begin{align*}
 \wh {\mb Z}_n(t) & =  \frac{(D f_t(h_0)[\psi^J])' [ S'  G_b^{-1}  S]^{-1}  S'  G_b^{-1/2}}{\sigma_n(f_t)} \left(\frac{1}{\sqrt n} \sum_{i=1}^n G_b^{-1/2} b^K(W_i) u_i\right)~, \\
 \mb Z_n(t) & =  \frac{(D f_t(h_0)[\psi^J])' [ S'  G_b^{-1}  S]^{-1}  S'  G_b^{-1/2}}{\sigma_n(f_t)}  \mathcal Z_n~~\text{where}~\mathcal Z_n \sim N(0,\Omega^o )~\text{with}~\Omega^o = G_b^{-1/2}\Omega G_b^{-1/2}.
\end{align*}
\textbf{Step 1: Uniform Bahadur representation.}
By Assumption \ref{a-functional}(a) or (b)(i)(ii), we have
\begin{eqnarray*}
 \sup_{t \in \mathcal T} \left| \sqrt n \frac{f_t(\wh h) - f_t(h_0)}{\wh\sigma_n(f_t)} - \wh{\mb Z}_n(t)\right|
  & \leq & \sup_{t \in \mathcal T} \left| \sqrt n \frac{ Df_t(h_0)[\wh h - \widetilde h]}{\sigma_n(f_t)} - \wh{\mb Z}_n(t)\right| + O_p(\eta_n)  \times \sup_{t \in \mathcal T} \left| \frac{\sigma_n(f_t)}{\wh \sigma_n(f_t)}  \right| \\
 & & +  \sup_{t \in \mathcal T} \left| \frac{\sigma_n(f_t)}{\wh\sigma_n(f_t)} -1 \right| \times \sup_{t \in \mathcal T} \left| \sqrt n \frac{ D f_t(h_0)[\wh h - \widetilde h]}{\sigma_n(f_t)} \right| \\
 & =: &  T_1 + T_2 + T_3\,.
\end{eqnarray*}
Control of $T_1$: As in the proof of Theorem \ref{t-dist},
\begin{eqnarray*}
 T_1 & = & \sup_{t \in \mathcal T} \left| \frac{(D f_t(h_0)[\psi^J])'((\wh G_b^{-1/2}\wh S)^-_l \wh G_b^{-1/2} G_b^{1/2} - (G_b^{-1/2}S)^-_l )(G_b^{-1/2} B'u/\sqrt n)}{\sigma_n(f_t)} \right| \\
 & = & \sup_{t \in \mathcal T} \left| \frac{[(D f_t(h_0)[\psi^J])' (G_b^{-1/2}S)^-_l ] G_b^{-1/2} S\{(\wh G_b^{-1/2}\wh S)^-_l \wh G_b^{-1/2} G_b^{1/2} - (G_b^{-1/2}S)^-_l \}(G_b^{-1/2} B'u/\sqrt n)}{\sigma_n(f_t)} \right| \\
 & \leq & \sup_{t \in \mathcal T} \frac{\left\| (D f_t(h_0)[\psi^J])'(G_b^{-1/2}S)^{-}_l \right\| \left\| G_b^{-1/2} S\{(\wh G_b^{-1/2}\wh S)^-_l \wh G_b^{-1/2} G_b^{1/2} - (G_b^{-1/2}S)^-_l \} \right\|  \left\|  G_b^{-1/2} B'u/\sqrt n \right\|}{ (\inf_w E[u_i^2|W_i = w])^{1/2}\left\| (D f_t(h_0)[\psi^J])'(G_b^{-1/2}S)^{-}_l  \right\| } \\
 & \leq & \underline \sigma^{-1} \left\| G_b^{-1/2} S\{(\wh G_b^{-1/2}\wh S)^-_l \wh G_b^{-1/2} G_b^{1/2} - (G_b^{-1/2}S)^-_l \} \right\| \left\|  G_b^{-1/2} B'u/\sqrt n \right\| \\
 & = & O_p( \tau_J \zeta \sqrt{(J\log J)/n} )=o_p(r_n)
\end{eqnarray*}
where the first inequality is by the Cauchy-Schwarz inequality, the second is by Assumption \ref{a-residuals}(iii), and the final line is by Lemmas \ref{lem-SGl2}(c) and \ref{lem-BuL2} and the equivalence $\tau_J \asymp s_{JK}^{-1}$ (see Lemma \ref{lem-ip}), and the last $o_p( r_n )$ is by Assumption \ref{a-Lipschitz}(ii.2).

Control of $T_2$: Lemma \ref{lem-varest-u} below shows that
\[
 \sup_{t \in \mathcal T} \left| \frac{\wh\sigma_n(f_t)}{\sigma_n(f_t)} -1 \right| = O_p( \delta_{V,n} + \eta_n') = o_p(1)
\]
from which it follows that $T_2 = O_p(\eta_n) \times O_p(1) = O_p(\eta_n)$.

Control of $T_3$: By Lemma \ref{lem-varest-u} below and the bound for $T_1$, we have:
\begin{align*}
 T_3 & = O_p( \delta_{V,n} + \eta_n') \times \sup_{t \in \mathcal T} \left| \sqrt n \frac{ Df_t(h_0)[\wh h - \widetilde h]}{\sigma_n(f_t)} \right| \\
 & = O_p( \delta_{V,n}  + \eta_n') \times \big[\sup_{t \in \mathcal T} | \wh{\mb Z}_n(t) | +o_p(r_n)\big]  \\
 & = O_p( \delta_{V,n}  + \eta_n' ) \times \big[o_p(r_n) + \sup_{t \in \mathcal T} | \mb Z_n(t)| +o_p(r_n)\big] \\
 & = O_p( \delta_{V,n}  + \eta_n' ) \times \big[ o_p(r_n) + O_p(c_n) \big]
\end{align*}
where the second-last line is by display (\ref{e-Zpt1}) step 2 below and the final line is by Lemma \ref{lem-esup} below. Therefore we have proved:
\begin{align}
 \sup_{t \in \mathcal T} \left| \sqrt n \frac{f_t(\wh h) - f_t(h_0)}{\wh\sigma_n(f_t)} - \wh{\mb Z}_n(t)\right| & =O_p( \tau_J \zeta \sqrt{(J\log J)/n} ) + O_p(\eta_n) + O_p( \delta_{V,n} + \eta_n' ) \times \big[ o_p(r_n) + O_p(c_n) \big] \notag \\
 & = o_p( r_n ) \label{e-unifzhat}
\end{align}
where the final line is by Assumption \ref{a-Lipschitz}(ii.2).

\textbf{Step 2: Approximating $\wh{\mb Z}_n(t)$ by a Gaussian process $\mb Z_n(t)$.} We use Yurinskii's coupling \cite[Theorem 10, p. 244]{PollardUGMTP} to show that there exists a sequence of $N(0,\Omega ^o)$ random vectors $\mathcal Z_n$ such that
\begin{equation}\label{e-yurinskii}
 \left\| \frac{1}{\sqrt n} \sum_{i=1}^n G_b^{-1/2} b^K(W_i) u_i - \mathcal Z_n  \right\|_{\ell^2} = o_p ( r_n) \,.
\end{equation}
By Assumption \ref{a-residuals}(iv) we have
\begin{eqnarray*}
  \sum_{i=1}^n E[ \| n^{-1/2} G_b^{-1/2} b^K(W_i) u_i \|^3_{\ell^2}] \lesssim n^{-1/2} \zeta_{b,K} E[\|G_b^{-1/2} b^K(W_i)\|^2_{\ell^2}] = \frac{\zeta_{b,K} K}{\sqrt n} = O \left( \frac{\zeta_{b,K} J}{\sqrt n} \right) \,.
\end{eqnarray*}
Existence of $\mathcal Z_n$ follows under the condition (Assumption \ref{a-Lipschitz}(ii.1))
\[
 \frac{\zeta_{b,K} J^2}{r_n^3 \sqrt n} = o(1)\,.
\]
The process $\mb Z_n(t)$ is a centered Gaussian process with the covariance function
\begin{equation} \label{e-cov-fn}
 E[\mb Z_n(t_1)\mb Z_n(t_2)] = \frac{(D f_{t_1}(h_0)[\psi^J])'[S' G_b^{-1} S]^{-1} S' G_b^{-1} \Omega G_b^{-1} S[S' G_b^{-1} S]^{-1} Df_{t_2}(h_0)[\psi^J]}{\sigma_n(f_{t_1})\sigma_n(f_{t_2})} \,.
\end{equation}

Now observe that
\begin{equation} \label{e-unif-equiv}
 \sup_{t \in \mathcal T} \left\| \frac{(Df_t(h_0)[\psi^J])' [ S'  G_b^{-1}  S]^{-1}  S'  G_b^{-1/2}}{\sigma_n(f_t)} \right\|_{\ell^2} = \sup_{t \in \mathcal T} \frac{{s_n(f_t)}}{\sigma_n(f_t)} \asymp 1
\end{equation}
by Assumption \ref{a-residuals}(i)(iii).
Therefore,
\begin{equation}\label{e-Zpt1}
 \sup_{t \in \mathcal T} \left| \wh{\mb Z}_n(t) - \mb Z_n(t) \right| = o_p(r_n)
\end{equation}
by equations (\ref{e-yurinskii}) and (\ref{e-unif-equiv}) and Cauchy-Schwarz.
\end{proof}

\begin{lemma}\label{lem-varest-u}
Let Assumptions \ref{a-data}(iii), \ref{a-residuals}(i)--(iii), \ref{a-sieve}(ii)(iii) and \ref{a-approx}(i) hold, $\tau_J \zeta \sqrt{(\log n)/n} = o(1)$, and Assumption \ref{a-functional}(b)(iii) hold (with $\eta_n'=0$ if $f_t(\cdot)$ is linear). Let $\|\wh h - h_0\|_\infty = O_p(\delta_{h,n})$ with $\delta_{h,n} = o(1)$. Then:
\[
 \sup_{t \in \mathcal T} \left| \frac{\sigma_n(f_t)}{\wh\sigma_n(f_t)} -1 \right| = O_p( \delta_{V,n} + \eta_n')=o_p(1).
\]
\end{lemma}

\begin{proof}[\textbf{Proof of Lemma \ref{lem-varest-u}}]
The proof follows by identical arguments to the proof of Lemma \ref{lem-varest}.
\end{proof}

\begin{proof}[\textbf{Proof of Theorem \ref{t-dist-u}}]
Recall that
\begin{equation*}
 \mb Z_n^*(t) = \frac{(D f_t(\wh h)[\psi^J])' [\wh S' \wh G_b^{-1} \wh S]^{-1} \wh S' \wh G_b^{-1}}{\wh\sigma (f_t)} \left(\frac{1}{\sqrt n} \sum_{i=1}^n  b^K(W_i) \wh u_i \varpi_i \right) \quad \mbox{for each $t \in \mathcal T$} \,.
\end{equation*}
\textbf{Step 1: Approximating $\mb Z_n^*(t)$ by a Gaussian process $\widetilde {\mb Z}_n^*$.} Each of the terms $n^{-1/2} G_b^{-1/2} b^K(W_i) \wh  u_i \varpi_i$ is centered under $\mb P^*$ because $E[\varpi_i|Z^n] = 0$ for $i=1,\ldots,n$. Moreover,
\[
 \sum_{i=1}^n E[ (n^{-1/2} G_b^{-1/2} b^K(W_i) \wh  u_i \varpi_i)(n^{-1/2} G_b^{-1/2} b^K(W_i) \wh  u_i \varpi_i)'| Z^n] = \wh \Omega^o
\]
where $G_b^{-1/2} \wh \Omega G_b^{-1/2} = \wh \Omega^o$, and
\[
 \sum_{i=1}^n E[ \| n^{-1/2} G_b^{-1/2} b^K(W_i) \wh u_i \varpi_i \|^3_{\ell^2}|Z^n] \lesssim n^{-3/2} \sum_{i=1}^n E[\| G_b^{-1/2} b^K(W_i)\|_{\ell^2}^2 | \wh u_i |^3]
\]
because $E[|\varpi_i|^3|Z^n] < \infty$ uniformly in $i$, and where
\[
  n^{-3/2} \sum_{i=1}^n E[\| G_b^{-1/2} b^K(W_i)\|_{\ell^2}^2 | \wh u_i |^3] \lesssim \frac{\zeta_{b,K} K}{\sqrt n}
\]
holds wpa1 (by Markov's inequality using $|\wh u_i|^3 \lesssim |u_i|^3 + \|\wh h - h_0\|_\infty^3$ and Assumption \ref{a-residuals}(iv)). A second application of Yurinskii's coupling conditional on the data $Z^n$ then yields existence of a sequence of $N(0,\wh \Omega^o)$ random vectors $\mathcal Z_n^*$ such that
\[
 \left\| \frac{1}{\sqrt n} \sum_{i=1}^n G_b^{-1/2} b^K(W_i) \wh u_i \varpi_i - \mathcal Z_n^* \right\|_{\ell^2} = o_{p^*} ( r_n )
\]
wpa1. Therefore:
\begin{equation} \label{e-Zpt2}
 \sup_{t \in \mathcal T} \left| \mb Z_n^*(t) - \frac{(Df_t(\wh h)[\psi^J])' [\wh S' \wh G_b^{-1} \wh S]^{-1} \wh S' \wh G_b^{-1/2}}{\wh\sigma_n(f_t)}\mathcal Z_n^* \right| = o_{p^*}(r_n)
\end{equation}
wpa1. Now observe that we can define a centered Gaussian process $\widetilde {\mb Z}_n^*$ under $\mb P^*$ by
\[
 \widetilde {\mb Z}_n^*(t) \equiv \frac{(D f_t( h_0)[\psi^J])' [ S'  G_b^{-1}  S]^{-1}  S'  G_b^{-1/2}}{\sigma_n(f_t)} (\Omega^o)^{1/2} (\wh \Omega^o)^{-1/2} \mathcal Z_n^*
\]
which has the same covariance function as $\mb Z_n$ (see equation (\ref{e-cov-fn})) whenever $\wh \Omega^o$ is invertible (which it is wpa1). Therefore, by Lemma \ref{lem-varest-boot} below we have:
\begin{equation} \label{e-Zpt3}
 \sup_{t \in \mathcal T} \left|  \frac{(D f_t(\wh h)[\psi^J])' [\wh S' \wh G_b^{-1} \wh S]^{-1} \wh S' \wh G_b^{-1/2}}{\wh\sigma_n(f_t)} \mathcal Z_n^* - \widetilde{\mb Z}_n^*(t)\right| = o_p^*(r_n)
\end{equation}
wpa1. It follows from equations (\ref{e-Zpt2}) and (\ref{e-Zpt3}) and Assumption \ref{a-Lipschitz}(ii) that
\begin{align}
 \sup_{t \in \mathcal T} \left| \mb Z_n^*(t) - \widetilde {\mb Z}_n^*(t) \right| \quad = \quad o_{p^*} ( r_n ) + o_p^*(r_n) \quad = \quad   o_p^*(r_n) \label{e-Zpt5}
\end{align}
wpa1.

\textbf{Step 2: Consistency.} By Lemma \ref{lem-bahadur} and display (\ref{e-Zpt1}), we have:
\[
 \sup_{t \in \mathcal T} \left| \frac{\sqrt n (f_t(\wh h) - f_t(h_0))}{\wh\sigma_n(f_t)} - \mb Z_n(t)\right|
   =  o_p(r_n) + o_p(r_n)
  = o_p(r_n) \,.
\]
Therefore, we may choose a sequence of positive constants $\epsilon_n$ with $\epsilon_n = o(1)$ such that
\begin{equation} \label{e-tstat-u-exp}
 \sup_{t \in \mathcal T} \left| \frac{\sqrt n (f_t(\wh h) - f_t(h_0))}{\wh\sigma_n(f_t)} - \mb Z_n(t)\right| \leq \epsilon_n r_n
\end{equation}
holds wpa1. By an anti-concentration inequality \cite[Theorem 2.1]{CCK-AC} and Lemma \ref{lem-esup} below, we have:
\begin{eqnarray*}
 \sup_{s \in \mb R}  \mb P\left( \sup_{t \in \mathcal T} |\mb Z_n (t) - s| \leq \epsilon_n r_n \right) \lesssim \epsilon_n r_n E[\sup_{t \in \mathcal T} |\mb Z_n(t)|]\lesssim \epsilon_n r_n c_n = o(1)
\end{eqnarray*}
due to $r_n c_n \lesssim  1$ (Assumption \ref{a-Lipschitz}(ii.1)). This, together with (\ref{e-tstat-u-exp}), yields:
\begin{equation} \label{e-diffprob}
 \sup_{s \in \mb R} \left| \mb P \left( \sup_{t \in \mathcal T} \left|\frac{\sqrt n (f_t(\wh h) - f_t(h_0))}{\wh\sigma_n(f_t)} \right| \leq s \right) - \mb P\left( \sup_{t \in \mathcal T} |\mb Z_n(t)| \leq s\right) \right| = o(1)\,.
\end{equation}

Moreover, by (\ref{e-Zpt5}) we may choose a sequence of positive constants $\epsilon_n'$ with $\epsilon_n' = o(1)$ such that
\[
 \sup_{t \in \mathcal T} \left| \mb Z_n^*(t) - \widetilde {\mb Z}_n^*(t) \right| \leq \epsilon_n' r_n
\]
holds wpa1. Similar arguments then yield:
\[
 \sup_{s \in \mb R}  \mb P^*\left( \sup_{t \in \mathcal T} |\widetilde{\mb Z}_n^*(t) - s| \leq \epsilon_n' r_n\right) \lesssim \epsilon_n' = o(1)
\]
wpa1. This, together with equation (\ref{e-Zpt5}), yields:
\begin{equation} \label{e-diffprob1}
 \sup_{s \in \mb R} \left| \mb P^*\left( \sup_{t \in \mathcal T} |\mb Z_n^*(t)| \leq s\right) - \mb P^*\left( \sup_{t \in \mathcal T} |\widetilde{\mb Z}_n^*(t)| \leq s\right) \right| = o_p(1) \,.
\end{equation}
The result is immediate from equations (\ref{e-diffprob}) and (\ref{e-diffprob1}) and the fact that
\[
 \mb P\left( \sup_{t \in \mathcal T} |\mb Z_n(t)| \leq s\right) = \mb P^*\left( \sup_{t \in \mathcal T} |\widetilde{\mb Z}_n^*(t)| \leq s\right) \quad \mbox{wpa1 in $\mb P$}
\]
holds uniformly in $s$.
\end{proof}

\begin{lemma}\label{lem-esup}
Let Assumption \ref{a-Lipschitz}(i) hold. Then: $E[{ \sup_{t \in \mathcal T}} |\mb Z_n(t)|]  \lesssim c_n$ and $\sup_{t \in \mathcal T} |\mb Z_n(t)| = O_p (c_n)$.
\end{lemma}

\begin{proof}[\textbf{Proof of Lemma \ref{lem-esup}}]
Observe that $d_n(t_1,t_2) := E[ (\mb Z_n(t_1) - \mb Z_n(t_2))^2]^{1/2}$. By Corollary 2.2.8 of \cite{vanderVaartWellner} and Assumption \ref{a-Lipschitz}(i), there exists a universal constant $C$ such that
\[
 E[{\textstyle \sup_t} |\mb Z_n(t)|] \leq E[|\mb Z_n(\bar t)|] + C \int_0^\infty \sqrt{\log N(\mathcal T,d_n,\epsilon)} \,\mathrm d \epsilon
\]
for any $\bar t \in \mathcal T$, where $E[|\mb Z_n(\bar t)|]  = \sqrt{2/\pi}$ because $\mb Z_n(\bar t) \sim N(0,1)$. Therefore, $E[{\textstyle \sup_t} |\mb Z_n(t)|] \lesssim c_n$. The second result follows by Markov's inequality.
\end{proof}

\begin{lemma}\label{lem-varest-boot}
Let Assumptions \ref{a-data}(iii), \ref{a-residuals}, \ref{a-sieve}(ii)(iii), \ref{a-approx}(i) and \ref{a-Lipschitz} hold, $\tau_J \zeta \sqrt{(\log n)/n} = o(1)$ and $\|\wh h - h_0\|_\infty = O_p(\delta_{h,n})$ with $\delta_{h,n} = o(1)$. Let Assumption \ref{a-functional}(b)(iii) hold with $\eta_n' \sqrt J = o(r_n)$ for nonlinear $f_t()$. Let $\mathcal Z_n^*$ and $\widetilde {\mb Z}_n^*(t)$ be as in the proof of Theorem \ref{t-dist-u}. Then:
\[
 \sup_{t \in \mathcal T} \left|  \frac{(Df_t(\wh h)[\psi^J])' [\wh S' \wh G_b^{-1} \wh S]^{-1} \wh S' \wh G_b^{-1/2}}{\wh\sigma_n(f_t)} \mathcal Z_n^* - \widetilde {\mb Z}_n^*(t)\right| = o_p^*(r_n) \quad \mbox{wpa1 in $\mb P$}.
\]
\end{lemma}

\begin{proof}[\textbf{Proof of Lemma \ref{lem-varest-boot}}]
First note that because $\mathcal Z_n^* \sim N(0,\wh \Omega^o)$ and the minimum and maximum eigenvalues of $\wh \Omega^o$ are uniformly bounded away from $0$ and $\infty$ wpa1 (by Lemma \ref{lem-omega} and Assumptions \ref{a-residuals}(i)(iii)), we have $\|\mathcal Z_n^*\| = O_{p^*}(\sqrt K)$ wpa1 by Chebyshev's inequality.

Now, writing out term by term we have:
\begin{align*}
 & \sup_{t \in \mathcal T} \left| \left( \frac{(Df_t(\wh h)[\psi^J])' [\wh S' \wh G_b^{-1} \wh S]^{-1} \wh S' \wh G_b^{-1/2}}{\wh\sigma_n(f_t)} - \frac{(Df_t(h_0)[\psi^J])' [ S'  G_b^{-1}  S]^{-1}  S'  G_b^{-1/2}(\Omega^o)^{1/2} (\wh \Omega^o)^{-1/2}}{\sigma_n(f_t)} \right) \mathcal Z_n^* \right| \\
 & \leq \sup_{t \in \mathcal T} \left|  \frac{ \left( Df_t(\wh h)[\psi^J] - Df_t(h_0)[\psi^J] \right)' (G_b^{-1/2}S)^-_l G_b^{-1/2}S [\wh S' \wh G_b^{-1} \wh S]^{-1} \wh S' \wh G_b^{-1/2}}{\wh\sigma_n(f_t)} \mathcal Z_n^* \right| \\
 & + \sup_{t \in \mathcal T} \left|  \frac{(D f_t(h_0)[\psi^J])'  \left( [\wh S' \wh G_b^{-1} \wh S]^{-1} \wh S' \wh G_b^{-1/2}- [ S'  G_b^{-1}  S]^{-1}  S'  G_b^{-1/2} (\Omega^o)^{1/2} (\wh \Omega^o)^{-1/2}\right) }{\sigma_n(f_t)} \mathcal Z_n^* \right| \times \sup_{t \in \mathcal T} \frac{\sigma_n(f_t)}{\wh\sigma_n(f_t)} \\
 & + \sup_{t \in \mathcal T}  \left| \frac{\sigma_n(f_t)}{\wh\sigma_n(f_t)}-1 \right| \times \sup_{t \in \mathcal T}   \left| \frac{(D f_t( h_0)[\psi^J])' [ S'  G_b^{-1}  S]^{-1}  S'  G_b^{-1/2}(\Omega^o)^{1/2} (\wh \Omega^o)^{-1/2}}{\sigma_n(f_t)}  \mathcal Z_n^* \right| \quad =: T_1 + T_2 + T_3\,.
\end{align*}

Control of $T_1$: By Cauchy-Schwarz, we have:
\begin{align*}
 T_1 & \leq \sup_{t \in \mathcal T} \frac{\|\Pi_K T(\wh v_n(f_t) - v_n(f_t))\|_{L^2(W)}}{\sigma_n(f_t)} \times \sup_{t \in \mathcal T} \frac{\sigma_n(f_t)}{\wh\sigma_n(f_t)} \times \left\|G_b^{-1/2}S [\wh S' \wh G_b^{-1} \wh S]^{-1} \wh S' \wh G_b^{-1/2} \right\|_{\ell^2} \times \|(\Omega^o)^{1/2} (\wh \Omega^o)^{-1/2}\|_{\ell^2}  \times \left\| \mathcal Z_n^*  \right\|_{\ell^2} \\
 & = o_p(\eta_n') \times O_p(1) \times O_p(1) \times O_p(1) \times O_{p^*}(\sqrt K)= o_p^*(r_n)
\end{align*}
where the first term is by Assumption \ref{a-functional}(b)(iii) (or zero if the $f_t$ are linear functionals), the second term is by Lemma \ref{lem-varest-u}, the third is by Lemma \ref{lem-SGl2}(c) (using the fact that $s_{JK} \asymp \tau_J$, see Lemma \ref{lem-ip}) and the fact that $\|G_b^{-1/2}S [ S'  G_b^{-1}S]^{-1}  S' G_b^{-1/2} \|_{\ell^2} =1$, and the fourth term is by Lemma \ref{lem-omega}. Therefore, $T_1 = O_{p^*} (\eta_n' \sqrt J) $ wpa1 (since $K \asymp J$), and is therefore $= o_p^*(r_n)$ wpa1 by the condition stated in this Lemma.

Control of $T_2$: Let $\Delta \mb Z_n(t)$ denote the Gaussian process (under $\mb P^*$) defined by
\[
 \Delta \mb Z_n(t) =  \frac{(D f_t(h_0)[\psi^J])'  \left( [\wh S' \wh G_b^{-1} \wh S]^{-1} \wh S' \wh G_b^{-1/2}- [ S'  G_b^{-1}  S]^{-1}  S'  G_b^{-1/2} (\Omega^o)^{1/2} (\wh \Omega^o)^{-1/2}\right) }{\wh\sigma_n(f_t)} \mathcal Z_n^*
\]
for each $t \in \mathcal T$. The intrinsic semi-metric $\Delta d_n(t_1,t_2)$ of $\Delta \mb Z_n(t)$ is $\Delta d_n(t_1,t_2)^2 =E^*[(\Delta \mb Z_n(t_1)-\Delta \mb Z_n(t_2))^2]$ for each $t_1,t_2 \in \mathcal T$, where $E^*$ denotes expectation under the measure $\mb P^*$. Observe that:
\begin{align*}
 \Delta d_n(t_1,t_2)
 = & \bigg\| \bigg( \frac{Df_{t_1}( h_0)[\psi^J] }{\sigma_n(f_{t_1})} - \frac{D f_{t_2}( h_0)[\psi^J] }{\sigma_n(f_{t_2})} \bigg)'  [ S'  G_b^{-1}  S]^{-1}  S'  G_b^{-1/2}(\Omega^o)^{1/2}\\
 & \times (\Omega^o)^{-1/2} G_b^{-1/2} S \left\{ [\wh S' \wh G_b^{-1} \wh S]^{-1} \wh S' \wh G_b^{-1/2}- [ S'  G_b^{-1}  S]^{-1}  S'  G_b^{-1/2} (\Omega^o)^{1/2} (\wh \Omega^o)^{-1/2} \right\}  (\wh \Omega^o)^{1/2} \bigg\|_{\ell^2} \\
 \lesssim & d_n(t_1,t_2) \times   \left\| G_b^{-1/2} S'\left\{ [\wh S' \wh G_b^{-1} \wh S]^{-1} \wh S' \wh G_b^{-1/2}- [ S'  G_b^{-1}  S]^{-1}  S'  G_b^{-1/2} (\Omega^o)^{1/2} (\wh \Omega^o)^{-1/2} \right\} \right\|_{\ell^2}
\end{align*}
wpa1, where the first line uses the fact that $\wh \Omega^o$ is invertible wpa1 and the second line uses the fact that $\Omega^o$ and $\wh \Omega^o$ have eigenvalue uniformly bounded away from $0$ and $\infty$ wpa1. It follows by Lemma \ref{lem-SGl2}(c) and Lemmas \ref{lem-omega} and \ref{lem-sqrtm} that
\begin{align*}
 & \left\| G_b^{-1/2} S'\left\{ [\wh S' \wh G_b^{-1} \wh S]^{-1} \wh S' \wh G_b^{-1/2}- [ S'  G_b^{-1}  S]^{-1}  S'  G_b^{-1/2} (\Omega^o)^{1/2} (\wh \Omega^o)^{-1/2} \right\} \right\|_{\ell^2} \\
 & \leq \left\| G_b^{-1/2} S'\left\{ [\wh S' \wh G_b^{-1} \wh S]^{-1} \wh S' \wh G_b^{-1/2}- [ S'  G_b^{-1}  S]^{-1}  S'  G_b^{-1/2} \right\}\right\|_{\ell^2} + \left\| I -  (\Omega^o)^{1/2} (\wh \Omega^o)^{-1/2} \right\|_{\ell^2} \\
 & = O_p( \tau_J \zeta\sqrt{(\log J)/n} ) +  O_p\Big( \big( \zeta_{b,K}^{(2+\delta)/\delta} \sqrt{(\log K)/n} \big)^{\delta/(1+\delta)} + \delta_{h,n}\Big) \\
 & = O_p( \delta_{V,n} ) \,.
\end{align*}
Therefore,
\[
 \Delta d_n(t_1,t_2) \leq O_p( \delta_{V,n} ) \times d_n(t_1,t_2)
\]
wpa1. Moreover, by similar arguments we have
\[
 \sup_{t \in \mathcal T} E^*[(\Delta \mb Z_n(t))^2]^{1/2} \lesssim  O_p( \delta_{V,n} )
\]
wpa1. Therefore, we can scale $\Delta \mb Z_n(t)$ by dividing through by a sequence of positive constants of order $\delta_{V,n}$ to obtain
\[
 E^*[\sup_{t \in \mathcal T}|\Delta \mb Z_n(t)| ] \lesssim O_p(  \delta_{V,n}  )\times c_n
\]
wpa1 by identical arguments to the proof of Lemma \ref{lem-esup}. Therefore,
\[
 T_2 \leq O_p(  \delta_{V,n} \times c_n ) \times  \sup_{t \in \mathcal T} \frac{\sigma_n(f_t)}{\wh\sigma_n(f_t)} = O_p( \delta_{V,n} \times c_n) \times O_p(1)
\]
wpa1 by Lemma \ref{lem-varest-u} and so $T_2 = o_{p^*} (r_n)$ under Assumption \ref{a-Lipschitz}(ii.2).

Control of $T_3$: The second term in $T_3$ is the supremum of a Gaussian process with the same distribution (under $\mb P^*$) as $\mb Z_n(t)$ (under $\mb P$). Therefore, by Lemmas \ref{lem-varest-u} and \ref{lem-esup} we have:
\begin{align*}
 T_3 & = O_p( \delta_{V,n} + \eta_n')  \times O_{p^*} (c_n)
\end{align*}
and so $T_3 = o_p^*(r_n)$ wpa1 under Assumption \ref{a-Lipschitz}(ii.2).
\end{proof}

\begin{proof}[\textbf{Proof of Remark \ref{rmk-a6suff}}]
For any $t_1,t_2 \in \mathcal T$ we have:
\begin{eqnarray}
 d_n(t_1,t_2)
 & \leq & \frac{2}{\sigma_n(f_{t_1}) \vee \sigma_n(f_{t_2})} \left\|\Omega^{1/2} G_b^{-1} S [S' G_b^{-1} S]^{-1} \left( Df_{t_1}(h_0)[\psi^J] - Df_{t_2}(h_0)[\psi^J] \right) \right\|_{\ell^2} \notag \\
 & \leq & \frac{2 \overline \sigma s_{JK}^{-1}}{\ul \sigma_n} \left\| G_\psi^{-1/2}  \left( D f_{t_1}(h_0)[\psi^J] -  Df_{t_2}(h_0)[\psi^J] \right) \right\|_{\ell^2} \notag \\
 & \lesssim & \frac{ \tau_J  \Gamma_n }{ \ul \sigma_n } \| t_1 - t_2 \|_{\ell^2}^{\gamma_n} \label{e-vlipbd}
\end{eqnarray}
where the first inequality is because $\|x /\|x\| - y/\|y\|\| \leq 2 \|x - y\|/(\|x\| \vee \|y\| )$ whenever $\|x\|,\|y\| \neq 0$ and the third is by the equivalence $s_{JK}^{-1} \asymp \tau_J$ (see Lemma \ref{lem-ip}). By (\ref{e-vlipbd}) and compactness of $\mathcal T$, we have $N(\mathcal T,d_n,\epsilon) \leq C (\tau_J \Gamma_n/(\epsilon \ul \sigma_n))^{d_T/\gamma_n} \vee 1$ for some finite constant $C$.
\end{proof}

\begin{proof}[\textbf{Proof of Corollary \ref{cor-ucb}}]
We verify the conditions of Lemma \ref{lem-bahadur} (or Theorem \ref{t-dist-u}). By assumption we may take $\ul \sigma_n \asymp \tau_J J^{a}$ with $a = \frac{1}{2} + \frac{|\alpha|}{d}$. Assumption \ref{a-functional} is therefore satisfied with $\eta_n = \sqrt n \tau_J^{-1} J^{-(p/d+1/2)}$ by Remark \ref{rmk-a5suff}(a') and Lemma \ref{lem-bias}.

The continuity condition in Remark \ref{rmk-a6suff} holds with $\Gamma_n = O(J^{a'})$ for some $a' > 0$ and $\gamma_n = 1$  since $\Psi_J$ is spanned by a B-spline basis of order $\gamma > (p \vee 2+|\alpha|)$ \cite[Section 5.3]{DeVoreLorentz}. Assumption \ref{a-Lipschitz}(i) therefore holds with $c_n = O(\sqrt{\log J})$ by Remark \ref{rmk-a6suff} because $(\tau_J \Gamma_n /(\epsilon \ul \sigma_n)) \lesssim (J^{a'-a} \epsilon^{-1})$. We can therefore take $r_n=(\log J)^{-\kappa}$ for $\kappa \in [1/2,1]$ in Assumption \ref{a-Lipschitz}(ii). The first condition in Assumption \ref{a-Lipschitz}(ii) then holds provided $J^5(\log J )^{6\kappa} /n = o(1)$. Since $\eta_n'=0$, the second condition in Assumption \ref{a-Lipschitz}(ii) holds provided
\begin{align*}
 \tau_J J \sqrt{(\log J)/n}+\sqrt n \tau_J^{-1} J^{-(p/d+1/2)} +  \Big( [J^{\frac{2+\delta}{2\delta}} \sqrt{(\log J)/n }]^{\frac{\delta}{1+\delta}} + J^{-p/d} + \tau_J \sqrt{J(\log J)/n} \Big) \sqrt{\log J} = o((\log J)^{-\kappa})
\end{align*}
(using Corollary \ref{c-deriv} for $\delta_{h,n}$). In applying Corollary \ref{c-deriv} we require that the conditions $\tau_J J/\sqrt n = O(1)$ and $J^{(2+\delta)/2\delta} \sqrt{(\log n)/n} = o(1)$ hold. Finally to apply Theorem \ref{t-dist-u} we also need $\tau_J J \sqrt{(\log J)/n} = o(1)$. Sufficient conditions for all these restrictions on $J$ are provided in the statement of this corollary. In particular, we note that $[ J^{\frac{2+\delta}{2\delta}} \sqrt{(\log n)/n}](\log J)^{\frac{1+\delta}{\delta}}$ decreases as $\delta >0$ increases. Hence the condition $J^5(\log n)^{6\kappa} /n = o(1)$ (for $\kappa \in [1/2, 1]$) implies that $[J^{\frac{2+\delta}{2\delta}} \sqrt{(\log n)/n }]^{\frac{\delta}{1+\delta}} (\log J)^{\kappa +0.5} =o(1)$ holds for all $\delta \geq 1$.
\end{proof}

\subsection{Proofs for Section \ref{s-inference}}

\begin{proof}[\textbf{Proof of Theorem \ref{t-cs}}]
The result will follow from Theorem \ref{t-dist}. Assumption \ref{a-residuals}(i)--(iii)(iv') is satisfied under Assumption CS(iii). Assumption \ref{a-sieve}(i)(ii)(iii) is satisfied by Assumption CS (iv) and the second part of Assumption CS(v), noting that $\zeta_{\psi,J} = O(\sqrt J)$ and $\zeta_{b,K} = O(\sqrt K) = O(\sqrt J)$. Since the basis spanning $\Psi_J$ is a Riesz basis for $T$ Assumption \ref{a-approx} is satisfied with $\tau_J \asymp \mu_J^{-1}$.

It remains to verify Assumption \ref{a-functional-prime}'(b). By the Riesz basis property and Assumption \ref{a-residuals}(i)--(iii) we have $[\sigma_n(f_{CS})]^2 \asymp \sum_{j=1}^J (a_j/\mu_j)^2 \lesssim J \mu_J^{-2}$ (see Section 6 of \cite{ChenPouzo2014}). For $f_{CS}$ we have
\[
 Df_{CS}(h_0)[h - h_0] = \int_{0}^{1} \left( \{h(\mf p(t),\mf y-S_{\mf y}(t)) - h_0(t,\mf y-S_{\mf y}(t))\}e^{-\int_{0}^t \partial_2 h_0(\mf p(v),\mf y-S_{\mf y}(v))\mf p'(v)\,\mathrm dv} \mf p'(t) \right)\mathrm dt
\]
\cite[p. 1471]{HausmanNewey1995} which is clearly a linear functional (Assumption \ref{a-functional-prime}'(b)(i)). Note that $\sigma_n(f_{CS}) \lesssim \sqrt{J} \mu_J^{-1}$ and Assumption CS(v) together imply $\mu_J^{-1}J^{3/2} \sqrt{(\log J)/n} = o(1)$. This, $p > 2$ and Corollary \ref{c-deriv} together imply that $\|\wh h - h_0\|_{B^2_{\infty,\infty}} = o_p(1)$ and $\|\widetilde h - h_0\|_{B^2_{\infty,\infty}} = o_p(1)$, and
\begin{align}
 \|\wh h - h_0\|_\infty & = O_p\left(J^{-p/2} + \mu_J^{-1} \sqrt{(J \log J)/n} \right), &
 \| \wh h - h_0\|_{B^1_{\infty,\infty}} & = O_p\left( \sqrt J \left( J^{-p/2} + \mu_J^{-1} \sqrt{(J \log J)/n} \right) \right) \label{e-rate-0} \\
 \|\wt h - h_0\|_\infty & = O_p\left(J^{-p/2} \right), &
 \| \wt h - h_0\|_{B^1_{\infty,\infty}} & = O_p\left( \sqrt J \left( J^{-p/2} \right) \right)~. \label{e-rate-1}
\end{align}
Applying Lemma A1 of \cite{HausmanNewey1995}, we obtain
\begin{align*}
 \left| f_{CS}(\wh h) - f_{CS}(h_0) - Df_{CS}(h_0)[\wh h - h_0] \right| & = O_p\left( \sqrt J \left( J^{-p} + \mu_J^{-2} \frac{J \log J}{n} \right) \right)  \\
  \left|  Df_{CS}(h_0)[\widetilde h - h_0] \right| & = O_p(J^{-p/2}) \,.
\end{align*}
Since $p>2$, Assumption CS(v) guarantees that Assumption \ref{a-functional-prime}'(b)(ii) holds with
\[
 \eta_n = \frac{\sqrt n}{\sigma_n(f_{CS})} \times \bigg( J^{-p/2} +\mu_J^{-2} \frac{J^{3/2} \log J}{n} \bigg)=o(1)\,.
\]
Finally, for Assumption \ref{a-functional-prime}'(b)(iii), we have
\[
 \frac{\|\Pi_K T( \wh v_n(f_{CS}) - v_n(f_{CS}))\|_{L^2(W)}}{\sigma_n(f_{CS})}  \lesssim \frac{\tau_J  \sqrt{\sum_{j=1}^J \left( Df_{CS}(\wh h)[(G_\psi^{-1/2} \psi^J)_j] - Df_{CS}(h_0) [(G_\psi^{-1/2} \psi^J)_j]\right)^2 }}{\sigma_n(f_{CS})}
\]
where $\tau_J \asymp \mu_J^{-1}$. Moreover,
\begin{align*}
 \Big| Df_{CS}(\wh h)[(G_\psi^{-1/2} \psi^J)_j] - Df_{CS}(h_0)[(G_\psi^{-1/2} \psi^J)_j] \Big| \lesssim \sqrt J \times  O_p\left( \sqrt J \left( J^{-p/2} + \mu_J^{-1} \sqrt{(J \log J)/n} \right) \right)
\end{align*}
(uniformly in $j = 1,\ldots J$) by Lemma A1 of \cite{HausmanNewey1995}. Therefore
\begin{align*}
\frac{\|\Pi_K T( \wh v_n(f_{CS}) - v_n(f_{CS}))\|_{L^2(W)}}{\sigma_n(f_{CS})} \lesssim \frac{ J^{3/2} \mu_J^{-1} O_p\left (J^{-p/2} + \mu_J^{-1} \sqrt{(J\log J)/n} \right)}{\Big(\sum_{j=1}^J (a_j/\mu_j)^2 \Big)^{1/2}} =O_p (\eta_n')
\end{align*}
which is $o_p(1)$ by Assumption CS(v). Finally we note that the condition $\mu_J^{-1}J \sqrt{(\log J)/n} = o(1)$ of Theorem \ref{t-dist} is trivially implied by $\mu_J^{-1}J^{3/2} \sqrt{(\log J)/n} = o(1)$ (which is in turn implied by Assumption CS(v)). This proves the result.
\end{proof}

\begin{proof}[\textbf{Proof of Corollary \ref{c-cs}}]
For \textbf{Result (1)}, since $\sigma_n(f_{CS})  \asymp J^{(a+\varsigma+1)/2}$, the first part of Assumption CS(v) is satisfied provided
\[
 \frac{\sqrt n}{J^{(a+\varsigma+1)/2}} \left( J^{-p/2} + J^{\varsigma+ 2}\sqrt{(\log J)}/n \right) = o(1)
\]
for which a sufficient condition is $nJ^{-(p+a+\varsigma+1)} = o(1)$ and $J^{3+\varsigma-a}(\log n)/n = o(1)$. Moreover, the condition $\mu_J^{-1}J^{3/2}\sqrt{(\log J)/n} = o(1)$ is implied by $J^{3+\varsigma-(a \wedge 0)}(\log n)/n = o(1)$. The condition $J^{3+\varsigma-(a \wedge 0)} (\log n)/n = o(1)$ also implies that $J^{(2+\delta)/(2\delta)} \sqrt{(\log n)/n} = o(1)$ holds whenever $\delta \geq 2/(2 + \varsigma - (a \wedge 0))$.

For \textbf{Result (2)}, we have $\sigma_n(f_{CS})^2 \gtrsim a_J^2/\mu_J^2 \asymp \exp( J^{\varsigma/2}+ a \log J )$. Take $J = (\log (n/(\log n)^\varrho))^{2/\varsigma}$. Then
\begin{align*}
 \sigma_n(f_{CS})^2 & \gtrsim \exp\left( \log (n/(\log n)^\varrho) + \log [(\log (n/(\log n)^\varrho))^{2a/\varsigma} ] \right) \\
 & = \exp\left( \log [ n/(\log n)^\varrho \times (\log (n/(\log n)^\varrho))^{2a/\varsigma} ] \right) \\
 & =  n/(\log n)^\varrho \times (\log (n/(\log n)^\varrho))^{2a/\varsigma}
\end{align*}
and so
\[
 \sigma_n(f_{CS}) \gtrsim \frac{\sqrt n}{(\log n)^{\varrho/2}} \times (\log (n/(\log n)^\varrho))^{a/\varsigma} \,.
\]
The first part of Assumption CS(v) is then satisfied provided
\begin{align*}
 \frac{(\log n)^{\varrho/2}}{(\log (n/(\log n)^\varrho))^{a/\varsigma}} & \Big( (\log (n/(\log n)^\varrho))^{-p/\varsigma}   + (\log n)^{-\varrho} \times (\log (n/(\log n)^\varrho))^{4/\varsigma} \log \log n \Big) = o(1)
\end{align*}
which holds provided $2p > \varrho \varsigma - 2a$ and $\varrho \varsigma > 8 - 2 a$. The condition $J^{(2+\delta)/(2\delta)} \sqrt{(\log n)/n} = o(1)$ holds for any $\delta > 0$. The remaining condition $\mu_J^{-1} J^{3/2}\sqrt{(\log J)/n} = o(1)$ is implied by
\[
  \frac{\sqrt n}{(\log n)^{\varrho/2}} (\log (n/(\log n)^\varrho))^{3/\varsigma} \sqrt{(\log\log n)/n}  = o(1)
\]
for which a sufficient condition is $\varrho \varsigma > 6$. Now, we may always choose $\varrho > 0$ so that $\varrho \varsigma > 6 \vee (8-2a)$. The remaining condition then holds provided $2p > 6 \vee (8-2a)- 2a$.
\end{proof}

\begin{proof}[\textbf{Proof of Theorem \ref{t-dl}}]
The proof follows by identical arguments to those of Theorem \ref{t-cs}, noting that
\[
 f_{DL}(h) - f_{CS}(h) = (\mf p^1 - \mf p^0)  h(\mf p^1,\mf y)
\]
and so
\begin{align*}
 f_{DL}(\wh h) - f_{DL}(h_0) &= f_{CS}(\wh h) - f_{CS}(h_0) + (\mf p^1 - \mf p^0) \left( \wh h(\mf p^1,\mf y) - h_0(\mf p^1,\mf y)  \right) \\
 Df_{DL}(\wh h)[h - h_0]  & = D f_{CS}(h_0)[h - h_0]   + (\mf p^1 - \mf p^0) \left( \wh h(\mf p^1,\mf y) - h_0(\mf p^1,\mf y)  \right) \\
 D f_{DL}(\wh h)[v] - Df_{DL}(h_0)[v] & = Df_{CS}(\wh h)[v] - Df_{CS}(h_0)[v]
\end{align*}
where clearly $|(\mf p^1 - \mf p^0) ( \wh h(\mf p^1,\mf y) - h_0(\mf p^1,\mf y)  )| \leq \mathrm{const} \times \|\wh h - h_0\|_\infty$. Since $\sigma_n (f_{DL}) \asymp \mu_J^{-1} \sqrt J$, the stated conditions on $J$ in this theorem imply that Assumption CS(v) holds.
\end{proof}

\begin{proof}[\textbf{Proof of Theorem \ref{t-apcs}}]
The result will follow from Theorem \ref{t-dist}, and is very similar to that of Theorem \ref{t-cs}.
Assumptions \ref{a-residuals}(i)--(iii)(iv'), \ref{a-sieve}(i)(ii)(iii), and \ref{a-approx} are verified as in the proof of Theorem \ref{t-cs}. It remains to verify Assumption \ref{a-functional-prime}'(b). As in the proof of Theorem \ref{t-cs} we have  $\tau_J \asymp \mu_J^{-1}$ and $[\sigma_n(f_A)]^2 \asymp \sum_{j=1}^J (a_j/\mu_j)^2$ (see Section 6 of \cite{ChenPouzo2014}). Simple expansion of $f_A$ yields
\[
 Df_A(h_0)[h - h_0] = \int w(\mf p) e^{h_0(\log \mf p,\log \mf y)}(h(\log \mf p,\log \mf y) - h_0(\log \mf p,\log \mf y)) \, \mathrm d \mf p
\]
which is clearly a linear functional (Assumption \ref{a-functional-prime}'(b)(i)), and
\begin{align*}
  & \left| f_A(\wh h) - f_A(h_0) - Df_A(h_0)[\wh h - h_0] \right| \\
  & = \int w(\mf p) \Big(e^{\wh h(\log  \mf p,\log  \mf y) - h_0(\log  \mf p,\log  \mf y)} - 1 - \big( \wh h(\log  \mf p,\log  \mf y) - h_0(\log  \mf p,\log \mf y) \big) \Big)e^{h_0(\log  \mf p,\log  \mf y)}\, \mathrm d \mf p \,.
\end{align*}
Therefore, by Corollary \ref{c-deriv} we have
\begin{align*}
 \left| f_A(\wh h) - f_A(h_0) - Df_A(h_0)[\wh h - h_0] \right| & = O_p \left( J^{-p} + \mu_J^{-2} \frac{J \log J}{n} \right) \\
  \left|  Df_A(h_0)[\widetilde h - h_0] \right|  & = O_p(J^{-p/2})\,.
\end{align*}
Since $p>0$, the stated conditions on $J$ in this theorem guarantees that Assumption \ref{a-functional-prime}'(b)(ii) holds with
\[
 \eta_n = \frac{\sqrt n}{\sigma_n(f_{A})} \times \bigg( J^{-p/2} +\mu_J^{-2} \frac{J \log J}{n} \bigg)=o(1)\,.
\]
Finally, for Assumption \ref{a-functional-prime}'(b)(iii), we have
\begin{align*}
 \frac{\|\Pi_K T(\wh v_n(f_A) - v_n(f_A))\|_{L^2(W)}}{\sigma_n(f_A)}
 \lesssim \frac{\tau_J}{\sigma_n(f_A)} \sqrt{ \left( \sum_{j=1}^J \left( Df_A(\wh h)[(G_\psi^{-1/2} \psi^J)_j] - Df_A(h_0)[(G_\psi^{-1/2} \psi^J)_j] \right)^2 \right)}
\end{align*}
where $\tau_J \asymp \mu_J^{-1}$ and where a first-order Taylor expansion of $Df_A$ yields
\begin{align*}
 \left| Df_A(\wh h)[(G_\psi^{-1/2} \psi^J)_j] - Df_A(h_0)[(G_\psi^{-1/2} \psi^J)_j]\right| & \lesssim \|(G_\psi^{-1/2} \psi^J)_j\|_\infty  \times \| \wh h - h_0\|_\infty \\
 & = O_p\left( \sqrt J \left( J^{-p/2} + \mu_J^{-1} \sqrt{(J\log J)/n} \right) \right) \,.
\end{align*}
It follows that
\[
 \frac{\|\Pi_K T(\wh v_n(f_A) - v_n(f_A))\|_{L^2(W)}}{[\sigma_n(f_A)]} \lesssim  \frac{ \mu_J^{-1}J \times  O_p\left (J^{-p/2} + \mu_J^{-1} \sqrt{(J\log J)/n}  \right) }{\Big(\sum_{j=1}^J (a_j/\mu_j)^2 \Big)^{1/2}}=O_p (\eta_n')
\]
which is $o_p(1)$ by the displayed condition on $J$ in this theorem. Finally we note that the condition $\mu_J^{-1}J \sqrt{(\log J)/n} = o(1)$ of Theorem \ref{t-dist} is implied by the displayed condition on $J$ in this theorem and the fact that $\sigma_n(f_A) \lesssim \sqrt{J} \mu_J^{-1}$. This proves the result.
\end{proof}

\begin{proof}[\textbf{Proof of Theorem \ref{t-ucb-cs}}]
We verify the conditions of Lemma \ref{lem-bahadur} and Theorem \ref{t-dist-u}. Assumptions \ref{a-data} and \ref{a-residuals} are satisfied by Assumption CS(i)(ii)(iv) and U-CS(i). Assumption \ref{a-sieve}(iii) is satisfied by Assumption U-CS(iii). Assumption \ref{a-approx}(i) is satisfied by the Riesz basis condition. For Assumption \ref{a-functional}(b), we check the conditions of Remark \ref{rmk-a5suff}(b'). It is clear that $Df_{CS,t}[h-h_0]$ (see display (\ref{e-Df_CS})) is a linear functional of $h-h_0$ for each $t \in \mc T$. As in the proof of Theorem \ref{t-cs} we have $[\sigma_n(f_{CS,t})]^2 \asymp \sum_{j=1}^J (a_{j,t}/\mu_j)^2$ uniformly in $t$. Thus, $\ul \sigma_n \lesssim \sqrt{J} \mu_J^{-1}$. This and Assumption U-CS(iv.1) together imply $\mu_J^{-1}J^{3/2} \sqrt{(\log J)/n} = o(1)$. Also we note that the first part of Assumption U-CS(iii) and $\delta \geq 1$ imply that $J^{(2+\delta)/(2\delta)}\sqrt{(\log n)/n} = o(1)$ holds. These results and $p > 2$ and Corollary \ref{c-deriv} together imply that $\|\wh h - h_0\|_{B^2_{\infty,\infty}} = o_p(1)$ and $\|\widetilde h - h_0\|_{B^2_{\infty,\infty}} = o_p(1)$, and equations (\ref{e-rate-0}) and (\ref{e-rate-1}) hold.
Therefore, $\wh h$ and $\widetilde h$ are within an $\epsilon$ neighborhood (in H\"older norm of smoothness 2) of $h_0$ wpa1.  As $\mc T = [\ul {\mf p}^0 , \ol {\mf p}^0] \times [\ul {\mf p}^1 , \ol {\mf p}^1] \times [\ul{\mf y}, \ol{\mf y}]$ where the intervals $[\ul {\mf p}^0 , \ol {\mf p}^0]$ and $[\ul {\mf p}^1 , \ol {\mf p}^1]$ are in the interior of the support of $\mf P_i$ and $[\ul{\mf y}, \ol{\mf y}]$ is in the interior of the support of $\mf Y_i$ and $h_0 \in B_\infty(p,L)$ with $p > 2$ and $0 < L < \infty$, it is straightforward to extend Lemma A1 of \cite{HausmanNewey1995} to show
\begin{align*}
 \sup_{t \in \mc T} \left| f_{CS,t}(\wh h) - f_{CS,t}(h_0) - Df_{CS,t}(h_0)[\wh h - h_0] \right| & = O_p\left( \sqrt J \left( J^{-p} + \mu_J^{-2} \frac{J \log J}{n} \right) \right)~,  \\
 \sup_{t \in \mc T} \left|  Df_{CS,t}(h_0)[\widetilde h - h_0] \right| & = O_p(J^{-p/2})
\end{align*}
by Corollary \ref{c-deriv}. Since $\ul \sigma_n \lesssim \sqrt{J} \mu_J^{-1}$, Assumption U-CS(iii) guarantees that Assumption \ref{a-functional}(b)(ii) holds with
\[
 \eta_n = \frac{\sqrt n}{ \ul \sigma_n } \times \bigg( J^{-p/2} +\mu_J^{-2} \frac{J^{3/2} \log J}{n} \bigg)\,.
\]
For Assumption \ref{a-functional}(b)(iii), we have
\begin{align*}
  \sup_{t \in \mc T} \frac{\|\Pi_K T(\wh v_n(f_{CS,t}) - v_n(f_{CS,t}))\|_{L^2(W)}}{[\sigma_n(f_{CS,t})]}
  \lesssim \sup_{t \in \mc T} \tau_J \frac{ \sqrt{\sum_{j=1}^J \left( Df_{CS,t}(\wh h)[(G_\psi^{-1/2} \psi^J)_j] - Df_{CS,t}(h_0) [(G_\psi^{-1/2} \psi^J)_j]\right)^2} }{[\sigma_n(f_{CS,t})]}
\end{align*}
where $\tau_J \asymp \mu_J^{-1}$. By straightforward extension of Lemma A1 of \cite{HausmanNewey1995} and (\ref{e-rate-0}) and (\ref{e-rate-1}):
\begin{align*}
  \sup_{t \in \mc T} \left| Df_{CS,t}(\wh h)[(G_\psi^{-1/2} \psi^J)_j] - Df_{CS,t}(h_0)[(G_\psi^{-1/2} \psi^J)_j]\right|
 & \lesssim \sqrt J \times  O_p\left( \sqrt J \left( J^{-p/2} + \mu_J^{-1} \sqrt{(J \log J)/n} \right) \right)
\end{align*}
whence Assumption \ref{a-functional}(b)(iii) holds with
\begin{align*}
 \eta_n' & =  \frac{ J^{3/2} \mu_J^{-1}}{ \ul \sigma_n } \times \left( J^{-p/2} + \mu_J^{-1} \sqrt{J(\log J)/n}  \right)\,.
\end{align*}
which is $o(1)$ by Assumption U-CS(iv). This verifies Assumption \ref{a-functional}(b).

Finally, Assumption \ref{a-Lipschitz}(i) holds with $c_n = O(\sqrt{\log J})$ by Assumption U-CS(ii) and Remark \ref{rmk-a6suff}. For Assumption \ref{a-Lipschitz}(ii) we take $r_n = [\log J]^{-1/2}$. Assumption \ref{a-Lipschitz}(ii.1) then holds provided $J^5(\log J )^{3} /n = o(1)$. Assumption \ref{a-Lipschitz}(ii.2) holds provided
\begin{align*}
 \tau_J J \sqrt{(\log J)/n}+\eta_n +  \Big( [J^{\frac{2+\delta}{2\delta}} \sqrt{(\log J)/n }]^{\frac{\delta}{1+\delta}} + J^{-p/d} + \tau_J \sqrt{J(\log J)/n} +\eta_n' \Big) \sqrt{\log J} = o((\log J)^{-1/2})
\end{align*}
(using Corollary \ref{c-deriv} for $\delta_{h,n}$), which is satisfied provided
\begin{align*}
 \tau_J J \sqrt{(\log J)/n}+\eta_n +  \eta_n' \sqrt{\log J} = o((\log J)^{-1/2})
\end{align*}
which is in turn implied by Assumption U-CS(iii) and U-CS(iv.1) and the property $\ul \sigma_n \lesssim \sqrt{J} \mu_J^{-1}$. Thus Lemma \ref{lem-bahadur} applies to $f_t = f_{CS,t}$ with a rate $r_n = [\log J]^{-1/2}$.

Next we note that the condition $ \eta_n' \sqrt J = o((\log J)^{-1/2})$ needed for Theorem \ref{t-dist-u} is directly implied by Assumption U-CS(iv.2).
\end{proof}

\begin{proof}[\textbf{Proof of Theorem \ref{t-ucb-dl}}]
Follows by similar arguments to the proofs of Theorems \ref{t-dl} and \ref{t-ucb-cs}, noting that
\[
 |Df_{DL,t_1}(h_0)[h] - Df_{DL,t_2}(h_0)[h]| \leq |Df_{CS,t_1}(h_0)[h] - Df_{CS,t_2}(h_0)[h]| + |(\mf p^1_1 - \mf p^0_1)h(\mf p^1_1,\mf y_1) - (\mf p^1_2 - \mf p^0_2)h(\mf p^1_2,\mf y_2)|
\]
and so $c_n = O(\sqrt{\log J})$ by Assumption U-CS(ii) and Remark \ref{rmk-a6suff} (see the proof of Corollary \ref{cor-ucb}). We can then take $r_n = [\log J]^{-1/2}$.
\end{proof}

\subsection{Proofs for Appendix \ref{s-l2}}

\begin{proof}[\textbf{Proof of Theorem \ref{t-l2}}]
As with the proof of Theorem \ref{t-rate}, we first decompose the error into three parts:
\begin{eqnarray*}
 \|\wh h - h_0\|_{L^2(X)} & \leq & \|\wh h - \widetilde h\|_{L^2(X)} + \|\widetilde h - \Pi_J h_0\|_{L^2(X)} + \|\Pi_J h_0 - h_0\|_{L^2(X)} \\
 & =: & T_1 + T_2 + \|h_0 - \Pi_J h_0 \|_{L^2(X)}\,.
\end{eqnarray*}
To prove \textbf{Result (1)} it is enough to show that $T_2 \leq O_p(1) \times \|h_0 - \Pi_J h_0\|_{L^2(X)}$. To do this, bound
\begin{eqnarray*}
 T_2 & \leq & \|G_\psi^{1/2} (S'G_b^{-1/2})^-_l G_b^{-1/2} B'(H_0 - \Psi c_J)/n\|_{\ell^2} \\
 & &  + \|G_\psi^{1/2} \{(\wh G_b^{-1/2} \wh S)^-_l\wh G_b^{-1/2}G_b^{1/2} - (G_b^{-1/2} S)^-_l\} G_b^{-1/2} B'(H_0 - \Psi c_J)/n\|_{\ell^2} \quad =: \quad T_{21} + T_{22}\,.
\end{eqnarray*}
For $T_{21}$,
\begin{eqnarray*}
 T_{21} & \leq & s_{JK}^{-1} \| G_b^{-1/2} B'(H_0 - \Psi c_J)/n\|_{\ell^2} \\
 & \leq & O_p ( \tau_J \zeta_{b,K}/\sqrt n ) \times \|h_0 - \Pi_J h_0\|_{L^2(X)} + \tau_J \| \Pi_K T (h_0 - \Pi_J h_0) \|_{L^2(W)} \\
 & \leq & O_p ( \tau_J \zeta_{b,K}/\sqrt n ) \times \|h_0 - \Pi_J h_0\|_{L^2(X)} + \tau_J \|  T (h_0 - \Pi_J h_0) \|_{L^2(W)} \\
 & = & O_p(1)\times \|h_0 - \Pi_J h_0\|_{L^2(X)}
\end{eqnarray*}
where the second line is by Lemma \ref{lem-BHJ} and the relations $J \asymp K$  and $\tau_J \asymp s_{JK}^{-1}$, and the final line is by Assumption \ref{a-approx}(ii) and the condition $\tau_J \zeta \sqrt{(\log J)/n} = o(1)$. Similarly,
\begin{eqnarray*}
 T_{22} & \leq & \|G_\psi^{1/2} \{(\wh G_b^{-1/2} \wh S)^-_l\wh G_b^{-1/2}G_b^{1/2} - (G_b^{-1/2} S)^-_l\} \|_{\ell^2} \| G_b^{-1/2} B'(H_0 - \Psi c_J)/n\|_{\ell^2} \\
 & \leq & O_p(s_{JK}^{-2} \zeta \sqrt{(\log J)/n} ) \times \left( O_p (  \zeta_{b,K}/\sqrt n ) \times \|h_0 - \Pi_J h_0\|_{L^2(X)} +  \| T (h_0 - \Pi_J h_0) \|_{L^2(W)} \right) \\
 & = & o_p( \tau_J \zeta \sqrt{(\log J)/n} )^2 \times  \|h_0 - \Pi_J h_0\|_{L^2(X)} + O_p(1) \times \tau_J  \| T (h_0 - \Pi_J h_0) \|_{L^2(W)} \\
 & = & O_p(1) \times \|h_0 - \Pi_J h_0\|_{L^2(X)}
\end{eqnarray*}
where the second line is by Lemmas \ref{lem-BHJ} and \ref{lem-SGl2}(b) and the relations  $J \asymp K$  and $\tau_J \asymp s_{JK}^{-1}$, and the final line is by the condition $\tau_J \zeta \sqrt{(\log J)/n} = o(1)$ and Assumption \ref{a-approx}(ii). This proves Result (1).

To prove \textbf{Result (2)} it remains to control $T_1$. To do this, bound
\begin{eqnarray*}
 T_1 & \leq & \|G_\psi^{1/2} (S'G_b^{-1/2})^-_l G_b^{-1/2} B'u/n\|_{\ell^2} + \|G_\psi^{1/2} \{(\wh G_b^{-1/2} \wh S)^-_l\wh G_b^{-1/2}G_b^{1/2} - (G_b^{-1/2} S)^-_l\} G_b^{-1/2} B'u/n\|_{\ell^2} \\
 & =: & T_{11} + T_{12}\,.
\end{eqnarray*}
For $T_{11}$, by definition of $s_{JK}$ and Lemma \ref{lem-BuL2} we have:
\begin{equation*}
 T_{11}  \leq  s_{JK}^{-1} \|G_b^{-1/2} B'u/n\|_{\ell^2}  =  O_p (s_{JK}^{-1} \sqrt{K/n})  =  O_p (\tau_J \sqrt{J/n})
\end{equation*}
where the final line is because $J \asymp K$ and $\tau_J \asymp s_{JK}^{-1}$ (Lemma \ref{lem-ip}). Similarly,
\begin{eqnarray*}
 T_{12} & \leq & \|G_\psi^{1/2} \{(\wh G_b^{-1/2} \wh S)^-_l\wh G_b^{-1/2}G_b^{1/2} - (G_b^{-1/2} S)^-_l\} \|_{\ell^2} \|G_b^{-1/2} B'u/n\|_{\ell^2} \\
 & = &  O_p(\tau_J^2 \zeta \sqrt{(\log J)/n}) \times O_p(\sqrt{J/n}) \\
 & = &  o_p(\tau_J \sqrt{J/n} )
\end{eqnarray*}
where the second line is by Lemmas \ref{lem-BuL2} and \ref{lem-SGl2}(b) and the relations $J \asymp K$  and $\tau_J \asymp s_{JK}^{-1}$, and the final line is by the condition $\tau_J \zeta \sqrt{(\log J)/n} = o(1)$.
\end{proof}

\begin{proof}[\textbf{Proof of Corollary \ref{c-rate2}}]
Analogous to the proof of Corollary \ref{c-deriv}.
\end{proof}

\begin{proof}[\textbf{Proof of Theorem \ref{t-lb-l2}}]
As in the proof of Theorem \ref{npiv lower bound}, it suffices to prove a lower bound for the Gaussian reduced-form NPIR model (\ref{e-npir}). Theorem \ref{t-lb-l2-npir} below does just this.
\end{proof}

\begin{theorem}\label{t-lb-l2-npir}
Let Condition LB hold with $B_2(p,L)$ in place of $B_\infty(p,L)$ hold for the NPIR model (\ref{e-npir}) with a random sample $\{(W_i,Y_i)\}_{i=1}^n$.  Then for any $0 \leq |\alpha| < p$:
\begin{equation*}
 \liminf_{n \to \infty} \inf_{\wh g_n} \sup_{h \in B_2(p,L)} \mb P_h \left( \|\wh g_n - \partial^\alpha h\|_{L^2(X)} \geq c n^{-(p-|\alpha|)/(2(p+\varsigma)+d)} \right) \geq c'>0
\end{equation*}
in the mildly ill-posed case, and
\begin{equation*}
 \liminf_{n \to \infty} \inf_{\wh g_n} \sup_{h \in B_2(p,L)} \mb P_h \left( \|\wh g_n - \partial^\alpha h\|_{L^2(X)} \geq c(\log n)^{-(p-|\alpha|)/\varsigma} \right) \geq c'>0
\end{equation*}
in the severely ill-posed case,
in the severely ill-posed case, where $\inf_{\wh g_n}$ denotes the infimum over all estimators of $\partial^\alpha h$ based on the sample of size $n$, $\sup_{h \in B_2(p,L)} \mb P_h$ denotes the sup over $h \in B_2(p,L)$  and distributions  $(W_i,u_i)$ which satisfy Condition LB with $\nu$ fixed, and the finite positive constants $c, c'$ depend only on $p,L,d,\varsigma$ and $\sigma_0$.
\end{theorem}

\begin{proof}[\textbf{Proof of Theorem \ref{t-lb-l2-npir}}]
We use similar arguments to the proof of Theorem \ref{npir lower bound}, using Theorem 2.5 of \cite{Tsybakov2009} (see Theorem \ref{t-tsybakov}). Again, we first explain the scalar $(d = 1)$ case in detail. Let $\{\phi_{j,k},\psi_{j,k}\}_{j,k}$ be a wavelet basis of regularity $\gamma > p$ for $L^2([0,1])$ as described in Appendix \ref{ax-basis}.

By construction, the support of each interior wavelet is an interval of length $2^{-j}(2r-1)$. Thus for all $j$ sufficiently large (hence the $\liminf$ in our statement of the Lemma) we may choose a set $M \subset \{r,\ldots,2^j-r-1\}$ of interior wavelets with cardinality $\#(M) \asymp 2^j$ such that $\mbox{support}(\psi_{j,m}) \cap \mbox{support}(\psi_{j,m'}) = \emptyset$ for all $m,m' \in M$ with $m \neq m'$.

Take $g_0 \in B(p,L/2)$ and for each $m \in M$ define $\theta = \{\theta_m\}_{m \in M}$ where each $\theta_m \in \{0,1\}$ and define
\[
 h_\theta = g_0 + c_0 2^{-j(p+1/2)}\sum_{m \in M} \theta_m \psi_{j,m}
\]
for each $\theta$, where $c_0$ is a positive constant to be defined subsequently. Note that this gives $2^{(\#(M))}$ such choices of $h_\theta$. By the equivalence $\|\cdot\|_{B^p_{2,2}} \asymp \|\cdot\|_{b^p_{2,2}}$, for each $\theta$ we have:
\begin{eqnarray*}
  \|h_\theta\|_{B^p_{2,2}} & \leq & L/2 + \left\| c_0 2^{-j(p+1/2)}\sum_{m \in M} \theta_m \psi_{j,m}\right\|_{B^p_{2,2}} \\
  & \leq & L/2 + \mbox{const} \times \left\| c_0 2^{-j(p+1/2)}\sum_{m \in M} \theta_m \psi_{j,m}\right\|_{b^p_{2,2}} \\
  & = & L/2 + \mbox{const} \times c_0  2^{-j(p+1/2)} \left( \sum_{m \in M} \theta_m ^2 2^{2jp} \right)^{1/2} \\
  & \leq & L/2 + \mbox{const} \times c_0  \,.
\end{eqnarray*}
Therefore, we can choose $c_0$ sufficiently small that $h_\theta \in B_2(p,L)$ for each $\theta$.

Since $\psi_{j,m} \in C^\gamma$ with $\gamma > |\alpha|$ is compactly supported and $X_i$ has density bounded away from $0$ and $\infty$, we have $\|2^{j/2}\psi^{(|\alpha|)}(2^j x-m)\|_{L^2(X)} \asymp 1$ (uniformly in $m$). By this and the disjoint support of the $\psi_{j,m}$, for each $\theta,\theta'$ we have:
\begin{eqnarray*}
 \| \partial^\alpha h_\theta - \partial^\alpha h_{\theta'}\|_{L^2(X)} & = & c_02^{-j(p-|\alpha| + 1/2)}\left( \sum_{m \in M} (\theta_m - \theta'_m)^2 \|2^{j/2}\psi^{(|\alpha|)}(2^j \,\cdot\,-m)\|_{L^2(X)}^2 \right)^{1/2} \\
 & \gtrsim & c_02^{-j(p-|\alpha|+ 1/2)}\sqrt{\rho(\theta,\theta')}
\end{eqnarray*}
where $\rho(\theta,\theta')$ is the Hamming distance between $\theta$ and $\theta'$. Take $j$ large enough that $\#(M)  \geq 8$. By the Varshamov-Gilbert bound \cite[Lemma 2.9]{Tsybakov2009} we may choose a subset $\theta^{(0)}, \theta^{(1)},\ldots,\theta^{(M^*)}$ such that $\theta^{0} = (0,\ldots,0)$, $\rho(\theta^{(a)},\theta^{(b)}) \geq \#(M)/8 \gtrsim 2^j$ for all $0 \leq a < b \leq M^*$ and $M^* \geq 2^{\#(M)/8}$ (recall that there were $2^{(\#(M))}$ distinct vectors $\theta$ and $\#(M) \asymp 2^j$). For each $m \in \{0,1,\ldots,m^*\}$ let $h_m := h_{\theta^{(m)}}$. Then we have
\[
 \| \partial^\alpha h_m - \partial^\alpha h_{m'}\|_{L^2(X)}  \gtrsim c_02^{-j(p-|\alpha|)}
\]
for each $0 \leq m < m' \leq M^*$.

For each $0 \leq m \leq M^*$, let $P_m$ denote the joint distribution of $\{(W_i,Y_i)\}_{i=1}^n$ with $Y_i = T h_m(W_i) + u_i$ for the Gaussian NPIR model (\ref{e-npir}). It follows from Condition LB(ii)(iii) that for each $1 \leq m \leq M^*$ the KL distance $K(P_m,P_0)$ is
\begin{eqnarray*}
 K(P_m,P_0) & \leq & \frac{1}{2}\sum_{i=1}^n (c_0 2^{-j(p+\frac{1}{2})})^2 E\left[\frac{(T \sum_{k \in M} \theta^{(m)}_k (\psi_{j,k}(W_i)))^2 }{\sigma^2(W_i)}\right] \\
 & \leq & \frac{n}{2\underline \sigma^2}(c_0 2^{-j(p+\frac{1}{2})})^2 \nu(2^j)^2 \sum_{k \in M} (\theta^{(m)}_k)^2  \|\psi_{j,k}\|_{L^2(X)}^2 \\
 & \lesssim & n c_0^2 2^{-2jp} \nu(2^j)^2
\end{eqnarray*}
where the final line is because $X_i$ has density bounded away from $0$ and $\infty$ and $\sum_{k \in M} (\theta^{(m)}_k)^2 \leq \#(M) \asymp 2^j$ for each $1 \leq m \leq M^*$.

In the mildly ill-posed case ($\nu(2^j) =  2^{-j \varsigma}$) we choose $2^{j} \asymp n^{1/(2(p+\varsigma)+1)}$. This yields:
\begin{align*}
 K(P_m,P_0) &\lesssim c_0^2 n^{\frac{1}{2(p+\varsigma)+1}} \quad \mbox{ uniformly in $m$}\\
 \log(M^*) &\gtrsim 2^j \asymp n^{1/(2(p+\varsigma)+1)} \,.
\end{align*}
since $M^* \geq 2^{\#(M)/8}$ and $\#(M) \asymp 2^j$.

In the severely ill-posed case ($\nu(2^j)= \exp(- \frac{1}{2} 2^{j \varsigma})$) we choose $2^j = (c_1 \log n)^{1/\varsigma}$ with $c_1 > 1$. This yields:
\begin{align*}
 K(P_m,P_0) &\lesssim c_0^2 n^{-(c_1-1)} \quad \mbox{ uniformly in $m$}\\
 \log(M^*) &\gtrsim  2^j \asymp (\log n)^{1/\varsigma} \,.
\end{align*}
In both the mildly and severely ill-posed cases, the result follows by choosing $c_0$ sufficiently small that both $\|h_m\|_{B^p_{2,2}} \leq L$ and $K(P_m,P_0) < \frac{1}{8} \log (M^*)$ hold uniformly in $m$ for all $n$ sufficiently large. All conditions of Theorem 2.5 of \cite{Tsybakov2009} are satisfied and hence we obtain the lower bound result.

In the multivariate case ($d > 1$) we let $\widetilde \psi_{j,k,G}(x)$ denote an orthonormal tensor-product wavelet for $L^2([0,1]^d)$ at resolution level $j$.
We construct a family of submodels analogously to the univariate case, setting $h_\theta = g_0 + c_0 2^{-j(p+d/2)}\sum_{m \in M} \theta_m \widetilde  \psi_{j,m,G}$ where $\psi_{j,m}$ is now the product of $d$ interior univariate wavelets at resolution level $j$ with $G = (w_\psi)^d$ (see Appendix \ref{ax-basis}) and where $\#(M) \asymp 2^{jd}$. We then use the Varshamov-Gilbert bound to reduce this to a family of models $h_m$ with $0 \leq m \leq M^*$ and $M^* \asymp 2^{jd}$. We then have:
\begin{equation*}
 \|\partial^\alpha h_m - \partial^\alpha h_{m'}\|_\infty \gtrsim c_0 2^{-j(p-|\alpha|)}
\end{equation*}
for each $0 \leq m < m' \leq M^*$, and
\begin{eqnarray*}
 K(P_m,P_0) & \lesssim &  n(c_0 2^{-j(p+d/2)})^2 \nu(2^j)^2
\end{eqnarray*}
for each $1 \leq m \leq M^*$, where $\nu(2^j) = 2^{-j \varsigma}$ in the mildly ill-posed case and $\nu(2^j) \asymp \exp(-  2^{j \varsigma})$ in the severely ill-posed case. We choose $2^j \asymp n^{1/(2(p+\varsigma)+d)}$ in the mildly ill-posed case and $2^j = (c_1 \log n)^{1/\varsigma}$ in the severely ill-posed case.
The result follows as in the univariate case.
\end{proof}

\subsection{Proofs for Appendix \ref{s-quad}}

\begin{proof}[\textbf{Proof of Theorem \ref{t-lbquad}}]
As in the proof of Theorem \ref{npiv lower bound}, this follows from the lower bound for NPIR in Theorem \ref{t-npirlb}.
\end{proof}

 The following is a slightly stronger ``in probability'' version of Lemma 1 in \cite{YuAFLC}, which is used to prove Theorem \ref{t-npirlb}. Let $\mc P$ be a family of probability measures, let $\theta(P)$ be a parameter with values in a pseudo-metric space $(\mc D,d)$ for some distribution $P \in \mc P$, and let $\hat \theta(P)$ be an estimator of $\theta(P)$ taking values in $(\mc D,d)$. If $\theta \in \mc D$ and $D \subset \mc D$, we let $d(\theta,D) = \inf_{\theta' \in D} d(\theta,\theta')$. Let $co(\mc P)$ denote the convex hull of a set of measures $\mc P$. Finally, if $\mb P,\mb Q \in \mc P$ we let $\|\mb P - \mb Q\|_{TV}$ denote the total variation distance and $\mr{aff}(\mb P,\mb Q) = 1-\|\mb P - \mb Q\|_{TV}$ denote the affinity between $\mb P$ and $\mb Q$.

\begin{lemma}\label{lem-lecam}
Suppose there are subsets $D_1,D_2 \subset \mc D$ that are $2\delta$ separated for some $\delta > 0$ (i.e. $d(s_1,s_2) \geq 2 \delta$ for all $s_1 \in D_1$ and $s_2 \in D_2$) and subsets $\mb P_1,\mb P_2 \subset \mc P$ for which $\theta(\mb P) \in D_1$ for all $\mb P \in \mc P_1$ and $\theta(\mb P) \in D_2$ for all $\mb P \in \mc P_2$. Then:
\[
 2\sup_{\mb P \in \mc P} \mb P( d(\wh \theta, \theta(\mb P)) \geq \delta ) \geq \sup_{\mb P_1 \in co(\mc P_1),\mb P_2 \in co(\mc P_2)} \mr{aff}( \mb P_1 , \mb P_2) \,.
\]
\end{lemma}

\begin{proof}[\textbf{Proof of Lemma \ref{lem-lecam}}]
We proceed as in the proof of Lemma 1 in \cite{YuAFLC}. Let $P_1 \in \mc P_1$ and $P_2 \in \mc P_2$. Then:
\begin{eqnarray*}
 2\sup_{\mb P \in \mc P} \mb P( d(\wh \theta, \theta(\mb P)) \geq \delta )
 & \geq & \mb P_1( d(\wh \theta, \theta(\mb P_1)) \geq \delta ) +  \mb P_2( d(\wh \theta, \theta(\mb P_2)) \geq \delta ) \\
 & \geq & \mb P_1( d(\wh \theta, D_1) \geq \delta ) +  \mb P_2( d(\wh \theta, D_2) \geq \delta )\,.
\end{eqnarray*}
Since the inequality $2\sup_{\mb P \in \mc P} \mb P( d(\wh \theta, \theta(\mb P)) \geq \delta ) \geq \mb P_1( d(\wh \theta, D_1) \geq \delta ) +  \mb P_2( d(\wh \theta, D_2) \geq \delta )$ holds for any fixed $\mb P_1 \in \mc P_1$ and $\mb P_2 \in \mc P_2$, it must also hold for any $\mb P_1 \in co(\mc P_1)$ and $\mb P_2 \in co(\mc P_2)$. Also note that
\begin{eqnarray*}
 \ind\{ d(\wh \theta, D_1) \geq \delta \} +  \ind\{ d(\wh \theta, D_2) \geq \delta \} & \geq & \ind\{ d(\wh \theta, D_1)+ d(\wh \theta, D_2) \geq 2\delta \} \\
 & \geq & \ind\{ d(D_1, D_2) \geq 2\delta \} \quad = \quad  1
\end{eqnarray*}
because $d(D_1,D_2) \geq 2 \delta$. Now by definition of $\alpha(\cdot,\cdot)$, for any $\mb P_1 \in co(\mc P_1)$ and $\mb P_2 \in co(\mc P_2)$ we have:
\begin{eqnarray*}
 2\sup_{\mb P \in \mc P} \mb P( d(\wh \theta, \theta(\mb P)) \geq \delta ) & \geq & \mb P_1( d(\wh \theta, D_1) \geq \delta ) +  \mb P_2( d(\wh \theta, D_2) \geq \delta ) \\
 & \geq & \inf\{ \mb P_1 f + \mb P_2 g : f,g \mbox{ non negative and measurable with } f + g \geq 1\} \\
 & = & \mr{aff}(\mb P_1,\mb P_2)\,.
\end{eqnarray*}
The result follows by taking the supremum of the right-hand side over $\mb P_1$ and $\mb P_2$.
\end{proof}

\begin{theorem}\label{t-npirlb}
Let Condition LB hold with $B_2(p,L)$ in place of $B_\infty(p,L)$ for the NPIR model (\ref{e-npir}) with a random sample $\{(W_i,Y_i)\}_{i=1}^n$.Then for any $0 \leq |\alpha| < p$:
\[
 \liminf_{n \to \infty} \inf_{\wh g_n} \sup_{h \in B_2(p,L)} \mb P_h \left( |\wh g_n - f(h) | > c r_n \right) \geq c' > 0
\]
where
\[
 r_n = \left[ \begin{array}{ll}
 n^{-1/2} & \mbox{in the mildly ill-posed case when } p \geq \varsigma + 2|\alpha| + d/4 \\
 n^{-4(p-|\alpha|)/(4(p+\varsigma)+d)} & \mbox{in the mildly ill-posed case when } \varsigma < p < \varsigma + 2|\alpha| + d/4 \\
 (\log n)^{-2(p-|\alpha|)/\varsigma} & \mbox{in the severely ill-posed case,}
 \end{array} \right.
\]
$\inf_{\wh g_n}$ denotes the infimum over all estimators of $f(h_0)$ based on the sample of size $n$, $\sup_{h \in B_2(p,L)} \mb P_h$ denotes the sup over $h \in B_2(p,L)$  and distributions  $(W_i,u_i)$ which satisfy Condition LB with $\nu$ fixed, and the finite positive constants $c, c'$ do not depend on $n$.
\end{theorem}

\begin{proof}[\textbf{Proof of Theorem \ref{t-npirlb}}]
We first prove the result for the scalar ($d=1$) case, then describe the modifications required in the multivariate case.

Let $\{\phi_{j,k},\psi_{j,k}\}_{j,k}$ be a CDV wavelet basis of regularity $\gamma > p$ for $L^2([0,1])$,  as described in Appendix \ref{ax-basis}. As in the proof of Theorem \ref{t-lb-l2-npir}, we choose a set $M \subset \{r,\ldots,2^j-r-1\}$ of interior wavelets with cardinality $\mf m := \#(M) \asymp 2^j$ such that $\mbox{support}(\psi_{j,m}) \cap \mbox{support}(\psi_{j,m'}) = \emptyset$ for all $m,m' \in M$ with $m \neq m'$. Let $\theta = \{\theta_m\}_{m \in M}$ where each $\theta_m \in \{-1,1\}$ and for each $\theta \in \{-1,1\}^{\mf m}$ define:
\[
 h_\theta =  \sum_{m \in M} \frac{\theta_m c_0 2^{-j p}}{\sqrt{\mf m}} \psi_{j,m} \,.
\]
and let $h_0 = 0$. By the equivalence  $\|\cdot\|_{b^p_{2,2}} \asymp \|\cdot\|_{B^p_{2,2}}$, we have:
\begin{eqnarray*}
  \left\| h_\theta \right\|_{B^p_{2,2}}
  \quad \lesssim \quad  \|h_\theta\|_{b^p_{2,2}}
  \quad = \quad \left( 2^{2jp} \sum_{m \in M} \frac{ \theta_m^2 c_0^2 2^{-2j p}}{\mf m} \right)^{1/2}
  \quad = \quad c_0 \,.
\end{eqnarray*}
Therefore, we may choose $c_0$ sufficiently small that $h_\theta \in B_2(p,L)$ for all $\theta \in \{-1,1\}^{\mf m}$.

Let $\psi^{(|\alpha|)}$ denote the $|\alpha|$th derivative of $\psi$. By disjoint support of the $\psi_{j,m}(x) = 2^{j/2} \psi(2^j x - m)$, $\mu(x) \geq \ul \mu > 0$, and a change of variables, we have:
\begin{eqnarray}
 |f(h_\theta) - f(h_0)| & = & \mf{m}^{-1} \sum_{m \in M} \int ( c_0 \theta_m 2^{-jp} \psi_{j,m}^{(|\alpha|)}(x))^2 \mu(x)\, \mr dx \notag \\
  & \gtrsim & \mf{m}^{-1} \sum_{m \in M} \int ( c_0 \theta_m 2^{-jp} \psi_{j,m}^{(|\alpha|)}(x))^2 \, \mr dx \notag \\
 & = & c_0^2 2^{-2jp} \mf{m}^{-1} \sum_{m \in M}  \int  2^{(2|\alpha|+1)j} \psi^{(|\alpha|)}(2^j x - m)^2 \,\mr dx \notag \\
 & = & c_0^2 2^{-2j(p-|\alpha|)}\int \psi^{(|\alpha|)}(u)^2 \,\mr du \quad  \gtrsim \quad c_0^2 2^{-2j(p-|\alpha|)} \,. \notag
 \end{eqnarray}
Therefore, there exists a constant $c_* > 0$ such that
\begin{equation} \label{e-fdist}
 |f(h_\theta) - f(h_0)| > 2c_* 2^{-2j(p-|\alpha|)}
\end{equation}
holds for all for each $\theta \in \{-1,1\}^{\mf m}$ whenever $j$ is sufficiently large.

Let $P_0$ (respectively $P_\theta$) denote the joint distribution of $\{(W_i,Y_i)\}_{i=1}^n$ with $Y_i = T h_0(W_i) + u_i$ (respectively $Y_i = T h_\theta(W_i) + u_i$) for the Gaussian NPIR model (\ref{e-npir}) where, under Condition LB, we may assume that $X_i$ and $W_i$ have uniform marginals and that the joint density $f_{XW}(x,w)$ of $(X_i,W_i)$ has wavelet expansion
\[
 f_{XW}(x,w) = \sum_{k=0}^{2^{r_0}-1} \lambda_{r_0} \varphi_{r_0,k}(x)\varphi_{r_0,k}(w) + \sum_{j=r_0}^\infty  \sum_{k=0}^{2^{j}-1} \lambda_{j} \psi_{j,k}(x) \psi_{j,k}(w) \,.
\]
Observe that
\[
 T \psi_{j,k}(w) = \int \psi_{j,k}(x) f_{XW}(x,w)\,\mr dx = \lambda_j \psi_{j,k}(w)
\]
for each $0 \leq k \leq 2^{j}-1$ and each $j \geq r_0$ ($r_0$ is fixed) and that $|\lambda_j| \asymp \nu(2^j)$ by Condition LB(iii). Let $P^*$ denote the mixture distribution obtained by assigning weight $2^{-\mf m}$ to $P_\theta$ for each of the $2^{\mf m}$ realizations of $\theta$. Lemma \ref{lem-tv} yields
\begin{equation} \label{e-tv0}
 \|P^* - P_0\|^2_{TV} \lesssim \frac{ n^2 2^{-4jp}\nu(2^j)^4 }{\mf m} \,.
\end{equation}

In the mildly ill-posed case ($\nu(2^j) = 2^{j \varsigma}$) we have
\[
 \|P^* - P_0\|^2_{TV} \lesssim  n^2  2^{-j(4(p+\varsigma)+1)}
\]
because $\mf m \asymp 2^j$. Choose $2^j \asymp c n^{2/(4(p+\varsigma)+1)}$ with $c$ sufficiently small so $\|P^* - P_0\|_{TV}  \leq 1-\epsilon$ for some $1>\epsilon > 0$ and all $n$ large enough, whence:
\begin{eqnarray}
 \mr{aff}(P^*,P_0) \quad = \quad  1-\| P^*-  P_0\|_{TV}  \quad \geq \quad \epsilon \label{e-aff1}
\end{eqnarray}
for all $n$ sufficiently large. It now follows by Lemma \ref{lem-lecam} and equations (\ref{e-fdist}) and (\ref{e-aff1}) that for all $n$ sufficiently large, any estimator $\wh g_n$ of $f(h)$ obeys the bound
\begin{eqnarray} \label{e-lb-aff}
 \sup_{h \in B_2(p,L)} \mb P_h \left( |\wh g_n - f(h) | > c_* 2^{-2j(p-|\alpha|)} \right) & \geq & \epsilon/2
\end{eqnarray}
where $2^{-2j(p-|\alpha|)} \asymp n^{-4(p-|\alpha|)/(4(p+\varsigma)+1)}$. This is slower than $n^{-1/2}$ whenever $p \leq \varsigma+ 2|\alpha| + 1/4$.

In the severely ill-posed case ($\nu(2^j) = \exp(-\frac{1}{2}2^{\varsigma j})$) we choose $2^j = (c \log n)^{1/\varsigma}$ for some $c \in (0,1)$. This yields  $\|P^* - P_0\|_{TV} = o(1)$ by (\ref{e-tv0}) and hence there exists $\epsilon > 0$ such that $\mr{aff}(P^*,P_0) \geq \epsilon$ for all $n$ sufficiently large. Then by Lemma \ref{lem-lecam} and equation (\ref{e-fdist}), for all $n$ sufficiently large, any estimator $\wt f_n$ of $f(h)$ obeys the same bound (\ref{e-lb-aff}) with $2^{-2j(p-|\alpha|)} \asymp (\log n)^{-2(p-|\alpha|)/\varsigma}$.

In the multivariate case ($d > 1$) we let $\wt \psi_{j,k,G}(x)$ denote an orthonormal tensor-product wavelet for $L^2([0,1]^d)$ at resolution level $j$, as described in Appendix \ref{ax-basis}. We may choose a subset $M$ of $\{0,\ldots,2^j-1\}^d$ with $\mf m := \# (M) \asymp 2^{dj}$ for which each $m \in M$ indexes a tensor-product of interior wavelets of the form $2^{j/2} \psi(2^j x_l - m_i)$, which we denote by $\wt \psi_{j,m}(x)$, such that $\wt \psi_{j,m}$ and $\wt \psi_{j,m'}$ have disjoint support for each $m ,m' \in M$ with $m \neq m'$. For each $\theta \in \{-1,1\}^{\mf m}$ we define
\[
 h_\theta = \sum_{m \in M} \frac{\theta_m c_0 2^{-j p}}{\sqrt{\mf m}} \wt \psi_{j,m}(x)
\]
with $c_0$ sufficiently small such that $h_\theta \in B_2(p,L)$ for each $\theta$.  Let $h_0(x) = 0$ for all $x \in [0,1]^d$. By disjoint support of the $\wt \psi_{j,m}$ and a change of variables, we have:
\begin{eqnarray*}
 |f(h_\theta) - f(h_0)|
 & = & c_0^2 2^{-2jp} \mf{m}^{-1} \sum_{m \in M}  \int \left(  \prod_{i = 1}^d    2^{(2\alpha_i+1)j} \psi^{(\alpha_i)}(2^j x_i - m_i)^2 \right) \mu(x)\,\mr dx \\
  & \gtrsim & c_0^2 2^{-2jp} \mf{m}^{-1} \sum_{m \in M} \int \left(  \prod_{i = 1}^d    2^{(2\alpha_i+1)j} \psi^{(\alpha_i)}(2^j x_i - m_i)^2 \right)\,\mr dx \quad \gtrsim \quad c_0^2 2^{-2j(p-|\alpha|)} \,.
 \end{eqnarray*}
Letting $P_0$, $P_\theta$, and $P^*$ be defined analogously to in the univariate case, we let $X_i$ and $W_i$ have uniform marginals on $[0,1]^d$ and their joint density $f_{XW}(x,w)$ has wavelet expansion
\begin{equation}
 f_{XW}(x,w) = \sum_{j = r_0}^\infty  \sum_{G \in G_{j,r_0}} \sum_{k} \lambda_{j} \wt \psi_{j,k,G}(x) \wt \psi_{j,k,G}(w)
\end{equation}
with $|\lambda_j| \asymp \nu(2^j)$.
Lemma \ref{lem-tv} again yields
\[
\| P^* -  P_0\|^2_{TV} \lesssim \frac{ n^2 2^{-4jp} \nu(2^j)^4 }{\mf m} \,.
\]
The result follows by choosing $2^j \asymp c n^{2/(4(p+\varsigma)+d)}$ with sufficiently small $c$ in the mildly ill-posed case and $2^j = (c \log n)^{1/\varsigma}$ for some $c \in (0,1)$ in the severely ill-posed case.
\end{proof}

\begin{lemma}\label{lem-tv}
Let the Condition LB hold with $B_2(p,L)$ in place of $B_\infty(p,L)$ for the NPIR model (\ref{e-npir}), let $P^*$ and $P_0$ be as described in the proof of Theorem \ref{t-npirlb}, and let $2^{-jp} \nu(2^j) = o(1)$. Then:
\[
\|P^* - P_0\|^2_{TV} \lesssim \frac{ n^2 2^{-4jp} \nu(2^j)^4 }{\mf m} \,.
\]
\end{lemma}

\begin{proof}[\textbf{Proof of Lemma \ref{lem-tv}}]
We prove the result for the multivariate case. For each $\theta \in \{-1,1\}^{\mf m}$, the density of $P_\theta$ with respect to $P_0$ is
\begin{eqnarray*}
 \frac{\mr d P_\theta}{\mr d P_0} & = & \prod_{i=1}^n \exp \left\{ \frac{-1}{2\sigma_0^2}  \left(  [T (h_\theta - h_0)(W_i)]^2 - 2 u_i[T (h_\theta - h_0)(W_i)] \right) \right\} \\
 & = & \prod_{i=1}^n \exp \left\{ \frac{-1}{2\sigma_0^2}  \left(   \left[  \sum_{m \in M} \frac{\theta_m c_0 2^{-j p}\lambda_j}{\sqrt{\mf m}} \wt \psi_{j,m}(W_i)\right]^2 - 2 u_i \left[ \sum_{m \in M} \frac{\theta_m c_0 2^{-j p} \lambda_j }{\sqrt{\mf m}} \wt \psi_{j,m}(W_i) \right] \right) \right\}  \,.
\end{eqnarray*}
Since the $\wt \psi_{j,m}$ have disjoint support, we have
\begin{eqnarray*}
 \frac{\mr d P_\theta}{\mr d P_0}  & = & \prod_{i=1}^n \exp \left\{ -\frac{1}{2} \sum_{m \in M} \frac{  c_0^2 2^{-2jp} \lambda_j^2}{\sigma_0^2 \mf m}  \wt \psi_{j,m}(W_i)^2 + \frac{u_i}{\sigma_0} \sum_{m \in M} \theta_m \frac{ c_0  2^{-j p}\lambda_j }{\sigma_0 \sqrt{\mf m}} \wt \psi_{j,m}(W_i) \right\} \\
 & = &   \prod_{i=1}^n  \exp\left\{ \sum_{m \in M} \left( -\frac{1}{2}\Delta_{i,j,m}^2 +  \theta_m \frac{u_i}{\sigma_0} \Delta_{i,j,m} \right) \right\}
\end{eqnarray*}
where $u_i \sim N(0,\sigma_0^2)$ under $P_0$ and
\[
 \Delta_{i,j,m} = \frac{c_0 2^{-jp} \lambda_j}{\sigma_0 \sqrt{\mf m}} \wt \psi_{j,m}(W_i) \,.
\]
Therefore:
\[
 \frac{\mr d P_\theta}{\mr d P_0} = \prod_{i=1}^n \left(1 +  A_{i,j}(\theta) \right)
\]
where
\begin{align*}
 A_{i,j}(\theta) & = \exp\left\{ \sum_{m \in M} \left( -\frac{1}{2}\Delta_{i,j,m}^2 +  \theta_m \frac{u_i}{\sigma_0} \Delta_{i,j,m} \right) \right\} - 1 \\
 & =  \left\{ \sum_{m \in M} \exp \left( -\frac{1}{2}\Delta_{i,j,m}^2 +  \theta_m \frac{u_i}{\sigma_0} \Delta_{i,j,m} \right) \right\} - 1
\end{align*}
and the second line is again by disjoint support of the $\wt \psi_{j,m}$ (which implies $\Delta_{i,j,m}$ is nonzero for at most one $m$ for each $i$).

Let $E_0$ be expectation under the measure $P_0$ and observe that $E_0[A_{i,j}(\theta)] = 0$ for each $\theta \in \{-1,1\}^{\mf m}$. For each $\theta,\theta'$ we define the vector $\kappa_{\theta,\theta'} \in \mb R^n$ whose $i$th element is:
\begin{align*}
 \kappa_{\theta,\theta'}(i) & =  E_0[A_{i,j}(\theta) A_{i,j}(\theta')] \\
  & = \sum_{m \in M} \sum_{m' \in M} E_0 \left[ \exp\left\{-\frac{1}{2}\Delta_{i,j,m}^2 -\frac{1}{2}\Delta_{i,j,m'}^2 +  \frac{u_i}{\sigma_0}(\theta_m \Delta_{i,j,m} + \theta'_{m'} \Delta_{i,j,m'})\right\} - 1\right] \\
 & = \sum_{m \in M} \sum_{m' \in M} E_0 \left[ e^{\theta_m \theta'_{m'} \Delta_{i,j,m} \Delta_{i,j,m'}} - 1\right] \\
 & = \sum_{m \in M} E_0 \left[ e^{\theta_m \theta'_{m} \Delta_{i,j,m}^2} - 1\right] \,.
\end{align*}
where the final line is again by disjoint support of the $\wt \psi_{j,m}$. Using  $|\wt \psi_{j,k}| \lesssim 2^{dj/2}$, $E_0 [\widetilde \psi_{j,k}(X_i)^2] = 1$, and $\mf m \asymp 2^{dj}$, it is straightforward to derive the bounds:
\begin{eqnarray}
 |\Delta_{i,j,m}|
 & \lesssim & \frac{ c_0 2^{-jp}\lambda_j}{ \sigma_0} \label{e-D1} \\
 E_0[\Delta_{i,j,m}^2] & = & \frac{ c_0^2 2^{-2jp}\lambda_j^2 }{ \sigma_0^2\mf m } \label{e-D2} \,.
\end{eqnarray}
By Taylor's theorem:
\begin{eqnarray*}
 \kappa_{\theta,\theta'}(i)  & = & \sum_{m \in M} \theta_m \theta'_{m} E_0[\Delta_{i,j,m}^2] + \sum_{m \in M}  \frac{1}{2} E_0[\Delta_{i,j,m}^4] + r_3(\theta,\theta') \\
 & = & \sum_{m \in M} \theta_m \theta'_{m} E_0[\Delta_{i,j,m}^2] + r_2(\theta,\theta')
\end{eqnarray*}
where $r_2$ and $r_3$ are remainder terms. Using the Lagrange remainder formula and (\ref{e-D1}) and (\ref{e-D2}), we may deduce that
\begin{eqnarray*}
 r_2(\theta,\theta')
 \quad \leq \quad \frac{1}{2} \sum_{m} E_0 \left[ e^{\Delta_{i,j,m}^2}  \Delta_{i,j,m}^4 \right]
 \quad \lesssim \quad \exp\left\{\frac{C c_0^2 2^{-2jp}\lambda_j^2}{\sigma_0^2}\right\}\frac{c_0^4 2^{-4jp}\lambda_j^4 }{\sigma_0^4 \mf m}\,.
\end{eqnarray*}
and
\begin{eqnarray*}
 r_3(\theta,\theta')
 \quad \leq \quad \frac{1}{6} \sum_{m} E_0 \left[ e^{\Delta_{i,j,m}^2}  \Delta_{i,j,m}^6 \right]
 \quad \lesssim \quad \exp\left\{\frac{C c_0^2 2^{-2jp}\lambda_j^2}{\sigma_0^2 }\right\}\frac{c_0^6 2^{-6jp}\lambda_j^6 }{\sigma_0^6 \mf m} .
\end{eqnarray*}
where $C$ is a finite positive constant and we again used the fact that $\Delta_{i,j,m}$ is nonzero for at most one $m$ for each $i$.

Lemma 22 of \cite{Pollard-metrics} provides the bound:
\[
 \| P^* -  P_0\|^2_{TV} \leq 2^{-2\mf m} \sum_{\theta \in\{-1,1\}^{\mf m}} \sum_{\theta' \in\{-1,1\}^{\mf m}} \Upsilon (\kappa_{\theta,\theta'})
\]
where for any vector $c = (c_1,\ldots,c_n)' \in \mb R^n$ the function $\Upsilon (c)$ is defined as
\[
 \Upsilon(c) = -1 + \prod_{i=1}^n (1+c_i) = \sum_{i=1}^n c_i + \sum_{i_1=1}^n \sum_{i_2=i_1 +1}^n c_{i_1} c_{i_2} + \mbox{higher-order terms}
\]
where the higher-order terms are sums over triples, quadruples, etc, with all distinct indices, up to $c_1 c_2 \ldots c_n$.
Therefore:
\begin{align}
 \|P^* - P_0\|^2_{TV} & \leq  2^{-2\mf m} \sum_{\theta,\theta'} \Bigg[ \sum_{i=1}^n \left( \sum_{m \in M} \theta_m \theta'_{m} E_0[\Delta_{i,j,m}^2] + \sum_{m \in M}  \frac{1}{2} E_0[\Delta_{i,j,m}^4] + r_3(\theta,\theta') \right) \notag \\
 &  \quad + \sum_{i_1=1}^n \sum_{i_2=i_1 +1}^n \Bigg\{ \left( \sum_{m \in M} \theta_m \theta'_{m} E_0[\Delta_{i_1,j,m}^2] + r_2(\theta,\theta') \right)  \left( \sum_{m' \in M} \theta_{m'} \theta'_{m'} E_0[\Delta_{i_2,j,m'}^2] + r_2(\theta,\theta') \right) \Bigg\} \notag \\
 &  \quad + \mbox{higher order terms} \Bigg] \,. \label{e-tv1}
\end{align}
Since $\sum_{\theta,\theta'} \theta_m \theta'_{m'} = 0$ for all $m,m' \in M$ and $\sum_{\theta,\theta'} 1 = 2^{2\mf m}$, the first-order sum in  (\ref{e-tv1}) is:
\begin{align}
 & 2^{-2\mf m} \sum_{i=1}^n  \sum_{\theta,\theta'} \left( \sum_{m}  \theta_m \theta'_{m'} E_0[\Delta_{i,j,m}^2] + \frac{1}{2} \sum_{m} E_0[\Delta_{i,j,m}^4] + r_3(\theta,\theta') \right) \notag \\
 & = \frac{n}{2}\sum_{m}  E_0 \left[   \Delta_{i,j,m}^4 \right] + n 2^{-2\mf m} \sum_{\theta,\theta'} r_3(\theta,\theta') \notag \\
 & \lesssim \frac{n  c_0^4 2^{-4jp}\lambda_j^4 }{ \sigma_0^4 \mf m}  \left( 1 + \exp\left\{\frac{c_0^2 2^{-2jp} \lambda_j^2}{\sigma_0^2}\right\}\frac{c_0^2 2^{-2jp}\lambda_j^2 }{\sigma_0^2 \mf m}  \right)\label{e-tv11}\,.
\end{align}
Also observe that
\[
 \sum_{\theta} \sum_{\theta'} (\theta_m \theta_m') (\theta_{m'} \theta_{m'}') =
 \left\{ \begin{array}{ll} 0 & \mbox{if $m \neq m'$} \\
 2^{2\mf m} & \mbox{if $m = m'$.} \end{array} \right.
\]
The second-order sum in (\ref{e-tv1}) is therefore:
\begin{align}
 & \frac{n(n-1)}{2} \left( \mf m E_0[\Delta_{i,j,m}^2]^2 + O( \mf m \max_{m \in M} E[\Delta_{i,j,m}^2] \times \max_{\theta,\theta'} r_2(\theta,\theta')) + O(\max_{\theta,\theta'} r_2(\theta,\theta')^2) \right) \notag \\
 & \asymp \frac{n^2 c_0^4 2^{-4jp} \lambda_j^4 }{\sigma_0^4 \mf m} \left( 1 + O\left( \exp\left\{\frac{2C c_0^2 2^{-2jp}\lambda_j^2 }{\sigma_0^2}\right\} \left(  2^{-2jp} \lambda_j^2 + \frac{2^{-4jp } \lambda_j^4}{\mf m} \right) \right) \right)  \,.\label{e-tv21}
 \end{align}
The higher-order terms in (\ref{e-tv1}) will be of asymptotically smaller order because $2^{-jp} \lambda_j \asymp 2^{-jp} \nu(2^j) = o(1)$  .  Substituting (\ref{e-tv11}) and (\ref{e-tv21}) into (\ref{e-tv1}) yields:
\begin{eqnarray*}
 \|P^* - P_0\|^2_{TV} \quad \lesssim \quad \frac{ n^2  2^{-4jp}\lambda_j^4 }{\mf m} \left( 1 +o(1)  \right) \quad \asymp \quad \frac{ n^2  2^{-4jp}\nu(2^j)^4 }{\mf m}
\end{eqnarray*}
as required.
\end{proof}

\subsection{Proofs for Appendix \ref{ax-basis}}

\begin{proof}[\textbf{Proof of Lemma \ref{lem-spline1}}]
Part (a) is equation (3.4) on p. 141 of \cite{DeVoreLorentz}. For part (b), let $v \in \mb R^J$, let $f_X(x)$ denote the density of $X_i$ and let $\underline f_X = \inf_x f_X(x)$ and $\overline f_X = \sup_x f_X(x)$. Then for any $v \in \mb R^J$:
\begin{eqnarray*}
 v'E[\psi^J(X_i)\psi^J(X_i)']v	& \geq & \underline f_X  \int_0^1 (\psi^J(x)'v)^2\,\mathrm dx
 \\
 & \geq & \underline f_X  c_1^2 \frac{\min_{-r+1 \leq j \leq m} (t_{j+r}-t_r)}{r} \|v\|_{\ell^2}^2 \\
 & \geq & \underline f_X c_1^2 \frac{c_2 J^{-1}}{r} \|v\|^2_{\ell^2}
\end{eqnarray*}
for some finite positive constant $c_1$, where the first inequality is by Assumption \ref{a-data}(i), the second is by Theorem 4.2 (p. 145) of \cite{DeVoreLorentz} with $p = 2$, and the third is by uniform boundedness of the mesh ratio. By the variational characterization of eigenvalues of selfadjoint matrices, we have:
\begin{eqnarray*}
 \lambda_{\min}(G_\psi) = \min_{v \in \mb R^J, v \neq 0} \frac{v'E[\psi^J(X_i)\psi^J(X_i)']v}{\|v\|_{\ell^2}^2}  \geq \underline f_X c_1^2 \frac{c_2 J^{-1}}{r}
\end{eqnarray*}
This establishes the upper bound on $\lambda_{\min}(G_\psi)^{-1}$. The proof of the lower bound for $\lambda_{\max}(G_\psi)^{-1}$ follows analogously
by Theorem 4.2 (p. 145) of \cite{DeVoreLorentz} with $p = 2$. Part (c) then follows directly from part (b).
\end{proof}

\begin{proof}[\textbf{Proof of Lemma \ref{lem-spline2}}]
The $\ell^1$ norm of the tensor product of vectors equals the product of the $\ell^1$ norms of the factors, whence part (a) follows from Lemma \ref{lem-spline1}. As $\psi^J(x)$ is formed as the tensor-product of univariate B-splines, each element of $\psi^J(x)$ is of the form $\prod_{l=1}^{d} \psi_{J i_l}(x_l)$ where $\psi_{J i_l}(x_l)$ denotes the $i_l$th element of the vector of univariate B-splines. Let $v \in \mb R^J$. We may index the elements of $v$ by the multi-indices $i_1,\ldots,i_d \in \{1,\ldots,m+r\}^{d}$. By boundedness of $f_X$ away from zero and Fubini's theorem, we have:
\begin{eqnarray*}
 v'E[\psi^J(X_i)\psi^J(X_i)']v	& \geq & \underline f_X  \int_0^1\cdots \int_0^1 \Big( \sum_{i_1,\ldots,i_l} v_{i_1\cdots i_d} \prod_{l=1}^d \psi_{J i_l}(x_l) \Big)^2\,\mathrm dx_1 \cdots \mathrm dx_d
 \\
 & = & \underline f_X  \int_0^1\cdots \int_0^1 \sum_{i_2,\ldots,i_l} \sum_{j_2,\ldots,j_l} \left( \prod_{l=2}^d \psi_{J i_l}(x_l) \right) \left( \prod_{l=2}^{d} \psi_{J j_l}(x_l) \right) \\
 & & \quad \left\{ \int_0^1 \sum_{i_1} \sum_{j_1} v_{i_1\cdots i_{d}} v_{j_1 \cdots j_{d}} \psi_{J i_1}(x_1) \psi_{J j_1}(x_1) \mathrm dx_1 \right\} \,\mathrm dx_2 \cdots \mathrm dx_d \,.
\end{eqnarray*}
Applying Theorem 4.2 (p. 145) of \cite{DeVoreLorentz} to the term in braces, and repeating for $x_2,\ldots,x_d$, we have:
\begin{eqnarray*}
 v'E[\psi^J(X_i)\psi^J(X_i)']v & \geq & \underline f_X \left(c_1^2 \frac{\min_{-r+1 \leq j \leq m} (t_{j+r}-t_r)}{r} \right)^d \|v\|_{\ell^2}^2  \\
 & \geq & \underline f_X c_1^{2d} \frac{c_2^d J^{-1}}{r^d} \|v\|^2_{\ell^2}
\end{eqnarray*}
where the second inequality is by uniform boundedness of the mesh ratio. The rest of the proof follows by identical arguments to Lemma \ref{lem-spline1}.
\end{proof}

\begin{proof}[\textbf{Proof of Lemma \ref{lem-wavelet1}}]
Each of the interior $\varphi_{j,k}$ and $\psi_{j,k}$ have support $[2^{-j}(-N+1+k),2^{-j}(N+k)]$, therefore $\varphi_{j,k}(x) \neq 0$ (respectively $\psi_{j,k}(x) \neq 0$) for less than or equal to $2N$ interior $\varphi_{j,k}$ (resp. $\psi_{j,k}$) and for any $x \in [0,1]$. Further, there are only $N$ left and right $\varphi_{j,k}$ and $\psi_{j,k}$. Therefore, $\varphi_{j,k}(x) \neq 0$ (respectively $\psi_{j,k}(x) \neq 0$) for less than or equal to $3N$ of the $\varphi_{j,k}$  (resp. $\psi_{j,k}$) at resolution level $j$ for each $x \in [0,1]$. By construction of the basis, each of $\varphi$, $\varphi^l_{j,k}$, $\psi^l_{j,k}$ for $k = 0,\ldots,N-1$ and $\varphi^r_{j,-k}$, $\psi^r_{j,-k}$ for $k =1,\ldots,N$ are continuous and therefore attain a finite maximum on $[0,1]$. Therefore, each of the $\varphi_{j,k}$ and $\psi_{j,k}$ are uniformly bounded by some multiple of $2^{j/2}$ and so:
\[
 \xi_{\psi,J} \lesssim 3N \times (\underbrace{2^{L_0/2}}_{\mbox{for the $\varphi_{L_0,k}$}} + \underbrace{2^{L_0/2}}_{\mbox{for the $\psi_{L_0,k}$}} + \ldots + \underbrace{2^{L/2}}_{\mbox{for the $\psi_{L,k}$}} ) \lesssim 2^{L/2}
\]
The result then follows because $J = 2^{L+1}$. For part (b), because $f_X$ is uniformly bounded away from $0$ and $\infty$ and the wavelet basis is orthonormal for $L^2[0,1]$, we have
\[
 v'E[\psi^J(X_i)\psi^J(X_i)'] v \asymp  v'  \left(  \int_0^1\psi^J(x)\psi^J(x)'\,\mathrm dx \right) v = \|v\|^2_{\ell^2}
\]
and so all eigenvalues of $G_\psi$ are uniformly (in $J$) bounded away from $0$ and $\infty$. Part (c) follows directly.
\end{proof}

\begin{proof}[\textbf{Proof of Lemma \ref{lem-wavelet2}}]
Lemma \ref{lem-wavelet1} implies that each of the factor vectors in the tensor product at level $j$ has $\ell^1$ norm of order $O(2^{dj/2})$ uniformly for $x = (x_1,\ldots,x_d)' \in [0,1]^d$ and in $j$. There are at most $2^d$ such tensor products at each resolution level. Therefore, $\xi_{\psi,J} = O(2^{dL/2}) = O(\sqrt J)$ since $J = O(2^{dL})$. Parts (b) and (c) follow by the same arguments of the proof of Lemma \ref{lem-wavelet1} since the tensor-product basis is orthonormal for $L^2([0,1]^d)$.
\end{proof}

\subsection{Proofs for Appendix \ref{ax-supp}}
\
\begin{proof}[\textbf{Proof of Lemma \ref{lem-inv}}]
$\|A^{-1} - I_r\|_{\ell^2} = \| A^{-1}(A - I_r) \|_{\ell^2}  \leq \|A^{-1}\|_{\ell^2} \|A - I_r\|_{\ell^2}$.
\end{proof}

\begin{proof}[\textbf{Proof of Lemma \ref{lem-linv}}]
The first assertion is immediate by Theorem 3.3 of \cite{Stewart1977} and definition of $A^-_l$ and $B^-_l$. For the second part, Weyl's inequality implies that $s_{\min}(B) \geq \frac{1}{2}s_{\min}(A)$ whenever $\|A - B\|_{\ell^2} \leq \frac{1}{2}s_{\min}(A)$.
\end{proof}

\begin{proof}[\textbf{Proof of Lemma \ref{lem-mpnorm}}]
$\|A_l^-\|_{\ell^2}^2 = \lambda_{\max}( A_l^- (A_l^-)') = \lambda_{\max}( (A'A)^{-1} ) = 1/\lambda_{\min}(A'A) = s_{\min}(A)^{-2}$.
\end{proof}

\begin{proof}[\textbf{Proof of Lemma \ref{lem-projm}}]
The result follows from \cite{LiLiCui} (see also \cite{Stewart1977}).
\end{proof}

\begin{proof}[\textbf{Proof of Lemma \ref{lem-matl2}}]
We prove the results for $\wh S^o$; convergence of $\wh G_\psi^o$ and $\wh G_b^o$ is proved in Lemma 2.1 of \cite{ChenChristensen-reg}. Note that
\[
 \wh S^o - S^o = \sum_{i=1}^n n^{-1} G_b^{-1/2}\{b^K(W_i)\psi^J(X_i)' - E[b^K(W_i)\psi^J(X_i)']\}G_\psi^{-1/2} =: \sum_{i=1}^n \Xi_i^o
\]
where $\|\Xi_i\|_{\ell^2} \leq 2n^{-1} \zeta_{b,K} \zeta_{\psi,J}$. Also,
\begin{eqnarray*}
 \left\|\sum_{i=1}^n E[\Xi_i^o \Xi_i^{o\prime}] \right\|_{\ell^2} & \leq & n^{-1} \|E[G_b^{-1/2}b^K(W_i)\psi^J(X_i)'G_\psi^{-1} \psi^J(X_i)b^K(W_i)'G_b^{-1/2}]\|_{\ell^2} \\
 & \leq & n^{-1} \zeta_{\psi,J}^2 \|E[G_b^{-1/2}b^K(W_i)b^K(W_i)'G_b^{-1/2}]\|_{\ell^2}  \\
 & = & n^{-1} \zeta_{\psi,J}^2 \|I\|_{\ell^2} \\
 & = & n^{-1} \zeta_{\psi,J}^2
\end{eqnarray*}
by the fact that $\|I\|_{\ell^2} = 1$. An identical argument yields the bound $\|\sum_{i=1}^n E[\Xi_i^{o\prime} \Xi_i^o] \|_{\ell^2} \leq n^{-1} \zeta_{b,K}^2$. Applying a Bernstein inequality for random matrices \cite[Theorem 1.6]{Tropp2012} yields
\[
 \mb P \left( \|\wh S^o - S^o\|_{\ell^2} > t \right) \leq 2\exp\left\{\log K - \frac{-t^2/2}{(\zeta_{b,K}^2 \vee \zeta_{\psi,J}^2)/n+ 2 \zeta_{b,K} \zeta_{\psi,J} t/(3n)} \right\} \,.
\]
The convergence rate $\|\wh S^o - S^o\|_{\ell^2}$ from this inequality under appropriate choice of $t$.
\end{proof}

\begin{proof}[\textbf{Proof of Lemma \ref{lem-BHJ}}]
Let $\widetilde b^K(x) = G_b^{-1/2} b^K(x)$ and denote $\widetilde b^K(x)' = (\widetilde b_{K1}(x),\ldots,\widetilde b_{KK}(x))$. As the summands have expectation zero, we have
\begin{eqnarray*}
 E\left[ \|G_b^{-1/2}\{B'(H_0 - H_J)/n - E[b^K(W_i)(h_0(X_i) - h_J(X_i))]\}\|_{\ell^2}^2 \right]
 & \leq & \frac{1}{n} E \left[ \sum_{k=1}^K (\widetilde b_{Kk}(W_i))^2 (h_0(X_i) - h_J(X_i))^2 \right] \\
 & \leq & \frac{K}{n} \|h_0 - h_J\|^2_\infty \wedge \frac{\zeta_{b,K}^2}{n} \|h_0 - h_J\|_{L^2(X)}^2\,.
\end{eqnarray*}
The result follows by Chebyshev's inequality.
\end{proof}

\begin{proof}[\textbf{Proof of Lemma \ref{lem-SGl2}}]
We begin by rewriting the target in terms of the orthonormalized matrices
\begin{eqnarray}
 (\wh G_b^{-1/2} \wh S)^-_l\wh G_b^{-1/2}G_b^{1/2} - (G_b^{-1/2} S)^-_l
 & = & G_\psi^{-1/2}\{ (\wh S^{o\prime} \wh G_b^{o-1}\wh S^o)^{-1}\wh S^{o\prime} \wh G_b^{o-1}   - (S^{o\prime}S^o)^{-1}S^{o\prime} \} \notag \\
 & = & G_\psi^{-1/2}\{ ((\wh G_b^o)^{-1/2}\wh S^o)^-_l ( \wh G_b^o)^{-1/2} - (S^o)^-_l \} \,. \label{e-maintarget}
\end{eqnarray}
We first bound the term in braces. By the triangle inequality,
\begin{equation}
 \begin{array}{rcl}
 & & \|((\wh G_b^o)^{-1/2}\wh S^o)^-_l ( \wh G_b^o)^{-1/2} - (S^o)^-_l\|_{\ell^2} \\
 & & \quad \leq \quad \|((\wh G_b^o)^{-1/2}\wh S^o)^-_l - (S^o)^-_l \|_{\ell^2} \|( \wh G_b^o)^{-1/2}\|_{\ell^2}  + \|( \wh G_b^o)^{-1/2} - I\|_{\ell^2}\|(S^o)^-_l\|_{\ell^2} \,.\label{e-target}
 \end{array}
\end{equation}
Lemma \ref{lem-matl2} provides that
\begin{eqnarray}
 \|\wh G_b^o - I_K \|_{\ell^2} & = & O_p (\zeta_{b,K} \sqrt{(\log K)/n}) \label{e-Gcgce} \\
 \|\wh S^o - S^o\|_{\ell^2} & = & O_p ((\zeta_{b,K} \vee \zeta_{\psi,J})\sqrt{(\log K)/n}) \label{e-Scgce} \,.
\end{eqnarray}
Let $\mathcal A_n$ denote the event upon which $\|\wh G_b^o - I_K \|_{\ell^2} \leq \frac{1}{2}$ and note that $\mb P(\mathcal A_n^c) = o(1)$ because $\|\wh G_b^o - I_K \|_{\ell^2}= o_p(1)$. Then by Lemmas \ref{lem-inv} and \ref{lem-sqrtm} we have
\begin{eqnarray*}
 \|(\wh G_b^o)^{-1/2} - I_K\|_{\ell^2} & \leq & \sqrt 2 \|(\wh G_b^o)^{1/2} - I_K\|_{\ell^2} \\
 & \leq &  \frac{2}{1+\sqrt 2} \|\wh G_b^o - I_K\|_{\ell^2} \notag
\end{eqnarray*}
on $\mathcal A_n$. It follows by expression (\ref{e-Gcgce}) and the fact that $\mb P(\mathcal A_n^c) = o(1)$ that
\begin{equation}\label{e-Gbinvrate}
 \|(\wh G_b^o)^{-1/2} - I\|_{\ell^2} = O_p (\zeta_{b,K} \sqrt{(\log K)/n})
\end{equation}
which in turn implies that $\|(\wh G_b^o)^{-1/2}\| = 1+o_p(1)$.

To bound $\|((\wh G_b^o)^{-1/2}\wh S^o)^-_l - (S^o)^-_l \|_{\ell^2}$, it follows by equations (\ref{e-Scgce}) and (\ref{e-Gbinvrate}) and the fact that $\|S^o\|_{\ell^2} \leq 1$ that:
\begin{eqnarray}
 \|(\wh G_b^o)^{-1/2}\wh S^o - S^o\|_{\ell^2} & \leq & \|(\wh G_b^o)^{-1/2} - I_K\|_{\ell^2} \|\wh S^o\|_{\ell^2} + \|\wh S^o - S^o\|_{\ell^2} \notag \\
 & = & O_p ((\zeta_{b,K} \vee \zeta_{\psi,J})\sqrt{(\log K)/n}) \label{e--lbound}\,.
\end{eqnarray}
Let $\mathcal A_{n,1} \subseteq \mathcal A_n$ denote the event on which $\|(\wh G_b^o)^{-1/2}\wh S^o - S^o\|_{\ell^2} \leq \frac{1}{2}s_{JK}$ and note that $\mb P(\mathcal A_{n,1}^c) = o(1)$ by virtue of the condition $s_{JK}^{-1} (\zeta_{b,K} \vee \zeta_{\psi,J})\sqrt{(\log K)/n} = o(1)$. Lemma \ref{lem-linv} provides that
\begin{equation} \label{e-stewart}
 \|((\wh G_b^o)^{-1/2}\wh S^o)^-_l - (S^o)^-_l \|_{\ell^2} \leq 2(1+\sqrt 5) s_{JK}^{-2}  \|(\wh G_b^o)^{-1/2}\wh S^o - S^o\|_{\ell^2}
\end{equation}
 on $\mathcal A_{n,1}$, and so
 \begin{equation} \label{e-asa}
  \|((\wh G_b^o)^{-1/2}\wh S^o)^-_l - (S^o)^-_l \|_{\ell^2} = O_p(s_{JK}^{-2}(\zeta_{b,K} \vee \zeta_{\psi,J})\sqrt{(\log K)/n})
\end{equation}
by (\ref{e--lbound}) and (\ref{e-stewart}). It follows from equations (\ref{e-asa}) and (\ref{e-maintarget}) that:
\[
  \| (\wh G_b^{-1/2} \wh S)^-_l\wh G_b^{-1/2}G_b^{1/2} - (G_b^{-1/2} S)^-_l  \|_{\ell^2}
  = O_p \Big(s_{JK}^{-2} (\zeta_{b,K} \vee \zeta_{\psi,J}) \sqrt{(\log K)/(n e_J)}\Big)
\]
which, together with the condition $J \leq K = O(J)$, proves part (a). Part (b) follows similarly.

For part (c), we pre and post multiply terms in the product by $G_b^{-1/2}$ and $G_\psi^{-1/2}$ to obtain:
\begin{eqnarray}
 & &  \| G_b^{-1/2} S\{(\wh G_b^{-1/2} \wh S)^-_l\wh G_b^{-1/2}G_b^{1/2} - (G_b^{-1/2} S)^-_l\}\|_{\ell^2} \notag \\
 & & \quad = \quad \| S^o  [\wh S^{o \prime} (\wh G_b^o)^{-} \wh S^o]^{-} \wh S^{o \prime} (\wh G_b^o)^{-} - S^o[S^{o\prime} S^o]^{-1} S^{o\prime}\|_{\ell^2} \notag \\
 & & \quad \leq \quad\|  S^o  [\wh S^{o \prime} (\wh G_b^o)^{-} \wh S^o]^{-} \wh S^{o \prime} (\wh G_b^o)^{-1/2} ((\wh G_b^o)^{-1/2} - I_K )\|_{\ell^2} \notag \\
 & & \quad \quad \quad + \| (S^o - (\wh G_b^o)^{-1/2}\wh S^o)  [\wh S^{o \prime} (\wh G_b^o)^{-} \wh S^o]^{-} \wh S^{o \prime} (\wh G_b^o)^{-1/2}\|_{\ell^2} \notag \\
 & & \quad \quad \quad + \| (\wh G_b^o)^{-1/2}\wh S^o  [\wh S^{o \prime} (\wh G_b^o)^{-} \wh S^o]^{-} \wh S^{o \prime} (\wh G_b^o)^{-1/2} - S^o[S^{o\prime} S^o]^{-1} S^{o\prime}\|_{\ell^2} \label{e-ssa1}\,.
\end{eqnarray}
Note that $\|((\wh G_b^o)^{-1/2}\wh S^o)^-_l\|_{\ell^2} \leq 2 s_{JK}^{-1}$ on $\mathcal A_{n,1}$ by Lemma \ref{lem-mpnorm}, so
\begin{equation} \label{e-ssa}
 \|((\wh G_b^o)^{-1/2}\wh S^o)^-_l\|_{\ell^2} = O_p (s_{JK}^{-1}) \,.
\end{equation}
It follows by substituting (\ref{e-Gbinvrate}), (\ref{e--lbound}), and (\ref{e-ssa}) into (\ref{e-ssa1}) that
\begin{eqnarray}
 & &  \|G_b^{-1/2} S\{(\wh G_b^{-1/2} \wh S)^-_l\wh G_b^{-1/2}G_b^{1/2} - (G_b^{-1/2} S)^-_l\}\|_{\ell^2} \notag \\
  & & \quad \leq \quad O_p (s_{JK}^{-1} (\zeta_{b,K} \vee \zeta_{\psi,J})\sqrt{(\log K)/n}) \notag \\
  & & \quad \quad \quad  + \|(\wh G_b^o)^{-1/2} \wh S^o  [\wh S^{o \prime} (\wh G_b^o)^{-} \wh S^o]^{-} \wh S^{o \prime} (\wh G_b^o)^{-1/2} - S^o[S^{o\prime} S^o]^{-1} S^{o\prime}\|_{\ell^2}\,. \label{e-partc1}
\end{eqnarray}
The remaining term on the right-hand side of (\ref{e-partc1}) is the $\ell^2$ norm of the difference between the orthogonal projection matrices associated with $S^o$ and $(\wh G_b^o)^{-1/2}\wh S^o$. Applying Lemma \ref{lem-projm}, we obtain:
\[
 \|(\wh G_b^o)^{-1/2}\wh S^o  [\wh S^{o \prime} (\wh G_b^o)^{-} \wh S^o]^{-} \wh S^{o \prime} (\wh G_b^o)^{-1/2} - S^o[S^{o\prime} S^o]^{-1} S^{o\prime}\|_{\ell^2} \leq 2 s_{JK}^{-1} \|(\wh G_b^o)^{-1/2}\wh S^o - S^o\|_{\ell^2}
\]
on $\mathcal A_{n,1}$. Result (c) then follows by (\ref{e--lbound}) and (\ref{e-partc1}).
\end{proof}

\singlespacing
\putbib
}
\end{bibunit}

\end{document}